\documentclass[leqno,a4paper,12pt]{article}

\usepackage{latexsym}
\usepackage{amsmath}
\usepackage{amsthm}
\usepackage{amssymb}
\usepackage{enumerate}
\usepackage{multicol}
\usepackage{graphicx}

\allowdisplaybreaks[4]

\theoremstyle{plain}
\newtheorem{theorem}{Theorem}[section]
\newtheorem{proposition}[theorem]{Proposition}
\newtheorem{lemma}[theorem]{Lemma}
\newtheorem{corollary}[theorem]{Corollary}

\theoremstyle{definition}

\newtheorem{remark}[theorem]{Remark}
\numberwithin{equation}{section}

\def\ba{{\mathbf{a}}}
\def\sA{{\mathsf{A}}}

\def\bb{{\mathbf{b}}}

\def\cB{{\mathcal{B}}}

\def\C{{\mathbb{C}}}
\def\cC{{\mathcal{C}}}

\def\be{{\mathbf{e}}}

\def\sF{{\mathsf{F}}}

\def\cG{{\mathcal{G}}}

\def\cH{{\mathcal{H}}}
\def\sH{{\mathsf{H}}}

\def\bk{{\mathbf{k}}}

\def\cM{{\mathcal{M}}}

\def\N{{\mathbb{N}}}
\def\sN{{\mathsf{N}}}

\def\cO{{\mathcal{O}}}

\def\bp{{\mathbf{p}}}

\def\bq{{\mathbf{q}}}
\def\cQ{{\mathcal{Q}}}

\def\R{{\mathbb{R}}}
\def\cR{{\mathcal{R}}}

\def\S{{\mathbb{S}}}
\def\cS{{\mathcal{S}}}
\def\bs{{\mathbf{s}}}

\def\sS{{\mathsf{S}}}

\def\T{{\mathbb{T}}}
\def\bT{{\mathbf{T}}}
\def\bt{{\mathbf{t}}}

\def\cU{{\mathcal{U}}}
\def\bu{{\mathbf{u}}}

\def\cV{{\mathcal{V}}}
\def\sV{{\mathsf{V}}}
\def\bv{{\mathbf{v}}}

\def\bw{{\mathbf{w}}}

\def\bW{{\mathbf{W}}}
\def\cW{{\mathcal{W}}}
\def\sW{{\mathsf{W}}}

\def\bx{{\mathbf{x}}}
\def\bX{{\mathbf{X}}}

\def\hbx{{\hat{\mathbf{x}}}}

\def\by{{\mathbf{y}}}
\def\bY{{\mathbf{Y}}}

\def\hby{{\hat{\mathbf{y}}}}

\def\Z{{\mathbb{Z}}}
\def\bz{{\mathbf{z}}}
\def\bZ{{\mathbf{Z}}}

\def\bphi{{\mbox{\boldmath$\phi$}}}

\def\opsi{{\overline{\psi}}}

\def\spin{{\{\uparrow,\downarrow\}}}
\def\ua{{\uparrow}}
\def\da{{\downarrow}}

\def\la{{\lambda}}
\def\bla{{\mbox{\boldmath$\lambda$}}}

\def\O{{\Omega}}
\def\o{{\omega}}

\def\eps{{\varepsilon}}

\def\g{{\gamma}}

\def\G{{\Gamma}}
\def\s{{\sigma}}

\def\D{{\Delta}}

\def\<{{\langle}}
\def\>{{\rangle}}

\def\Tr{\mathop{\mathrm{Tr}}\nolimits}

\def\Map{\mathop{\mathrm{Map}}}
\def\dis{\mathop{\mathrm{dis}}\nolimits}

\def\sgn{\mathop{\mathrm{sgn}}\nolimits}

\def\supp{\mathop{\mathrm{supp}}\nolimits}

\def\Im{\mathop{\mathrm{Im}}}
\def\Re{\mathop{\mathrm{Re}}}

\def\b0{{\mathbf{0}}}

\def\frah{{\left(\frac{1}{h}\right)}}

\def\0betah{{[0,\beta)_h}}

\begin{document}

\title{Superconducting phase in the BCS model with imaginary magnetic field}

\author{Yohei Kashima \medskip\\
Graduate School of Mathematical Sciences, \\
the University of Tokyo, 
Komaba, Tokyo 153-8914, Japan\\
kashima@ms.u-tokyo.ac.jp}

\date{}

\maketitle

\begin{abstract} We prove that in the reduced quartic BCS model with an
 imaginary external magnetic field a spontaneous $U(1)$-symmetry
 breaking (SSB) and an off-diagonal long range order (ODLRO) occur. The
 system is defined on a hyper-cubic lattice with periodic boundary
 conditions at positive temperature. In the free part of the Hamiltonian
 we assume the nearest-neighbor hopping. The chemical potential is fixed
 so that the free Fermi surface does not degenerate. The term
 representing the interaction between electrons' spin and the imaginary
 external magnetic field is the $z$-component of the spin operator
 multiplied by a pure imaginary parameter. The SSB and the ODLRO are
 shown in the infinite-volume limit of the thermal average over the
 full Fermionic Fock space. The magnitude of the negative coupling
 constant must be larger than a certain value so that the gap equation
 is solvable. The gap equation is different from that of the
 conventional mean field BCS model because of the presence of the
 imaginary magnetic field. By adjusting the imaginary magnetic field
 this model shows the SSB and the ODLRO in high temperature, weak
 coupling regimes where the conventional reduced BCS model does not show
 these phenomena. The proof is based on Grassmann Gaussian integral formulations
 and a double-scale integration scheme to analytically control the
 formulations.
\footnote{
\textit{keywords and phrases.} The BCS model, spontaneous symmetry
 breaking, off-diagonal long range order, Grassmann integral formulation,
 tree expansion}
\end{abstract}

\tableofcontents

\section{Introduction}\label{sec_introduction}

\subsection{Introductory remarks}\label{subsec_remarks}
In 1957 (\cite{BCS}) Bardeen, Cooper and Schrieffer proposed a microscopic theory of
superconductivity which is widely known as the BCS theory today. 
As the importance of the BCS theory was recognized, many began to
mathematically verify effective approximations made in the
theory. Proving a superconductivity within the fundamental principles
proposed by Bardeen, Cooper and Schrieffer has been a stimulating topic
in mathematical physics until today. See e.g. the review article
\cite{HS} for a recent trend of the subject. Many papers concerning
mathematical physics of the BCS theory have been published since the
early stage. The historical review \cite{BP0} reported in 100th year 
after the discovery of superconductivity is enlightening. 
However, if we focus our attention on a basic simple
question whether the BCS model shows superconductivity
characterized by spontaneous $U(1)$-symmetry breaking (SSB) and 
off-diagonal long range order (ODLRO), we notice that there are
unexpectedly few mathematical results answering this question. Here the
BCS model is meant to be the Hamiltonian consisting of a kinetic part,
quadratic in Fermionic operators, describing free movements of electrons
and an interacting part, quartic in Fermionic operators, describing a
long range interaction between Cooper pairs. We also require SSB and
ODLRO to be shown in the infinite-volume limit of the thermal average
over the full Fermionic Fock space. 

In the strong coupling limit of the reduced BCS model, where the free
part is the number operator multiplied by the chemical potential only and the
interacting part is a product of the Cooper pair operators, a SSB and
an ODLRO in the above sense were proved by a $C^*$-algebraic approach by
Bru and de Siqueire Pedra in \cite{BP1}. The model considered in
\cite{BP1} is allowed to contain the Hubbard type on-site interaction as
well. The same authors also extended their $C^*$-algebraic framework to be
applicable to the BCS model having a non-constant kinetic term and gave
a mathematical sense of ODLRO in a limit of the finite systems under
periodic boundary conditions in \cite{BP2}. Before \cite{BP1} many
researchers had continued their efforts to analyze the BCS model in the 
quasi-spin formulation at positive temperature. The achievements of 
these authors are listed in the references of \cite{BP0}. Here we refer 
to the original article \cite{T} where the equivalence between
correlation functions in the strong coupling limit of the reduced BCS 
model and those in the mean field BCS model was proved in the quasi-spin formulation. See also 
\cite{DP} for an analysis of the BCS model with non-constant free
dispersion relations in the quasi-spin formulation at positive temperature.
It should be remarked that the thermal average in the quasi-spin
formulation amounts to the average
over a proper subspace of the full Fermionic Fock space. There were
also attempts to demonstrate SSB in the BCS model in Grassmann
integral formulations. When the grand canonical partition function of
the reduced BCS model is formulated into a Grassmann Gaussian integral,
a quartic Grassmann polynomial resembling the BCS interaction appears to
be integrated with a time-variable in its action. By artificially
dividing the single time-integral into a double time-integral and thus
deriving the so-called doubly reduced BCS model, Lehmann showed that a
SSB occurs in a form of the Schwinger function in \cite{Leh}. In
\cite{M} Mastropietro extended the approach based on the Grassmann
integral formulation and showed that a SSB occurs in the Schwinger
function where the interaction is of the doubly reduced BCS type
tempered by time-integration with a Kac potential.
 The insertion of the Kac potential is in effect an
interpolation between the doubly reduced BCS interaction and the reduced
BCS interaction in the Grassmannian level. The gap equation in these
studies is equal to that of the mean field BCS model.

Despite a long history of the research we can hardly find a thoroughly
explicit demonstration of SSB and ODLRO in the full BCS model.
In this situation this paper is devoted to demonstrating them in the BCS model in
a non-standard parameter region of the complex plane. We will prove that 
a SSB and an ODLRO occur in the reduced BCS model in a way fulfilling 
the above-mentioned requirements, provided a term
representing the interaction between electrons' spin and an imaginary
external magnetic field is added to the Hamiltonian. More precisely, the
interacting term with the imaginary magnetic field is given by the
$z$-component of the spin operator multiplied by a pure imaginary
parameter. The model is initially defined on a finite hyper-cubic lattice with 
periodic boundary conditions. In the free part of the Hamiltonian we
assume the nearest-neighbor hopping. We restrict the range of the
chemical potential so that the free Fermi surface does not
degenerate. The magnitude of the negative coupling constant must be
larger than a certain value so that the gap equation has a positive
solution. At the same time it must be smaller than a certain value so
that our perturbative treatments make sense. Thus, there are two kinds of
constraint on the magnitude of the negative coupling constant. It is
due to a fine tuning of the imaginary magnetic field that we can
actually choose a coupling constant satisfying both the constraints.
The gap equation is different from that of the  conventional mean field
BCS model because of the insertion of the imaginary magnetic
field. Consequently it turns out that the SSB and the ODLRO can occur
in high temperature, weak coupling regimes where these phenomena do not
show up in the conventional reduced BCS model. The presence of the
imaginary magnetic field breaks the hermiticity of the whole
Hamiltonian. However, it will be proved as a part of our main results that
the grand canonical partition function takes a real positive value in 
parameter regions where our analytical methods are valid.

From a technical view point this paper is seen to be a continuation of
the Grassmann integration approach by Lehmann (\cite{Leh}) and
Mastropietro (\cite{M}). In the Grassmannian level we divide the reduced
BCS interaction into the doubly reduced BCS interaction and the
correction term, transform the doubly reduced interaction into an
integral of  quadratic Grassmann polynomials by means of the
Hubbard-Stratonovich transformation and estimate the Grassmann Gaussian
integral having the quartic correction term in its exponent by the
tree expansion. The use of the Hubbard-Stratonovich transformation was
motivated by \cite{Leh}, \cite{M} and an important division technique of
truncated expectations which works to produce an extra inverse volume
factor was influenced by \cite{M}. However, there are notable
differences between the conclusions of this paper and those of the
preceding articles. This paper starts with Fermionic operators defined
on the Fock space and concludes the SSB and the ODLRO in the
infinite-volume limit of the full trace thermal expectations, while the
conclusions of \cite{Leh}, \cite{M} concern limit values of the
Schwinger functions on Grassmann algebras. In fact it is not yet known
how to realize the doubly reduced BCS interaction with or without the
Kac potential in a concrete form of Fermionic operators. Moreover, since our
gap equation is different from that of the mean field BCS model, the SSB
and the ODLRO take place in a parameter region where the
conventional BCS gap equation is not solvable and thus where SSB and ODLRO do not
emerge in the sense of \cite{Leh}, \cite{M}. As another
new technical aspect we show the convergence of the infinite-volume
limit of the finite-volume thermal expectations by sending the box size
to infinity without taking a subsequence. As the result the SSB and the
ODLRO can be claimed in the limit, not only in some accumulation points of
the finite-volume formulations. To prove this, it is essential to
establish the full convergent property of the Grassmann Gaussian
integral of the correction term which does not simply follow from a
uniform boundedness of the Grassmann integration and requires a detailed analysis
of the tree expansion of truncated expectations. The analysis is
performed in Subsection \ref{subsec_L_limit}.

The technical core of our construction is the estimation of the
correction to the doubly reduced BCS model. The estimation is completed
by a double-scale integration process over the Matsubara frequency. The
first integration involves a covariance with all but one time-momentum,
while the covariance in the second integration contains only one
time-momentum. We should declare in this remark that we use
Pedra-Salmhofer's type determinant bound (PS bound, \cite{PS}) to bound
the determinant of the first covariance with large Matsubara frequencies.
In general the application of the PS bound is a very efficient
alternative to a multi-scale 
integration procedure over large Matsubara frequencies. By applying it
one can prepare the input to the succeeding infrared integration process
by a simple single-scale integration. In this paper, which aims at
providing the first convincing proof of SSB and ODLRO in the BCS model
with an imaginary magnetic field, we decide to make the construction
simple and thus choose to apply the PS bound rather than go through a
self-contained but lengthy multi-scale Matsubara ultra-violet
integration. Also, in the interest of simplicity we do not perform a
multi-scale infrared integration to improve the dependency of possible
magnitude of the coupling constant on the temperature. As a consequence, 
this paper does not ensure that one can take
the temperature close to zero while keeping the magnitude of the
coupling constant positive. This may be seen as a shortcoming of this paper's
results. Since it needs to classify Grassmann polynomials
at each scale, the proof based on a multi-scale integration would be 
substantially longer. This thought together with
the hope that a simpler construction must be more convincing led us to
conclude our construction only by the double-scale integration. A
qualitative improvement of the temperature-dependency of the coupling
constant by means of a multi-scale infrared integration should 
be performed elsewhere.

The main novelty of this paper is to reveal the mathematical fact that
the insertion of the imaginary magnetic field in the BCS model makes it
possible to prove SSB and ODLRO in wide parameter regions. 
We should note that the extension of the external magnetic field to the
complex plane in many-body systems has been an important subject of
mathematical physics since the pioneering study by Lee and Yang
(\cite{YL}, \cite{LY}). 
At the same time the presence of the imaginary 
magnetic field admittedly makes it difficult to find a conventional physical
meaning of the model. It is interesting
and encouraging to know that the Lee-Yang zeros of the partition
function of the Ising model are being experimentally realized through a
time domain measurement of the spin system (\cite{PZWCDL}).
What appears particularly interesting to us in \cite{PZWCDL} is that the
partition function of the ferromagnetic Ising model with long range
interaction under an imaginary magnetic field was identified with the
coherence of a central spin coupled to the spin bath, which was
experimentally measured. See also \cite{WL}, \cite{WCPL} for the idea
behind the experiment \cite{PZWCDL}. A message from \cite{PZWCDL} is
that if the partition function with an imaginary magnetic field is
measurable, then extending the magnetic field into the complex plane is
not only a way to solve a mathematical problem but could be an analysis
of a model of the real world.  Amid the latest progress of physical
experiments our hope from a mathematical side is that the
superconducting phase in the BCS model with 
an imaginary magnetic field should be experimentally realized someday.

The contents of this paper are outlined as follows. In the rest of this
section we define the model and officially state the main results of
this paper. In Section \ref{sec_formulation} we formulate the grand
canonical partition function of the model Hamiltonian by means of the
Hubbard-Stratonovich transformation and the Grassmann Gaussian
integration. In Section \ref{sec_Grassmann_integration} we construct a
double-scale integration process in a generalized setting with the aim
of applying it to estimate the Grassmann Gaussian integral of the correction
term in the following section. In Section \ref{sec_proof_theorem} we
apply the general results obtained in the previous section to the actual
model problem and derive necessary bound properties. Then we 
show necessary convergent properties of the Grassmann Gaussian integral
of the correction term in the time-continuum, infinite-volume
limit. After these preparations we complete the proof of the main
theorem. In Appendix \ref{app_P_S_bound} we provide a short proof to 
Pedra-Salmhofer's type determinant bound used in our construction for
completeness. We also list notations which are used over multiple
sections for readers' convenience at the end of the paper.

\subsection{The model and the main results}\label{subsec_main_results}
Throughout the paper the spatial dimension is denoted by $d$. With
$L(\in \N=\{1,2,\cdots\})$ we define the spatial lattice $\G$ by
$\G:=\{0,1,2,\cdots,L-1\}^d$. For $(\bx,\s)\in\G\times \spin$ let
$\psi_{\bx\s}^*$, $\psi_{\bx\s}$ denote the Fermionic creation /
annihilation operator respectively. We impose periodic boundary
conditions on the finite-volume system. To describe the periodicity, it is
convenient to use the map $r_L:\Z^d\to \G$ which satisfies
$r_L(\bx)=\bx$ in $(\Z/L\Z)^d$ for any $\bx\in \Z^d$. For
$(\bx,\s)\in\Z^d\times \spin$ we identify $\psi_{\bx\s}^{(*)}$ with
$\psi_{r_L(\bx)\s}^{(*)}$. The free part $\sH_0$
of our Hamiltonian is defined by
\begin{align*}
\sH_0:=\sum_{\bx\in\G}\sum_{\s\in \spin}
\left((-1)^{hop}\sum_{j=1}^d(\psi_{\bx\s}^*\psi_{\bx+\be_j\s}+
\psi_{\bx\s}^*\psi_{\bx-\be_j\s})-\mu
\psi_{\bx\s}^*\psi_{\bx\s}\right),
\end{align*}
where $hop\in\{0,1\}$,
$\be_j$ $(j=1,2,\cdots,d)$ are the standard basis of $\R^d$ and
the real parameter $\mu$ is the chemical potential.
For simplicity we adopt the unit where the hopping amplitude is scaled
to be 1. We use the parameter $hop$ to treat the positive hopping and
the negative hopping at once. 
Moreover, we include the number operator multiplied by the
chemical potential in the free Hamiltonian. Also we restrict the hopping
of electrons to be only between nearest-neighbor sites. Throughout the
paper except Remark \ref{rem_dispersion_extension} we assume that
$$
\mu\in (-2d,2d)
$$ 
so that the free Fermi surface $\{\bk\in [0,2\pi)^d\ |\
(-1)^{hop}2\sum_{j=1}^d\cos k_j-\mu=0\}$ does not degenerate.
Only in Remark \ref{rem_dispersion_extension} we consider the degenerate case.  In \cite{BCS} the complex phonon-electron interaction was reduced into a
sum of product of 2 Cooper pair operators. We consider the reduced BCS
interaction with constant matrix element defined as follows.
$$
\sV:=\frac{U}{L^d}\sum_{\bx,\by\in\G}\psi_{\bx\ua}^*\psi_{\bx\da}^*\psi_{\by\da}\psi_{\by\ua},
$$
where $U$ is a real negative parameter controlling the strength of
non-local attraction between Cooper pairs. The BCS model $\sH$ is
defined by 
\begin{align*}
\sH&:=\sH_0+\sV\\
&=\sum_{\bx\in\G}\sum_{\s\in \spin}
\left((-1)^{hop}\sum_{j=1}^d(\psi_{\bx\s}^*\psi_{\bx+\be_j\s}+
\psi_{\bx\s}^*\psi_{\bx-\be_j\s})-\mu
\psi_{\bx\s}^*\psi_{\bx\s}\right)\\
&\quad +\frac{U}{L^d}\sum_{\bx,\by\in\G}\psi_{\bx\ua}^*\psi_{\bx\da}^*\psi_{\by\da}\psi_{\by\ua},\end{align*}
which is a self-adjoint operator on the
Fermionic Fock space $F_f(L^2(\G\times \spin))$.

In this paper we focus on two characteristics of superconductivity and
try to prove their existence in the infinite-volume
limit of the system.  One characteristic is spontaneous symmetry breaking (SSB). The other
characteristic of our interest is off-diagonal long range order (ODLRO). A
mathematical description of SSB is the following. 
We add a $U(1)$-symmetry breaking external field 
$$
\sF=\g \sum_{\bx\in
\G}(\psi_{\bx\ua}^*\psi_{\bx\da}^*+\psi_{\bx\da}\psi_{\bx\ua}),\quad \g\in\R
$$
to the system and observe the thermal expectation value of the pairing
operator in the limit $\g\to 0$ after taking the limit $L\to \infty$,
$$
\lim_{\g\to 0}\lim_{L\to \infty\atop L\in\N}\frac{\Tr(e^{-\beta
(\sH+\sF)}\psi_{\hbx\ua}^*\psi_{\hbx\da}^*)}{\Tr e^{-\beta (\sH+\sF)}}.
$$
Here the trace operation is taken over the Fermionic Fock space and $\beta(\in
\R_{>0})$ is the inverse temperature. If the expectation value converges
to a non-zero value, it is said that a SSB occurs in the system. This is
because the $U(1)$-gauge symmetry which the original system possesses remains
broken even after removing the symmetry-breaking external field. 
A long range correlation between Cooper pairs is explained by 
the behavior of the 4-point correlation function in the infinite-volume
limit,
$$
\lim_{L\to \infty\atop L\in\N}\frac{\Tr (e^{-\beta
\sH}\psi_{\hbx\ua}^*\psi_{\hbx\da}^*\psi_{\hby\da}\psi_{\hby\ua})}{\Tr e^{-\beta \sH}}.
$$
If the correlation function in the infinite-volume limit converges to a
non-zero value as the distance between $\hbx$ and $\hby$ goes to
infinity, the system is said to exhibit an ODLRO (see \cite{Y}). These
phenomena have been desired to be proven in the BCS model. Despite many years of research
after \cite{BCS}, the full rigorous demonstration of SSB and ODLRO in the BCS
model seems unexpectedly scarce. Amid this situation this paper is
devoted to revealing a new fact of the BCS model that a SSB and an ODLRO are
 present under an external imaginary magnetic field. More precise
 explanation of our plan is that we add the operator $i\theta \sS_z$
 $(\theta \in\R)$ to the BCS model $\sH$ and prove the existence of
 a SSB and an ODLRO. Here $\sS_z$ is the $z$-component of the spin operator.
$$
\sS_z:=\frac{1}{2}\sum_{\bx\in \G}(\psi_{\bx\ua}^*\psi_{\bx\ua}-\psi_{\bx\da}^*\psi_{\bx\da}).
$$
The term $i\theta \sS_z$ is formally interpreted as an interaction between
the imaginary magnetic field $(0,0,i\theta)$ and the electrons' spin.

Since adding $i\theta \sS_z$ to the Hamiltonian breaks hermiticity, we
do not know whether the partition function $\Tr e^{-\beta (\sH+i\theta
\sS_z+\sF)}$ remains non-zero. Thus, even the well-definedness of the 
thermal expectation is unclear. We know at least the following. 
Set
$$
\sA_1:=\psi_{\hbx\ua}^*\psi_{\hbx\da}^*,\quad 
\sA_2:=\psi_{\hbx\ua}^*\psi_{\hbx\da}^*\psi_{\hby\da}\psi_{\hby\ua}.
$$  

\begin{lemma}\label{lem_real_value_confirmation}
\begin{align*}
\Tr e^{-\beta(\sH+i\theta \sS_z+\sF)},\ \Tr(e^{-\beta(\sH+i\theta
 \sS_z+\sF)}\sA_1),\  \Tr(e^{-\beta(\sH+i\theta \sS_z+\sF)}\sA_2)\in\R
\end{align*}
and 
$$
\Tr(e^{-\beta(\sH+i\theta \sS_z+\sF)}\sA_1)= \Tr(e^{-\beta(\sH+i\theta \sS_z+\sF)}\sA_1^*).
$$
\end{lemma}
\begin{proof}
Let us define the transforms $\cU_1,\cU_2$ on $F_f(L^2(\G\times \spin))$
 by 
\begin{align*}
&\cU_1\O:=\O,\\
&\cU_1\psi_{\bx_1\s_1}^*\psi_{\bx_2\s_2}^*\cdots\psi_{\bx_n\s_n}^*\O:=
\psi_{\bx_1-\s_1}^*\psi_{\bx_2-\s_2}^*\cdots\psi_{\bx_n-\s_n}^*\O,\\
&\cU_2\O:=\O,\\
&\cU_2\psi_{\bx_1\s_1}^*\psi_{\bx_2\s_2}^*\cdots\psi_{\bx_n\s_n}^*\O:=
i^n\psi_{\bx_1\s_1}^*\psi_{\bx_2\s_2}^*\cdots\psi_{\bx_n\s_n}^*\O,\\
&(\forall n\in \N,\ (\bx_j,\s_j)\in\G\times \spin\ (j=1,2,\cdots,n))
\end{align*}
and by linearity, where $\O$ is the vacuum of the Fock space. The
 transforms $\cU_1,\cU_2$ are unitary. Moreover,
\begin{align*}
\Tr e^{-\beta(\sH+i\theta \sS_z+\sF)}&=\Tr e^{-\beta\cU_1(\sH+i\theta
 \sS_z+\sF)\cU_1^*}
=\Tr e^{-\beta(\sH-i\theta \sS_z-\sF)}=\Tr e^{-\beta\cU_2(\sH-i\theta
 \sS_z-\sF)\cU_2^*}\\
&=\Tr e^{-\beta(\sH-i\theta \sS_z+\sF)}
=\overline{\Tr e^{-\beta(\sH+i\theta \sS_z+\sF)}}.
\end{align*}
Thus, $\Tr e^{-\beta(\sH+i\theta \sS_z+\sF)}\in\R$. The other claims can
 be checked in the same way. 
\end{proof}
It will be proved as a part of the main theorem that
$\Tr e^{-\beta(\sH+i\theta \sS_z+\sF)}>0$ for sufficiently
large $L$.
The next lemma tells us that it suffices to analyze the system for
$\theta \in [0,2\pi/\beta]$.

\begin{lemma}\label{lem_periodicity_with_phase}
Assume that $\theta\in\R$, $\theta'\in (-2\pi/\beta,2\pi/\beta]$ and
 $\theta=\theta'$ in $\R/\frac{4\pi}{\beta}\Z$. Then,
\begin{align*}
&\Tr e^{-\beta(\sH+i\theta \sS_z+\sF)}=\Tr e^{-\beta(\sH+i|\theta'|
 \sS_z+\sF)},\\
&\Tr (e^{-\beta(\sH+i\theta \sS_z+\sF)}\sA_j)=\Tr (e^{-\beta(\sH+i|\theta'|
 \sS_z+\sF)}\sA_j),\ (j=1,2).
\end{align*}
\end{lemma}
\begin{proof}
Note that the operator $\sS_z$ commutes with $\sH$, $\sF$,
 $\psi_{\hbx\ua}^*\psi_{\hbx\da}^*$, $\psi_{\hbx\da}\psi_{\hbx\ua}$ for
 any $\hbx\in\Z^d$. The trace operation over the Fock space can be
 decomposed into the sum of the trace over each eigenspace of
 $\sS_z$. Since each eigenvalue of $\sS_z$ belongs to $\frac{1}{2}\Z$,
$$
\Tr e^{-\beta(\sH+i\theta \sS_z+\sF)}=\Tr(e^{-\beta
 (\sH+\sF)}e^{-i\beta\theta \sS_z})
=\Tr(e^{-\beta (\sH+\sF)}e^{-i\beta\theta' \sS_z})
=\Tr e^{-\beta(\sH+i\theta' \sS_z+\sF)}.
$$ 
Moreover, by Lemma \ref{lem_real_value_confirmation},
$$
\Tr e^{-\beta(\sH+i\theta' \sS_z+\sF)}=\Tr e^{-\beta(\sH+i|\theta'| \sS_z+\sF)}.
$$ 
Thus, the first equality is obtained. The other equalities can be
 derived in the same way.
\end{proof}
From here we always assume that 
$$\theta \in \left[0,\frac{2\pi}{\beta}\right).$$ 
In this
paper we do not treat the case $\theta =2\pi/\beta$. As in this case
the free partition function vanishes (see Lemma
\ref{lem_free_partition_function}), we are unable to define the free
covariance which plays a central role in our analysis.

In order to officially state the main results of this paper, we should
make clear notations used in the statements.
Let $\|\cdot\|_{\R^d}$ denote the euclidean norm of $\R^d$. For a
function $f:\Z^d\times \Z^d\to \C$ and $c\in \C$ we write
$$
\lim_{\|\bx-\by\|_{\R^d}\to\infty}f(\bx,\by)=c
$$
if for any $\eps\in\R_{>0}$ there exists $r\in\R_{>0}$ such that
$|f(\bx,\by)-c|< \eps$ for any $\bx,\by\in\Z^d$ satisfying
$\|\bx-\by\|_{\R^d}\ge r$. For a proposition $P$ let $1_P$ be $1$ if $P$
is true, $0$ otherwise. We define the function $e:\R^d\to \R$ by
$$
e(\bk):=(-1)^{hop}2\sum_{j=1}^d\cos k_j-\mu,\quad \bk=(k_1,k_2,\cdots,k_d)\in\R^d,
$$
which is in fact the free dispersion relation. To estimate possible
magnitude of the coupling constant, we use the function
$g_d:(0,\infty)\to\R$ defined as follows.
\begin{align}
g_d(x):=1_{d\ge
 2}(\log(x^{-1}+1))^{\frac{d}{d+1}}x^{-\frac{1}{d+1}}
+1_{d=1}(4-\mu^2)^{-\frac{1}{2}}\log(x^{-1}+1).\label{eq_magnitude_function}
\end{align}
The main result of this paper is the following.
\newpage

\begin{theorem}\label{thm_main_theorem}
Assume that $\beta\in \R_{>0}$, $U\in\R_{<0}$. 
There exist constants $c_1(d)\in\R_{>0}$, $c_2(d)\in (0,1]$ depending
 only on $d$ such that the following statements hold true.
\begin{enumerate}[(i)]
\item\label{item_partition_positivity}
Assume that $\theta\in
     [0,2\pi/\beta)$ and 
\begin{align*}
|U|<c_2(d)\left(1+\beta^{d+3}+(1+\beta^{-1})g_d\left(\left|\frac{\theta}{2}-\frac{\pi}{\beta}\right|\right)\right)^{-2}.
\end{align*}
Then, there exists $L_0\in\N$ such that 
\begin{align*}
\Tr e^{-\beta(\sH+i\theta \sS_z+\sF)}\in\R_{>0},\quad (\forall
 L\in\N\text{ with }L\ge L_0,\ \g\in [0,1]).
\end{align*}
\item\label{item_superconducting_phase}
Assume that 
$\theta\in
     [0,2\pi/\beta)$ and 
\begin{align}
&c_1(d)(2d-|\mu|)^{1-d}\beta\left|\frac{\theta}{2}-\frac{\pi}{\beta}\right|\label{eq_sufficient_condition_superconductivity}\\
&\quad \cdot\left(1_{|\frac{\theta}{2}-\frac{\pi}{\beta}|\le \frac{1}{2}(2d-|\mu|)}+
1_{|\frac{\theta}{2}-\frac{\pi}{\beta}|> \frac{1}{2}(2d-|\mu|)}(2d-|\mu|)^{-1}\left|\frac{\theta}{2}-\frac{\pi}{\beta}\right|\right)\notag\\
&<|U|<c_2(d)\left(1+\beta^{d+3}+(1+\beta^{-1})g_d\left(\left|\frac{\theta}{2}-\frac{\pi}{\beta}\right|\right)\right)^{-2}.\notag
\end{align}
Then, there uniquely exists $\D\in \R$ such that $\D>0$ and 
\begin{align}
-\frac{2}{|U|}+\frac{1}{(2\pi)^d}\int_{[0,2\pi]^d}d\bk\frac{\sinh(\beta\sqrt{e(\bk)^2+\D^2})}{(\cos(\beta\theta/2)+\cosh(\beta\sqrt{e(\bk)^2+\D^2}))\sqrt{e(\bk)^2+\D^2}}=0.
\label{eq_gap_equation}
\end{align}
Moreover, 
\begin{align}
&\lim_{L\to \infty\atop L\in \N}\left(-\frac{1}{\beta L^d}\log(\Tr
 e^{-\beta(\sH+i\theta \sS_z)})\right)\label{eq_free_energy_density}\\
&=\frac{\D^2}{|U|}-\frac{1}{\beta(2\pi)^d}\int_{[0,2\pi]^d}d\bk\log\Bigg(2\cos\left(\frac{\beta\theta}{2}\right)e^{-\beta
 e(\bk)}\notag\\
&\qquad\qquad\qquad\qquad\qquad\qquad\qquad +e^{\beta(\sqrt{e(\bk)^2+\D^2}-e(\bk))}
+e^{-\beta(\sqrt{e(\bk)^2+\D^2}+e(\bk))}\Bigg).\notag\\
&\lim_{\g\searrow 0\atop \g\in (0,1]}\lim_{L\to \infty\atop L\in\N}\frac{\Tr(e^{-\beta
(\sH+i\theta \sS_z+\sF)}\psi_{\hbx\ua}^*\psi_{\hbx\da}^*)}{\Tr e^{-\beta (\sH+i\theta \sS_z+\sF)}}
\label{eq_SSB}\\
&=\lim_{\g\searrow 0\atop \g\in (0,1]}\lim_{L\to \infty\atop L\in\N}\frac{\Tr(e^{-\beta
(\sH+i\theta \sS_z+\sF)}\psi_{\hbx\da}\psi_{\hbx\ua})}{\Tr e^{-\beta (\sH+i\theta \sS_z+\sF)}}=-\frac{\D}{|U|}.\notag\\
&\lim_{\|\hbx-\hby\|_{\R^d}\to\infty}\lim_{L\to \infty\atop L\in\N}\frac{\Tr (e^{-\beta
(\sH+i\theta \sS_z)}\psi_{\hbx\ua}^*\psi_{\hbx\da}^*\psi_{\hby\da}\psi_{\hby\ua})}{\Tr
 e^{-\beta (\sH+i\theta \sS_z)}}=\frac{\D^2}{U^2}.\label{eq_ODLRO}
\end{align}
\item\label{item_no_superconductivity}
Assume that $\theta\in
     [0,\pi/\beta)$ and 
\begin{align*}
|U|<c_2(d)\left(1+\beta^{d+3}+(1+\beta^{-1})g_d\left(\left|\frac{\theta}{2}-\frac{\pi}{\beta}\right|\right)\right)^{-2}.
\end{align*}
Then, for any $\D\in \R$ the equation \eqref{eq_gap_equation} does not
     hold. Moreover, the statements \eqref{eq_free_energy_density},
     \eqref{eq_SSB}, \eqref{eq_ODLRO} hold with $\D=0$.
\item\label{item_confirmation_condition_superconductivity}
For any $\beta\in\R_{>0}$ there exists $\delta \in\R_{>0}$ such that 
\begin{align}
&c_1(d)(2d-|\mu|)^{1-d}\beta\left|\frac{\theta}{2}-\frac{\pi}{\beta}\right|
\label{eq_sufficient_condition_superconductivity_extract}\\
&\quad \cdot\left(1_{|\frac{\theta}{2}-\frac{\pi}{\beta}|\le \frac{1}{2}(2d-|\mu|)}+
1_{|\frac{\theta}{2}-\frac{\pi}{\beta}|> \frac{1}{2}(2d-|\mu|)}(2d-|\mu|)^{-1}\left|\frac{\theta}{2}-\frac{\pi}{\beta}\right|\right)\notag\\
&<c_2(d)\left(1+\beta^{d+3}+(1+\beta^{-1})g_d\left(\left|\frac{\theta}{2}-\frac{\pi}{\beta}\right|\right)\right)^{-2}\notag
\end{align}
for any $\theta  \in [2\pi/\beta -\delta,2\pi/\beta)$. Thus,
     for any $\theta  \in [2\pi/\beta
     -\delta,2\pi/\beta)$ there exists $U\in \R_{<0}$ such that
     \eqref{eq_sufficient_condition_superconductivity} holds.
\item\label{item_imply_no_superconductivity}
Assume that $\theta\in
     [\pi/\beta,2\pi/\beta)$ and
     \eqref{eq_sufficient_condition_superconductivity} holds. Then, 
\begin{align*}
|U| <c_2(d)\left(1+\beta^{d+3}+(1+\beta^{-1})g_d\left(\left|\frac{\eta}{2}-\frac{\pi}{\beta}\right|\right)\right)^{-2}
\end{align*}
for any $\eta  \in [0,\pi/\beta)$ and thus the conclusions of
     \eqref{item_no_superconductivity} hold with these $\beta$, $U$ and
     $\eta$ in place of $\theta$.
\end{enumerate}
\end{theorem}
 
The claims \eqref{item_confirmation_condition_superconductivity},
\eqref{item_imply_no_superconductivity} can be proved here. Set
$\Theta:=|\theta/2-\pi/\beta|$. 
If $\Theta \le \frac{1}{2}(2d-|\mu|)$, the inequality
\eqref{eq_sufficient_condition_superconductivity_extract} is
equivalently written as follows.
\begin{align*}
&(1+\beta^{d+3})\Theta^{\frac{1}{2}}
+1_{d=1}(1+\beta^{-1})
(4-\mu^2)^{-\frac{1}{2}}(\log(\Theta^{-1}+1))\Theta^{\frac{1}{2}}\\
&+1_{d\ge
 2}(1+\beta^{-1})(\log(\Theta^{-1}+1))^{\frac{d}{d+1}}\Theta^{\frac{1}{2}-\frac{1}{d+1}}\\&< (c_1(d)^{-1}(2d-|\mu|)^{d-1}\beta^{-1}c_2(d))^{\frac{1}{2}}.
\end{align*}
Since the left-hand side converges to 0 as $\Theta\searrow 0$, the
claim \eqref{item_confirmation_condition_superconductivity} holds true.
The claim \eqref{item_imply_no_superconductivity} follows from the fact that
$g_d:(0,\infty)\to \R$ is decreasing.

\begin{remark}
The implication of the claim
 \eqref{item_confirmation_condition_superconductivity} is that at any
 temperature we can choose $\theta\in [\pi/\beta,2\pi/\beta)$ and the
 negative coupling constant $U$ so that SSB and ODLRO occur in the system. Moreover, since\begin{align*}
\lim_{\theta\nearrow
 \frac{2\pi}{\beta}}c_2(d)
\left(1+\beta^{d+3}+(1+\beta^{-1})g_d\left(\left|\frac{\theta}{2}-\frac{\pi}{\beta}\right|\right)\right)^{-2}=0,
\end{align*}
for any small $U_0\in\R_{>0}$ and at any positive temperature we can choose
 $\theta\in [\pi/\beta,2\pi/\beta)$ and the negative coupling constant
 $U$ so that $|U|\le U_0$ and SSB and ODLRO occur in the
 system. In other words, at arbitrarily high temperature, for
 arbitrarily weak coupling SSB and ODLRO take place in the system with
 an imaginary magnetic field. 
\end{remark}

\begin{remark}
The implication of the claim \eqref{item_imply_no_superconductivity} is
 the following. Assume that
 \eqref{eq_sufficient_condition_superconductivity} holds with some
 $(U,\beta,\theta)\in \R_{<0}\times \R_{>0}\times
 [\pi/\beta,2\pi/\beta)$. Then, SSB and ODLRO occur in the system with
 $(U,\beta,\theta)$ by the claim \eqref{item_superconducting_phase}, while SSB and
 ODLRO do not occur in the system with $(U,\beta,\eta)$ for any $\eta\in
 [0,\pi/\beta)$. By the claim
 \eqref{item_confirmation_condition_superconductivity} we can always choose 
$(U,\beta,\theta)\in \R_{<0}\times \R_{>0}\times
 [\pi/\beta,2\pi/\beta)$ such that
 \eqref{eq_sufficient_condition_superconductivity} holds with
 $(U,\beta,\theta)$. By fixing these $U$, $\beta$ and taking $\eta$ to
 be 0 we can conclude in
 other words that superconductivity characterized by SSB and ODLRO
 emerges in the BCS model with an imaginary magnetic field and it
  does not emerge in the BCS model without an imaginary
 magnetic field.
\end{remark}

\begin{remark}
The condition 
\begin{align*}
|U|<c_2(d)\left(1+\beta^{d+3}+(1+\beta^{-1})g_d\left(\left|\frac{\theta}{2}-\frac{\pi}{\beta}\right|\right)\right)^{-2}
\end{align*}
is necessary to ensure that our double-scale integration converges. Thus,
 the claim \eqref{item_no_superconductivity} especially means that
within the analytical framework of the present paper we cannot
 prove the existence of superconductivity in the
 conventional reduced BCS Hamiltonian,  which is the case $\theta=0$.
\end{remark}

\begin{remark}
There is no essential reason to choose the spatial lattice to be
 $\{0,1,\cdots,L-1\}^d$. One can prove that all the partition functions
 and the thermal expectations in the theorem are equivalent to those
 defined in the system on the spatial lattice $\{0,1,\cdots,L-1\}^d+\ba$
 with the periodic boundary conditions for any $\ba\in\Z^d$.
\end{remark}

\begin{remark} We introduce the parameter $hop(\in \{0,1\})$ to treat
 the model with positive hopping and the model with negative hopping at
 the same time. However, if $L\in 2\N$, by the unitary transform
 $\psi_{\bx\s}^*\to (-1)^{\sum_{j=1}^dx_j}\psi_{\bx\s}^*$
 $(\bx=(x_1,\cdots,x_d)\in\G)$ we can change the sign of hopping by
 keeping all the other terms unchanged. Thus, the role of the parameter
 $hop$ seems not essential. We add it for completeness. It causes no
 technical complication.
\end{remark}

\begin{remark}\label{rem_dispersion_extension} 
The reason why we only consider the nearest-neighbor hopping in the free 
Hamiltonian is that in this case the free dispersion relation
 takes the relatively simple form $e(\bk)$ which allows us to make explicit
the condition
 \eqref{eq_sufficient_condition_superconductivity}.
We made this choice to claim the main
 results of this paper simply and explicitly. 
In fact the condition 
\begin{align}
|U|>&c_1(d)(2d-|\mu|)^{1-d}\beta\left|\frac{\theta}{2}-\frac{\pi}{\beta}\right|\label{eq_coupling_constant_lower_bound}\\
&\cdot\left(1_{|\frac{\theta}{2}-\frac{\pi}{\beta}|\le \frac{1}{2}(2d-|\mu|)}+
1_{|\frac{\theta}{2}-\frac{\pi}{\beta}|> \frac{1}{2}(2d-|\mu|)}(2d-|\mu|)^{-1}\left|\frac{\theta}{2}-\frac{\pi}{\beta}\right|\right)\notag
\end{align}
in \eqref{eq_sufficient_condition_superconductivity} is a 
 sufficient condition for the inequality
\begin{align*}
-\frac{2}{|U|}+\frac{1}{(2\pi)^d}\int_{[0,2\pi]^d}d\bk\frac{\sinh(\beta|e(\bk)|)}{(\cos(\beta\theta/2)+\cosh(\beta
 e(\bk)))|e(\bk)|}>0,
\end{align*}
which is a necessary and sufficient condition for the existence of a
 positive solution to the gap equation \eqref{eq_gap_equation}.
The theorem can be claimed under the condition 
\begin{align*}
|U|>2\left(\frac{1}{(2\pi)^d}\int_{[0,2\pi]^d}d\bk\frac{\sinh(\beta|e(\bk)|)}{(\cos(\beta\theta/2)+\cosh(\beta
 e(\bk)))|e(\bk)|}\right)^{-1}
\end{align*}
in place of the condition
 \eqref{eq_coupling_constant_lower_bound}. We should also remark that
 apart from a multiplication of irrelevant positive constant,
 the term $g_d(|\theta/2-\pi/\beta|)$ is derived as an upper bound on the integral
\begin{align}
\int_{[0,2\pi]^d}d\bk\frac{1}{\sqrt{|\theta/2-\pi/\beta|^2+e(\bk)^2}}.\label{eq_dispersion_integral_determinant_bound}
\end{align}
In fact we can replace the term $g_d(|\theta/2-\pi/\beta|)$ in the
 condition 
\begin{align*}
|U|<c_2(d)\left(1+\beta^{d+3}+(1+\beta^{-1})g_d\left(\left|\frac{\theta}{2}-\frac{\pi}{\beta}\right|\right)\right)^{-2}
\end{align*}
by the integral \eqref{eq_dispersion_integral_determinant_bound}.
The validity of these modifications would be clearly seen after completing
 the proof of Theorem \ref{thm_main_theorem}. 
 Since it is important
 in this approach to guarantee the solvability of the gap equation and
 the convergence of Grassmann integrations at the same time, the
 sufficient conditions should be explicitly comparable.
We decide not to pursue the issue of
 generalization of the dispersion relation or the whole Hamiltonian in
 this paper. Here we list the lemmas which use
 the specific form of the free dispersion relation and eventually lead
 to the condition
 \eqref{eq_sufficient_condition_superconductivity}. These are Lemma
 \ref{lem_covariance_crucial_inequality}, Lemma
 \ref{lem_estimation_free_fermi_surface} and Lemma
 \ref{lem_sufficient_condition_gap_equation}. 

However, based on the above modifications of the crucial conditions, let
 us see that the results hold for the degenerate case $\mu\in \{2d,
 -2d\}$ as the least extension of the theorem. By letting $\Theta$,
 $c(d)$ denote $|\theta/2-\pi/\beta|$, a positive constant depending
 only on $d$ respectively we observe in this case that 
\begin{align*}
&\int_{[0,2\pi]^d}d\bk\frac{\sinh(\beta|e(\bk)|)}{(\cos(\beta\theta/2)+\cosh(\beta
 e(\bk)))|e(\bk)|}\\
&\ge 
c(d)\beta^{-1} \int_{[0,2\pi]^d}d\bk
 \frac{1_{|e(\bk)|\le\beta^{-1}}}{e(\bk)^2+\Theta^2}\ge c(d)\beta^{-1}\int_0^{1}dr r^{d-1}\frac{1_{r^2\le
 \beta^{-1}}}{r^4+\Theta^2}\\
&\ge c(d)\beta^{-1}\big(
1_{\Theta\le \min\{1,\beta^{-1}\}}\\
&\qquad\qquad\quad\cdot(1_{d\le
 3}\Theta^{\frac{d}{2}-2}+1_{d=4}\log(\min\{1,\beta^{-1}\}\Theta^{-1}) +1_{d\ge
 5}(\min\{1,\beta^{-1}\})^{\frac{d}{2}-2})\\
&\quad+1_{\Theta>\min\{1,\beta^{-1}\}}(\min\{1,\beta^{-1}\})^{\frac{d}{2}}\Theta^{-2}\big),\\
&\int_{[0,2\pi]^d}d\bk
 \frac{1}{\sqrt{\Theta^2+e(\bk)^2}}\le c(d)\int_0^1dr\frac{r^{d-1}}{\sqrt{\Theta^2+r^4}}\\
&\le c(d)\big( 1_{\Theta \le
 1}(1_{d=1}\Theta^{-\frac{1}{2}}+1_{d=2}\log(1+\Theta^{-1})+1_{d\ge
 3})+1_{\Theta >1}\Theta^{-1}\big).
\end{align*}
Define the function $g_{d,s}:(0,\infty)\to\R$ by
\begin{align*}
g_{d,s}(x):=1_{x\le 1}(1_{d=1} x^{-\frac{1}{2}}
+ 1_{d=2}(\log 2)^{-1}\log(1+x^{-1})+1_{d\ge 3})+1_{x>1}x^{-1}.
\end{align*}
Here we inserted $\log 2$ in order to make the function
 decreasing. Then, by using the lower, upper bounds obtained above the
 condition \eqref{eq_sufficient_condition_superconductivity} is modified as follows. 
\begin{align}
&c_1(d)\beta\big(
1_{\Theta\le
 \min\{1,\beta^{-1}\}}\label{eq_coupling_constant_upper_lower_degenerate}\\
&\qquad\quad\cdot(1_{d\le
 3}\Theta^{2-\frac{d}{2}}+1_{d=4}\log(\min\{1,\beta^{-1}\}\Theta^{-1})^{-1} +1_{d\ge
 5}(\min\{1,\beta^{-1}\})^{2-\frac{d}{2}})\notag\\
&\quad+1_{\Theta>\min\{1,\beta^{-1}\}}(\min\{1,\beta^{-1}\})^{-\frac{d}{2}}\Theta^{2}\big)\notag\\
&<|U|<c_2(d)(1+\beta^{d+3}+(1+\beta^{-1})g_{d,s}(\Theta))^{-2}.\notag
\end{align}
Since with positive constants $c_1(d,\beta)$, $c_2(d,\beta)$ depending
 only on $d,\beta$,
\begin{align*}
&(\text{L.H.S of
 \eqref{eq_coupling_constant_upper_lower_degenerate}})
\le c_1(d,\beta)(1_{d\le 3}\Theta^{2-\frac{d}{2}}+1_{d=4}|\log
 \Theta|^{-1}+1_{d\ge 5}),\\
&(\text{R.H.S of
 \eqref{eq_coupling_constant_upper_lower_degenerate}})
\ge c_2(d,\beta)(1_{d=1}\Theta + 1_{d=2}|\log\Theta|^{-2}+1_{d\ge 3})
\end{align*}
for small $\Theta$, we can find $U\in \R_{<0}$ satisfying 
the condition
 \eqref{eq_coupling_constant_upper_lower_degenerate} for small $\Theta$
in the case $d=1,2,3,4$. We can
 expect that the claims parallel to those of Theorem
 \ref{thm_main_theorem} hold for $\mu \in \{-2d,2d\}$ in the case $d=1,2,3,4$. We should note that in the
 case $d=3,4$ the upper bound on $|U|$ can be independent of $\Theta$
 and thus we can take $\Theta$ arbitrarily close to zero.
\end{remark}

\section{Formulation}\label{sec_formulation}

In this section we will derive a finite-dimensional Grassmann integral
formulation of the grand canonical partition function. Then, by means
of the Hubbard-Stratonovich transformation we will transform it into 
a Gaussian integral formulation involving both Grassmann variables and
real variables, which will be analyzed as a central object in the
following sections.

\subsection{Grassmann algebra}\label{subsec_Grassmann_algebra}

To begin with, let us recall some basics of finite-dimensional Grassmann
integration. Let $S_0$ be a finite set and let $S:=S_0\times
\{1,-1\}$. In practice we will need to change the index set $S_0$ several
times during the construction. Here we do not fix any detail of
$S_0$. Let $R$ be the complex vector space spanned by the abstract basis
$\{\psi_X\ |\ X\in S\}$. We should remark that $\psi_X$ $(X\in S)$ are
not operators on the Fock space, though we use the same symbol as the
Fermionic annihilation operator. For any $X\in S_0$ we let $\opsi_X$,
$\psi_X$ denote $\psi_{(X,1)}$, $\psi_{(X,-1)}$ respectively. For $n\in
\N$, $\bigwedge^n R$ denotes the $n$-fold anti-symmetric tensor product
of $R$. We set $\bigwedge^0 R:=\C$ by convention. The Grassmann algebra
$\bigwedge R$ generated by $\{\psi_X\ |\ X\in S\}$ is the direct sum of 
$\bigwedge^n R$.
$$
\bigwedge R:=\bigoplus_{n=0}^{\sharp S}\bigwedge^n R.
$$

We will often work in a situation where $R$ is the direct sum $\bigoplus_{p=1}^{m}R^p$
of other vector spaces $R^p$ $(p=1,2,\cdots,m)$. We assume that 
the basis of $R^p$ is $\{\psi_X^p\ |\ X\in S\}$. For a function
$D:S_0^2\to \C$ 
the Grassmann Gaussian integral
$\int\cdot d\mu_{D}(\psi^1)$ is a linear map from $\bigwedge \left(
\bigoplus_{p=1}^mR^p\right)$ to $\bigwedge \left(
\bigoplus_{p=2}^mR^p\right)$ defined as follows. For $f\in \bigwedge \left(
\bigoplus_{p=2}^mR^p\right)$, $X_1,X_2,\cdots,X_a$,
$Y_1,Y_2,\cdots,Y_b\in S_0$,
\begin{align*}
&\int f\opsi_{X_1}^1\opsi_{X_2}^1\cdots \opsi_{X_a}^1\psi_{Y_b}^1\cdots
 \psi_{Y_2}^1\psi_{Y_1}^1d\mu_{D}(\psi^1)\\
&\quad:=\left\{\begin{array}{ll}\det(D(X_i,Y_j))_{1\le
					    i, j\le a}f & \text{if
					     }a=b,\\
0 & \text{otherwise,}\end{array}
\right.\\
&\int fd\mu_{D}(\psi^1):=f.
\end{align*}
Then by linearity and anti-symmetry the value $\int g
d\mu_{D}(\psi^1)$ is uniquely determined for any $g\in \bigwedge \left(
\bigoplus_{p=1}^mR^p\right)$. We can define $\int\cdot d\mu_D(\psi)$ as
a linear functional on $\bigwedge R$ in the same way. 

Exponential and logarithm of a Grassmann polynomial appear in many parts
of this paper. Let us recall their definitions. For $f\in \bigwedge R$
with the constant part $f_0(\in\C)$
$$
e^f:=e^{f_0}\sum_{n=0}^{\sharp S}\frac{1}{n!}(f-f_0)^n.
$$
If $f_0\in \C\backslash \R_{\le 0}$,
$$
\log f:=\log f_0+\sum_{n=1}^{\sharp S}\frac{(-1)^{n-1}}{n}\left(\frac{f-f_0}{f_0}\right)^n.
$$
Throughout the paper $\log \alpha $ for $\alpha\in \C\backslash \R_{\le
0}$ is assumed to be representing the principal value
$\log|\alpha|+i\theta$,
where $\theta\in (-\pi,\pi)$ satisfies $\alpha=|\alpha|e^{i\theta}$. See
e.g. \cite{FKT} for more properties of Grassmann algebra.

\subsection{One-band formulation}\label{subsec_one_band}

It is systematic to introduce artificial parameters $\la_1,\la_2\in\C$
and deal with the normalized partition function
\begin{align}
\frac{\Tr e^{-\beta (\sH+i\theta \sS_z+\sF+\sA)}}{\Tr e^{-\beta
 (\sH_0+i\theta \sS_z)}},\label{eq_partition_function}
\end{align}
where $\sA:=\la_1\sA_1+\la_2\sA_2$. 
From now we always assume that 
$$r_L(\hbx)\neq r_L(\hby).$$ 
We can assume this condition to prove Theorem \ref{thm_main_theorem}, since the theorem
concerns the limit $\|\hbx-\hby\|_{\R^d}\to \infty$ and 
for any $\hbx,\hby\in \Z^d$ with $\hbx\neq \hby$ there exists $N_0\in
\N$ such that $r_L(\hbx)\neq
r_L(\hby)$ for any $L\in \N$ with $L\ge N_0$. 
We will derive the thermal
expectation values of our interest by differentiating
\eqref{eq_partition_function} with $\la_1,\la_2$. We are going to formulate
\eqref{eq_partition_function} into a limit of finite-dimensional
Grassmann integration. First of all we should make sure that the
denominator is non-zero. Let us define the momentum lattice $\G^*$ by
$$
\G^*:=\left\{0,\frac{2\pi}{L},\frac{2\pi}{L}\cdot 2,\cdots,\frac{2\pi}{L}(L-1)
\right\}^d.
$$ 

\begin{lemma}\label{lem_free_partition_function}
\begin{align}
\Tr e^{-\beta (\sH_0+i\theta\sS_z)}&=\prod_{\bk\in\G^*}\left(1+2\cos\left(\frac{\beta
 \theta}{2}\right)e^{-\beta e(\bk)}+e^{-2\beta
 e(\bk)}\right)\label{eq_free_partition_function}\\
&=e^{-\beta \sum_{\bk\in \G^*}e(\bk)}2^{L^d}\prod_{\bk\in\G^*}
\left(\cos\left(\frac{\beta
 \theta}{2}\right)+\cosh(\beta e(\bk))\right)\neq 0.\notag
\end{align}
\end{lemma}
\begin{proof}
For $\s\in \spin$ we consider the Fermionic Fock space
 $F_f(L^2(\G\times\{\s\}))$ as a subspace of $F_f(L^2(\G\times
 \spin))$. Set 
\begin{align*}
&\sH_{0,\s}:=\sum_{\bx\in\G}
\left((-1)^{hop}\sum_{j=1}^d(\psi_{\bx\s}^*\psi_{\bx+\be_j\s}+
\psi_{\bx\s}^*\psi_{\bx-\be_j\s})-\mu
\psi_{\bx\s}^*\psi_{\bx\s}\right),\\
&\sN_{\s}:=\sum_{\bx\in\G}\psi_{\bx\s}^*\psi_{\bx\s}.
\end{align*}
By letting $\Tr_{\s}$ mean the trace operation over $F_f(L^2(\G\times
 \{\s\}))$ we have that
\begin{align*}
\Tr e^{-\beta (\sH_0+i\theta \sS_z)}&=\Tr_{\ua}e^{-\beta
 (\sH_{0,\ua}+i\frac{\theta}{2} \sN_{\ua})}
\Tr_{\da}e^{-\beta (\sH_{0,\da}-i\frac{\theta}{2} \sN_{\da})}
=|\Tr_{\ua}e^{-\beta
 (\sH_{0,\ua}+i\frac{\theta}{2}
 \sN_{\ua})}|^2\\
&=\prod_{\bk\in\G^*}|1+e^{-\beta(e(\bk)+i\frac{\theta}{2})}|^2,
\end{align*}
which is the right-hand side of \eqref{eq_free_partition_function}.
\end{proof}

The covariance in our Grassmann Gaussian integral formulation of
\eqref{eq_partition_function} is equal to a restriction of the following 
free two-point correlation function. For $(\bx,\s,s)$, $(\by,\tau,t)\in
\Z^d\times \spin\times [0,\beta)$, set
\begin{align*}
G(\bx\s s,\by \tau t):=\frac{\Tr(e^{-\beta(\sH_0+i\theta \sS_z)}(1_{s\ge
 t}\psi_{\bx\s}^*(s)\psi_{\by\tau}(t)-1_{s<t}\psi_{\by\tau}(t)\psi_{\bx\s}^*(s)))}{\Tr
 e^{-\beta(\sH_0+i\theta \sS_z)}},
\end{align*}
where $\psi_{\bx\s}^{(*)}(s):=e^{s(\sH_0+i\theta\sS_z)}\psi_{\bx
\s}^{(*)}e^{-s(\sH_0+i\theta\sS_z)}$. 
We introduce the finite index set of Grassmann algebra by discretizing
the time interval. With $h\in \frac{2}{\beta}\N$ we set
$$
[0,\beta)_h:=\left\{0,\frac{1}{h},\frac{2}{h},\cdots,\beta-\frac{1}{h}\right\},
$$
which is a discretization of $[0,\beta)$. We take the parameter $h$ from
$\frac{2}{\beta}\N$ rather than from $\frac{1}{\beta}\N$, since it is technically
convenient according to the earlier study \cite[\mbox{Appendix C}]{K9}.
The index sets of Grassmann algebra for our one-band model are defined
by
$$
J_0:=\G\times \spin\times[0,\beta)_h,\quad J:=J_0\times \{1,-1\}.
$$
The restriction $G|_{J_0^2}$ is the covariance of our Grassmann Gaussian
integral formulation. For simplicity let us omit the notation
$\cdot|_{J_0^2}$ in the following.

Let $\cW$ be the complex vector space spanned by the basis $\{\psi_X\ |\
X\in J\}$. Set $N:=4L^d\beta h$ so that $\sharp J=N$.
Here we state the Grassmann integral formulation of
\eqref{eq_partition_function} in the Grassmann algebra $\bigwedge
\cW$. For $r\in\R_{> 0}$ let $D(r)$ denote the open disk $\{z\in \C\ |\
|z|<r\}$.
Set
\begin{align*}
&\sV(\psi):=\frac{U}{hL^d}\sum_{\bx,\by\in\G}\sum_{s\in[0,\beta)_h}\opsi_{\bx\ua
 s}\opsi_{\bx\da
 s}\psi_{\by\da s}\psi_{\by\ua s},\\
&\sF(\psi):=\frac{\g}{h}\sum_{\bx\in \G}\sum_{s\in [0,\beta)_h}(\opsi_{\bx\ua
 s}\opsi_{\bx\da
 s}+\psi_{\bx\da s}\psi_{\bx\ua s}),\\
&\sA^1(\psi):=\frac{1}{h}\sum_{s\in [0,\beta)_h}\opsi_{r_L(\hbx)\ua
 s}\opsi_{r_L(\hbx)\da s},\\
&\sA^2(\psi):=\frac{1}{h}\sum_{s\in [0,\beta)_h}\opsi_{r_L(\hbx)\ua
 s}\opsi_{r_L(\hbx)\da s}\psi_{r_L(\hby)\da
 s}\psi_{r_L(\hby)\ua s},\\
&\sA(\psi):=\la_1\sA^1(\psi)+\la_2\sA^2(\psi).
\end{align*}
We let $\bla$ denote $(\la_1,\la_2)$ $(\in \C^2)$.
\begin{lemma}\label{lem_grassmann_formulation}
For any $r\in\R_{>0}$,
\begin{align}
\lim_{h\to\infty\atop
 h\in\frac{2}{\beta}\N}\sup_{\bla\in\overline{D(r)}^2}
\left|
\int e^{-\sV(\psi)-\sF(\psi)-\sA(\psi)}d\mu_G(\psi)
-\frac{\Tr e^{-\beta (\sH+i\theta\sS_z+\sF+\sA)}}{\Tr e^{-\beta
 (\sH_0+i\theta \sS_z)}}\right|=0.\label{eq_formulation_partition}
\end{align}
\end{lemma}

\begin{proof} The proof is close to the Grassmann integral formulation
 process in \cite{K9}, \cite{K10}, \cite{K14}. However, we sketch the
 procedure for readers' convenience. For any objects
 $\alpha_1,\alpha_2,\cdots,\alpha_n$ we let $\prod_{j=1}^n\alpha_j$
 denote
$\alpha_1\alpha_2\cdots \alpha_n$. This definition should be reminded especially
 when $\alpha_j$ $(j=1,2,\cdots, n)$ are non-commutative. For
 $\bx,\by\in\G$, $a\in \{0,-1,1\}$ set
\begin{align*}
V(\bx,\by,a):=&1_{a=0}\left(\frac{U}{L^d}+\la_21_{(\bx,\by)=(r_L(\hbx),r_L(\hby))}\right)+1_{a=1}\left(\frac{\g}{L^d}+\frac{\la_1}{L^d}1_{\bx=r_L(\hbx)}\right)\\
&+1_{a=-1}\frac{\g}{L^d}.
\end{align*}
The partition function \eqref{eq_partition_function} can be expanded as
 follows (see e.g. \cite[\mbox{Lemma B.3}]{K9}).
\begin{align}
&\frac{\Tr e^{-\beta (\sH+i\theta \sS_z+\sF+\sA)}}{\Tr e^{-\beta
 (\sH_0+i\theta \sS_z)}}\label{eq_perturbation_series}\\
&=1+\sum_{n=1}^{\infty}(-1)^n\prod_{j=1}^n\left(
\sum_{\bx_j,\by_j\in\G}\int_0^{\beta}ds_j\sum_{a_j\in \{0,1,-1\}}V(\bx_j,\by_j,a_j)
\right)1_{s_1>s_2>\cdots>s_n}1_{\sum_{j=1}^na_j=0}\notag\\
&\qquad\qquad\cdot
\<\prod_{j=1}^n(1_{a_j=0}\psi_{\bx_j\ua}^*(s_j)\psi_{\bx_j\da}^*(s_j)
\psi_{\by_j\da}(s_j)\psi_{\by_j\ua}(s_j)
+
1_{a_j=1}\psi_{\bx_j\ua}^*(s_j)\psi_{\bx_j\da}^*(s_j)\notag
\\
&\qquad\qquad\qquad\quad+
1_{a_j=-1}
\psi_{\by_j\da}(s_j)\psi_{\by_j\ua}(s_j))\>_0,\notag
\end{align}
where
$$
\<\cO\>_0:=\frac{\Tr(e^{-\beta(\sH_0+i\theta \sS_z)}\cO)}{\Tr
 e^{-\beta(\sH_0+i\theta \sS_z)}}
$$
for any operator $\cO$ on $F_f(L^2(\G\times \spin))$.
The constraint $1_{\sum_{j=1}^na_j=0}$ is due to the fact that
 $\sH_0+i\theta \sS_z$ conserves the particle number.

Assume that $\beta >s_1 > s_2>\cdots >s_n>0$ and $\sum_{j=1}^na_j=0$. We
 can choose $\{i_p\}_{p=1}^l$, $\{j_p\}_{p=1}^m$,
 $\{k_p\}_{p=1}^m\subset\{1,2,\cdots,n\}$ so that $l+2m=n$ and 
\begin{align*}
&i_1<i_2<\cdots <i_l,\quad a_{i_p}=0\ (\forall p\in \{1,2,\cdots,l\}),\\
&j_1<j_2<\cdots <j_m,\quad a_{j_p}=1\ (\forall p\in \{1,2,\cdots,m\}),\\
&k_1<k_2<\cdots <k_m,\quad a_{k_p}=-1\ (\forall p\in \{1,2,\cdots,m\}).
\end{align*}
Then, let us set
\begin{align*}
&\bX^0:=((\bx_{i_1},\ua,s_{i_1}),(\bx_{i_1},\da,s_{i_1}),\cdots,
(\bx_{i_l},\ua,s_{i_l}),(\bx_{i_l},\da,s_{i_l}))\in (\G\times\spin
 \times[0,\beta))^{2l},\\
&\bX^1:=((\bx_{j_1},\ua,s_{j_1}),(\bx_{j_1},\da,s_{j_1}),\cdots,
(\bx_{j_m},\ua,s_{j_m}),(\bx_{j_m},\da,s_{j_m}))\\
&\quad\in (\G\times\spin
 \times[0,\beta))^{2m},\\
&\bY^0:=((\by_{i_1},\ua,s_{i_1}),(\by_{i_1},\da,s_{i_1}),\cdots,
(\by_{i_l},\ua,s_{i_l}),(\by_{i_l},\da,s_{i_l}))\in (\G\times\spin
 \times[0,\beta))^{2l},\\
&\bY^1:=((\by_{k_1},\ua,s_{k_1}),(\by_{k_1},\da,s_{k_1}),\cdots,
(\by_{k_m},\ua,s_{k_m}),(\by_{k_m},\da,s_{k_m}))\\
&\quad \in (\G\times\spin
 \times[0,\beta))^{2m}.
\end{align*}
By giving a number to each component we write as follows.
$$(\bX^0,\bX^1)=(X_j)_{1\le j\le 2l+2m},\quad 
(\bY^0,\bY^1)=(Y_j)_{1\le j\le 2l+2m}. $$
Because of the assumption $s_1>\cdots>s_n$, the operator
\begin{align*}
&\prod_{j=1}^n(1_{a_j=0}\psi_{\bx_j\ua}^*(s_j)\psi_{\bx_j\da}^*(s_j)
\psi_{\by_j\da}(s_j)\psi_{\by_j\ua}(s_j)
+
1_{a_j=1}\psi_{\bx_j\ua}^*(s_j)\psi_{\bx_j\da}^*(s_j)\\
&\quad+
1_{a_j=-1}
\psi_{\by_j\da}(s_j)\psi_{\by_j\ua}(s_j))
\end{align*}
is already ordered with respect to the standard lexicographical order in
 the product set $[0,\beta)\times \{\text{particle},
 \text{hole}\}$. Thus, 
\begin{align*}
&\<\prod_{j=1}^n(1_{a_j=0}\psi_{\bx_j\ua}^*(s_j)\psi_{\bx_j\da}^*(s_j)
\psi_{\by_j\da}(s_j)\psi_{\by_j\ua}(s_j)
+
1_{a_j=1}\psi_{\bx_j\ua}^*(s_j)\psi_{\bx_j\da}^*(s_j)\\
&\quad+
1_{a_j=-1}
\psi_{\by_j\da}(s_j)\psi_{\by_j\ua}(s_j))\>_0\\
&=\det(G(X_i,Y_j))_{1\le i,j\le 2l+2m}.
\end{align*}
See e.g. \cite[\mbox{Lemma B.7,\ Lemma B.8,\ Lemma B.9}]{K9} for a proof
 of the above equality. By substituting
 this into \eqref{eq_perturbation_series} we have
\begin{align}
&\frac{\Tr e^{-\beta (\sH+i\theta \sS_z+\sF+\sA)}}{\Tr e^{-\beta
 (\sH_0+i\theta \sS_z)}}\label{eq_perturbation_series_determinant}\\
&=1+\sum_{n=1}^{\infty}(-1)^n\prod_{j=1}^n\left(
\sum_{\bx_j,\by_j\in\G}\int_0^{\beta}ds_j\sum_{a_j\in \{0,1,-1\}}V(\bx_j,\by_j,a_j)
\right)1_{s_1>s_2>\cdots>s_n}1_{\sum_{j=1}^na_j=0}
\notag\\
&\qquad\qquad\cdot \det(G(X_i,Y_j))_{1\le i,j\le 2l+2m}.\notag
\end{align}
Define the function $P(\bla)$ $(\bla\in\C^2)$ by the right-hand side of
 \eqref{eq_perturbation_series_determinant}. We also define its discrete
 analogue $P_h(\bla)$ by
\begin{align}
&P_h(\bla):=1+\sum_{n=1}^{\infty}(-1)^n\prod_{j=1}^n\left(\frac{1}{h}
\sum_{\bx_j,\by_j\in\G}\sum_{s_j\in \0betah}\sum_{a_j\in \{0,1,-1\}}V(\bx_j,\by_j,a_j)
\right)\label{eq_discrete_perturbation_series}\\
&\qquad\qquad\qquad\qquad\cdot 1_{s_1>s_2>\cdots>s_n}1_{\sum_{j=1}^na_j=0}
\det(G(X_i,Y_j))_{1\le i,j\le 2l+2m}.\notag
\end{align}
Based on the fact that the function 
$$
\bs\mapsto 1_{s_1>s_2>\cdots>s_n} \det(G(X_i,Y_j))_{1\le i,j\le 2l+2m}:[0,\beta)^n\to\C
$$
is continuous almost everywhere in $[0,\beta)^n$ and the uniform bounds
\begin{align}
&\prod_{j=1}^n\left(\int_0^{\beta}ds_j\right)1_{s_1>s_2>\cdots>s_n}\le
 \frac{\beta^n}{n!},\notag\\
&|\det(G(W_i,Z_j))_{1\le i,j\le n}|\le
 \frac{2^{2L^d}e^{(2n+1)\beta\|\sH_0+i\theta \sS_z\|_{\cB(F_f)}}}
{|\Tr e^{-\beta (\sH_0+i\theta
 \sS_z)}|},\label{eq_full_determinant_bound}\\
&(\forall n\in\N,\ W_j,Z_j\in \G\times \spin\times[0,\beta)\ (j=1,2,\cdots,n)),\notag
\end{align}
where $\|\cdot\|_{\cB(F_f)}$ is the operator norm of operators on
 $F_f(L^2(\G\times \spin))$, we can prove that for any $r\in \R_{>0}$,
 $n\in \N$, 
\begin{align*}
\lim_{h\to \infty\atop h\in \frac{2}{\beta}\N}\sup_{\bla\in\overline{D(r)}^2}
\Bigg|&
\prod_{j=1}^n\left(\frac{1}{h}
\sum_{\bx_j,\by_j\in\G}\sum_{s_j\in \0betah}\sum_{a_j\in \{0,1,-1\}}V(\bx_j,\by_j,a_j)
\right)\\
&\quad\cdot 1_{s_1>s_2>\cdots>s_n}1_{\sum_{j=1}^na_j=0}
\det(G(X_i,Y_j))_{1\le i,j\le 2l+2m}\\
&-
\prod_{j=1}^n\left(
\sum_{\bx_j,\by_j\in\G}\int_0^{\beta}ds_j
\sum_{a_j\in \{0,1,-1\}}V(\bx_j,\by_j,a_j)
\right)\\
&\qquad\cdot 1_{s_1>s_2>\cdots>s_n}1_{\sum_{j=1}^na_j=0}
\det(G(X_i,Y_j))_{1\le i,j\le 2l+2m}\Bigg|=0
\end{align*}
and by the dominated convergence theorem that
\begin{align}
\lim_{h\to\infty\atop h\in
 \frac{2}{\beta}\N}\sup_{\bla\in\overline{D(r)}^2}|P_h(\bla)-P(\bla)|=0.\label{eq_formulation_time_continuum_pre}
\end{align}
By the definition of the Grassmann Gaussian integral, it holds inside
 \eqref{eq_discrete_perturbation_series} that
\begin{align*}
\det(G(X_i,Y_j))_{1\le i,j\le
 2l+2m}=\int\prod_{j=1}^n(&1_{a_j=0}\opsi_{\bx_j\ua s_j}\opsi_{\bx_j\da
 s_j}\psi_{\by_j\da s_j}\psi_{\by_j\ua s_j}
+
1_{a_j=1}\opsi_{\bx_j\ua s_j}\opsi_{\bx_j\da
 s_j}\\
&+1_{a_j=-1}
\psi_{\by_j\da s_j}\psi_{\by_j\ua s_j})
d\mu_{G}(\psi).
\end{align*}
By substituting this we observe that 
\begin{align*}
P_h(\bla)=1+\sum_{n=1}^{2L^2\beta h}&\frac{(-1)^n}{n!}\prod_{j=1}^n\left(\frac{1}{h}
\sum_{\bx_j,\by_j\in\G}\sum_{s_j\in \0betah}\sum_{a_j\in \{0,1,-1\}}V(\bx_j,\by_j,a_j)
\right)1_{j\neq k\to s_j\neq s_k}\\
&\cdot\int\prod_{j=1}^n(1_{a_j=0}\opsi_{\bx_j\ua s_j}\opsi_{\bx_j\da
 s_j}\psi_{\by_j\da s_j}\psi_{\by_j\ua s_j}
+
1_{a_j=1}\opsi_{\bx_j\ua s_j}\opsi_{\bx_j\da
 s_j}\\
&\qquad\quad+1_{a_j=-1}
\psi_{\by_j\da s_j}\psi_{\by_j\ua s_j})
d\mu_{G}(\psi).
\end{align*}
Note that if we drop the constraint $1_{j\neq k\to s_j\neq s_k}$, the
 right-hand side is equal to \\
$\int
 e^{-\sV(\psi)-\sF(\psi)-\sA(\psi)}d\mu_G(\psi)$.
By using the estimate
$$
\prod_{j=1}^n\left(\frac{1}{h}\sum_{s_j\in \0betah}\right)1_{\exists
 j\exists k (j\neq k\land s_j=s_k)}\le 1_{n\ge
	 2}\left(\begin{array}{c}n\\2\end{array}\right)\frac{\beta^{n-1}}{h}
$$
and the uniform bound \eqref{eq_full_determinant_bound} we can prove
 that for any $r\in \R_{>0}$
\begin{align*}
\lim_{h\to\infty\atop
 h\in\frac{2}{\beta}\N}\sup_{\bla\in\overline{D(r)}^2}
\left|P_h(\bla)-
\int e^{-\sV(\psi)-\sF(\psi)-\sA(\psi)}d\mu_G(\psi)\right|=0.
\end{align*}
This convergence property and \eqref{eq_formulation_time_continuum_pre}
 imply \eqref{eq_formulation_partition}.
\end{proof}

As the second step we decompose the quartic Grassmann polynomial
$\sV(\psi)$ into quadratic polynomials and a quartic correction term by means of
the Hubbard-Stratonovich transformation. Let us define $\sV_+(\psi)$,
$\sV_-(\psi)$, $\sW(\psi)\in \bigwedge \cW$ by 
\begin{align*}
&\sV_+(\psi):=\frac{|U|^{\frac{1}{2}}}{\beta^{\frac{1}{2}}L^{\frac{d}{2}}h}\sum_{\bx\in\G}\sum_{s\in\0betah}\opsi_{\bx\ua
 s}\opsi_{\bx\da s},\\
&\sV_-(\psi):=\frac{|U|^{\frac{1}{2}}}{\beta^{\frac{1}{2}}L^{\frac{d}{2}}h}\sum_{\bx\in\G}\sum_{s\in\0betah}\psi_{\bx\da
 s}\psi_{\bx\ua s},\\
&\sW(\psi):=\frac{U}{\beta L^{d}h^2}\sum_{\bx,\by\in\G}\sum_{s,t\in\0betah}
\opsi_{\bx\ua
 s}\opsi_{\bx\da s}\psi_{\by\da
 t}\psi_{\by\ua t}.
\end{align*}

\begin{lemma}\label{lem_H_S_transformation}
\begin{align}
&\int e^{-\sV(\psi)-\sF(\psi)-\sA(\psi)}d\mu_G(\psi)\label{eq_H_S_transformation}\\
&=
\frac{1}{\pi}\int_{\R^2}d\phi_1d\phi_2 e^{-|\phi|^2}
\int e^{-\sV(\psi)+\sW(\psi)
-\sF(\psi)-\sA(\psi)+\phi\sV_+(\psi)+\overline{\phi}\sV_-(\psi)}d\mu_G(\psi),\notag
\end{align}
where $\phi:=\phi_1+i\phi_2$, $|\phi|:=\sqrt{\phi_1^2+\phi_2^2}$.
\end{lemma}

\begin{proof}
For $f_j(\psi)\in\bigwedge \cW$, $g_j\in L^1(\R^2)$ $(j=1,2,\cdots,n)$
 we can define the Grassmann polynomial
 $\int_{\R^2}d\phi_1d\phi_2\sum_{j=1}^ng_j(\phi_1,\phi_2)f_j(\psi)$ by
\begin{align*}
\int_{\R^2}d\phi_1d\phi_2\sum_{j=1}^ng_j(\phi_1,\phi_2)f_j(\psi):=\sum_{j=1}^n\left(
\int_{\R^2}d\phi_1d\phi_2 g_j(\phi_1,\phi_2)\right)f_j(\psi).
\end{align*}
Bearing this definition in mind, the Hubbard-Stratonovich transformation
 gives that
\begin{align}
e^{\sV_+(\psi)\sV_-(\psi)}=\frac{1}{\pi}\int_{\R^2}d\phi_1d\phi_2e^{-|\phi|^2+\phi
 \sV_+(\psi)+\overline{\phi}\sV_-(\psi)}.\label{eq_H_S_transformation_partial}
\end{align}
This equality can be confirmed without difficulty. In fact,
\begin{align*}
&\frac{1}{\pi}\int_{\R^2}d\phi_1d\phi_2e^{-|\phi|^2+\phi
 \sV_+(\psi)+\overline{\phi}\sV_-(\psi)}\\
&=\sum_{m,n=0}^{N/2}\frac{1}{\pi(2m)!(2n)!}(\sV_+(\psi)+\sV_-(\psi))^{2m}
(i\sV_+(\psi)-i\sV_-(\psi))^{2n}\\
&\quad\cdot\int_{\R}d\phi_1e^{-\phi_1^2}\phi_1^{2m}
\int_{\R}d\phi_2e^{-\phi_2^2}\phi_2^{2n}\\
&=\sum_{m,n=0}^{N/2}\frac{2^{-2m-2n}}{m!n!}(\sV_+(\psi)+\sV_-(\psi))^{2m}
(i\sV_+(\psi)-i\sV_-(\psi))^{2n}\\
&=e^{\frac{1}{4}(\sV_+(\psi)+\sV_-(\psi))^2+\frac{1}{4}(i\sV_+(\psi)-i\sV_-(\psi))^2}=e^{\sV_+(\psi)\sV_-(\psi)}.
\end{align*}
By substituting the equality $\sV_+(\psi)\sV_-(\psi)=-\sW(\psi)$ and
 \eqref{eq_H_S_transformation_partial} we can derive the result.
\end{proof}

\subsection{Two-band formulation}\label{subsec_two_band}

To complete the formulation, we will include the quadratic terms
$\sV_+(\psi)$, $\sV_-(\psi)$, $\sF(\psi)$ in the covariance. This
procedure leads to another Grassmann integration where Grassmann
algebra is indexed by the band index $\{1,2\}$ rather than the spin
$\spin$. To this end, let us introduce some notations. We define the new index
sets $I_0$, $I$ by 
$$
I_0:=\{1,2\}\times \G \times \0betah,\quad I:=I_0\times\{1,-1\}.
$$ 
Let $\cV$ be the complex vector space spanned by the basis $\{\psi_X\ |\
X\in I\}$. Then, define the Grassmann polynomials $V(\psi)$, $W(\psi)$, 
 $A^1(\psi)$, $A^2(\psi)$, 
$A(\psi)\in\bigwedge \cV$ by
\begin{align}
&V(\psi):=\frac{U}{L^dh}\sum_{\bx\in\G}\sum_{s\in\0betah}
\opsi_{1\bx
 s}\psi_{1\bx s}+
\frac{U}{L^dh}\sum_{\bx,\by\in\G}\sum_{s\in\0betah}\opsi_{1\bx
 s}\psi_{2\bx s}\opsi_{2\by s}\psi_{1\by s},\label{eq_2_band_Grassmann_V}\\
&W(\psi):=\frac{U}{\beta L^dh^2}\sum_{\bx,\by\in\G}\sum_{s,t\in\0betah}\opsi_{1\bx
 s}\psi_{2\bx s}\opsi_{2\by t}\psi_{1\by t},\label{eq_2_band_Grassmann_W}\\
&A^1(\psi):=\frac{1}{h}\sum_{s\in
 \0betah}\opsi_{1r_L(\hbx)s}\psi_{2r_L(\hbx)s},\label{eq_2_band_Grassmann_A_parts}\\
&A^2(\psi):=\frac{1}{h}\sum_{s\in
 \0betah}\opsi_{1r_L(\hbx)s}\psi_{2r_L(\hbx)s} \opsi_{2r_L(\hby)s}\psi_{1r_L(\hby)s},\notag\\
&A(\psi):=\la_1A^1(\psi)+\la_2A^2(\psi).
\label{eq_artificial_Grassmann_term}
\end{align}

In order to introduce the Grassmann Gaussian integral formulation, we need
to define its covariance.
To define the covariance as a free 2-point correlation
function, first we need to introduce a free Hamiltonian on the Fermionic Fock
space $F_f(L^2(\{1,2\}\times \G))$. For $\phi\in \C$ set
\begin{align}
&H_0(\phi):=\frac{1}{L^d}\sum_{\bx,\by\in \G}\sum_{\bk\in
 \G^*}e^{i\<\bk,\bx-\by\>}
\<\left(\begin{array}{c}\psi_{1\bx}^* \\ \psi_{2\bx}^*
	\end{array}\right),
\left(\begin{array}{cc} i\frac{\theta}{2}+e(\bk) & \phi \\ 
                         \overline{\phi} & i\frac{\theta}{2}-e(\bk)
\end{array}\right)
\left(\begin{array}{c}\psi_{1\by} \\ \psi_{2\by}
	\end{array}\right)\>,\label{eq_definition_one_band_free_hamiltonian}
\end{align}
where $\psi^*_{\rho\bx}$ $(\psi_{\rho\bx})$ is the creation
(annihilation) operator in $F_f(L^2(\{1,2\}\times \G))$. 
Because of the presence of $i\frac{\theta}{2}$, $H_0(\phi)$ is not
self-adjoint. Therefore, it may not be appropriate to call $H_0(\phi)$
Hamiltonian. We included the imaginary magnetic field inside only for
conciseness. The covariance of the 2-band formulation is the restriction
of the free 2-point correlation function $C(\phi):(\{1,2\}\times
\Z^d\times [0,\beta))^2\to \C$ defined by
\begin{align}\label{eq_2_band_covariance_original_definition}
C(\phi)(\rho\bx s,\eta\by t):=\frac{\Tr (e^{-\beta H_0(\phi)}(1_{s\ge
 t}\psi_{\rho\bx}^*(s)\psi_{\eta\by}(t)- 1_{s<
 t}\psi_{\eta\by}(t)\psi_{\rho\bx}^*(s)))}{\Tr e^{-\beta H_0(\phi)}},
\end{align}
where
$\psi_{\rho\bx}^{(*)}(s):=e^{sH_0(\phi)}\psi_{\rho\bx}^{(*)}e^{-sH_0(\phi)}$.
Here again we identify $\psi_{\rho\bx}^{(*)}$ with $\psi_{\rho
r_L(\bx)}^{(*)}$ for $\bx\in\Z^d$. 
Since $H_0(\phi)$ is not self-adjoint, the denominator could be zero. We
have to make sure that this is not the case.

\begin{lemma}\label{lem_partition_function_2_band}
\begin{enumerate}[(i)]
\item\label{item_partition_function_2_band}
\begin{align*}
\Tr e^{-\beta H_0(\phi)}&=\prod_{\bk\in \G^*}\prod_{\delta\in
 \{1,-1\}}
\left(1+e^{-\beta(i\frac{\theta}{2}+\delta\sqrt{e(\bk)^2+|\phi|^2})}
\right)\\
&=e^{-i\frac{\beta\theta}{2}L^d}2^{L^d}\prod_{\bk\in
 \G^*}\left(\cos\left(\frac{\beta\theta}{2}\right)+\cosh\left(\beta\sqrt{e(\bk)^2+|\phi|^2}\right)\right)
\neq 0.
\end{align*}
\item\label{item_covariance_2_band}
For any $(\rho,\bx,s),(\eta,\by,t)\in \{1,2\}\times \Z^d\times [0,\beta)$,
\begin{align}
&C(\phi)(\rho\bx s,\eta \by t)\label{eq_covariance_2_band}\\
&=\frac{1}{L^d}\sum_{\bk\in
 \G^*}e^{i\<\bk,\bx-\by\>}e^{(s-t)(i\frac{\theta}{2}I_2+E(\phi)(\bk))}\notag\\
&\qquad\qquad\cdot (1_{s\ge t}(I_2+e^{\beta(i\frac{\theta}{2}I_2+E(\phi)(\bk))})^{-1}
-
1_{s< t}(I_2+e^{-\beta(i\frac{\theta}{2}I_2+E(\phi)(\bk))})^{-1})(\rho,\eta),\notag
\end{align}
where $I_2$ is the
$2\times 2$ unit matrix and 
\begin{align}
E(\phi)(\bk):=\left(\begin{array}{cc} e(\bk) & \overline{\phi} \\
                                     \phi & -e(\bk) \end{array}
\right).
\label{eq_dispersion_matrix}
\end{align}
\end{enumerate}
\end{lemma}

\begin{proof} \eqref{item_partition_function_2_band}:
Since the materials will be used later, we describe
 the derivation in some detail. Define the $(\phi,\bk)$-dependent
 $2\times 2$ matrix $U(\phi)(\bk)$ as follows. 
\begin{align}
U(\phi)(\bk):=1_{\phi=0}I_2+1_{\phi\neq 0}
\left(\frac{\bX(\phi)(\bk)}{\|\bX(\phi)(\bk)\|_{\C^2}},\frac{\bY(\phi)(\bk)}{\|\bY(\phi)(\bk)\|_{\C^2}}\right),\label{eq_unitary_for_diagonalization}
\end{align}
where
\begin{align*}
\bX(\phi)(\bk):=\left(\begin{array}{c} \overline{\phi} \\
		      \sqrt{e(\bk)^2+|\phi|^2}-e(\bk)\end{array}\right),\quad
\bY(\phi)(\bk):=\left(\begin{array}{c} -\overline{\phi} \\
		      \sqrt{e(\bk)^2+|\phi|^2}+e(\bk)\end{array}\right)
\end{align*}
and $\|\cdot\|_{\C^2}$ is the norm of $\C^2$ induced by the hermitian
 inner product.
Moreover, set 
\begin{align}
e(\phi)(\bk):=1_{\phi=0}e(\bk)+1_{\phi\neq 0}\sqrt{e(\bk)^2+|\phi|^2}.
\label{eq_full_eigen_value}
\end{align}
One can check that $U(\phi)(\bk)$ is unitary and 
\begin{align}
U(\phi)(\bk)^*E(\phi)(\bk)U(\phi)(\bk)=\left(\begin{array}{cc}
					e(\phi)(\bk) & 0 \\ 
                                        0 & -e(\phi)(\bk)
\end{array}\right).
\label{eq_dispersion_matrix_diagonalization}
\end{align}
Note that 
$$
H_0(\phi)=\frac{1}{L^d}\sum_{\bx,\by\in \G}\sum_{\bk\in
 \G^*}e^{i\<\bk,\bx-\by\>}
\<\left(\begin{array}{c}\psi_{1\bx}^* \\ \psi_{2\bx}^*
	\end{array}\right), \left(i\frac{\theta}{2}I_2+E(\overline{\phi})(\bk)\right)
\left(\begin{array}{c}\psi_{1\by} \\ \psi_{2\by}
	\end{array}\right)\>.
$$
With the matrix $U(\phi)(\bk)$ one can define a unitary transform
 $\cU(\phi)$ on $F_f(L^2(\{1,2\}\times \G))$ satisfying that 
\begin{align}
&\cU(\phi)H_0(\phi)\cU(\phi)^*\label{eq_band_diagonalization}\\
&=\frac{1}{L^d}\sum_{\bx,\by\in \G}\sum_{\bk\in
 \G^*}e^{i\<\bk,\bx-\by\>}
\<\left(\begin{array}{c}\psi_{1\bx}^* \\ \psi_{2\bx}^*
	\end{array}\right),
\left(\begin{array}{cc} i\frac{\theta}{2} + e(\phi)(\bk) & 0 \\
                                        0 & i\frac{\theta}{2} - e(\phi)(\bk)
\end{array}\right)
\left(\begin{array}{c}\psi_{1\by} \\ \psi_{2\by}
	\end{array}\right)\>,\notag\\
&\cU(\phi)\psi_{\rho\bx}^*\cU(\phi)^*=\frac{1}{L^d}\sum_{\by\in
 \G}\sum_{\eta\in \{1,2\}}\sum_{\bk\in \G^*}e^{-i\<\bk,\bx-\by\>}\overline{U(\overline{\phi})(\bk)(\rho,\eta)}\psi_{\eta\by}^*.\label{eq_unitary_characterization}
\end{align}
Since $\Tr e^{-\beta H_0(\phi)}=\Tr e^{-\beta
 \cU(\phi)H_0(\phi)\cU(\phi)^*}$ and $\cU(\phi)H_0(\phi)\cU(\phi)^*$ is
 diagonalized with the band index, the result follows.

\eqref{item_covariance_2_band}: By the periodicity of both sides of
 \eqref{eq_covariance_2_band} we can restrict the spatial variables to
 $\G$ during the proof. By \eqref{eq_band_diagonalization},
 \eqref{eq_unitary_characterization} and
 $\overline{U(\overline{\phi})(\bk)}=U(\phi)(\bk)$, 
\begin{align*}
&C(\phi)(\rho\bx s,\eta\by t)\\
&=\frac{1}{L^{2d}}\sum_{\bk,\bp\in\G^*}\sum_{\bx',\by'\in
 \G}\sum_{\rho',\eta'\in
 \{1,2\}}e^{-i\<\bk,\bx-\bx'\>+i\<\bp,\by-\by'\>}
U(\phi)(\bk)(\rho,\rho')\overline{U(\phi)(\bp)(\eta,\eta')}\\
&\quad \cdot\frac{\Tr (e^{-\beta \cU(\phi)H_0(\phi)\cU(\phi)^*}(1_{s\ge
 t}\tilde{\psi}_{\rho'\bx'}^*(s)\tilde{\psi}_{\eta'\by'}(t)- 1_{s<
 t}\tilde{\psi}_{\eta'\by'}(t)\tilde{\psi}_{\rho'\bx'}^*(s)))}{\Tr e^{-\beta \cU(\phi)H_0(\phi)\cU(\phi)^*}},
\end{align*}
where $\tilde{\psi}_{\rho\bx}^{(*)}(s):=e^{s
 \cU(\phi)H_0(\phi)\cU(\phi)^*}\psi_{\rho\bx}^{(*)}
e^{-s \cU(\phi)H_0(\phi)\cU(\phi)^*}$. Since
 $\cU(\phi)H_0(\phi)\cU(\phi)^*$ is diagonalized with the band index, it
 can be derived by a standard procedure (see e.g. \cite[\mbox{Appendix B}]{K9})
 that
\begin{align*}
&\frac{\Tr (e^{-\beta \cU(\phi)H_0(\phi)\cU(\phi)^*}(1_{s\ge
 t}\tilde{\psi}_{\rho'\bx'}^*(s)\tilde{\psi}_{\eta'\by'}(t)- 1_{s<
 t}\tilde{\psi}_{\eta'\by'}(t)\tilde{\psi}_{\rho'\bx'}^*(s)))}{\Tr
 e^{-\beta \cU(\phi)H_0(\phi)\cU(\phi)^*}}\\
&=\frac{1_{\rho'=\eta'}}{L^d}\sum_{\bk'\in
 \G^*}e^{i\<\bk',\bx'-\by'\>}e^{(s-t)(i\frac{\theta}{2}+(-1)^{1_{\rho'=2}}e(\phi)(\bk'))}\\
&\qquad\qquad\cdot \left(\frac{1_{s\ge
 t}}{1+e^{\beta(i\frac{\theta}{2}+(-1)^{1_{\rho'=2}}e(\phi)(\bk'))}}-\frac{1_{s<
 t}}{1+e^{-\beta(i\frac{\theta}{2}+
 (-1)^{1_{\rho'=2}}e(\phi)(\bk'))}}\right).
\end{align*}
Then by using the equality $e(\phi)(\bk)=e(\phi)(-\bk)$ we obtain that 
\begin{align}
&C(\phi)(\rho\bx s,\eta\by t)\label{eq_covariance_2_band_pre}\\
&=\frac{1}{L^{d}}\sum_{\bk\in\G^*}\sum_{\xi\in
 \{1,2\}}e^{i\<\bk,\bx-\by\>}
U(\phi)(\bk)(\rho,\xi)\overline{U(\phi)(\bk)(\eta,\xi)} e^{(s-t)(i\frac{\theta}{2}+(-1)^{1_{\xi=2}}e(\phi)(\bk))}\notag\\
&\qquad\qquad\cdot \left(\frac{1_{s\ge
 t}}{1+e^{\beta(i\frac{\theta}{2}+(-1)^{1_{\xi=2}}e(\phi)(\bk))}}-\frac{1_{s<
 t}}{1+e^{-\beta(i\frac{\theta}{2}+
 (-1)^{1_{\xi=2}}e(\phi)(\bk))}}\right).\notag
\end{align}
Then by using \eqref{eq_dispersion_matrix_diagonalization} again we reach the
 claimed equality.
\end{proof}

The following lemma will form the basis of our analysis which eventually
leads to the proof of Theorem \ref{thm_main_theorem}.

\begin{lemma}\label{lem_grassmann_formulation_2_band}
The following statements hold true for any $r\in\R_{>0}$.
\begin{enumerate}[(i)]
\item\label{item_integrand_time_continuum_limit}
For any $\phi\in\C$,
$$
\lim_{h\to\infty\atop h\in\frac{2}{\beta}\N}\int e^{-V(\psi)+W(\psi)-A(\psi)}d\mu_{C(\phi)}(\psi)
$$
converges in $C(\overline{D(r)}^2)$ as a sequence of function with the
     variable 
     $\bla(\in \overline{D(r)}^2)$.
\item\label{item_integrand_integrability}
The $C(\overline{D(r)}^2)$-valued function 
\begin{align*}
(\phi_1,\phi_2)\mapsto &e^{-\frac{\beta L^d}{|U|}|\phi-\g|^2}
 \frac{\prod_{\bk\in\G^*}(\cos(\beta\theta/2)+\cosh(\beta\sqrt{e(\bk)^2+|\phi|^2}))}{\prod_{\bk\in\G^*}(\cos(\beta\theta/2)+\cosh(\beta e(\bk)))}\\
&\cdot \lim_{h\to\infty\atop h\in\frac{2}{\beta}\N}\int e^{-V(\psi)+W(\psi)-A(\psi)}d\mu_{C(\phi)}(\psi) 
\end{align*}
belongs to $L^1(\R^2,C(\overline{D(r)}^2))$.
\item\label{item_grassmann_formulation_2_band}
\begin{align}
&\frac{\Tr e^{-\beta (\sH+i\theta \sS_z+\sF+\sA)}}{\Tr e^{-\beta
 (\sH_0+i\theta \sS_z)}}\label{eq_grassmann_formulation_2_band}\\
&=\frac{\beta L^d}{\pi |U|}\int_{\R^2}d\phi_1d\phi_2
e^{-\frac{\beta L^d}{|U|}|\phi-\g|^2}
 \frac{\prod_{\bk\in\G^*}(\cos(\beta\theta/2)+\cosh(\beta\sqrt{e(\bk)^2+|\phi|^2}))}{\prod_{\bk\in\G^*}(\cos(\beta\theta/2)+\cosh(\beta e(\bk)))}\notag\\
&\quad\cdot \lim_{h\to\infty\atop h\in\frac{2}{\beta}\N}\int
 e^{-V(\psi)+W(\psi)-A(\psi)}d\mu_{C(\phi)}(\psi).\notag\\
 &\frac{\Tr (e^{-\beta (\sH+i\theta \sS_z+\sF)}\sA_j)}{\Tr e^{-\beta
 (\sH_0+i\theta \sS_z)}}\label{eq_grassmann_formulation_2_band_correlation}\\
&=\frac{L^d}{\pi |U|}\int_{\R^2}d\phi_1d\phi_2
e^{-\frac{\beta L^d}{|U|}|\phi-\g|^2}
 \frac{\prod_{\bk\in\G^*}(\cos(\beta\theta/2)+\cosh(\beta\sqrt{e(\bk)^2+|\phi|^2}))}{\prod_{\bk\in\G^*}(\cos(\beta\theta/2)+\cosh(\beta e(\bk)))}\notag\\
&\quad\cdot \lim_{h\to\infty\atop h\in\frac{2}{\beta}\N}\int
 e^{-V(\psi)+W(\psi)}A^j(\psi)d\mu_{C(\phi)}(\psi),\notag\\
&(j=1,2).\notag
\end{align}
\end{enumerate}
\end{lemma}

\begin{remark} At this point we do not prove that we can change the
 order of the integration over $\R^2$ and the limit operation $h\to
 \infty$ in \eqref{eq_grassmann_formulation_2_band},
 \eqref{eq_grassmann_formulation_2_band_correlation}. It suffices to
 establish a suitable uniform bound on 
\begin{align*}
\int
 e^{-V(\psi)+W(\psi)-A(\psi)}d\mu_{C(\phi)}(\psi),\quad
\int
 e^{-V(\psi)+W(\psi)}A^j(\psi)d\mu_{C(\phi)}(\psi)\quad (j=1,2)
\end{align*}
with $\phi$, $h$ 
in order to ensure that these operations are exchangeable. Later we will
 prove the uniform bound \eqref{eq_grassmann_formulations_all_bound} and thus we will be able to
 exchange these operations in \eqref{eq_grassmann_formulation_2_band}
 with $\bla=(0,0)$ and in \eqref{eq_grassmann_formulation_2_band_correlation}.
 It is also possible to use 
 Pedra-Salmhofer's type determinant bound Proposition
 \ref{prop_P_S_bound_application} to directly establish a desirable uniform
 boundedness of these Grassmann integrals.
\end{remark}

\begin{proof}[Proof of Lemma \ref{lem_grassmann_formulation_2_band}]
We decompose the Grassmann polynomial $\sW(\psi)$ in the right-hand side
 of \eqref{eq_H_S_transformation} temporarily by the
 Hubbard-Stratonovich transformation. Set
\begin{align*}
&\sW_+(\psi):=\frac{i|U|^{\frac{1}{2}}}{\beta^{\frac{1}{2}}L^{\frac{d}{2}}h}\sum_{\bx\in\G}\sum_{s\in\0betah}\opsi_{\bx\ua
 s}\opsi_{\bx\da s},\\
&\sW_-(\psi):=\frac{i|U|^{\frac{1}{2}}}{\beta^{\frac{1}{2}}L^{\frac{d}{2}}h}\sum_{\bx\in\G}\sum_{s\in\0betah}\psi_{\bx\da
 s}\psi_{\bx\ua s}.
\end{align*}
For the same reason as the equality
 \eqref{eq_H_S_transformation_partial} holds, the following equality
 holds true.
\begin{align}
(\text{R.H.S of }\eqref{eq_H_S_transformation})
&=\frac{1}{\pi^2}\int_{\R^2}d\phi_1d\phi_2 \int_{\R^2}d\xi_1d\xi_2 e^{-|\phi|^2-|\xi|^2}\label{eq_double_H_S_transformation}\\
&\quad\cdot \int e^{-\sV(\psi)-\sF(\psi)-\sA(\psi)+\phi
 \sV_+(\psi)+\overline{\phi}\sV_-(\psi)+\xi
 \sW_+(\psi)+\overline{\xi}\sW_-(\psi)}d\mu_G(\psi),\notag
\end{align}
where we set $\xi:=\xi_1+i\xi_2$. Let us transform the Grassmann integral
 inside the Gaussian integral. 
By expanding each exponential of the
 Grassmann polynomials and using the determinant bound
 \eqref{eq_full_determinant_bound} we can derive that
\begin{align}
&\left|
\int e^{-\sV(\psi)-\sF(\psi)-\sA(\psi)+\phi
 \sV_+(\psi)+\overline{\phi}\sV_-(\psi)+\xi
 \sW_+(\psi)+\overline{\xi}\sW_-(\psi)}d\mu_G(\psi)\right|\label{eq_decomposed_formulation_bound}\\
&\le \max\left\{
1,\frac{2^{2L^d}e^{\beta \|\sH_0+i\theta \sS_z\|_{\cB(F_f)}}}{|\Tr
 e^{-\beta (\sH_0+i\theta\sS_z)}|}
\right\}\notag\\
&\quad\cdot e^{|U|\beta L^d D^2+2|\g|\beta L^d D+|\la_1|\beta D+|\la_2|\beta
 D^2+2|\phi||U|^{\frac{1}{2}}\beta^{\frac{1}{2}}L^{\frac{d}{2}}D+
2|\xi||U|^{\frac{1}{2}}\beta^{\frac{1}{2}}L^{\frac{d}{2}}D},\notag
\end{align}
where $D:=e^{2\beta \|\sH_0+i\theta \sS_z\|_{\cB(F_f)}}$.
The same argument as in the proof of Lemma
 \ref{lem_grassmann_formulation} proves that for any $r\in\R_{>0}$
\begin{align}
\lim_{h\to\infty\atop h\in \frac{2}{\beta}\N}\sup_{\bla\in
 \overline{D(r)}^2}
\Bigg| &\int e^{-\sV(\psi)-\sF(\psi)-\sA(\psi)+\phi
 \sV_+(\psi)+\overline{\phi}\sV_-(\psi)+\xi
 \sW_+(\psi)+\overline{\xi}\sW_-(\psi)}d\mu_G(\psi)\label{eq_decomposed_time_continuum_limit}\\
&-\frac{\Tr e^{-\beta (\sH+i\theta \sS_z+\sF+\sA-\phi
 \sV_+-\overline{\phi}\sV_--\xi\sW_+-\overline{\xi}\sW_-)}}{\Tr
 e^{-\beta (\sH_0+i\theta \sS_z)}}
\Bigg|=0,\notag
\end{align}
where $\sV_+$, $\sV_-$, $\sW_+$, $\sW_-$ are operators on
 $F_f(L^2(\G\times\spin))$ defined by
\begin{align*}
&\sV_+:=\frac{|U|^{\frac{1}{2}}}{\beta^{\frac{1}{2}}L^{\frac{d}{2}}}\sum_{\bx\in\G}\psi_{\bx\ua}^*\psi_{\bx\da}^*,\quad
\sV_-:=\frac{
|U|^{\frac{1}{2}}}{\beta^{\frac{1}{2}}L^{\frac{d}{2}}}\sum_{\bx\in\G}\psi_{\bx\da}\psi_{\bx\ua},\\
&\sW_+:=i\sV_+,\quad \sW_-:=i\sV_-.
\end{align*}

Here we introduce the band index $\{1,2\}$ and relate the partition
 function in \eqref{eq_decomposed_time_continuum_limit} to a partition
 function in the Fermionic Fock space $F_f(L^2(\{1,2\}\times \G))$. Let
 us give a number to each $\bx\in \G$ so that we can write
 $\G=\{\bx_j\}_{j=1}^{L^d}$. Define the linear map $\cU$ from
 $F_f(L^2(\G\times \spin))$ to $F_f(L^2(\{1,2\}\times \G))$ by
\begin{align*}
&\cU\O:=\prod_{j=1}^{L^d}\psi_{2\bx_j}^*\O_2,\\
&\cU(\psi_{\bx_{i_1}\ua}^*\psi_{\bx_{i_2}\ua}^*\cdots
 \psi_{\bx_{i_l}\ua}^*
\psi_{\bx_{j_1}\da}^*\psi_{\bx_{j_2}\da}^*\cdots \psi_{\bx_{j_m}\da}^*
 \O)\\
&:=\psi_{1\bx_{i_1}}^*\psi_{1\bx_{i_2}}^*\cdots
 \psi_{1\bx_{i_l}}^*
\psi_{2\bx_{j_1}}\psi_{2\bx_{j_2}}\cdots \psi_{2\bx_{j_m}}\cU \O,\\
&(\forall i_1,i_2,\cdots,i_l,j_1,j_2,\cdots,j_m\in \{1,2,\cdots,L^d\})
\end{align*}
and by linearity. Here $\O$, $\O_2$ are the vacuum of 
$F_f(L^2(\G\times \spin))$, $F_f(L^2(\{1,2\}\times\G))$ respectively.
We can see that the map $\cU$ is unitary and 
\begin{align}
\cU\psi_{\bx\ua}^*\cU^*=\psi_{1\bx}^*,\quad  
\cU\psi_{\bx\ua}\cU^*=\psi_{1\bx},\quad 
\cU\psi_{\bx\da}^*\cU^*=\psi_{2\bx},\quad  
\cU\psi_{\bx\da}\cU^*=\psi_{2\bx}^*,\quad (\forall
 \bx\in\G).\label{eq_unitary_creation_annihilation}
\end{align}
Let us define the operators $V$, $A$, $W_+$, $W_-$ on
 $F_f(L^2(\{1,2\}\times \G))$ by
\begin{align*}
&V:=\frac{U}{L^d}\sum_{\bx\in\G}\psi_{1\bx}^*\psi_{1\bx}
-\frac{U}{L^d}\sum_{\bx,\by\in\G}\psi_{1\bx}^*\psi_{2\by}^*\psi_{2\bx}\psi_{1\by},\\
&A:=\la_1\psi_{1\hbx}^*\psi_{2\hbx}-
\la_2\psi_{1\hbx}^*\psi_{2\hby}^*\psi_{2\hbx}\psi_{1\hby},\\
&W_+:=\frac{i|U|^{\frac{1}{2}}}{\beta^{\frac{1}{2}}L^{\frac{d}{2}}}\sum_{\bx\in\G}\psi_{1\bx}^*\psi_{2\bx},\quad
 W_-:=\frac{i|U|^{\frac{1}{2}}}{\beta^{\frac{1}{2}}L^{\frac{d}{2}}}\sum_{\bx\in\G}\psi_{2\bx}^*\psi_{1\bx}.
\end{align*}
We can see from \eqref{eq_definition_one_band_free_hamiltonian}, \eqref{eq_unitary_creation_annihilation} that
\begin{align*}
&\cU(\sH+i\theta
 \sS_z+\sF+\sA-\phi\sV_+-\overline{\phi}\sV_--\xi
 \sW_+-\overline{\xi}\sW_-)\cU^*\\
&=H_0(\phi')+V+A-\xi W_+-\overline{\xi}
 W_-+\sum_{\bk\in \G^*}e(\bk)-i\frac{\theta}{2}L^d,
\end{align*}
where we set $\phi':=\g-|U|^{\frac{1}{2}}\beta^{-\frac{1}{2}}L^{-\frac{d}{2}}\phi$.
Therefore,
\begin{align*}
&\frac{\Tr e^{-\beta (\sH+i\theta \sS_z+\sF+\sA-\phi
 \sV_+-\overline{\phi}\sV_--\xi\sW_+-\overline{\xi}\sW_-)}}{\Tr
 e^{-\beta (\sH_0+i\theta \sS_z)}}\\
&=\frac{e^{-\beta (\sum_{\bk\in \G^*}e(\bk)-i\frac{\theta}{2}L^d)}\Tr e^{-\beta H_0(\phi')}}{\Tr
 e^{-\beta (\sH_0+i\theta \sS_z)}}
\frac{\Tr e^{-\beta (H_0(\phi')+V+A-\xi W_+-\overline{\xi} W_-)}}{\Tr
 e^{-\beta H_0(\phi')}}.
\end{align*}
For conciseness, set
\begin{align*}
B(\phi):=
 \frac{\prod_{\bk\in\G^*}(\cos(\beta\theta/2)+\cosh(\beta\sqrt{e(\bk)^2+|\phi|^2}))}{\prod_{\bk\in\G^*}(\cos(\beta\theta/2)+\cosh(\beta e(\bk)))}.
\end{align*}
Recalling Lemma \ref{lem_free_partition_function} and Lemma
 \ref{lem_partition_function_2_band}, we observe that
\begin{align}
 \frac{\Tr e^{-\beta (\sH+i\theta \sS_z+\sF+\sA-\phi
 \sV_+-\overline{\phi}\sV_--\xi\sW_+-\overline{\xi}\sW_-)}}{\Tr
 e^{-\beta (\sH_0+i\theta \sS_z)}}
=B(\phi')
\frac{\Tr e^{-\beta (H_0(\phi')+V+A-\xi W_+-\overline{\xi} W_-)}}{\Tr
 e^{-\beta H_0(\phi')}}.\label{eq_partition_equality_inside}
\end{align}
The normalized partition function for the 2-band Hamiltonian 
$H_0(\phi')+V+A-\xi W_+-\overline{\xi} W_-$ can be formulated into the
 time-continuum limit of a Grassmann Gaussian integral in $\bigwedge \cV$
 in the same way as the proof of Lemma
 \ref{lem_grassmann_formulation}. 
Here we especially need to make sure that the creation operators are on the left of
 the annihilation operators in $V+A-\xi W_+-\overline{\xi} W_-$.
The result is that for any
 $r\in\R_{>0}$
\begin{align}
\lim_{h\to\infty\atop h\in \frac{2}{\beta}\N}\sup_{\bla\in
 \overline{D(r)}^2}
\Bigg| &\int e^{-V(\psi)-A(\psi)+\xi W_+(\psi)+\overline{\xi}W_-(\psi)}
d\mu_{C(\phi')}(\psi)\label{eq_2_band_formulation_inside}\\
&-\frac{\Tr e^{-\beta (H_0(\phi')+V+A-\xi W_+-\overline{\xi} W_-)}}{\Tr
 e^{-\beta H_0(\phi')}}
\Bigg|=0,\notag
\end{align}
where 
\begin{align*}
W_+(\psi):=\frac{i|U|^{\frac{1}{2}}}{\beta^{\frac{1}{2}}L^{\frac{d}{2}}h}\sum_{\bx\in\G}\sum_{s\in\0betah}\opsi_{1\bx
 s}\psi_{2\bx s},\quad W_-(\psi):=\frac{i|U|^{\frac{1}{2}}}{\beta^{\frac{1}{2}}L^{\frac{d}{2}}h}\sum_{\bx\in\G}\sum_{s\in\0betah}\opsi_{2\bx
 s}\psi_{1\bx s}.
\end{align*}
By considering the original definition
 \eqref{eq_2_band_covariance_original_definition} we can see that $C(\phi)$
 has a determinant bound like
 \eqref{eq_full_determinant_bound}. We can expand each exponential
 of the Grassmann polynomials to derive that
\begin{align}
&\left|
\int e^{-V(\psi)-A(\psi)+\xi
 W_+(\psi)+\overline{\xi}W_-(\psi)}d\mu_{C(\phi')}(\psi)\right|\label{eq_decomposed_partition_function_bound}\\
&\le \max\left\{
1,\frac{2^{2L^d}e^{\beta \|H_0(\phi')\|_{\cB(F_{f,2})}}}{|\Tr
 e^{-\beta H_0(\phi')}|}
\right\} e^{|U|\beta D_2+
|U|\beta L^d D_2^2+|\la_1|\beta  D_2+|\la_2|\beta
  D_2^2+
2|\xi||U|^{\frac{1}{2}}\beta^{\frac{1}{2}}L^{\frac{d}{2}}D_2},\notag
\end{align}
where 
$ \|\cdot \|_{\cB(F_{f,2})}$ denotes the operator norm of operators on
 $F_f(L^2(\{1,2\}\times \G))$ and $D_2:=e^{2\beta \|H_0(\phi')\|_{\cB(F_{f,2})} }$.

Here let us put these pieces together. By
 \eqref{eq_decomposed_formulation_bound},
\eqref{eq_decomposed_time_continuum_limit} and
\eqref{eq_partition_equality_inside} we can apply the dominated
 convergence theorem in $L^1(\R^2,C(\overline{D(r)}^2))$, 
$L^1(\R^4,C(\overline{D(r)}^2))$ to prove that
\begin{align}
&\lim_{h\to\infty\atop h\in \frac{2}{\beta}\N}\sup_{\bla\in
 \overline{D(r)}^2}\label{eq_1_band_Grassmann_2_band_partition}\\
&\cdot\Bigg|\frac{1}{\pi}\int_{\R^2}d\xi_1d\xi_2 e^{-|\xi|^2}
 \int e^{-\sV(\psi)-\sF(\psi)-\sA(\psi)+\phi
 \sV_+(\psi)+\overline{\phi}\sV_-(\psi)+\xi
 \sW_+(\psi)+\overline{\xi}\sW_-(\psi)}d\mu_G(\psi)\notag\\
&\quad-\frac{1}{\pi}\int_{\R^2}d\xi_1d\xi_2 e^{-|\xi|^2}B(\phi')
\frac{\Tr e^{-\beta (H_0(\phi')+V+A-\xi W_+-\overline{\xi} W_-)}}{\Tr
 e^{-\beta H_0(\phi')}}
\Bigg|=0,\quad (\forall \phi\in\C),\notag\\
&\lim_{h\to\infty\atop h\in \frac{2}{\beta}\N}\sup_{\bla\in
 \overline{D(r)}^2}\label{eq_1_band_Grassmann_2_band_partition_double}\\
&\cdot\Bigg|\frac{1}{\pi^2}\int_{\R^2}d\phi_1d\phi_2\int_{\R^2}d\xi_1d\xi_2 e^{-|\phi|^2-|\xi|^2}\notag\\
&\qquad\cdot \int e^{-\sV(\psi)-\sF(\psi)-\sA(\psi)+\phi
 \sV_+(\psi)+\overline{\phi}\sV_-(\psi)+\xi
 \sW_+(\psi)+\overline{\xi}\sW_-(\psi)}d\mu_G(\psi)\notag\\
&\quad-\frac{1}{\pi^2}\int_{\R^2}d\phi_1d\phi_2\int_{\R^2}d\xi_1d\xi_2
 e^{-|\phi|^2-|\xi|^2} B(\phi')
\frac{\Tr e^{-\beta (H_0(\phi')+V+A-\xi W_+-\overline{\xi} W_-)}}{\Tr
 e^{-\beta H_0(\phi')}}\Bigg|=0.\notag
\end{align}
By \eqref{eq_2_band_formulation_inside},
 \eqref{eq_decomposed_partition_function_bound} the dominated
 convergence theorem in $L^1(\R^2,C(\overline{D(r)}^2))$ ensures that
\begin{align*}
\lim_{h\to\infty\atop h\in \frac{2}{\beta}\N}\sup_{\bla\in
 \overline{D(r)}^2}\Bigg|&\frac{1}{\pi}\int_{\R^2}d\xi_1d\xi_2 e^{-|\xi|^2}
 \int e^{-V(\psi)-A(\psi)+\xi
 W_+(\psi)+\overline{\xi}W_-(\psi)}d\mu_{C(\phi')}(\psi)\notag\\
&\quad-\frac{1}{\pi}\int_{\R^2}d\xi_1d\xi_2 e^{-|\xi|^2}
\frac{\Tr e^{-\beta (H_0(\phi')+V+A-\xi W_+-\overline{\xi} W_-)}}{\Tr
 e^{-\beta H_0(\phi')}}
\Bigg|=0,\quad (\forall \phi\in\C),
\end{align*}
or by using the Hubbard-Stratonovich transformation again,
\begin{align}
\lim_{h\to\infty\atop h\in \frac{2}{\beta}\N}\sup_{\bla\in
 \overline{D(r)}^2}\Bigg|&
 \int e^{-V(\psi)+W(\psi)-A(\psi)}d\mu_{C(\phi')}(\psi)\label{eq_2_band_2_band}\\
&\quad-\frac{1}{\pi}\int_{\R^2}d\xi_1d\xi_2 e^{-|\xi|^2}
\frac{\Tr e^{-\beta (H_0(\phi')+V+A-\xi W_+-\overline{\xi} W_-)}}{\Tr
 e^{-\beta H_0(\phi')}}
\Bigg|=0,\quad (\forall \phi\in\C),\notag
\end{align}
which implies the claim \eqref{item_integrand_time_continuum_limit}. 
We can deduce from \eqref{eq_decomposed_formulation_bound},
\eqref{eq_1_band_Grassmann_2_band_partition} that the 
$C(\overline{D(r)}^2)$-valued function 
\begin{align*}
&(\phi_1,\phi_2)\mapsto \frac{e^{-|\phi|^2}}{\pi}\int_{\R^2}d\xi_1d\xi_2
 e^{-|\xi|^2}B(\phi')\frac{\Tr e^{-\beta (H_0(\phi')+V+A-\xi W_+-\overline{\xi} W_-)}}{\Tr e^{-\beta H_0(\phi')}}
\end{align*}
belongs to $L^1(\R^2,C(\overline{D(r)}^2))$. By combining this fact with
 \eqref{eq_2_band_2_band} we see that the $C(\overline{D(r)}^2)$-valued
 function
\begin{align*}
(\phi_1,\phi_2)\mapsto \frac{e^{-|\phi|^2}}{\pi}B(\phi')\lim_{h\to
 \infty\atop h\in \frac{2}{\beta}\N}\int e^{-V(\psi)+W(\psi)-A(\psi)}d\mu_{C(\phi')}(\psi)
\end{align*}
belongs to $L^1(\R^2,C(\overline{D(r)}^2))$. Then, by changing $\phi$ to
 $|U|^{-\frac{1}{2}}\beta^{\frac{1}{2}}L^{\frac{d}{2}}(\g-\phi)$ we see
 that the claim \eqref{item_integrand_integrability} holds. 
Moreover, by \eqref{eq_formulation_partition},
 \eqref{eq_H_S_transformation},
\eqref{eq_double_H_S_transformation},
\eqref{eq_1_band_Grassmann_2_band_partition_double},
\eqref{eq_2_band_2_band} and changing $\phi$ to
 $|U|^{-\frac{1}{2}}\beta^{\frac{1}{2}}L^{\frac{d}{2}}(\g-\phi)$,
\begin{align}\label{eq_grassmann_formulation_2_band_pre}
&\frac{\Tr e^{-\beta (\sH+i\theta\sS_z+\sF+\sA)}}{\Tr
 e^{-\beta (\sH_0+i\theta \sS_z)}}\\
&=
\frac{1}{\pi}\int_{\R^2}d\phi_1d\phi_2
 e^{-|\phi|^2}B(\phi')
\lim_{h\to\infty\atop h\in \frac{2}{\beta}\N}
 \int
 e^{-V(\psi)+W(\psi)-A(\psi)}d\mu_{C(\phi')}(\psi)\notag\\
&=\frac{\beta L^d}{\pi|U|}\int_{\R^2}d\phi_1d\phi_2
 e^{-\frac{\beta L^d}{|U|}|\phi-\g|^2}B(\phi)
\lim_{h\to\infty\atop h\in \frac{2}{\beta}\N}
 \int
 e^{-V(\psi)+W(\psi)-A(\psi)}d\mu_{C(\phi)}(\psi),\notag
\end{align}
which is \eqref{eq_grassmann_formulation_2_band}.
Furthermore, by Cauchy's integral formula, 
\begin{align}
&\frac{\Tr (e^{-\beta (\sH+i\theta\sS_z+\sF)}\sA_j)}{\Tr
 e^{-\beta (\sH_0+i\theta \sS_z)}}=-\frac{1}{2\pi
 i\beta}\oint_{|\la_j|=r}d\la_j\frac{1}{\la_j^2}
\frac{\Tr e^{-\beta (\sH+i\theta\sS_z+\sF+\la_j\sA_j)}}{\Tr
 e^{-\beta (\sH_0+i\theta \sS_z)}}\label{eq_grassmann_formulation_2_band_correlation_pre}\\
&=-\frac{L^d}{\pi|U|}\int_{\R^2}d\phi_1d\phi_2\frac{1}{2\pi
 i}\oint_{|\la_j|=r}d\la_j\frac{1}{\la_j^2}
 e^{-\frac{\beta L^d}{|U|}|\phi-\g|^2}B(\phi)\notag\\
&\qquad\cdot\lim_{h\to\infty\atop h\in \frac{2}{\beta}\N}
\int e^{-V(\psi)+W(\psi)-\la_j A^j(\psi)}d\mu_{C(\phi)}(\psi)\notag\\
&=-\frac{L^d}{\pi|U|}
\int_{\R^2}d\phi_1d\phi_2
  e^{-\frac{\beta L^d}{|U|}|\phi-\g|^2}
B(\phi)\notag\\
&\qquad\cdot \lim_{h\to\infty\atop h\in \frac{2}{\beta}\N}
 \frac{1}{2\pi
 i}\oint_{|\la_j|=r}d\la_j
\frac{1}{\la_j^2}
 \int
 e^{-V(\psi)+W(\psi)-\la_j A^j(\psi)}d\mu_{C(\phi)}(\psi)\notag\\
&=\frac{L^d}{\pi|U|}\int_{\R^2}d\phi_1d\phi_2
  e^{-\frac{\beta L^d}{|U|}|\phi-\g|^2} B(\phi)
\lim_{h\to\infty\atop h\in \frac{2}{\beta}\N}
  \int
 e^{-V(\psi)+W(\psi)}A^j(\psi)d\mu_{C(\phi)}(\psi),\notag
\end{align}
which is \eqref{eq_grassmann_formulation_2_band_correlation}.
Note that the claim \eqref{item_integrand_integrability}
justifies the change of order of the integrals in the 2nd equality. The
 uniform convergence property claimed in  \eqref{item_integrand_time_continuum_limit} justifies the change of
 order of the integral and the limit operation in the 3rd equality.
\end{proof}

\section{Estimation of Grassmann integration}\label{sec_Grassmann_integration}

Thanks to Lemma \ref{lem_grassmann_formulation_2_band}, our objective is
set to analyze the Grassmann integral 
$$
\int e^{-V(\psi)+W(\psi)-A(\psi)}d\mu_{C(\phi)}(\psi)
$$
with $\phi(\in\C)$ fixed. We especially need to find out which term will
remain relevant after taking the limit $h\to\infty$, $L\to \infty$. To
achieve this aim, it is efficient to generalize the problem to some
extent so that we can describe the basic mechanism of convergence
without taking care of a bunch of physical parameters. In the next
section we will substitute the physical parameters into the general
results obtained in this section. 

In \eqref{eq_2_band_Grassmann_V}, \eqref{eq_2_band_Grassmann_W} the
Grassmann polynomials $V(\psi)$, $W(\psi)$ were defined with
the coupling constant $U(\in\R_{<0})$. It is convenient for our analysis
to extend the coupling constant to be a complex parameter. To avoid
confusion, let us use the notation $V(u)(\psi)$, $W(u)(\psi)$ when we
consider a complex parameter $u$ in place of $U$. More precisely, we set for
$u\in\C$
\begin{align*}
&V(u)(\psi):=
\frac{u}{L^dh}\sum_{\bx\in\G}\sum_{s\in
 [0,\beta)_h}\opsi_{1\bx s}\psi_{1 \bx s}+
\frac{u}{L^dh}\sum_{\bx,\by\in\G}\sum_{s\in
 [0,\beta)_h}\opsi_{1\bx s}\psi_{2 \bx s}\opsi_{2\by s}\psi_{1\by s},\\
&W(u)(\psi):=\frac{u}{\beta L^dh^2}\sum_{\bx,\by\in\G}\sum_{s,t\in
 [0,\beta)_h}\opsi_{1\bx s}\psi_{2 \bx s}\opsi_{2\by t}\psi_{1\by t}
\end{align*}
so that $V(U)(\psi)=V(\psi)$, $W(U)(\psi)=W(\psi)$. 

\subsection{Preliminaries}\label{subsec_preliminaries}

In addition to the brief introduction of Grassmann algebra in Subsection
\ref{subsec_Grassmann_algebra} here we need to define more notations,
notational conventions and other tools necessary in the forthcoming
 analysis. We keep using $I_0$ and $I$ defined in Subsection
 \ref{subsec_two_band} as the index sets of Grassmann algebra.
Let us admit that for any set $S$, $n\in\N$ and $\bX$ belonging to the
product set $S^n$, $X_j$ denotes the $j$-th component of $\bX$. 
Thus, $\bX$ is equal to $(X_1,X_2,\cdots,X_n)$. We will use this
notational rule, which helps to shorten formulas, without any additional comment.
In many occasions we will apply this 
rule to the sets $I_0^n$, $I^n$. Size of a Grassmann polynomial can be
measured through norms on its kernel functions. Thus, it is
important to organize various notions concerning kernel functions. For $n\in\N$
let $\S_n$ denote the set of permutations over $\{1,2,\cdots,n\}$. For
$\bX\in I^n$ and $\s\in \S_n$ we let $\bX_{\s}$ denote
$(X_{\s(1)},X_{\s(2)},\cdots, X_{\s(n)})$. For a function $f:I^n\to \C$
we call it anti-symmetric if 
$$
f(\bX)=\sgn(\s)f(\bX_{\s}),\quad (\forall \bX\in I^n,\ \s\in\S_n).
$$
For a function $g:I^m\times I^n\to \C$ we call it bi-anti-symmetric if 
\begin{align*}
g(\bX,\bY)=\sgn(\s)\sgn(\tau)g(\bX_{\s},\bY_{\tau}),\quad (\forall
 (\bX,\bY)\in I^m\times I^n,\ \s\in\S_m,\ \tau\in \S_n).
\end{align*}
For any function $f:I^n\to \C$ $(n\in\N_{\ge 2})$ we define the norms
$\|f\|_{1,\infty}$, $\|f\|_1$ by
\begin{align*}
&\|f\|_{1,\infty}:=\sup_{j\in \{1,2,\cdots,n\}}\sup_{X_0\in
 I}\frah^{n-1}\sum_{\bX\in I^{j-1}}\sum_{\bY\in
 I^{n-j}}|f(\bX,X_0,\bY)|,\\
&\|f\|_1:=\frah^n\sum_{\bX\in I^n}|f(\bX)|.
\end{align*}

Since anti-symmetric functions on $I^2$ play special roles in our
analysis, we need to define other kinds of norm on them. For
this purpose as well as other later use let us introduce a few
notational conventions. 
For $\bX=(\rho_1\bx_1s_1\xi_1,\rho_2\bx_2s_2\xi_2,$
$\cdots,\rho_n\bx_ns_n\xi_n)\in 
(\{1,2\}\times \Z^d\times\frac{1}{h}\Z\times\{1,-1\})^n$,
$s\in\frac{1}{h}\Z$,
we set
\begin{align}
\bX+s:=(\rho_1\bx_1(s_1+s)\xi_1,\rho_2\bx_2(s_2+s)\xi_2,
\cdots,\rho_n\bx_n(s_n+s)\xi_n).\label{eq_time_translation_notational_convention}
\end{align}
Similarly for $\bX=(\rho_1\bx_1s_1,\cdots,\rho_n\bx_ns_n)\in 
(\{1,2\}\times \Z^d\times\frac{1}{h}\Z)^n$,
$s\in\frac{1}{h}\Z$
we set
$\bX+s:=(\rho_1\bx_1(s_1+s),\cdots,\rho_n\bx_n(s_n+s))$.
Define the index set $I^0$ by 
$$
I^0:=\{1,2\}\times\G\times\{0\}\times\{1,-1\}.
$$
It follows that for $\bX\in (I^0)^n$, $s\in [0,\beta)_h$, $\bX+s\in I^n$.
With these notational rules we define the norms 
$\|\cdot\|_{1,\infty}'$, $\|\cdot\|$ on anti-symmetric functions on
$I^2$ as follows. For any anti-symmetric function $g:I^2\to\C$,
\begin{align*}
\|g\|_{1,\infty}':=\sup_{X_0\in I\atop s\in \0betah}\sum_{X\in
 I^0}|g(X_0,X+s)|,\quad
\|g\|:=\|g\|_{1,\infty}'+\beta^{-1}\|g\|_{1,\infty}.
\end{align*}

We will also deal with bi-anti-symmetric functions on the product set $I^m\times I^n$. By
considering that these functions are defined on $I^{m+n}$ the norms
$\|\cdot\|_{1,\infty}$, $\|\cdot\|_1$ can be defined on them. We will
need to measure these functions coupled with another anti-symmetric
function on $I^2$. The measurement will be carried out in terms of the
following quantities. For a bi-anti-symmetric function $f_{m,n}:I^m\times I^n\to \C$
$(m,n\in \N_{\ge 2})$ and an anti-symmetric function $g:I^2\to \C$ we
set
\begin{align*}
&[f_{m,n},g]_{1,\infty}\\
&:=\max\Bigg\{\sup_{X_0\in I}\frah^{m-1}\sum_{\bX\in
 I^{m-1}}
\Bigg\{\sup_{Y_0\in I}\frah^n\sum_{\bY\in
 I^n}|f_{m,n}((X_0,\bX),\bY)||g(Y_0,Y_1)|\Bigg\},\\
&\qquad\qquad \sup_{Y_0\in I}\frah^{n-1}\sum_{\bY\in
 I^{n-1}}
\Bigg\{\sup_{X_0\in I}\frah^{m}\sum_{\bX\in
 I^m}|f_{m,n}(\bX,(Y_0,\bY))||g(X_0,X_1)|\Bigg\}\Bigg\},\\
&[f_{m,n},g]_{1}:=\frah^{m+n}\sum_{\bX\in I^m\atop \bY\in
 I^n}|f_{m,n}(\bX,\bY)||g(X_1,Y_1)|.
\end{align*}

For $\bX\in I^n$ let $\psi_{\bX}$ denote $\psi_{X_1}\psi_{X_2}\cdots\psi_{X_n}$.
Anti-symmetry of Grassmann variables implies that for any
$f(\psi)\in\bigwedge\cV$ there uniquely exist $f_0\in \C$ and
anti-symmetric functions $f_n:I^n\to\C$ $(n=1,2,\cdots,N)$ such that
$$
f(\psi)=\sum_{n=0}^N\frah^n\sum_{\bX\in I^n}f_n(\bX)\psi_{\bX}.
$$
Based on this fact we admit that for $f(\psi)\in\bigwedge\cV$, $f_n$
$(n=0,1,\cdots,N)$ denote the unique anti-symmetric kernels of
$f(\psi)$. A norm can be defined in the vector space $\bigwedge
\cV$ by defining a norm in every space of
anti-symmetric kernels. Finite dimensionality of $\bigwedge \cV$ implies
that $\bigwedge \cV$ is a Banach space with the norm. Then, by
considering as a Banach-space-valued function
the standard notions
such as continuity, differentiability and analyticity of a Grassmann
polynomial parameterized by real or complex variables are defined. 
Since we introduce various norms on anti-symmetric kernels, it
is clearer to define these notions without specifying a norm on
$\bigwedge \cV$. 
We say that a sequence of elements of $\bigwedge \cV$, $f^m(\psi)$
($m=1,2,\cdots$) converges in $\bigwedge \cV$ if each anti-symmetric
kernel function of $f^m(\psi)$ converges point-wise, or more precisely
 $\lim_{m\to\infty}f^m_n(\bX)$ converges in $\C$ for any $n\in
 \{0,1,\cdots,N\}$, $\bX\in I^n$. 
For a domain $O$ of
$\R^m$ or $\C^m$ and $f(\bz)(\psi)\in \bigwedge \cV$ parameterized by
$\bz\in \overline{O}$ we say that $f(\bz)(\psi)$ is continuous with
$\bz$ in $\overline{O}$, differentiable with $\bz$ in $O$ and analytic
with $\bz$ in $O$ if so is $f(\bz)_n(\bX)$ for any $n\in
 \{0,1,\cdots,N\}$, $\bX\in I^n$. Moreover, when it is differentiable, for
$j\in\{1,2,\cdots,m\},\bz=(z_1,z_2,\cdots,z_m)\in O$ we define the
Grassmann polynomial $\frac{\partial}{\partial z_j}f(\bz)(\psi)\in
\bigwedge \cV$ by
\begin{align*}
\frac{\partial}{\partial
 z_j}f(\bz)(\psi):=\sum_{n=0}^N\frah^n\sum_{\bX\in I^n}
\frac{\partial}{\partial z_j}f(\bz)_n(\bX)\psi_{\bX}.
\end{align*}

The single-scale integration is well-described in terms of trees. 
We refer to the clear statement of the tree formula
 with a self-contained proof presented in \\
\cite[\mbox{Theorem 3,\ Appendix A}]{SW}. 
We should also lead the readers to the references of \cite{SW} for more
original versions of such expansion techniques known as the
Brydges-Battle-Federbush formula. To state the formula, we need to recall the definition of Grassmann
left-derivatives. Let $\cV^j$ be the vector
space spanned by the basis $\{\psi_X^j\ |\ X\in I\}$ for $j=1,2,\cdots,n$.
For $p\in \{1,\cdots,n\}$, $X\in I$ the Grassmann
left-derivative $\partial/\partial \psi_{X}^p$ is a linear transform on
$\bigwedge(\cV^1\oplus\cdots\oplus \cV^n)$ defined by
\begin{align*}
&\frac{\partial}{\partial
 \psi_{X}^p}(\psi_{X_1}^{p_1}\cdots\psi_{X_j}^{p_j}
\psi_X^p\psi_{X_{j+1}}^{p_{j+1}}\cdots \psi_{X_m}^{p_m}):=(-1)^j
\psi_{X_1}^{p_1}\cdots\psi_{X_j}^{p_j}
\psi_{X_{j+1}}^{p_{j+1}}\cdots \psi_{X_m}^{p_m},\\
&\frac{\partial}{\partial
 \psi_{X}^p}(\psi_{X_1}^{p_1}\cdots\psi_{X_j}^{p_j}
\psi_{X_{j+1}}^{p_{j+1}}\cdots \psi_{X_m}^{p_m}):=0
\end{align*}
for any $(p_j,X_j)\in \{1,2,\cdots,n\}\times I$ satisfying
$(p_j,X_j)\neq (p,X)$ $(j=1,2,\cdots,m)$ and by
linearity. The Grassmann left-derivative $\partial/\partial \psi_{X}$
$(X\in I)$ can be defined as a linear transform on $\bigwedge(\cV\oplus
\cV^1\oplus\cdots \oplus \cV^n)$ in the same way.

Let $\bigwedge_{even}\cV$ denote a subspace of $\bigwedge \cV$
consisting of even polynomials. More precisely, 
$$
\bigwedge_{even}\cV:=\bigoplus_{n=0}^{N/2}\bigwedge^{2n}\cV.
$$
For a covariance $\cC:I_0^2\to\C$ and $f^j(\psi)\in \bigwedge_{even}\cV$
 $(j=1,2,\cdots,n)$
the Grassmann polynomial
$$
\log\left(\int e^{\sum_{j=1}^nz_jf^j(\psi+\psi^1)}d\mu_{\cC}(\psi^1)
\right)
$$
is well-defined and analytic with $(z_1,z_2,\cdots,z_n)$ in a neighborhood of the origin.
The strategy of single-scale analysis is to expand the logarithm into the
Taylor series around the origin and estimate each order term. 
To this end we need to know a formula for 
$$
\frac{1}{n!}\prod_{j=1}^n\left(\frac{\partial}{\partial z_j}\right)
\log\left(\int e^{\sum_{j=1}^nz_jf^j(\psi+\psi^1)}d\mu_{\cC}(\psi^1)
\right)\Bigg|_{z_j=0\atop (\forall j\in\{1,2,\cdots,n\})}.
$$
The formula for $n=1$ can be derived from the definition as follows.
\begin{align*}
\frac{d}{dz}\log\left(\int
 e^{zf^1(\psi+\psi^1)}d\mu_{\cC}(\psi^1)\right)\Bigg|_{z=0}&=\int f^1(\psi+\psi^1)d\mu_{\cC}(\psi^1)\\
&=e^{-\sum_{\bX\in I_0^2}\cC(\bX)\frac{\partial}{\partial
 \opsi_{X_1}^1}\frac{\partial}{\partial
 \psi_{X_2}^1}}f^1(\psi+\psi^1)\Big|_{\psi^1=0}.
\end{align*}
The tree formula characterizes the derivative for $n\in\N_{\ge 2}$.
\begin{align*}
&\frac{1}{n!}\prod_{j=1}^n\left(\frac{\partial}{\partial z_j}\right)
\log\left(\int e^{\sum_{j=1}^nz_jf^j(\psi+\psi^1)}d\mu_{\cC}(\psi^1)
\right)\Bigg|_{z_j=0\atop (\forall j\in\{1,2,\cdots,n\})}\\
&=\frac{1}{n!}\sum_{T\in \T(\{1,2,\cdots,n\})}\prod_{\{p,q\}\in
 T}(\D_{p,q}(\cC)+\D_{q,p}(\cC))\\
&\quad\cdot\int_{[0,1]^{n-1}}d\bs \sum_{\s\in
 \S_n(T)}\varphi(T,\s,\bs)
e^{\sum_{a,b=1}^nM(T,\s,\bs)_{a,b}\D_{a,b}(\cC)}
\prod_{j=1}^nf^j(\psi^j+\psi)\Bigg|_{\psi^j=0\atop(\forall
 j\in\{1,2,\cdots,n\})},
\end{align*}
where $\T(\{1,2,\cdots,n\})$ is the set of all trees over the vertices
$\{1,2,\cdots,n\}$, 
$$
\D_{p,q}(\cC):=-\sum_{\bX\in I^2_0}\cC(\bX)\frac{\partial}{\partial
\opsi_{X_1}^p}\frac{\partial}{\partial \psi_{X_2}^q},
$$
$\S_n(T)$ is a $T$-dependent subset of $\S_n$, $\varphi(T,\s,\cdot)$ is
a real non-negative function on $[0,1]^{n-1}$ depending on $T\in
\T(\{1,2,\cdots,n\})$, $\s\in \S_n(T)$ and 
$(M(T,\s,\bs)_{a,b})_{1\le a,b\le n}$ is a $(T,\s,\bs)$-dependent real
symmetric non-negative matrix which satisfies that 
\begin{align*}
&M(T,\s,\bs)_{a,a}=1,\\
&(\forall a\in \{1,\cdots,n\},\ T\in \T(\{1,\cdots,n\}),\ \s\in
\S_n(T),\ \bs\in [0,1]^{n-1}).\\
&\bs\mapsto M(T,\s,\bs)_{a,b}\text{ is continuous in }[0,1]^{n-1},\\
&(\forall a,b\in \{1,\cdots,n\},\ T\in \T(\{1,\cdots,n\}),\ \s\in
\S_n(T)).
\end{align*}
Moreover, the function $\varphi(T,\s,\cdot)$ satisfies that
\begin{align}
\int_{[0,1]^{n-1}}d\bs\sum_{\s\in \S_n(T)}\varphi(T,\s,\bs)=1,\quad
 (\forall T\in \T(\{1,2,\cdots,n\})).\label{eq_phi_good_property}
\end{align}

Because of the property \eqref{eq_phi_good_property} the function
$\varphi(T,\s,\cdot)$ does not affect our estimation of Grassmann
polynomials in practice. We can deduce from the fact that the matrix
$M(T,\s,\bs)$ is real symmetric non-negative and all the diagonal
elements are 1 that there are $\bv_1,\cdots,\bv_n\in\R^n$ such
that $\|\bv_i\|_{\R^n}=1$ $(i=1,\cdots,n)$ and
\begin{align}
M(T,\s,\bs)_{a,b}=\<\bv_a,\bv_b\>,\quad(\forall a,b\in \{1,\cdots,n\}).
\label{eq_gram_representation_primitive}
\end{align}
Thus,
\begin{align}
&|M(T,\s,\bs)_{a,b}|\le 1,\label{eq_exponent_matrix_bound}\\
&(\forall a,b\in \{1,\cdots,n\},\ T\in \T(\{1,\cdots,n\}),\ \s\in
 \S_n(T),\ \bs\in [0,1]^{n-1}).\notag
\end{align}

To systematize our estimation, let us define operators on
Grassmann algebras which are slight generalization of the above formulas.
For $p,q\in \Z$, set 
$$
\D_{\{p,q\}}(\cC):=\sum_{\bX\in
I^2}\tilde{\cC}(\bX)\frac{\partial}{\partial \psi_{X_1}^p}\frac{\partial}{\partial \psi_{X_2}^q},
$$
where $\tilde{\cC}:I^2\to\C$ is the anti-symmetric extension of $\cC$
defined by
\begin{align}
&\tilde{\cC}((X,\xi),(Y,\zeta)):=\frac{1}{2}(1_{(\xi,\zeta)=(1,-1)}\cC(X,Y)-1_{(\xi,\zeta)=(-1,1)}\cC(Y,X)),\label{eq_anti_symmetric_extension}\\
&(\forall X,Y\in I_0,\ \xi,\zeta\in \{1,-1\}).\notag
\end{align}
We can see that 
\begin{align*}
&-2\D_{\{p,q\}}(\cC)=\D_{p,q}(\cC)+\D_{q,p}(\cC),\\
 &-\sum_{p,q=1}^nM(T,\s,\bs)_{p,q}\D_{\{p,q\}}(\cC)=
\sum_{p,q=1}^nM(T,\s,\bs)_{p,q}\D_{p,q}(\cC).
\end{align*}
For $S=\{s_1,s_2,\cdots,s_n\}(\subset \Z)$ with $\sharp S=n\ge 2$ let
$\T(S)$ denote the set of all trees over the vertices
$\{s_1,s_2,\cdots,s_n\}$. Using these notations we set 
for $S=\{s_1,s_2,\cdots,s_n\}(\subset\Z)$, if $n=1$,
\begin{align*}
&Tree(S,\cC):=e^{\D_{\{s_1,s_1\}(\cC)}},
\end{align*}
if $n\ge 2$,
\begin{align*}
&Tree(S,\cC):=(-2)^{n-1}\sum_{T\in \T(S)}\prod_{\{p,q\}\in T}\D_{\{p,q\}}(\cC)
\int_{[0,1]^{n-1}}d\bs \sum_{\s\in
 \S_n(T)}\varphi(T,\s,\bs)\\
&\qquad\qquad\qquad\quad\cdot e^{-\sum_{a,b=1}^nM(T,\s,\bs)_{a,b}\D_{\{s_a,s_b\}}(\cC)}.
\end{align*}
It follows that for any $n\in \N$, $f^j(\psi)\in \bigwedge_{even}\cV$
$(j=1,2,\cdots,n)$,
\begin{align}
&\frac{1}{n!}\prod_{j=1}^n\left(\frac{\partial}{\partial z_j}\right)\log\left(\int
 e^{\sum_{j=1}^nz_jf^j(\psi+\psi^1)}d\mu_{\cC}(\psi^1)\right)\Bigg|_{z_j=0\atop(\forall
 j\in\{1,2,\cdots,n\})}\label{eq_tree_formula_basic}\\
&=\frac{1}{n!}Tree(\{1,2,\cdots,n\},\cC)\prod_{j=1}^nf^j(\psi+\psi^j)
\Bigg|_{\psi^j=0\atop(\forall
 j\in\{1,2,\cdots,n\})}.\notag
\end{align}
When $f^j(\psi)=f(\psi)$ for any $j\in \{1,2,\cdots,n\}$, the formula
\eqref{eq_tree_formula_basic} implies that
\begin{align}
&\frac{1}{n!}\left(\frac{d}{dz}\right)^n\log\left(\int
 e^{zf(\psi+\psi^1)}d\mu_{\cC}(\psi^1)\right)\Bigg|_{z=0}\label{eq_tree_formula_simple}\\
&=\frac{1}{n!}Tree(\{1,2,\cdots,n\},\cC)\prod_{j=1}^nf(\psi+\psi^j)
\Bigg|_{\psi^j=0\atop(\forall
 j\in\{1,2,\cdots,n\})}.\notag
\end{align}

The following inequalities will be crucially important.
\begin{align}
&\left|e^{\D_{\{s_1,s_1\}}(\cC)}\psi_{\bX}^{s_1}\big|_{\psi^{s_1}=0}
\right|
\le \left\{\begin{array}{ll} 0 &\text{if $m$ is odd,}\\
                             \displaystyle\sup_{Y_j,Z_j\in I_0\atop
			      (j=1,2,\cdots,\frac{m}{2})}\left|
\det(\cC(Y_i,Z_j))_{1\le i,j\le \frac{m}{2}}
\right|   &\text{if $m$ is even,}
\end{array}
\right.\label{eq_free_integration_bound_general}\\
&(\forall m\in\N,\ \bX\in I^m).\notag\\
&\left|
e^{-\sum_{a,b=1}^nM(T,\s,\bs)_{a,b}\D_{\{s_a,s_b\}}(\cC)}\prod_{j=1}^n\psi_{\bX_j}^{s_j}
\Bigg|_{\psi^{s_j}=0\atop(\forall
 j\in\{1,2,\cdots,n\})}
\right|\label{eq_determinant_bound_general}\\
&\le  \left\{\begin{array}{ll} 0 &\text{if $m$ is odd,}\\
\displaystyle \sup_{\bu_j,\bv_j\in \C^n\text{ with
 }\|\bu_j\|_{\C^n},\|\bv_j\|_{\C^n}\le 1\atop
(j=1,\cdots,\frac{m}{2})}                             
\sup_{Y_j,Z_j\in I_0\atop
			      (j=1,\cdots,\frac{m}{2})} & \\
\qquad\cdot\left|
\det(\<\bu_i,\bv_j\>_{\C^n}\cC(Y_i,Z_j))_{1\le i,j\le \frac{m}{2}}
\right|   &\text{if $m$ is even,}
\end{array}
\right.\notag\\
&(\forall m_j\in \{0,1,\cdots,N\},\ \bX_j\in I^{m_j}\
 (j=1,2,\cdots,n),\notag\\
&\quad  T\in\T(\{s_1,s_2,\cdots,s_n\}),\ \s\in \S_n(T),\ \bs\in [0,1]^{n-1}),\notag
\end{align}
where $m:=\sum_{j=1}^nm_j$ and $\<\cdot,\cdot\>_{\C^n}$ is the hermitian
inner product and $\|\cdot\|_{\C^n}$ is the norm induced by
$\<\cdot,\cdot\>_{\C^n}$. The inequality
\eqref{eq_determinant_bound_general} is based on the Gram representation \eqref{eq_gram_representation_primitive} of the
matrix $M(T,\s,\bs)$. See e.g. \cite[\mbox{Lemma 4.5}]{K9} for details
of how to derive an inequality of this kind.

\subsection{General estimation}\label{subsec_general_estimation}

Here we estimate Grassmann polynomials produced by applying the operator\\
$Tree(S,\cC)$ to given Grassmann polynomials. As explained in the
beginning of the section our purpose here is to summarize generic
structures of single-scale integrations. Let us introduce some notions 
which are necessary to describe properties of the Grassmann input and the
covariances
 of the single-scale
integrations. To describe periodicity and translation invariance with the time
variable, we define the map $r_{\beta}:\frac{1}{h}\Z\to [0,\beta)_{h}$
by the condition that $r_{\beta}(s)\in \0betah$ and $r_{\beta}(s)=s$ in
$\frac{1}{h}\Z/\beta \Z$ for $s\in \frac{1}{h}\Z$. Then we define
the map $\cR_{\beta}$ from $(\{1,2\}\times\G\times
\frac{1}{h}\Z\times\{1,-1\})^n$ to $I^n$ by
$$\cR_{\beta}(\rho_1\bx_1s_1\xi_1,\cdots,\rho_n\bx_ns_n\xi_n):=
(\rho_1\bx_1r_{\beta}(s_1)\xi_1,\cdots,\rho_n\bx_nr_{\beta}(s_n)\xi_n).$$
We will sometimes consider $\cR_{\beta}$ as the map from 
$(\{1,2\}\times\G\times
\frac{1}{h}\Z)^n$ to $I_0^n$ satisfying that
$$
\cR_{\beta}(\rho_1\bx_1s_1,\cdots,\rho_n\bx_ns_n)= (\rho_1\bx_1r_{\beta}(s_1),\cdots,\rho_n\bx_nr_{\beta}(s_n)),
$$
by admitting the notational abuse. The meaning of the map
$\cR_{\beta}$ should be understood from the context. 

We assume that the covariance $\cC:I_0^2\to \C$ satisfies that
\begin{align}
&\cC(\cR_{\beta}(\bX+s))=\cC(\bX),\quad \left(\forall \bX\in I_0^2,\
 s\in
 \frac{1}{h}\Z\right),\label{eq_time_translation_generic_covariance}\\
&|\det(\<\bu_i,\bv_j\>_{\C^m}\cC(X_i,Y_j))_{1\le i,j\le n}|\le
 D^n,\label{eq_determinant_bound_generic}\\
&(\forall m,n\in\N,\ \bu_i,\bv_i\in\C^m\text{ with
 }\|\bu_i\|_{\C^m},\|\bv_i\|_{\C^m}\le 1,\ X_i,Y_i\in I_0\
 (i=1,2,\cdots,n)),\notag
\end{align}
where $D$ is a fixed positive constant. The condition
\eqref{eq_time_translation_generic_covariance} might appear unnatural if
$\cC$ is thought to be a sum over the Matsubara frequency. However, one
can modify such a covariance to satisfy
\eqref{eq_time_translation_generic_covariance} by a simple gauge transform.

One implication of the property
\eqref{eq_time_translation_generic_covariance} is that
\begin{align}
&Tree(\{1,2,\cdots,n\},\cC)\prod_{j=1}^n\psi_{\cR_{\beta}(\bX_j+s)}^j\Bigg|_{\psi^j=0\atop(\forall
 j\in\{1,2,\cdots,n\})}\label{eq_time_translation_tree_formula}\\
&=
Tree(\{1,2,\cdots,n\},\cC)\prod_{j=1}^n\psi_{\bX_j}^j\Bigg|_{\psi^j=0\atop(\forall
 j\in\{1,2,\cdots,n\})},\notag\\
&\left(\forall n\in\N,\ m_j\in \{0,1,\cdots,N\},\ \bX_j\in I^{m_j}\
 (j=1,2,\cdots,n),\ s\in\frac{1}{h}\Z\right).\notag
\end{align}

For $j\in \N$ let $F^j(\psi)\in \bigwedge_{even}\cV$ be such that its
anti-symmetric kernels $F_m^j:I^m\to \C$ $(m=2,4,\cdots,N)$ satisfy
\begin{align}
F_m^j(\cR_{\beta}(\bX+s))=F_m^j(\bX),\quad \left(\forall \bX\in I^m,\
 s\in\frac{1}{h}\Z\right).\label{eq_time_translation_generic}
\end{align}
In this subsection we will give the Grassmann polynomials $F^j(\psi)$
$(j\in \N)$ as the input to the single-scale integrations. 

For $n\in\N$ we define $A^{(n)}(\psi)\in \bigwedge_{even}\cV$ by
\begin{align*}
A^{(n)}(\psi):=Tree(\{1,2,\cdots,n\},\cC)\prod_{j=1}^nF^j(\psi^j+\psi)\Bigg|_{\psi^j=0\atop(\forall
 j\in\{1,2,\cdots,n\})}.
\end{align*}
Since $Tree(S,\cC)$ consists of Grassmann left-derivatives of even
degree, it is clear that the output belongs to $\bigwedge_{even}\cV$ if
so does the input.

For conciseness of formulas we let $\|f_0\|_{1,\infty}=\|f_0\|_1:=|f_0|$
for the constant term $f_0$ of $f(\psi)\in\bigwedge\cV$. We admit this
notational convention throughout this section. The next lemma is the
simplest among other lemmas in this subsection.

\begin{lemma}\label{lem_tree_bound}
For any $m\in\{2,4,\cdots,N\}$, $n\in \N$ the anti-symmetric kernel
 $A_m^{(n)}(\cdot)$ satisfies
 \eqref{eq_time_translation_generic}. Moreover the following
 inequalities hold for any $m\in \{0,2,\cdots, N\}$, $n\in \N_{\ge 2}$.
\begin{align}
&\|A_m^{(1)}\|_{1,\infty}\le 
\sum_{p=m}^{N}\left(\frac{N}{h}\right)^{1_{m=0\land p\neq 0}}\left(\begin{array}{c}p\\ m\end{array}\right)
D^{\frac{p-m}{2}}\|F_p^1\|_{1,\infty}.\label{eq_tree_1_1_infinity}\\
&\|A_m^{(1)}\|_{1}\le  \sum_{p=m}^{N}\left(\begin{array}{c}p\\ m\end{array}\right)
D^{\frac{p-m}{2}}\|F_p^1\|_{1}.\label{eq_tree_1_1}\\
&\|A_m^{(n)}\|_{1,\infty}\le \left(\frac{N}{h}\right)^{1_{m=0}}
(n-2)!D^{-n+1-\frac{m}{2}}2^{-2m}\|\tilde{\cC}\|_{1,\infty}^{n-1}\label{eq_tree_1_infinity}\\
&\qquad\qquad\quad\cdot\prod_{j=1}^n\left(\sum_{p_j=2}^N2^{3p_j}D^{\frac{p_j}{2}}\|F_{p_j}^j\|_{1,\infty}
\right)1_{\sum_{j=1}^np_j-2(n-1)\ge m}.\notag\\
&\|A_m^{(n)}\|_{1}\le (n-2)!
 D^{-n+1-\frac{m}{2}}2^{-2m}\|\tilde{\cC}\|_{1,\infty}^{n-1}
\sum_{p_1=2}^N2^{3p_1}D^{\frac{p_1}{2}}\|F_{p_1}^1\|_{1}\label{eq_tree_1}\\
&\qquad\qquad\cdot\prod_{j=2}^n\left(\sum_{p_j=2}^N2^{3p_j}D^{\frac{p_j}{2}}\|F_{p_j}^j\|_{1,\infty}
\right)1_{\sum_{j=1}^np_j-2(n-1)\ge m}.\notag
\end{align}
Here $\tilde{\cC}(:I^2\to \C)$ is the anti-symmetric extension of $\cC$
 defined as in \eqref{eq_anti_symmetric_extension}.
\end{lemma}
\begin{proof}
By anti-symmetry, 
\begin{align*}
A^{(1)}(\psi)=\sum_{p=0}^N\sum_{m=0}^p
\left(\begin{array}{c}p\\ m\end{array}\right)
\frah^p\sum_{\bX\in I^m\atop \bY\in
 I^{p-m}}F_p^1(\bY,\bX)Tree(\{1\},\cC)\psi_{\bY}^1\Big|_{\psi^1=0}\psi_{\bX}.
\end{align*}
Thus, by the uniqueness of anti-symmetric kernels, for any
 $m\in\{0,2,\cdots,N\}$, $\bX\in I^m$,
\begin{align}
A^{(1)}_m(\bX)=\sum_{p=m}^N
\left(\begin{array}{c}p\\ m\end{array}\right)
\frah^{p-m}\sum_{\bY\in
 I^{p-m}}F_p^1(\bY,\bX)Tree(\{1\},\cC)\psi_{\bY}^1\Big|_{\psi^1=0}.\label{eq_tree_1_1_kernel}
\end{align}
By using the translation invariant properties
\eqref{eq_time_translation_tree_formula},
 \eqref{eq_time_translation_generic} we see that for any $\bX\in I^m$,
$s\in \frac{1}{h}\Z$
\begin{align*}
&A^{(1)}_m(\cR_{\beta}(\bX+s))\\
&=\sum_{p=m}^N
\left(\begin{array}{c}p\\ m\end{array}\right)
\frah^{p-m}\sum_{\bY\in
 I^{p-m}}F_p^1(\bY,\cR_{\beta}(\bX+s))Tree(\{1\},\cC)\psi_{\bY}^1\Big|_{\psi^1=0}\\
&=\sum_{p=m}^N
\left(\begin{array}{c}p\\ m\end{array}\right)
\frah^{p-m}\sum_{\bY\in
 I^{p-m}}F_p^1(\cR_{\beta}((\bY,\bX)+s))Tree(\{1\},\cC)\psi_{\cR_{\beta}(\bY+s)}^1\Big|_{\psi^1=0}\\
&=A^{(1)}_m(\bX).
\end{align*}
Thus, $A_m^{(1)}$ 
 satisfies \eqref{eq_time_translation_generic}.
The inequalities
 \eqref{eq_tree_1_1_infinity}, \eqref{eq_tree_1_1} can be derived from
 \eqref{eq_tree_1_1_kernel} by using \eqref{eq_free_integration_bound_general},
 \eqref{eq_determinant_bound_generic}.

Next let us consider the case $n\ge 2$. By anti-symmetry, 
\begin{align*}
A^{(n)}(\psi)
=&\prod_{j=1}^n\Bigg(\sum_{p_j=2}^N\sum_{m_j=0}^{p_j-1}\left(\begin{array}{c}p_j\\
						       m_j\end{array}
\right)\frah^{p_j}\sum_{\bX_j\in I^{m_j}\atop \bY_j\in I^{p_j-m_j}}F_{p_j}^j(\bY_j,\bX_j)\Bigg)\\
&\cdot Tree(\{1,2,\cdots,n\},\cC)\prod_{j=1}^n\psi_{\bY_j}^j
\Bigg|_{\psi^j=0\atop(\forall
 j\in\{1,2,\cdots,n\})}
(-1)^{\sum_{j=1}^{n-1}m_j\sum_{k=j+1}^n(p_k-m_k)}\\
&\cdot \prod_{k=1}^n\psi_{\bX_k}.
\end{align*}
Thus, for $m\in \{0,2,\cdots,N\}$, $\bX\in I^m$,
\begin{align}
A_m^{(n)}(\bX)&=\frac{1}{m!}
\sum_{\s\in \S_m}\sgn(\s)
\prod_{j=1}^n\Bigg(\sum_{p_j=2}^N\sum_{m_j=0}^{p_j-1}\left(\begin{array}{c}p_j\\
						       m_j\end{array}
\right)\frah^{p_j-m_j}\sum_{\bY_j\in I^{p_j-m_j}}F_{p_j}^j(\bY_j,\bX_j')\Bigg)\label{eq_tree_1_higher_kernel}\\
&\quad\cdot Tree(\{1,2,\cdots,n\},\cC)\prod_{j=1}^n\psi_{\bY_j}^j
\Bigg|_{\psi^j=0\atop(\forall
 j\in\{1,2,\cdots,n\})}
(-1)^{\sum_{j=1}^{n-1}m_j\sum_{k=j+1}^n(p_k-m_k)}\notag\\
&\quad\cdot
 1_{(\bX_1',\bX_2',\cdots,\bX_n')=\bX_{\s}}1_{\sum_{j=1}^nm_j=m}1_{\sum_{j=1}^np_j-2(n-1)\ge m}.\notag
\end{align}
The constraint $\sum_{j=1}^np_j-2(n-1)\ge m$ is added since the operator
 $\prod_{\{p,q\}\in T}\D_{\{p,q\}}(\cC)$ inside
 $Tree(\{1,2,\cdots,n\},\cC)$ erases $2(n-1)$ Grassmann variables. Again
 by using \eqref{eq_time_translation_tree_formula},
 \eqref{eq_time_translation_generic} we can check that
 $A_m^{(n)}:I^m\to\C$ satisfies \eqref{eq_time_translation_generic}.

To establish upper bounds on the norms of $A_m^{(n)}$ we need to replace the sum
 over trees by a sum over possible degrees of trees. For $j\in \{1,2,\cdots,n\}$,
 $T\in \T(\{1,2,\cdots,n\})$ let $d_j(T)$ denote the degree of the
 vertex $j$ in $T$. The following calculation, which we will frequently
 refer to during this subsection, is based on Cayley's theorem on the number
 of trees with fixed degrees. For any $k_1,k_2,\cdots,k_n\in \N$,
\begin{align}
&\sum_{T\in
 \T(\{1,2,\cdots,n\})}\prod_{j=1}^n\left(\left(\begin{array}{c}k_j\\
					       d_j(T)\end{array}\right)d_j(T)!
\right)\label{eq_number_trees_pre}\\
&=\prod_{j=1}^n \left(\sum_{d_j=1}^{k_j}
\left(\begin{array}{c}k_j\\ d_j\end{array}\right)d_j!
\right)
\frac{(n-2)!}{\prod_{k=1}^n(d_k-1)!}1_{\sum_{j=1}^nd_j=2(n-1)}\notag\\
&\le(n-2)!\prod_{j=1}^n(k_j2^{k_j-1})\le
 (n-2)!2^{-n}2^{2\sum_{j=1}^{n}k_j}.\notag
\end{align}
Since we will need to deal with the case $n=1$ at the same time, let us 
give a meaning to the left-hand side of
 \eqref{eq_number_trees_pre} for $n=1$ and generalize the combinatorial
 estimate \eqref{eq_number_trees_pre} to be valid for any $n\in\N$. We
 assume that $\T(\{1\})=\{\{1\}\}$ and $d_1(\{1\})=0$ so that the left-hand
 side is 1. It follows from this convention that for any $n\in\N$
\begin{align}
\sum_{T\in
 \T(\{1,2,\cdots,n\})}\prod_{j=1}^n\left(\left(\begin{array}{c}k_j\\
					       d_j(T)\end{array}\right)d_j(T)!
\right)\le(1_{n=1}+1_{n\ge 2}
(n-2)!2^{-n})2^{2\sum_{j=1}^nk_j}.\label{eq_number_trees}
\end{align}

By \eqref{eq_phi_good_property}, \eqref{eq_determinant_bound_general},
 \eqref{eq_determinant_bound_generic},
\begin{align*}
|A_m^{(n)}(\bX)|
&\le \frac{2^{n-1}}{m!}\sum_{\s\in \S_m}\sum_{T\in\T(\{1,2,\cdots,n\})}\prod_{j=1}^n
\Bigg(
\sum_{p_j=2}^N\sum_{m_j=0}^{p_j-1}
\left(\begin{array}{c}p_j\\ m_j\end{array}\right)
\left(\begin{array}{c}p_j-m_j\\ d_j(T)\end{array}\right)
\frah^{p_j-m_j}\\
&\qquad\qquad\cdot
\sum_{\bX_j'\in I^{m_j},\bY_j\in I^{p_j-m_j-d_j(T)}\atop
\bZ_j\in I^{d_j(T)}}|F_{p_j}^j(\bY_j,\bZ_j,\bX_j')|
\Bigg)D^{\frac{1}{2}(\sum_{j=1}^np_j-2(n-1)-m)}\\
&\quad\cdot\left|
\prod_{\{p,q\}\in T}\D_{\{p,q\}}(\cC)\prod_{j=1}^n\psi_{\bZ_j}^j\right|
1_{(\bX_1',\bX_2',\cdots,\bX_n')=\bX_{\s}}1_{\sum_{j=1}^nm_j=m}1_{\sum_{j=1}^np_j-2(n-1)\ge
 m}.
\end{align*}
Note that
$$
\prod_{\{p,q\}\in T}\D_{\{p,q\}}(\cC)\prod_{j=1}^n\psi_{\bZ_j}^j
$$
creates at most $\prod_{j=1}^nd_j(T)!$ terms,
since 
$$\left(\sum_{X\in I}\frac{\partial}{\partial
 \psi_{X}^j}\right)^{d_j(T)}\psi_{\bZ_j}^{j}$$
 creates $d_j(T)!$ terms for $j\in \{1,2,\cdots,n\}$. For every $T\in \T(\{1,2,\cdots,n\})$ we
 consider the vertex 1 as the root. Then, by recursively estimating
 along the lines of $T$ from younger branches to the root and using
 \eqref{eq_number_trees} for $k_j=p_j-m_j$ $(j=1,2,\cdots,n)$
we observe that
\begin{align*}
&\|A_m^{(n)}\|_1\\
&\le 2^{n-1}\sum_{T\in\T(\{1,2,\cdots,n\})}
\sum_{p_1=2}^N\sum_{m_1=0}^{p_1-1}\left(\begin{array}{c}p_1\\ m_1\end{array}\right)
\left(\begin{array}{c}p_1-m_1\\ d_1(T)\end{array}\right)d_1(T)!
\|F_{p_1}^1\|_1\\
&\quad\cdot\prod_{j=2}^n
\Bigg(
\sum_{p_j=2}^N\sum_{m_j=0}^{p_j-1}
\left(\begin{array}{c}p_j\\ m_j\end{array}\right)
\left(\begin{array}{c}p_j-m_j\\ d_j(T)\end{array}\right)
d_j(T)!\\
&\qquad\quad\cdot
\sup_{X_0\in I}
\frah^{p_j}\sum_{\bX\in
 I^{p_j}}|F_{p_j}^j(\bX)||\tilde{\cC}(X_0,X_1)|\Bigg)
D^{\frac{1}{2}(\sum_{j=1}^np_j-2(n-1)-m)}\\
&\quad\cdot1_{\sum_{j=1}^nm_j=m}1_{\sum_{j=1}^np_j-2(n-1)\ge
 m}\\
&\le (n-2)!D^{-n+1-\frac{m}{2}}\|\tilde{\cC}\|_{1,\infty}^{n-1}
\sum_{p_1=2}^N\sum_{m_1=0}^{p_1-1}\left(\begin{array}{c}p_1\\
					m_1\end{array}\right)
2^{2(p_1-m_1)}D^{\frac{p_1}{2}}\|F_{p_1}^1\|_1\\
&\quad\cdot\prod_{j=2}^n\Bigg(
\sum_{p_j=2}^N\sum_{m_j=0}^{p_j-1}
\left(\begin{array}{c}p_j\\
					m_j\end{array}\right)
2^{2(p_j-m_j)}D^{\frac{p_j}{2}}
\|F_{p_j}^j\|_{1,\infty}
\Bigg)1_{\sum_{j=1}^nm_j=m}1_{\sum_{j=1}^np_j-2(n-1)\ge
 m},
\end{align*}
which is bounded from above by the right-hand side of \eqref{eq_tree_1}.
 
We can estimate $\|A_m^{(n)}\|_{1,\infty}$ from
 \eqref{eq_tree_1_higher_kernel} in a way parallel to the above
 argument. In the case $m\ge 2$, first we fix a component of $\bX(\in I^m)$. For
 fixed $\s\in\S_m$ there uniquely exists $j_1\in \{1,2,\cdots,n\}$ such
 that the fixed component is one component of the variable
 $\bX_{j_1}'(\in I^{m_{j_1}})$. Then we consider $j_1$ as the root of
 each tree and repeat the same recursive calculation as above to reach
 the claimed inequality \eqref{eq_tree_1_infinity}. The inequality 
\eqref{eq_tree_1_infinity} for $m=0$ follows from \eqref{eq_tree_1} for
 $m=0$ and $\|F_{p_1}^1\|_1\le \frac{N}{h}\|F_{p_1}^1\|_{1,\infty}$.
\end{proof}

In addition to $F^j(\psi)$ $(j=1,2,\cdots,n)$ we give a Grassmann
polynomial having bi-anti-symmetric kernels as one
piece of the input. Assume that we have bi-anti-symmetric functions
$F_{p,q}:I^p\times I^q\to \C$ $(p,q\in\{2,4,\cdots,N\})$ satisfying
\eqref{eq_time_translation_generic} and the following property.
For any function $g:[0,\beta)_h^p\to\C$, $h:[0,\beta)_h^q\to\C$ satisfying 
\begin{align}
&g(r_{\beta}(s_1+s),r_{\beta}(s_2+s),\cdots,r_{\beta}(s_p+s))=g(s_1,s_2,\cdots,s_p)\label{eq_test_function_time_translation}\\
&\left(\forall (s_1,s_2,\cdots,s_p)\in [0,\beta)_h^p,\ s\in
 \frac{1}{h}\Z\right),\notag\\
&h(r_{\beta}(s_1+s),r_{\beta}(s_2+s),\cdots,r_{\beta}(s_q+s))=h(s_1,s_2,\cdots,s_q)\notag\\&\left(\forall (s_1,s_2,\cdots,s_q)\in [0,\beta)_h^q,\ s\in
 \frac{1}{h}\Z\right),\notag
\end{align}
\begin{align}
&\sum_{(s_1,\cdots,s_p)\in
 [0,\beta)_h^p}F_{p,q}((\rho_1\bx_1s_1\xi_1,\cdots,\rho_p\bx_ps_p\xi_p),\bY)g(s_1,\cdots,s_p)=0,\label{eq_vanishing_property}\\
&(\forall \bY\in I^q,\ (\rho_j,\bx_j,\xi_j)\in
 \{1,2\}\times\G\times\{1,-1\}\ (j=1,2,\cdots,p)),\notag\\
&\sum_{(t_1,\cdots,t_q)\in
 [0,\beta)_h^q}F_{p,q}(\bX,
(\eta_1\by_1t_1\zeta_1,\cdots,\eta_q\by_qt_q\zeta_q))h(t_1,\cdots,t_q)=0,\notag\\
&(\forall \bX\in I^p,\ (\eta_j,\by_j,\zeta_j)\in
 \{1,2\}\times\G\times\{1,-1\}\ (j=1,2,\cdots,q)).\notag
\end{align}
We are going to analyze the Grassmann polynomial
$B^{(n)}(\psi)\in\bigwedge_{even}\cV$ $(n\in\N)$ defined by
\begin{align*}
B^{(n)}(\psi)
:=&\sum_{p,q=2}^{N}1_{p,q\in 2\N}\frah^{p+q}\sum_{\bX\in I^p,\bY\in
 I^q}F_{p,q}(\bX,\bY)Tree(\{1,2,\cdots,n+1\},\cC)\\
&\cdot(\psi^1+\psi)_{\bX}(\psi^2+\psi)_{\bY}\prod_{j=3}^{n+1}F^j(\psi^j+\psi)\Bigg|_{\psi^j=0\atop(\forall
 j\in\{1,2,\cdots,n+1\})}.
\end{align*}

\begin{lemma}\label{lem_tree_double_bound}
For any $m\in \{2,4,\cdots,N\}$, $n\in\N$ the anti-symmetric
 kernel $B_m^{(n)}(\cdot)$ satisfies
 \eqref{eq_time_translation_generic}. Moreover, for any
 $m\in\{0,2,\cdots,N\}$, $n\in \N_{\ge 2}$,
\begin{align}
&\|B_m^{(1)}\|_{1,\infty}
\le\left(\frac{N}{h}\right)^{1_{m=0}}
 D^{-1-\frac{m}{2}}\sum_{p_1,p_2=2}^{N}1_{p_1,p_2\in
 2\N}2^{2p_1+2p_2}D^{\frac{p_1+p_2}{2}}[F_{p_1,p_2},\tilde{\cC}]_{1,\infty}1_{p_1+p_2-2\ge
 m}.\label{eq_double_1_1_infinity}\\
&\|B_m^{(1)}\|_1\le
 D^{-1-\frac{m}{2}}\sum_{p_1,p_2=2}^{N}1_{p_1,p_2\in
 2\N}2^{2p_1+2p_2}D^{\frac{p_1+p_2}{2}}[F_{p_1,p_2},\tilde{\cC}]_11_{p_1+p_2-2\ge
 m}.\label{eq_double_1_1}\\
&\|B_m^{(n)}\|_{1,\infty}\le \left(\frac{N}{h}\right)^{1_{m=0}}
(n-1)!D^{-n-\frac{m}{2}}2^{-2m}\|\tilde{\cC}\|_{1,\infty}^{n-1}\label{eq_double_1_infinity}\\
&\qquad\qquad\quad\cdot\sum_{p_1,p_2=2}^{N}1_{p_1,p_2\in2\N}2^{3p_1+3p_2}D^{\frac{p_1+p_2}{2}}[F_{p_1,p_2},\tilde{\cC}]_{1,\infty}
\prod_{j=3}^{n+1}\left(\sum_{p_j=2}^N2^{3p_j}D^{\frac{p_j}{2}}\|F_{p_j}^j\|_{1,\infty}
\right)\notag\\
&\qquad\qquad\quad\cdot 1_{\sum_{j=1}^{n+1}p_j-2n\ge m}.\notag\\
&\|B_m^{(n)}\|_{1}\le 
(n-1)!D^{-n-\frac{m}{2}}2^{-2m}\|\tilde{\cC}\|_{1,\infty}^{n-1}\label{eq_double_1}\\
&\qquad\qquad\quad\cdot\sum_{p_1,p_2=2}^{N}1_{p_1,p_2\in2\N}2^{3p_1+3p_2}D^{\frac{p_1+p_2}{2}}[F_{p_1,p_2},\tilde{\cC}]_{1,\infty}
\prod_{j=3}^{n}\left(\sum_{p_j=2}^N2^{3p_j}D^{\frac{p_j}{2}}\|F_{p_j}^j\|_{1,\infty}
\right)\notag\\
&\qquad\qquad\quad\cdot
\sum_{p_{n+1}=2}^N2^{3p_{n+1}}D^{\frac{p_{n+1}}{2}}\|F_{p_{n+1}}^{n+1}\|_1
1_{\sum_{j=1}^{n+1}p_j-2n\ge m}.\notag
\end{align}
\end{lemma}

\begin{proof}
By anti-symmetry,
\begin{align*}
&B^{(n)}(\psi)\\
&=\sum_{p_1,p_2=2}^{N}1_{p_1,p_2\in2\N}\sum_{m_1=0}^{p_1-1}\sum_{m_2=0}^{p_2-1}
\left(\begin{array}{c}p_1\\
						       m_1\end{array}
\right)
\left(\begin{array}{c}p_2\\
						       m_2\end{array}
\right)\\
&\quad\cdot
\frah^{p_1+p_2}\sum_{\bX_1\in I^{m_1},\bY_1\in I^{p_1-m_1}\atop 
\bX_2\in I^{m_2},\bY_2\in
 I^{p_2-m_2}}F_{p_1,p_2}((\bY_1,\bX_1),(\bY_2,\bX_2))\\
&\quad\cdot 
\prod_{j=3}^{n+1}\Bigg(\sum_{p_j=2}^N\sum_{m_j=0}^{p_j-1}\left(\begin{array}{c}p_j\\
						       m_j\end{array}
\right)\frah^{p_j}\sum_{\bX_j\in I^{m_j}\atop \bY_j\in I^{p_j-m_j}}F_{p_j}^j(\bY_j,\bX_j)\Bigg)\\
&\quad\cdot Tree(\{1,2,\cdots,n+1\},\cC)\prod_{j=1}^{n+1}\psi_{\bY_j}^j
\Bigg|_{\psi^j=0\atop(\forall
 j\in\{1,2,\cdots,n+1\})}
(-1)^{\sum_{j=1}^{n}m_j\sum_{k=j+1}^{n+1}(p_k-m_k)}\\
&\quad\cdot \prod_{k=1}^{n+1}\psi_{\bX_k}.
\end{align*}
Then, the uniqueness of anti-symmetric kernels ensures that 
for $m\in \{0,2,\cdots,N\}$, $\bX\in I^m$,
\begin{align*}
&B^{(n)}_m(\bX)\\
&=\frac{1}{m!}\sum_{\s\in\S_m}\sgn(\s)
\sum_{p_1,p_2=2}^{N}1_{p_1,p_2\in2\N}\sum_{m_1=0}^{p_1-1}\sum_{m_2=0}^{p_2-1}
\left(\begin{array}{c}p_1\\
						       m_1\end{array}
\right)
\left(\begin{array}{c}p_2\\
						       m_2\end{array}
\right)\\
&\quad\cdot
\frah^{p_1+p_2-m_1-m_2}\sum_{\bY_1\in I^{p_1-m_1},\bY_2\in
 I^{p_2-m_2}}F_{p_1,p_2}((\bY_1,\bX_1'),(\bY_2,\bX_2'))\\
&\quad\cdot 
\prod_{j=3}^{n+1}\Bigg(\sum_{p_j=2}^N\sum_{m_j=0}^{p_j-1}\left(\begin{array}{c}p_j\\
						       m_j\end{array}
\right)\frah^{p_j-m_j}\sum_{\bY_j\in I^{p_j-m_j}}F_{p_j}^j(\bY_j,\bX_j')\Bigg)\\
&\quad\cdot Tree(\{1,2,\cdots,n+1\},\cC)\prod_{j=1}^{n+1}\psi_{\bY_j}^j
\Bigg|_{\psi^j=0\atop(\forall
 j\in\{1,2,\cdots,n+1\})}
(-1)^{\sum_{j=1}^{n}m_j\sum_{k=j+1}^{n+1}(p_k-m_k)}\\
&\quad\cdot 1_{\sum_{j=1}^{n+1}m_j=m}1_{\sum_{j=1}^{n+1}p_j-2n\ge m}1_{(\bX_1',\bX_2',\cdots,\bX_{n+1}')=\bX_{\s}}.
\end{align*}
The property \eqref{eq_time_translation_generic} of $B^{(n)}_m$
 follows from \eqref{eq_time_translation_tree_formula} and the property
 \eqref{eq_time_translation_generic} of the input.

By \eqref{eq_phi_good_property}, \eqref{eq_determinant_bound_general},
 \eqref{eq_determinant_bound_generic},
\begin{align}
&|B^{(n)}_m(\bX)|\label{eq_tree_double_inside_inequality}\\
&\le\frac{2^n}{m!}\sum_{\s\in\S_m}\sum_{T\in \T(\{1,2,\cdots,n+1\})}
\sum_{p_1,p_2=2}^{N}1_{p_1,p_2\in2\N}\notag\\
&\quad\cdot\sum_{m_1=0}^{p_1-1}\sum_{m_2=0}^{p_2-1}
\left(\begin{array}{c}p_1\\
						       m_1\end{array}
\right)
\left(\begin{array}{c}p_2\\
						       m_2\end{array}
\right)
\left(\begin{array}{c}p_1-m_1\\
						       d_1(T)\end{array}
\right)
\left(\begin{array}{c}p_2-m_2\\
						       d_2(T)\end{array}
\right)\notag\\
&\quad\cdot
\frah^{p_1+p_2-m_1-m_2}\sum_{\bY_1\in I^{p_1-m_1-d_1(T)},\bY_2\in
 I^{p_2-m_2-d_2(T)}\atop \bZ_1\in I^{d_1(T)},\bZ_2\in I^{d_2(T)}}|F_{p_1,p_2}((\bY_1,\bZ_1,\bX_1'),(\bY_2,\bZ_2,\bX_2'))|\notag\\
&\quad\cdot 
\prod_{j=3}^{n+1}\Bigg(\sum_{p_j=2}^N\sum_{m_j=0}^{p_j-1}\left(\begin{array}{c}p_j\\
						       m_j\end{array}
\right)
\left(\begin{array}{c}p_j-m_j\\
						       d_j(T)\end{array}
\right)\notag\\
&\qquad\qquad\cdot
\frah^{p_j-m_j}\sum_{\bY_j\in I^{p_j-m_j-d_j(T)}\atop \bZ_j\in I^{d_j(T)}}|F_{p_j}^j(\bY_j,\bZ_j,\bX_j')|\Bigg)\notag\\
&\quad\cdot D^{\frac{1}{2}(\sum_{j=1}^{n+1}p_j-2n-m)}\left|
\prod_{\{p,q\}\in T}\D_{\{p,q\}}(\cC)\prod_{j=1}^{n+1}\psi_{\bZ_j}^j
\right|\notag\\
&\quad\cdot 1_{\sum_{j=1}^{n+1}m_j=m}1_{\sum_{j=1}^{n+1}p_j-2n\ge m}1_{(\bX_1',\bX_2',\cdots,\bX_{n+1}')=\bX_{\s}}.\notag
\end{align}

Let us derive the inequality \eqref{eq_double_1_1} from
 \eqref{eq_tree_double_inside_inequality}. For any $m\in
 \{0,2,\cdots,N\}$,
\begin{align*}
\|B_m^{(1)}\|_1
\le& 2\sum_{p_1,p_2=2}^{N}1_{p_1,p_2\in2\N}
\sum_{m_1=0}^{p_1-1}\sum_{m_2=0}^{p_2-1}
\left(\begin{array}{c}p_1\\
						       m_1\end{array}
\right)
\left(\begin{array}{c}p_2\\
						       m_2\end{array}
\right)
(p_1-m_1)(p_2-m_2)\\
&\cdot [F_{p_1,p_2},\tilde{\cC}]_1 D^{\frac{1}{2}(\sum_{j=1}^{2}p_j-2-m)}1_{m_1+m_2=m}1_{p_1+p_2-2\ge m},
\end{align*}
which is less than or equal to the right-hand side of \eqref{eq_double_1_1}.
The inequality \eqref{eq_double_1_1_infinity} can be derived in the same
 way.

Let us consider the case that $n\ge 2$. We decompose
 \eqref{eq_tree_double_inside_inequality} as follows.
\begin{align}
|B^{(n)}_m(\bX)|\le \sum_{T\in \T(\{1,2,\cdots,n+1\})}B_m^{(n)}(T)(\bX),
\label{eq_tree_double_inside_decomposition}
\end{align}
\begin{align}
&B_m^{(n)}(T)(\bX)\label{eq_tree_double_inside_decomposition_expansion}\\
&:=\frac{2^n}{m!}\sum_{\s\in\S_m}
\sum_{p_1,p_2=2}^{N}1_{p_1,p_2\in2\N}\notag\\
&\quad\cdot\sum_{m_1=0}^{p_1-1}\sum_{m_2=0}^{p_2-1}
\left(\begin{array}{c}p_1\\
						       m_1\end{array}
\right)
\left(\begin{array}{c}p_2\\
						       m_2\end{array}
\right)
\left(\begin{array}{c}p_1-m_1\\
						       d_1(T)\end{array}
\right)
\left(\begin{array}{c}p_2-m_2\\
						       d_2(T)\end{array}
\right)\notag\\
&\quad\cdot
\frah^{p_1+p_2-m_1-m_2}\sum_{\bY_1\in I^{p_1-m_1-d_1(T)},\bY_2\in
 I^{p_2-m_2-d_2(T)}\atop \bZ_1\in I^{d_1(T)},\bZ_2\in I^{d_2(T)}}|F_{p_1,p_2}((\bY_1,\bZ_1,\bX_1'),(\bY_2,\bZ_2,\bX_2'))|\notag\\
&\quad\cdot 
\prod_{j=3}^{n+1}\Bigg(\sum_{p_j=2}^N\sum_{m_j=0}^{p_j-1}\left(\begin{array}{c}p_j\\
						       m_j\end{array}
\right)
\left(\begin{array}{c}p_j-m_j\\
						       d_j(T)\end{array}
\right)\notag\\
&\qquad\qquad\cdot
\frah^{p_j-m_j}\sum_{\bY_j\in I^{p_j-m_j-d_j(T)}\atop \bZ_j\in I^{d_j(T)}}|F_{p_j}^j(\bY_j,\bZ_j,\bX_j')|\Bigg)\notag\\
&\quad\cdot D^{\frac{1}{2}(\sum_{j=1}^{n+1}p_j-2n-m)}\left|
\prod_{\{p,q\}\in T}\D_{\{p,q\}}(\cC)\prod_{j=1}^{n+1}\psi_{\bZ_j}^j
\right|\notag\\
&\quad\cdot 1_{\sum_{j=1}^{n+1}m_j=m}1_{\sum_{j=1}^{n+1}p_j-2n\ge m}1_{(\bX_1',\bX_2',\cdots,\bX_{n+1}')=\bX_{\s}}.\notag
\end{align}
Let us estimate $\|B_m^{(n)}(T)\|_1$. For $T\in \T(\{1,2,\cdots,n+1\})$
 we consider the vertex $n+1$ as the root of $T$. Without loss of
 generality we can assume that 
\begin{align}
&\text{the distance between 1 and $n+1$ is shorter than or equal
 to}\label{eq_tree_order_temporal_assumption}\\
 &\text{that
 between 2 and $n+1$ in $T$.}\notag
\end{align}
We can derive the same inequality by assuming otherwise. For $j\in
 \{1,2,\cdots,n+1\}$ let us introduce the conditions $P_j$, $Q$ as
 follows.
\begin{align*}
&P_j : \text{ The vertex 1 is on the shortest path between $j$ and 2 in
 $T$.}\\
&Q : \text{ The distance between 1 and 2 in $T$ is 1.}
\end{align*}
Then, only one of the following cases occurs.
$$
P_{n+1}\land Q \qquad P_{n+1}\land \lnot Q\qquad \lnot P_{n+1} 
$$
By recursively estimating along the lines of $T$ from younger branches
 to the root $n+1$ we obtain from \eqref{eq_tree_double_inside_decomposition_expansion}
 that for any $m\in \{0,2,\cdots,N\}$
\begin{align}
&\|B_m^{(n)}(T)\|_1\label{eq_double_1_pre}\\
&\le 2^n
\sum_{p_1,p_2=2}^{N}1_{p_1,p_2\in2\N}\notag\\
&\quad\cdot\sum_{m_1=0}^{p_1-1}\sum_{m_2=0}^{p_2-1}
\left(\begin{array}{c}p_1\\
						       m_1\end{array}
\right)
\left(\begin{array}{c}p_2\\
						       m_2\end{array}
\right)
\left(\begin{array}{c}p_1-m_1\\
						       d_1(T)\end{array}
\right)d_1(T)!
\left(\begin{array}{c}p_2-m_2\\
						       d_2(T)\end{array}
\right)d_2(T)!
\notag\\
&\quad\cdot\prod_{j=3}^n\Bigg(\sum_{p_j=2}^N\sum_{m_j=0}^{p_j-1}
\left(\begin{array}{c}p_j\\
						       m_j\end{array}
\right)
\left(\begin{array}{c}p_j-m_j\\
						       d_j(T)\end{array}
\right)d_j(T)!\notag\\
&\qquad\qquad\cdot \sup_{X_0\in I}\frah^{p_j}\sum_{\bX\in I^{p_j}}|F_{p_j}^j(\bX)||\tilde{\cC}(X_0,X_1)|
\Bigg)\notag\\
&\quad\cdot \sum_{p_{n+1}=2}^N\sum_{m_{n+1}=0}^{p_{n+1}-1}
\left(\begin{array}{c}p_{n+1}\\
						       m_{n+1}\end{array}
\right)
\left(\begin{array}{c}p_{n+1}-m_{n+1}\\
						       d_{n+1}(T)\end{array}
\right)d_{n+1}(T)!
\|F_{p_{n+1}}^{n+1}\|_1
D^{\frac{1}{2}(\sum_{j=1}^{n+1}p_j-2n-m)}\notag\\
&\quad\cdot\Bigg(1_{P_{n+1}\land Q}\sup_{X_0\in
 I}\frah^{p_1+p_2}\sum_{\bX\in I^{p_1}\atop \bY\in
 I^{p_2}}|F_{p_1,p_2}(\bX,\bY)||\tilde{\cC}(X_0,X_1)|
|\tilde{\cC}(X_2,Y_1)|\notag\\
&\qquad + (1_{P_{n+1}\land \lnot Q}+1_{\lnot P_{n+1}})
\sup_{X_0\in I}\Bigg(\frah^{p_1+p_2}\sum_{\bX\in I^{p_1}}\notag\\
&\qquad\qquad\qquad\qquad\qquad\qquad\qquad\cdot \sup_{Y_0\in I}
\sum_{\bY\in
 I^{p_2}}|F_{p_1,p_2}(\bX,\bY)||\tilde{\cC}(X_0,X_1)|
|\tilde{\cC}(Y_0,Y_1)|\Bigg)\Bigg)\notag\\
&\quad\cdot 1_{\sum_{j=1}^{n+1}m_j=m}1_{\sum_{j=1}^{n+1}p_j-2n\ge
 m}\notag\\
&\le 2^n D^{-n-\frac{m}{2}}\|\tilde{\cC}\|_{1,\infty}^{n-1}\notag\\
&\quad\cdot 
\sum_{p_1,p_2=2}^{N}1_{p_1,p_2\in2\N}D^{\frac{p_1+p_2}{2}}
\sum_{m_1=0}^{p_1-1}\sum_{m_2=0}^{p_2-1}
\left(\begin{array}{c}p_1\\
						       m_1\end{array}
\right)
\left(\begin{array}{c}p_2\\
						       m_2\end{array}
\right)\notag\\
&\qquad\cdot
\left(\begin{array}{c}p_1-m_1\\
						       d_1(T)\end{array}
\right)d_1(T)!
\left(\begin{array}{c}p_2-m_2\\
						       d_2(T)\end{array}
\right)d_2(T)!
[F_{p_1,p_2},\tilde{\cC}]_{1,\infty}\notag\\
&\quad\cdot\prod_{j=3}^n\Bigg(\sum_{p_j=2}^ND^{\frac{p_j}{2}}
\sum_{m_j=0}^{p_j-1}
\left(\begin{array}{c}p_j\\
						       m_j\end{array}
\right)
\left(\begin{array}{c}p_j-m_j\\
						       d_j(T)\end{array}
\right)d_j(T)!\|F_{p_j}^j\|_{1,\infty}\Bigg)\notag\\
&\quad\cdot \sum_{p_{n+1}=2}^ND^{\frac{p_{n+1}}{2}}
\sum_{m_{n+1}=0}^{p_{n+1}-1}
\left(\begin{array}{c}p_{n+1}\\
						       m_{n+1}\end{array}
\right)
\left(\begin{array}{c}p_{n+1}-m_{n+1}\\
						       d_{n+1}(T)\end{array}
\right)d_{n+1}(T)!
\|F_{p_{n+1}}^{n+1}\|_1\notag\\
&\quad\cdot
 1_{\sum_{j=1}^{n+1}m_j=m}1_{\sum_{j=1}^{n+1}p_j-2n\ge
 m}.\notag
\end{align}
When $P_{n+1}$ holds, the first inequality above can be derived
 smoothly. The point of the derivation of the first inequality when
 $P_{n+1}$ does not hold is to complete the recursive estimation along
 the branch containing the vertex 2 before the recursive estimation
 along the branch containing the vertex 1. Here we use the assumption 
 \eqref{eq_tree_order_temporal_assumption} to exclude that the vertex 2
 is on the shortest path between $n+1$ and 1 in $T$. By applying
 \eqref{eq_number_trees} we can derive \eqref{eq_double_1} from
 \eqref{eq_tree_double_inside_decomposition}, \eqref{eq_double_1_pre}. Note that \eqref{eq_double_1} for $m=0$
 implies \eqref{eq_double_1_infinity} for $m=0$, since
 $\|B_0^{(n)}\|_1=\|B_0^{(n)}\|_{1,\infty}=|B_0^{(n)}|$ and 
$\|F_{p_{n+1}}^{n+1}\|_1\le\frac{N}{h} \|F_{p_{n+1}}^{n+1}\|_{1,\infty}$.

In order to derive the claimed upper bound on $\|B_m^{(n)}\|_{1,\infty}$
 for $m\ge 2$ from
 \eqref{eq_tree_double_inside_decomposition_expansion}, we fix
 $T\in \T(\{1,2,\cdots,n+1\})$ and the first component $X_1$ of the
 variable $\bX(\in I^m)$. For any $\s\in \S_m$ there uniquely exists
 $j_1\in \{1,2,\cdots,n+1\}$ such that $X_1$ is a component of
 $\bX_{j_1}'$. Then, we consider the vertex $j_1$ as the root of the tree $T$. Again
 without loss of generality we can assume that 
\begin{align}
&\text{the distance between 1 and $j_1$ is shorter than or equal
 to}\label{eq_tree_order_temporal_assumption_second}\\
 &\text{that
 between 2 and $j_1$ in $T$.}\notag
\end{align}
Only one of the following cases happens.
$$
P_{j_1}\land Q \land j_1=1\quad\ P_{j_1}\land Q \land j_1\neq 1\quad\ 
P_{j_1}\land \lnot Q \land j_1=1\quad\  
P_{j_1}\land \lnot Q \land j_1\neq 1\quad\ \lnot P_{j_1} 
$$
By the recursive estimation from younger branches to the root $j_1$
 we deduce that
\begin{align*}
&\|B_m^{(n)}(T)\|_{1,\infty}\\
&\le \frac{2^n}{m!}\sum_{\s\in\S_m}
\sum_{p_1,p_2=2}^{N}1_{p_1,p_2\in2\N}\notag\\
&\quad\cdot\sum_{m_1=0}^{p_1-1}\sum_{m_2=0}^{p_2-1}
\left(\begin{array}{c}p_1\\
						       m_1\end{array}
\right)
\left(\begin{array}{c}p_2\\
						       m_2\end{array}
\right)
\left(\begin{array}{c}p_1-m_1\\
						       d_1(T)\end{array}
\right)d_1(T)!
\left(\begin{array}{c}p_2-m_2\\
						       d_2(T)\end{array}
\right)d_2(T)!
\notag\\
&\quad\cdot\prod_{j=3}^{n+1}\Bigg(\sum_{p_j=2}^N\sum_{m_j=0}^{p_j-1}
\left(\begin{array}{c}p_j\\
						       m_j\end{array}
\right)
\left(\begin{array}{c}p_j-m_j\\
						       d_j(T)\end{array}
\right)d_j(T)!\Bigg)D^{\frac{1}{2}(\sum_{j=1}^{n+1}p_j-2n-m)}\\
&\quad\cdot 1_{\sum_{j=1}^{n+1}m_j=m}1_{\sum_{j=1}^{n+1}p_j-2n\ge
 m}\notag\\
&\quad\cdot\Bigg(
1_{P_{j_1}\land Q\land j_1=1}
 \sup_{X_0\in I}\Bigg(\frah^{p_1+p_2-1}\sum_{\bX\in I^{p_1-1}\atop \bY\in
 I^{p_2}}|F_{p_1,p_2}((X_0,\bX),\bY)||\tilde{\cC}(X_1,Y_1)|\Bigg)\notag\\
&\qquad\qquad\cdot \prod_{j=3}^{n+1}\Bigg(\sup_{X_0\in I}
\frah^{p_j}\sum_{\bX\in I^{p_j}}|F_{p_j}^j(\bX)||\tilde{\cC}(X_0,X_1)|
\Bigg)\\
&\qquad + 1_{P_{j_1}\land Q\land j_1\neq 1}
\sup_{X_0\in I}\Bigg(\frah^{p_1+p_2}\sum_{\bX\in I^{p_1}\atop \bY\in
 I^{p_2}}|F_{p_1,p_2}(\bX,\bY)||\tilde{\cC}(X_0,X_1)|
|\tilde{\cC}(X_2,Y_1)|\Bigg)\\
&\qquad\qquad\cdot 
\|F_{p_{j_1}}^{j_1}\|_{1,\infty}\prod_{j=3\atop j\neq j_1}^{n+1}\Bigg(\sup_{X_0\in
 I}\frah^{p_j}\sum_{\bX\in I^{p_j}}|F_{p_j}^j(\bX)||\tilde{\cC}(X_0,X_1)|
\Bigg)\notag\\
&\qquad + 1_{P_{j_1}\land \lnot Q\land j_1= 1}\\
&\qquad\qquad\cdot\sup_{X_0\in I}\Bigg(\frah^{p_1+p_2-1}\sum_{\bX\in I^{p_1-1}}\sup_{Y_0\in
 I}\Bigg(\sum_{\bY\in I^{p_2}}|F_{p_1,p_2}((X_0,\bX),\bY)||\tilde{\cC}(Y_0,Y_1)|
\Bigg)\Bigg)\\
&\qquad\qquad\cdot \prod_{j=3}^{n+1}\Bigg(\sup_{X_0\in I}
\frah^{p_j}\sum_{\bX\in I^{p_j}}|F_{p_j}^j(\bX)||\tilde{\cC}(X_0,X_1)|
\Bigg)\\
&\qquad + (1_{P_{j_1}\land \lnot Q\land j_1\neq 1}+1_{\lnot P_{j_1}})\\
&\qquad\qquad\cdot \sup_{X_0\in I}\Bigg(\frah^{p_1+p_2}\sum_{\bX\in
 I^{p_1}}
\sup_{Y_0\in
 I}\Bigg(\sum_{\bY\in I^{p_2}}|F_{p_1,p_2}(\bX,\bY)||\tilde{\cC}(X_0,X_1)||\tilde{\cC}(Y_0,Y_1)|
\Bigg)\Bigg)\\
&\qquad\qquad\cdot \|F_{p_{j_1}}^{j_1}\|_{1,\infty}\prod_{j=3\atop j\neq j_1}^{n+1}\Bigg(\sup_{X_0\in I}
\frah^{p_j}\sum_{\bX\in I^{p_j}}|F_{p_j}^j(\bX)||\tilde{\cC}(X_0,X_1)|
\Bigg)\Bigg)\\
&\le 2^n D^{-n-\frac{m}{2}}\|\tilde{\cC}\|_{1,\infty}^{n-1}\\
&\quad\cdot 
\sum_{p_1,p_2=2}^{N}1_{p_1,p_2\in2\N}D^{\frac{p_1+p_2}{2}}
\sum_{m_1=0}^{p_1-1}\sum_{m_2=0}^{p_2-1}
\left(\begin{array}{c}p_1\\
						       m_1\end{array}
\right)
\left(\begin{array}{c}p_2\\
						       m_2\end{array}
\right)\notag\\
&\qquad\cdot
\left(\begin{array}{c}p_1-m_1\\
						       d_1(T)\end{array}
\right)d_1(T)!
\left(\begin{array}{c}p_2-m_2\\
						       d_2(T)\end{array}
\right)d_2(T)!
[F_{p_1,p_2},\tilde{\cC}]_{1,\infty}\notag\\
&\quad\cdot\prod_{j=3}^{n+1}\Bigg(\sum_{p_j=2}^ND^{\frac{p_j}{2}}
\sum_{m_j=0}^{p_j-1}
\left(\begin{array}{c}p_j\\
						       m_j\end{array}
\right)
\left(\begin{array}{c}p_j-m_j\\
						       d_j(T)\end{array}
\right)d_j(T)!\|F_{p_j}^j\|_{1,\infty}\Bigg)\notag\\
&\quad\cdot 
 1_{\sum_{j=1}^{n+1}m_j=m}1_{\sum_{j=1}^{n+1}p_j-2n\ge
 m}.
\end{align*}
Again when $P_{j_1}$ does not hold, the assumption
 \eqref{eq_tree_order_temporal_assumption_second} excludes that the
 vertex 2 is on the shortest path between $j_1$ and $1$ in $T$ so that
 we can carry
 out the recursive estimation along the branch containing the vertex 2
 before that along the branch containing the vertex 1. Combining this
 inequality with \eqref{eq_tree_double_inside_decomposition} and
 \eqref{eq_number_trees} results in \eqref{eq_double_1_infinity} for
 $m\ge 2$.
\end{proof}

Next we consider the Grassmann polynomials $E^{(n)}(\psi)\in
\bigwedge_{even}\cV$ $(n\in \N)$ defined as follows.
\begin{align*}
E^{(n)}(\psi)
:=&\sum_{p,q=2}^{N}1_{p,q\in 2\N}\frah^{p+q}\sum_{\bX\in I^p\atop
 \bY\in I^q}F_{p,q}(\bX,\bY)\\
&\cdot
 Tree(\{s_j\}_{j=1}^{m+1},\cC)(\psi^1+\psi)_{\bX}\prod_{j=2}^{m+1}F^{s_j}(\psi^{s_j}+\psi)\Bigg|_{\psi^{s_j}=0\atop(\forall
 j\in\{1,2,\cdots,m+1\})}\\
&\cdot
 Tree(\{t_k\}_{k=1}^{n-m},\cC)(\psi^1+\psi)_{\bY}\prod_{k=2}^{n-m}F^{t_k}(\psi^{t_k}+\psi)\Bigg|_{\psi^{t_k}=0\atop(\forall
 k\in\{1,2,\cdots,n-m\})},
\end{align*}
where the functions $F_{p,q}:I^p\times I^q\to\C$ $(p,q\in
\{2,4,\cdots,N\})$ are bi-anti-symmetric and
satisfy \eqref{eq_time_translation_generic},
\eqref{eq_vanishing_property} and 
\begin{align*}
&m\in \{0,1,\cdots,n-1\},\\
&1=s_1<s_2<\cdots <s_{m+1}\le n,\quad 1=t_1<t_2<\cdots <t_{n-m}\le n,\\
&\{s_j\}_{j=2}^{m+1}\cup \{t_k\}_{k=2}^{n-m}=\{2,3,\cdots,n\},\quad
 \{s_j\}_{j=2}^{m+1}\cap \{t_k\}_{k=2}^{n-m}=\emptyset.
\end{align*}
Here we assume that $\{s_j\}_{j=2}^{m+1}=\emptyset$ if $m=0$, $\{t_k\}_{k=2}^{n-m}=\emptyset$ if $m=n-1$. 

\begin{lemma}\label{lem_tree_divided_bound}
For any $n\in \N$, $a,b\in \{2,4,\cdots,N\}$ there exists a
 function $E_{a,b}^{(n)}:I^a\times I^b\to \C$ such that $E_{a,b}^{(n)}$
 is bi-anti-symmetric, satisfies \eqref{eq_time_translation_generic},
 \eqref{eq_vanishing_property} and 
$$
E^{(n)}(\psi)=\sum_{a,b=2}^{N}1_{a,b\in 2\N}\frah^{a+b}\sum_{\bX\in
 I^a\atop \bY\in I^b}E_{a,b}^{(n)}(\bX,\bY)\psi_{\bX}\psi_{\bY}.
$$
Moreover, the following inequalities hold for any $a,b\in
 \{2,4,\cdots,N\}$, $n\in\N_{\ge 2}$. 
\begin{align}
&\|E_{a,b}^{(1)}\|_{1,\infty}\le
 \sum_{p=a}^{N}\sum_{q=b}^{N}1_{p,q\in 2\N}
\left(\begin{array}{c} p \\ a \end{array}\right)
\left(\begin{array}{c} q \\ b \end{array}\right)
D^{\frac{1}{2}(p+q-a-b)}\|F_{p,q}\|_{1,\infty}.\label{eq_divided_1_1_infinity}\\
&\|E_{a,b}^{(1)}\|_{1}\le
 \sum_{p=a}^{N}\sum_{q=b}^{N}1_{p,q\in 2\N}
\left(\begin{array}{c} p \\ a \end{array}\right)
\left(\begin{array}{c} q \\ b \end{array}\right)
D^{\frac{1}{2}(p+q-a-b)}\|F_{p,q}\|_{1}.\label{eq_divided_1_1}\\
&\|E_{a,b}^{(n)}\|_{1,\infty}\le (1_{m\neq 0}(m-1)!+1_{m=0})(1_{m\neq
 n-1}(n-m-2)!+1_{m=n-1})\label{eq_divided_1_infinity}\\
&\qquad\qquad\quad\cdot 2^{-2a-2b}D^{-n+1-\frac{1}{2}(a+b)}\|\tilde{\cC}\|_{1,\infty}^{n-1}\sum_{p_1,q_1=2}^{N}1_{p_1,q_1\in
 2\N}2^{3p_1+3q_1}D^{\frac{p_1+q_1}{2}}\|F_{p_1,q_1}\|_{1,\infty}\notag\\
&\qquad\qquad\quad\cdot\prod_{j=2}^{m+1}\Bigg(
\sum_{p_j=2}^N2^{3p_j}D^{\frac{p_j}{2}}\|F_{p_j}^{s_j}\|_{1,\infty}\Bigg)
\prod_{k=2}^{n-m}\Bigg(
\sum_{q_k=2}^N2^{3q_k}D^{\frac{q_k}{2}}
\|F_{q_k}^{t_k}\|_{1,\infty}\Bigg)\notag\\
&\qquad\qquad\quad\cdot 1_{\sum_{j=1}^{m+1}p_j-2m\ge a}
1_{\sum_{k=1}^{n-m}q_k-2(n-m-1)\ge b}.\notag\\
&\|E_{a,b}^{(n)}\|_{1}
\le (1_{m\neq 0}(m-1)!+1_{m=0})(1_{m\neq
 n-1}(n-m-2)!+1_{m=n-1})\label{eq_divided_1}\\
&\qquad\qquad\cdot2^{-2a-2b}D^{-n+1-\frac{1}{2}(a+b)}\|\tilde{\cC}\|_{1,\infty}^{n-1} \sum_{p_1,q_1=2}^{N}1_{p_1,q_1\in
 2\N}2^{3p_1+3q_1}D^{\frac{p_1+q_1}{2}}\|F_{p_1,q_1}\|_{1,\infty}\notag\\
&\qquad\qquad\cdot\prod_{j=2}^{m+1}\Bigg(
\sum_{p_j=2}^N2^{3p_j}D^{\frac{p_j}{2}}
(1_{s_j\neq
 n}\|F_{p_j}^{s_j}\|_{1,\infty}+1_{s_j=n}\|F_{p_j}^{s_j}\|_1)\Bigg)\notag\\
&\qquad\qquad\cdot\prod_{k=2}^{n-m}\Bigg(
\sum_{q_k=2}^N2^{3q_k}D^{\frac{q_k}{2}}
(1_{t_k\neq
 n}\|F_{q_k}^{t_k}\|_{1,\infty}+1_{t_k=n}\|F_{q_k}^{t_k}\|_1)\Bigg)\notag\\
&\qquad\qquad\cdot 1_{\sum_{j=1}^{m+1}p_j-2m\ge a}
1_{\sum_{k=1}^{n-m}q_k-2(n-m-1)\ge b}.\notag
\end{align}
\end{lemma}

\begin{proof}
By using anti-symmetry we can transform $E^{(n)}(\psi)$ as follows.
\begin{align}
&E^{(n)}(\psi)\label{eq_divided_part_characterization}\\
 &=\sum_{p_1,q_1=2}^{N}1_{p_1,q_1\in 2\N}\sum_{u_1=0}^{p_1}
(1_{m=0}+1_{m\neq 0}1_{u_1\le
 p_1-1})\left(\begin{array}{c} p_1 \\ u_1\end{array}\right)
\frah^{u_1}\sum_{\bX_1\in I^{u_1}}\notag\\
&\quad\cdot \sum_{v_1=0}^{q_1}
(1_{m=n-1}+1_{m\neq n-1}1_{v_1\le
 q_1-1})\left(\begin{array}{c} q_1 \\ v_1\end{array}\right)
\frah^{v_1}\sum_{\bY_1\in I^{v_1}}\notag\\
&\quad\cdot \prod_{j=2}^{m+1}\left(\sum_{p_j=2}^N\sum_{u_j=0}^{p_j-1}
\left(\begin{array}{c} p_j \\ u_j\end{array}\right)
\frah^{u_j}\sum_{\bX_j\in I^{u_j}}\right)
\prod_{k=2}^{n-m}\left(\sum_{q_k=2}^N\sum_{v_k=0}^{q_k-1}
\left(\begin{array}{c} q_k \\ v_k\end{array}\right)
\frah^{v_k}\sum_{\bY_k\in I^{v_k}}\right)\notag\\
&\quad\cdot f_m^n((p_j)_{1\le j\le m+1},(u_j)_{1\le j\le
 m+1},(q_j)_{1\le j\le n-m},(v_j)_{1\le j\le
 n-m})\notag\\
&\quad\qquad ((\bX_1,\bX_2,\cdots,\bX_{m+1}),
(\bY_1,\bY_2,\cdots,\bY_{n-m}))\notag\\
&\quad\cdot
 \psi_{\bX_1}\psi_{\bX_2}\cdots\psi_{\bX_{m+1}}\psi_{\bY_1}\psi_{\bY_2}\cdots\psi_{\bY_{n-m}},\notag
\end{align}
where the function
$$
f_m^n((p_j)_{1\le j\le m+1},(u_j)_{1\le j\le
 m+1},(q_j)_{1\le j\le n-m},(v_j)_{1\le j\le
 n-m}):\prod_{j=1}^{m+1}I^{u_j}\times \prod_{k=1}^{n-m}I^{v_k}\to\C
$$
is defined by
\begin{align}
&f_m^n((p_j)_{1\le j\le m+1},(u_j)_{1\le j\le
 m+1},(q_j)_{1\le j\le n-m},(v_j)_{1\le j\le
 n-m})\label{eq_inside_function_characterization}\\
&\quad ((\bX_1,\bX_2,\cdots,\bX_{m+1}),
(\bY_1,\bY_2,\cdots,\bY_{n-m}))\notag\\
&:=\frah^{p_1+q_1-u_1-v_1}\sum_{\bW_1\in I^{p_1-u_1}}\sum_{\bZ_1\in
 I^{q_1-v_1}}F_{p_1,q_1}((\bW_1,\bX_1),(\bZ_1,\bY_1))\notag\\
&\quad\cdot \prod_{j=2}^{m+1}\left(
\frah^{p_j-u_j}\sum_{\bW_j\in
 I^{p_j-u_j}}F_{p_j}^{s_j}(\bW_j,\bX_j)\right)\notag\\
&\quad\cdot \prod_{k=2}^{n-m}
\left(
\frah^{q_k-v_k}\sum_{\bZ_k\in
 I^{q_k-v_k}}F_{q_k}^{t_k}(\bZ_k,\bY_k)\right)\notag\\
&\quad\cdot
 Tree(\{s_j\}_{j=1}^{m+1},\cC)\prod_{j=1}^{m+1}\psi_{\bW_j}^{s_j}\Bigg|_{\psi^{s_j}=0\atop(\forall
 j\in\{1,2,\cdots,m+1\})}\notag\\
&\quad\cdot Tree(\{t_k\}_{k=1}^{n-m},\cC)\prod_{k=1}^{n-m}\psi_{\bZ_k}^{t_k}\Bigg|_{\psi^{t_k}=0\atop(\forall
 k\in\{1,2,\cdots,n-m\})}\notag\\
&\quad\cdot
 (-1)^{\sum_{j=1}^mu_j\sum_{i=j+1}^{m+1}(p_i-u_i)+\sum_{k=1}^{n-m-1}v_k\sum_{i=k+1}^{n-m}(q_i-v_i)}.\notag
\end{align}
For simplicity, set $\bp:=(p_j)_{1\le j\le m+1}$,
$\bu:=(u_j)_{1\le j\le
 m+1}$, $\bq:=(q_j)_{1\le j\le n-m}$, $\bv:=(v_j)_{1\le j\le
 n-m}$.
Since the kernel of $E^{(n)}(\psi)$ inherits many properties from the
 function $f_m^n(\bp,\bu,\bq,\bv)$, we should study
 $f_m^n(\bp,\bu,\bq,\bv)$ first.

It follows from the property \eqref{eq_time_translation_generic} of
 $F_{p_1,q_1}$, $F^j$ and \eqref{eq_time_translation_tree_formula}
 that $f_m^n(\bp,\bu,\bq,\bv)(\cdot)$ satisfies
 \eqref{eq_time_translation_generic}. Assume that $\bu=\b0$ and define
 the function $g:I^{p_1}\to\C$ by
\begin{align*}
&g(\bX)\\
&:=\prod_{j=2}^{m+1}\left(\frah^{p_j}\sum_{\bW_j\in
 I^{p_j}}F_{p_j}^{s_j}(\bW_j)\right)
 Tree(\{s_j\}_{j=1}^{m+1},\cC)\psi_{\bX}^{s_1}\prod_{j=2}^{m+1}\psi_{\bW_j}^{s_j}
\Bigg|_{\psi^{s_j}=0\atop(\forall
 j\in\{1,2,\cdots,m+1\})}.
\end{align*}
By \eqref{eq_time_translation_tree_formula} and the property
 \eqref{eq_time_translation_generic} of $F^j$ the function $g$ satisfies
 \eqref{eq_time_translation_generic} too. Then, the property
 \eqref{eq_vanishing_property} of $F_{p_1,q_1}$ implies that
\begin{align*}
&f_m^n(\bp,\b0,\bq,\bv)(\bY_1,\bY_2,\cdots,\bY_{n-m})\\
&=\frah^{p_1+q_1-v_1}\sum_{\bW_1\in I^{p_1}}\sum_{\bZ_1\in
 I^{q_1-v_1}}F_{p_1,q_1}(\bW_1,(\bZ_1,\bY_1))g(\bW_1)\notag\\
&\quad\cdot \prod_{k=2}^{n-m}
\left(
\frah^{q_k-v_k}\sum_{\bZ_k\in
 I^{q_k-v_k}}F_{q_k}^{t_k}(\bZ_k,\bY_k)\right)\notag\\
&\quad\cdot Tree(\{t_k\}_{k=1}^{n-m},\cC)\prod_{k=1}^{n-m}\psi_{\bZ_k}^{t_k}\Bigg|_{\psi^{t_k}=0\atop(\forall
 k\in\{1,2,\cdots,n-m\})}
 (-1)^{\sum_{k=1}^{n-m-1}v_k\sum_{i=k+1}^{n-m}(q_i-v_i)}\\
&=0,\quad \left(\forall (\bY_1,\bY_2,\cdots,\bY_{n-m})\in \prod_{k=1}^{n-m}I^{v_k}\right).
\end{align*}
Similarly we can check that
 $f_m^n(\bp,\bu,\bq,\b0)\equiv 0$.

To confirm that $f_m^n(\bp,\bu,\bq,\bv)(\cdot,\cdot)$ satisfies
 \eqref{eq_vanishing_property}, let us take a function
 $h:\prod_{j=1}^{m+1}[0,\beta)_h^{u_j}\to\C$ satisfying
 \eqref{eq_test_function_time_translation}. Here let us temporarily
 extend the notational rule defined in
 \eqref{eq_time_translation_notational_convention} as follows.
For $\bX=(\rho_1\bx_1s_1\xi_1,\cdots,\rho_n\bx_n s_n\xi_n)\in
 (\{1,2\}\times\Z^d\times\frac{1}{h}\Z\times\{1,-1\})^n$,
 $\bt=(t_1,\cdots,t_n)\in (\frac{1}{h}\Z)^n$, set 
$$
\bX+\bt:=(\rho_1\bx_1(s_1+t_1)\xi_1,\cdots,\rho_n\bx_n(s_n+t_n)\xi_n).
$$
Then, we see that for any $\bX_j\in (I^0)^{u_j}$ $(j=1,2,\cdots,m+1)$, 
$\bY_k\in I^{v_k}$ $(k=1,2,\cdots,n-m)$,
\begin{align}
&\prod_{j=1}^{m+1}\Bigg(\sum_{\bs_j\in [0,\beta)_h^{u_j}}\Bigg)h(\bs_1,\bs_2,\cdots,\bs_{m+1})\label{eq_vanishing_property_check}\\
&\quad\cdot f_m^n(\bp,\bu,\bq,\bv)((\bX_1+\bs_1,\bX_2+\bs_2,\cdots,\bX_{m+1}+\bs_{m+1}),(\bY_1,\bY_2,\cdots,\bY_{n-m}))
\notag\\
&=\frah^{p_1+q_1-u_1-v_1}\sum_{\bW_1\in (I^0)^{p_1-u_1}\atop \bZ_1\in
 I^{q_1-v_1}}
\sum_{\bs_1\in [0,\beta)_h^{u_1}\atop \bt_1\in [0,\beta)_h^{p_1-u_1}}
F_{p_1,q_1}((\bW_1+\bt_1,\bX_1+\bs_1),(\bZ_1,\bY_1))\notag\\
&\quad\cdot h'(\bW_1,\bZ_1)(\bt_1,\bs_1),\notag
\end{align}
where 
\begin{align*}
h'(\bW_1,\bZ_1)(\bt_1,\bs_1)
&:=\prod_{j=2}^{m+1}\Bigg(\sum_{\bs_j\in [0,\beta)_h^{u_j}}\Bigg)h(\bs_1,\bs_2,\cdots,\bs_{m+1})\\
&\quad\cdot
\prod_{j=2}^{m+1}\left(
\frah^{p_j-u_j}\sum_{\bW_j\in
 I^{p_j-u_j}}F_{p_j}^{s_j}(\bW_j,\bX_j+\bs_j)\right)\notag\\
&\quad\cdot \prod_{k=2}^{n-m}
\left(
\frah^{q_k-v_k}\sum_{\bZ_k\in
 I^{q_k-v_k}}F_{q_k}^{t_k}(\bZ_k,\bY_k)\right)\notag\\
&\quad\cdot
 Tree(\{s_j\}_{j=1}^{m+1},\cC)\psi_{\bW_1+\bt_1}^{s_1}
\prod_{j=2}^{m+1}\psi_{\bW_j}^{s_j}\Bigg|_{\psi^{s_j}=0\atop(\forall
 j\in\{1,2,\cdots,m+1\})}\notag\\
&\quad\cdot Tree(\{t_k\}_{k=1}^{n-m},\cC)\prod_{k=1}^{n-m}\psi_{\bZ_k}^{t_k}\Bigg|_{\psi^{t_k}=0\atop(\forall
 k\in\{1,2,\cdots,n-m\})}\notag\\
&\quad\cdot
 (-1)^{\sum_{j=1}^mu_j\sum_{i=j+1}^{m+1}(p_i-u_i)+\sum_{k=1}^{n-m-1}v_k\sum_{i=k+1}^{n-m}(q_i-v_i)}.\notag
\end{align*}
The equality \eqref{eq_time_translation_tree_formula}, the
 property \eqref{eq_time_translation_generic} of $F^j$ and the property
 \eqref{eq_test_function_time_translation} of $h$ imply that
\begin{align*}
&h'(\bW_1,\bZ_1)(r_{\beta}(t_1+s),\cdots,r_{\beta}(t_{p_1-u_1}+s),r_{\beta}(s_1+s),\cdots,r_{\beta}(s_{u_1}+s))\\
&=h'(\bW_1,\bZ_1)(t_1,\cdots,t_{p_1-u_1},s_1,\cdots,s_{u_1}),\\
&\left(\forall t_j\in [0,\beta)_h\ (j=1,\cdots,p_1-u_1),\ s_k\in \0betah\
 (k=1,\cdots,u_1),\ s\in\frac{1}{h}\Z\right).
\end{align*}
Thus, the property \eqref{eq_vanishing_property} of $F_{p_1,q_1}$
 ensures 
that the right-hand side of \eqref{eq_vanishing_property_check}
 vanishes. By the same procedure as above we can check that for any
 function $\phi:\prod_{j=1}^{n-m}[0,\beta)_h^{v_j}\to\C$ satisfying
 \eqref{eq_test_function_time_translation}, $\bX_j\in I^{u_j}$
 $(j=1,2,\cdots,m+1)$, $\bY_k\in (I^0)^{v_k}$ $(k=1,2,\cdots,n-m)$,
\begin{align*}
&\prod_{k=1}^{n-m}\left(\sum_{\bt_k\in[0,\beta)_h^{v_k}}\right)\phi(\bt_1,\bt_2,\cdots,\bt_{n-m})\\
&\cdot f_m^n(\bp,\bu,\bq,\bv)((\bX_1,\bX_2,\cdots,\bX_{m+1}),(\bY_1+\bt_1,\bY_2+\bt_2,\cdots,\bY_{n-m}+\bt_{n-m}))=0.
\end{align*}

After these preparations we define the functions $E_{a,b}^{(n)}:I^a\times
 I^b\to\C$ $(a,b\in \{0,2,4,\cdots,N\})$ by
\begin{align}
E_{a,b}^{(n)}(\bX,\bY)
&:=\sum_{p_1,q_1=2}^{N}1_{p_1,q_1\in 2\N}\sum_{u_1=0}^{p_1}
(1_{m=0}+1_{m\neq 0}1_{u_1\le
 p_1-1})\left(\begin{array}{c} p_1 \\ u_1\end{array}\right)
\label{eq_divided_kernel_definition}\\
&\quad\cdot \sum_{v_1=0}^{q_1}
(1_{m=n-1}+1_{m\neq n-1}1_{v_1\le
 q_1-1})\left(\begin{array}{c} q_1 \\ v_1\end{array}\right)\notag\\
&\quad\cdot \prod_{j=2}^{m+1}\left(\sum_{p_j=2}^N\sum_{u_j=0}^{p_j-1}
\left(\begin{array}{c} p_j \\ u_j\end{array}\right)\right)
\prod_{k=2}^{n-m}\left(\sum_{q_k=2}^N\sum_{v_k=0}^{q_k-1}
\left(\begin{array}{c} q_k \\ v_k\end{array}\right)\right)\notag\\
&\quad\cdot f_m^n((p_j)_{1\le j\le m+1},(u_j)_{1\le j\le
 m+1},(q_j)_{1\le j\le n-m},(v_j)_{1\le j\le
 n-m})\notag\\
&\quad\qquad ((\bX_1',\bX_2',\cdots,\bX_{m+1}'),
(\bY_1',\bY_2',\cdots,\bY_{n-m}'))\notag\\
&\quad\cdot 1_{\sum_{j=1}^{m+1}u_j=a}1_{\sum_{k=1}^{n-m}v_k=b}
1_{\sum_{j=1}^{m+1}p_j-2m\ge a}1_{\sum_{k=1}^{n-m}q_k-2(n-m-1)\ge b}\notag\\
&\quad\cdot \frac{1}{a!b!}\sum_{\s\in\S_a\atop \tau\in
 \S_b}\sgn(\s)\sgn(\tau)
1_{(\bX_1',\bX_2',\cdots,\bX_{m+1}')=\bX_{\s}}
1_{(\bY_1',\bY_2',\cdots,\bY_{n-m}')=\bY_{\tau}}.\notag
\end{align}
By definition $E_{a,b}^{(n)}$ is bi-anti-symmetric. Moreover, it follows
 from the above study on the function $f_m^n(\bp,\bu,\bq,\bv)$ that
 $E_{a,b}^{(n)}$ satisfies \eqref{eq_time_translation_generic},
 \eqref{eq_vanishing_property} and that $E_{a,b}^{(n)}\equiv 0$ if $a=0$
 or $b=0$. Thus, by \eqref{eq_divided_part_characterization}
$$
E^{(n)}(\psi)=\sum_{a,b=2}^{N}1_{a,b\in 2\N}\frah^{a+b}\sum_{\bX\in
 I^a\atop \bY\in I^b}E_{a,b}^{(n)}(\bX,\bY)\psi_{\bX}\psi_{\bY}.
$$
  
To establish upper bounds on the integrals of $E_{a,b}^{(n)}$, let us
 study bound properties of $f_m^n(\bp,\bu,\bq,\bv)$. First we consider the
 case $n=1$. In this case it simply follows from
 \eqref{eq_free_integration_bound_general},
 \eqref{eq_determinant_bound_generic} that for any $\bX_1\in I^{u_1}$,
 $\bY_1\in I^{v_1}$
\begin{align*}
&|f_m^n(\bp,\bu,\bq,\bv)(\bX_1,\bY_1)|\\
&\le
 \frah^{p_1+q_1-u_1-v_1}\sum_{\bW_1\in I^{p_1-u_1}\atop \bZ_1\in
 I^{q_1-v_1}}|F_{p_1,q_1}((\bW_1,\bX_1),(\bZ_1,\bY_1))|D^{\frac{1}{2}(p_1+q_1-u_1-v_1)},
\end{align*}
and thus
\begin{align}
&\|f_m^n(\bp,\bu,\bq,\bv)\|_{1,\infty}\le
 D^{\frac{1}{2}(p_1+q_1-u_1-v_1)}\|F_{p_1,q_1}\|_{1,\infty},\label{eq_divided_inside_1_1_infinity}\\
&\|f_m^n(\bp,\bu,\bq,\bv)\|_{1}\le
 D^{\frac{1}{2}(p_1+q_1-u_1-v_1)}\|F_{p_1,q_1}\|_{1}.\label{eq_divided_inside_1_1}
\end{align}

Next let us assume that $n\ge 2$. By \eqref{eq_phi_good_property},
 \eqref{eq_determinant_bound_general},
 \eqref{eq_determinant_bound_generic}, for any $\bX_j\in I^{u_j}$
 $(j=1,2,\cdots,m+1)$, $\bY_j\in I^{v_j}$ $(j=1,2,\cdots,n-m)$,
\begin{align*}
&|f_m^n(\bp,\bu,\bq,\bv)((\bX_1,\bX_2,\cdots,\bX_{m+1}),(\bY_1,\bY_2,\cdots,\bY_{n-m}))|\\
&\le
 2^{n-1}\sum_{S\in\T(\{s_j\}_{j=1}^{m+1})}\sum_{T\in\T(\{t_k\}_{k=1}^{n-m})}\frah^{p_1+q_1-u_1-v_1}\sum_{\bW_1\in I^{p_1-u_1-d_1(S)}\atop \bW_1'\in I^{d_1(S)}}
\sum_{\bZ_1\in I^{q_1-v_1-d_1(T)}\atop \bZ_1'\in I^{d_1(T)}}\\
&\quad\cdot
\left(\begin{array}{c} p_1-u_1 \\ d_1(S)\end{array}\right)
\left(\begin{array}{c} q_1-v_1 \\ d_1(T)\end{array}\right)
|F_{p_1,q_1}((\bW_1,\bW_1',\bX_1),(\bZ_1,\bZ_1',\bY_1))|\\
&\quad\cdot 
\prod_{j=2}^{m+1}\Bigg(
\frah^{p_j-u_j}\sum_{\bW_j\in I^{p_j-u_j-d_{s_j}(S)}\atop \bW_j'\in
 I^{d_{s_j}(S)}}
\left(\begin{array}{c} p_j-u_j \\ d_{s_j}(S)\end{array}\right)
|F_{p_j}^{s_j}(\bW_j,\bW_j',\bX_j)|\Bigg)\\
&\quad\cdot 
\prod_{k=2}^{n-m}\Bigg(
\frah^{q_k-v_k}\sum_{\bZ_k\in I^{q_k-v_k-d_{t_k}(T)}\atop \bZ_k'\in
 I^{d_{t_k}(T)}}
\left(\begin{array}{c} q_k-v_k \\ d_{t_k}(T)\end{array}\right)
|F_{q_k}^{t_k}(\bZ_k,\bZ_k',\bY_k)|\Bigg)\\
&\quad\cdot D^{\frac{1}{2}(\sum_{j=1}^{m+1}p_j-2m-\sum_{j=1}^{m+1}u_j)
+\frac{1}{2}(\sum_{k=1}^{n-m}q_k-2(n-m-1)-\sum_{k=1}^{n-m}v_k)
}\\
&\quad\cdot \left|\prod_{\{p,q\}\in
 S}\D_{\{p,q\}}(\cC)\prod_{j=1}^{m+1}\psi_{\bW_j'}^{s_j}
\right|
\left|\prod_{\{p,q\}\in
 T}\D_{\{p,q\}}(\cC)\prod_{k=1}^{n-m}\psi_{\bZ_k'}^{t_k}
\right|\\
&=2^{n-1}
D^{-n+1-\frac{1}{2}(\sum_{j=1}^{m+1}u_j+\sum_{k=1}^{n-m}v_k)}
\sum_{S\in\T(\{s_j\}_{j=1}^{m+1})}\sum_{T\in\T(\{t_k\}_{k=1}^{n-m})}\\
&\quad\cdot
\frah^{p_1+q_1-u_1-v_1}\sum_{\bW_1\in I^{p_1-u_1-d_1(S)}\atop \bW_1'\in I^{d_1(S)}}
\sum_{\bZ_1\in I^{q_1-v_1-d_1(T)}\atop \bZ_1'\in I^{d_1(T)}}\\
&\quad\cdot
\left(\begin{array}{c} p_1-u_1 \\ d_1(S)\end{array}\right)
\left(\begin{array}{c} q_1-v_1 \\ d_1(T)\end{array}\right)D^{\frac{1}{2}(p_1+q_1)}
|F_{p_1,q_1}((\bW_1,\bW_1',\bX_1),(\bZ_1,\bZ_1',\bY_1))|\\
&\quad\cdot 
\prod_{j=2}^{m+1}\Bigg(
\frah^{p_j-u_j}\sum_{\bW_j\in I^{p_j-u_j-d_{s_j}(S)}\atop \bW_j'\in
 I^{d_{s_j}(S)}}
\left(\begin{array}{c} p_j-u_j \\ d_{s_j}(S)\end{array}\right)D^{\frac{p_j}{2}}
|F_{p_j}^{s_j}(\bW_j,\bW_j',\bX_j)|\Bigg)\\
&\quad\cdot 
\prod_{k=2}^{n-m}\Bigg(
\frah^{q_k-v_k}\sum_{\bZ_k\in I^{q_k-v_k-d_{t_k}(T)}\atop \bZ_k'\in
 I^{d_{t_k}(T)}}
\left(\begin{array}{c} q_k-v_k \\ d_{t_k}(T)\end{array}\right)D^{\frac{q_k}{2}}
|F_{q_k}^{t_k}(\bZ_k,\bZ_k',\bY_k)|\Bigg)\\
&\quad\cdot \left|\prod_{\{p,q\}\in
 S}\D_{\{p,q\}}(\cC)\prod_{j=1}^{m+1}\psi_{\bW_j'}^{s_j}
\right|
\left|\prod_{\{p,q\}\in
 T}\D_{\{p,q\}}(\cC)\prod_{k=1}^{n-m}\psi_{\bZ_k'}^{t_k}
\right|,
\end{align*}
where for consistency we admit that
\begin{align*}
&\left|\prod_{\{p,q\}\in
 S}\D_{\{p,q\}}(\cC)\prod_{j=1}^{m+1}\psi_{\bW_j'}^{s_j}
\right|:=1\text{ if }m=0,\\
&\left|\prod_{\{p,q\}\in
 T}\D_{\{p,q\}}(\cC)\prod_{k=1}^{n-m}\psi_{\bZ_k'}^{t_k}
\right|:=1\text{ if }m=n-1,
\end{align*}
since in these cases $S$ (or $T$) has no line.
For each $S\in\T(\{s_j\}_{j=1}^{m+1})$,
 $T\in\T(\{t_k\}_{k=1}^{n-m})$
we can draw a line between the vertex $s_1(=1)$ of $S$ and the vertex
 $t_1(=1)$ of $T$ to form a tree containing both $S$ and $T$. To
 estimate $\|f_m^n(\bp,\bu,\bq,\bv)\|_1$, we consider the vertex $n$ as
 the root of this large tree. Then, by repeating the recursive
 estimation from younger branches to the root $n$ and using the
 inequality \eqref{eq_number_trees} we deduce that
\begin{align}
&\|f_m^n(\bp,\bu,\bq,\bv)\|_1\label{eq_divided_inside_1}\\
&\le
 2^{n-1}
D^{-n+1-\frac{1}{2}(\sum_{j=1}^{m+1}u_j+\sum_{k=1}^{n-m}v_k)}
\sum_{S\in\T(\{s_j\}_{j=1}^{m+1})}\sum_{T\in\T(\{t_k\}_{k=1}^{n-m})}\notag\\
&\quad\cdot
\left(\begin{array}{c} p_1-u_1 \\ d_1(S)\end{array}\right)
\left(\begin{array}{c} q_1-v_1 \\ d_1(T)\end{array}\right)d_1(S)!d_1(T)!
D^{\frac{1}{2}(p_1+q_1)}\notag\\
&\quad\cdot \Bigg(1_{n\in \{s_j\}_{j=2}^{m+1}}\sup_{X_0\in
 I}\frah^{p_1+q_1}\sum_{\bX\in I^{p_1}\atop \bY\in
 I^{q_1}}|F_{p_1,q_1}(\bX,\bY)||\tilde{\cC}(X_0,X_1)|\notag\\
&\qquad + 1_{n\in \{t_k\}_{k=2}^{n-m}}
\sup_{Y_0\in
 I}\frah^{p_1+q_1}
\sum_{\bX\in I^{p_1}\atop \bY\in
 I^{q_1}}|F_{p_1,q_1}(\bX,\bY)||\tilde{\cC}(Y_0,Y_1)|\Bigg)\notag\\
&\quad\cdot 
\prod_{j=2}^{m+1}\Bigg(
\left(\begin{array}{c} p_j-u_j \\ d_{s_j}(S)\end{array}\right)d_{s_j}(S)!
D^{\frac{p_j}{2}}
\Bigg(
1_{s_j\neq n}\sup_{X_0\in I}\frah^{p_j}\sum_{\bX\in I^{p_j}}
|F_{p_j}^{s_j}(\bX)||\tilde{\cC}(X_0,X_1)|\notag\\
&\qquad\qquad\qquad\qquad\qquad\qquad\qquad\quad+1_{s_j= n}\|F_{p_j}^{s_j}\|_1
\Bigg)\Bigg)\notag\\
&\quad\cdot 
\prod_{k=2}^{n-m}\Bigg(
\left(\begin{array}{c} q_k-v_k \\ d_{t_k}(T)\end{array}\right)d_{t_k}(T)!
D^{\frac{q_k}{2}}
\Bigg(
1_{t_k\neq n}\sup_{X_0\in I}\frah^{q_k}\sum_{\bX\in I^{q_k}}
|F_{q_k}^{t_k}(\bX)||\tilde{\cC}(X_0,X_1)|\notag\\
&\qquad\qquad\qquad\qquad\qquad\qquad\qquad\quad+1_{t_k= n}\|F_{q_k}^{t_k}\|_1
\Bigg)\Bigg)\notag\\
&\le (1_{m+1\ge 2}(m-1)!+1_{m+1=1})(1_{n-m\ge 2}(n-m-2)!+1_{n-m=1})\notag\\
&\quad\cdot 2^{-2\sum_{j=1}^{m+1}u_j-2\sum_{k=1}^{n-m}v_k}
D^{-n+1-\frac{1}{2}(\sum_{j=1}^{m+1}u_j+\sum_{k=1}^{n-m}v_k)}
\|\tilde{\cC}\|_{1,\infty}^{n-1}\notag\\
&\quad\cdot 2^{2p_1+2q_1}D^{\frac{1}{2}(p_1+q_1)}\|F_{p_1,q_1}\|_{1,\infty}\notag\\
&\quad\cdot 
\prod_{j=2}^{m+1}(2^{2p_j}D^{\frac{p_j}{2}}(1_{s_j\neq n}\|F_{p_j}^{s_j}\|_{1,\infty}+
1_{s_j= n}\|F_{p_j}^{s_j}\|_{1}))\notag\\
&\quad\cdot 
\prod_{k=2}^{n-m}(2^{2q_k}D^{\frac{q_k}{2}}(1_{t_k\neq n}\|F_{q_k}^{t_k}\|_{1,\infty}+
1_{t_k= n}\|F_{q_k}^{t_k}\|_{1})).\notag
\end{align}
We should remark that the inequality 
\begin{align*}
&2^{n-1}(1_{m+1\ge 2}(m-1)!2^{-m-1}+1_{m+1=1})(1_{n-m\ge 2}(n-m-2)!2^{-n+m}
+1_{n-m=1})\\
&\le (1_{m+1\ge 2}(m-1)!+1_{m+1=1})(1_{n-m\ge 2}(n-m-2)!+1_{n-m=1})
\end{align*}
was used to derive the second inequality.

Estimation of $\|f_m^n(\bp,\bu,\bq,\bv)\|_{1,\infty}$ can be done
 similarly. In this case first we fix a component of
 $((\bX_1,\bX_2,\cdots,\bX_{m+1}),(\bY_1,\bY_2,\cdots,\bY_{n-m}))$.
 Then, there uniquely exists a vertex of the enlarged tree containing 
both $S$ and $T$ such that the fixed component is a variable of the
 function $F^j$ or $F_{p,q}$ on the vertex. We consider the vertex as the root of the enlarged tree and repeat the same recursive estimation as above. The
 result is that
\begin{align}
&\|f_m^n(\bp,\bu,\bq,\bv)\|_{1,\infty}\label{eq_divided_inside_1_infinity}\\
&\le (1_{m+1\ge 2}(m-1)!+1_{m+1=1})(1_{n-m\ge 2}(n-m-2)!+1_{n-m=1})\notag\\
&\quad\cdot 2^{-2\sum_{j=1}^{m+1}u_j-2\sum_{k=1}^{n-m}v_k}
D^{-n+1-\frac{1}{2}(\sum_{j=1}^{m+1}u_j+\sum_{k=1}^{n-m}v_k)}
\|\tilde{\cC}\|_{1,\infty}^{n-1}\notag\\
&\quad\cdot 2^{2p_1+2q_1}D^{\frac{1}{2}(p_1+q_1)}\|F_{p_1,q_1}\|_{1,\infty}
\prod_{j=2}^{m+1}(2^{2p_j}D^{\frac{p_j}{2}}\|F_{p_j}^{s_j}\|_{1,\infty})
\prod_{k=2}^{n-m}(2^{2q_k}D^{\frac{q_k}{2}}\|F_{q_k}^{t_k}\|_{1,\infty}).\notag
\end{align}

It follows from the definition \eqref{eq_divided_kernel_definition} that
\begin{align}
\|E_{a,b}^{(n)}\|_{norm}
&\le \sum_{p_1,q_1=2}^{N}1_{p_1,q_1\in 2\N}\sum_{u_1=0}^{p_1}
(1_{m=0}+1_{m\neq 0}1_{u_1\le
 p_1-1})\left(\begin{array}{c} p_1 \\ u_1\end{array}\right)\label{eq_divided_kernel_intermediate_bound}\\
&\quad\cdot \sum_{v_1=0}^{q_1}
(1_{m=n-1}+1_{m\neq n-1}1_{v_1\le
 q_1-1})\left(\begin{array}{c} q_1 \\ v_1\end{array}\right)\notag\\
&\quad\cdot \prod_{j=2}^{m+1}\left(\sum_{p_j=2}^N\sum_{u_j=0}^{p_j-1}
\left(\begin{array}{c} p_j \\ u_j\end{array}\right)\right)
\prod_{k=2}^{n-m}\left(\sum_{q_k=2}^N\sum_{v_k=0}^{q_k-1}
\left(\begin{array}{c} q_k \\ v_k\end{array}\right)\right)
 \|f_m^n(\bp,\bu,\bq,\bv)\|_{norm}\notag\\
&\quad\cdot 1_{\sum_{j=1}^{m+1}u_j=a}1_{\sum_{k=1}^{n-m}v_k=b}
1_{\sum_{j=1}^{m+1}p_j-2m\ge a}1_{\sum_{k=1}^{n-m}q_k-2(n-m-1)\ge b},\notag
\end{align}
where $norm=`1,\infty'$ or $norm=1$.
When $n=1$, by substituting \eqref{eq_divided_inside_1_1_infinity} into
 \eqref{eq_divided_kernel_intermediate_bound} with $norm=`1,\infty'$ we obtain
 \eqref{eq_divided_1_1_infinity}. By substituting
 \eqref{eq_divided_inside_1_1} into
 \eqref{eq_divided_kernel_intermediate_bound} with $norm=1$
  we obtain \eqref{eq_divided_1_1}. Assume that
 $n\ge 2$. By inserting \eqref{eq_divided_inside_1_infinity} into
 \eqref{eq_divided_kernel_intermediate_bound} with $norm=`1,\infty'$ we see that
\begin{align*}
\|E_{a,b}^{(n)}\|_{1,\infty}
&\le \sum_{p_1,q_1=2}^{N}1_{p_1,q_1\in 2\N}\sum_{u_1=0}^{p_1}
\left(\begin{array}{c} p_1 \\ u_1\end{array}\right)
\sum_{v_1=0}^{q_1}
\left(\begin{array}{c} q_1 \\ v_1\end{array}\right)\\
&\quad\cdot \prod_{j=2}^{m+1}\left(\sum_{p_j=2}^N\sum_{u_j=0}^{p_j-1}
\left(\begin{array}{c} p_j \\ u_j\end{array}\right)\right)
\prod_{k=2}^{n-m}\left(\sum_{q_k=2}^N\sum_{v_k=0}^{q_k-1}
\left(\begin{array}{c} q_k \\ v_k\end{array}\right)\right)\\
&\quad\cdot (1_{m+1\ge 2}(m-1)!+1_{m+1=1})(1_{n-m\ge
 2}(n-m-2)!+1_{n-m=1})\\
&\quad\cdot 2^{-2a-2b}D^{-n+1-\frac{1}{2}(a+b)}\|\tilde{\cC}\|_{1,\infty}^{n-1}
2^{2p_1+2q_1}D^{\frac{1}{2}(p_1+q_1)}\|F_{p_1,q_1}\|_{1,\infty}\\
&\quad\cdot\prod_{j=2}^{m+1}(2^{2p_j}D^{\frac{p_j}{2}}\|F_{p_j}^{s_j}\|_{1,\infty})
\prod_{k=2}^{n-m}(2^{2q_k}D^{\frac{q_k}{2}}\|F_{q_k}^{t_k}\|_{1,\infty})\\
&\quad\cdot 1_{\sum_{j=1}^{m+1}p_j-2m\ge a}1_{\sum_{k=1}^{n-m}q_k-2(n-m-1)\ge b},
\end{align*}
which gives \eqref{eq_divided_1_infinity}. By combining
 \eqref{eq_divided_inside_1} with \eqref{eq_divided_kernel_intermediate_bound} with $norm=1$ we obtain \eqref{eq_divided_1}.
\end{proof}

\subsection{Generalized
  covariances}\label{subsec_generalized_covariances}

To construct a double-scale integration process in a generalized
setting, here we list the assumptions on a couple of generalized covariances. Let
$c_0\in \R_{\ge 1}$, $D_c\in \R_{>0}$. We assume that covariances $\cC_0$, $\cC_1:I_0^2\to\C$ satisfy the following properties.
\begin{itemize}
\item $\cC_{1}$ satisfies
      \eqref{eq_time_translation_generic_covariance}.
\item 
\begin{align}
&\cC_{0}(\rho\bx s,\eta \by t)=\cC_{0}(\rho\bx 0,\eta \by
 0),\label{eq_time_independence}\ (\forall \rho,\eta\in \{1,2\},\ \bx,\by\in \G,\ s,t\in
 [0,\beta)_h).
\end{align}
\item 
\begin{align}
&|\det(\<\bu_i,\bv_j\>_{\C^m}\cC_l(X_i,Y_j))_{1\le i,j\le n}|\le
 c_0^n,\label{eq_determinant_bound}\\
&(\forall m,n\in\N,\ \bu_i,\bv_i\in\C^m\text{ with
 }\|\bu_i\|_{\C^m},\|\bv_i\|_{\C^m}\le 1,\notag\\
&\quad X_i,Y_i\in I_0\
 (i=1,2,\cdots,n),\ l\in \{0,1\}).\notag
\end{align}
\item 
\begin{align}
&\|\tilde{\cC_1}\|_{1,\infty}\le c_0.\label{eq_decay_bound}\\
&\|\tilde{\cC_1}\|\le c_0.\label{eq_decay_bound_coupled}\\
&\|\tilde{\cC_0}\|_{1,\infty}\le c_0D_c.\label{eq_decay_bound_final}
\end{align}
\end{itemize}
Here $\tilde{\cC}_l(:I^2\to \C)$ is the
anti-symmetric extension of $\cC_l$ defined as
in \eqref{eq_anti_symmetric_extension}. 
In practice $\cC_{1}$ will be replaced by the free covariance with 
many Matsubara frequencies and the covariance $\cC_{0}$
will be the free covariance containing only one Matsubara frequency closest to the parameter $\theta/2$.
The condition \eqref{eq_time_independence} requires $\cC_0$ to be
independent of the time variables, which may be seen as a strong
assumption at this point. If a covariance sums over only one
time-momentum, then by a gauge transform the covariance can be made
independent of the time variables. It will turn out that because of the
time-independence of $\cC_0$, only negligibly small data bounded by the
inverse volume factor remain after the double-scale integration of the
correction term.

\subsection{The first integration without the artificial
  term}\label{subsec_UV_without_artificial}

Our purpose here is to develop a single-scale analysis concerning the
single-scale integration
\begin{align*}
\log\left(
\int e^{-V(u)(\psi+\psi^1)+W(u)(\psi+\psi^1)}d\mu_{\cC_1}(\psi^1)
\right).
\end{align*}
In fact what we will analyze is an analytic continuation of the above
Grassmann polynomial which a priori makes sense only if the coupling
constant is sufficiently small. The analytically continued polynomial will
be related to the Grassmann integral of the correction term by the
identity theorem in Subsection \ref{subsec_double_scale}. In the
next subsection we will add the artificial term $-A(\psi)$ to the input
$-V(\psi)+W(\psi)$. 

To describe properties of the output of the above integration, we 
introduce a couple of sets of
$\bigwedge\cV$-valued functions. For sets $O,O'$ let $\Map(O,O')$
denote the set of maps from $O$ to $O'$. From now we use a
parameter $\alpha\in \R_{\ge 1}$ in many situations.
In this subsection kernels of Grassmann polynomials are parameterized by
$u\in \overline{D(r)}$. To describe uniform convergent properties of
the kernels, let us modify the norm $\|\cdot\|_{1,\infty}$ defined in
Subsection \ref{subsec_preliminaries} as follows. For $f\in
\Map(\overline{D(r)}, \Map(I^m,\C))$ we set 
$$
\|f\|_{1,\infty,r}:=\sup_{u\in\overline{D(r)}}\|f(u)\|_{1,\infty}.
$$
For notational consistency we set
$$
\|f\|_{1,\infty,r}:=\sup_{u\in\overline{D(r)}}|f(u)|
$$
for $f\in \Map(\overline{D(r)},\C)$ as well. 

With these notations, for $r\in\R_{>0}$ we
define the subset $\cQ(r)$ of $\Map(\overline{D(r)},\bigwedge_{even}\cV)$ as follows. $f$ belongs to
$\cQ(r)$ if and only if the following statements hold.
\begin{itemize}
\item $f\in \displaystyle\Map\left(\overline{D(r)},\bigwedge_{even}\cV\right)$.
\item $u\mapsto f(u)(\psi):\overline{D(r)}\to \bigwedge \cV$ is
 continuous in $\overline{D(r)}$ and analytic in $D(r)$.
\item For any $u\in \overline{D(r)}$ the anti-symmetric kernels
 $f(u)_m:I^m\to \C$ $(m=2,4,\cdots,N)$ satisfy
 \eqref{eq_time_translation_generic} and 
\begin{align}
&\frac{h}{N}\alpha^2 \|f_0\|_{1,\infty,r}\le
 L^{-d},\notag\\
&\sum_{m=2}^{N}c_0^{\frac{m}{2}}\alpha^m\|f_m\|_{1,\infty,r}\le
 L^{-d}.\label{eq_scale_volume_norm}
\end{align}
\end{itemize}
In short the set $\cQ(r)$ gathers Grassmann data bounded by
 $L^{-d}$. 

Next we define the subset $\cR(r)$ of $\Map(\overline{D(r)},\bigwedge_{even}\cV)$ as follows. $f$ belongs to
$\cR(r)$ if and only if the following statements hold.
\begin{itemize}
\item $f\in \displaystyle\Map\left(\overline{D(r)},\bigwedge_{even}\cV\right)$.
\item $u\mapsto f(u)(\psi):\overline{D(r)}\to \bigwedge \cV$ is
 continuous in $\overline{D(r)}$ and analytic in $D(r)$.
\item There exist $f_{p,q}\in \Map(\overline{D(r)},\Map(I^p\times
      I^q,\C))$ $(p,q\in \{2,4,\cdots,N\})$ such that for any $u\in
      \overline{D(r)}$, $p,q\in \{2,4,\cdots,N\}$,
      $f_{p,q}(u):I^p\times I^q\to\C$ is bi-anti-symmetric and satisfies
      \eqref{eq_time_translation_generic}, \eqref{eq_vanishing_property} and
\begin{align}
&f(u)(\psi)=\sum_{p,q=2}^{N}1_{p,q\in 2\N}\frah^{p+q}\sum_{\bX\in
 I^p\atop \bY\in
 I^q}f_{p,q}(u)(\bX,\bY)\psi_{\bX}\psi_{\bY},\notag\\
&\sum_{p,q=2}^{N}1_{p,q\in
 2\N}c_0^{\frac{1}{2}(p+q)}\alpha^{p+q}\|f_{p,q}\|_{1,\infty,r}\le 1.\label{eq_scale_norm}
\end{align}
\end{itemize}
In short the set $\cR(r)$ collects Grassmann data whose kernels have the good property
      \eqref{eq_vanishing_property}.

With fixed $r\in \R_{>0}$ let us define $V^{0-1,1},V^{0-2,1},V^{0,1}\in
      \Map(\overline{D(r)},\bigwedge_{even}\cV)$ as follows.
\begin{align*}
&V^{0-1,1}(u)(\psi):=\frah^{2}\sum_{\bX\in
 I^2}V_{2}^{0-1,1}(u)(\bX)\psi_{\bX},\\
&V^{0-2,1}(u)(\psi):=\frah^{4}\sum_{\bX,\bY\in
 I^2}V_{2,2}^{0-2,1}(u)(\bX,\bY)\psi_{\bX}\psi_{\bY},\\
&V^{0,1}(u)(\psi):=V^{0-1,1}(u)(\psi)+V^{0-2,1}(u)(\psi),\quad (u\in \overline{D(r)}),
\end{align*}
where
\begin{align}
&V_{2}^{0-1,1}(u)(\rho_1\bx_1 s_1\xi_1,\rho_2\bx_2
 s_2\xi_2)\label{eq_0_1_1_initial_kernel}\\
&:=-\frac{1}{2}uL^{-d}h1_{(\rho_1,\bx_1,s_1)=(\rho_2,\bx_2,s_2)}1_{\rho_1=1}
(1_{(\xi_1,\xi_2)=(1,-1)}-1_{(\xi_1,\xi_2)=(-1,1)}),\notag\\
&V_{2,2}^{0-2,1}(u)(\rho_1\bx_1 s_1\xi_1,\rho_2\bx_2
 s_2\xi_2,\eta_1\by_1t_1\zeta_1,\eta_2\by_2t_2\zeta_2)\label{eq_definition_initial_data_kernel}\\
&:=-\frac{1}{4}uL^{-d}h^21_{(\bx_1,s_1,\by_1,t_1)=(\bx_2,s_2,\by_2,t_2)}(h1_{s_1=t_1}-\beta^{-1})\notag\\
&\qquad\cdot \sum_{\s,\tau\in \S_2}\sgn(\s)\sgn(\tau)
1_{(\rho_{\s(1)},\rho_{\s(2)},\eta_{\tau(1)},\eta_{\tau(2)})=(1,2,2,1)}\notag\\
&\qquad\qquad\cdot 1_{(\xi_{\s(1)},\xi_{\s(2)},\zeta_{\tau(1)},\zeta_{\tau(2)})=(1,-1,1,-1)}
.\notag
\end{align}
One can check that $V_2^{0-1,1}(u):I^2\to\C$ is anti-symmetric, 
$V_{2,2}^{0-2,1}(u):I^2\times I^2\to \C$ is
bi-anti-symmetric and 
$V^{0,1}(u)(\psi)$ is equal to the initial data
      $-V(u)(\psi)+W(u)(\psi)$. 
Then, we define $V^{0-1-1,0},V^{0-1-2,0},V^{0-2,0}\in
      \Map(\overline{D(r)},\bigwedge_{even}\cV)$ as follows.
For any $n\in \N$, $u\in \overline{D(r)}$,
\begin{align*}
&V^{0-1-1,0,(n)}(u)(\psi)\\
&:=\frac{1}{n!}Tree(\{1,2,\cdots,n\},\cC_1)\prod_{j=1}^n\left(
\sum_{b_j\in \{1,2\}}V^{0-b_j,1}(u)(\psi^j+\psi)\right)\Bigg|_{\psi^{j}=0\atop(\forall
 j\in\{1,2,\cdots,n\})}
1_{\exists j(b_j=1)},\\
&V^{0-1-2,0,(n)}(u)(\psi)\\
&:=\frah^{4}\sum_{\bX,\bY\in
 I^2}V_{2,2}^{0-2,1}(u)(\bX,\bY)
 \frac{1}{n!}Tree(\{1,2,\cdots,n+1\},\cC_{1})\\
&\quad\cdot (\psi^1+\psi)_{\bX}
(\psi^2+\psi)_{\bY}
\prod_{j=3}^{n+1}V^{0-2,1}(u)(\psi^j+\psi)\Bigg|_{\psi^{j}=0\atop(\forall
 j\in\{1,2,\cdots,n+1\})},\\
&V^{0-2,0,(n)}(u)(\psi)\\
&:=\frac{1}{n!}\sum_{m=0}^{n-1}\sum_{(\{s_j\}_{j=1}^{m+1},
 \{t_k\}_{k=1}^{n-m})\in S(n,m)}
\frah^{4}\sum_{\bX,\bY\in
 I^2}V_{2,2}^{0-2,1}(u)(\bX,\bY)\\
&\qquad\cdot Tree(\{s_j\}_{j=1}^{m+1},\cC_{1})(\psi^{s_1}+\psi)_{\bX}\prod_{j=2}^{m+1}V^{0-2,1}(u)(\psi^{s_j}+\psi)
\Bigg|_{\psi^{s_j}=0\atop(\forall
 j\in\{1,2,\cdots,m+1\})}\\
&\qquad\cdot Tree(\{t_k\}_{k=1}^{n-m},\cC_{1})(\psi^{t_1}+\psi)_{\bY}\prod_{k=2}^{n-m}V^{0-2,1}(u)(\psi^{t_k}+\psi)
\Bigg|_{\psi^{t_k}=0\atop(\forall
 k\in\{1,2,\cdots,n-m\})},
\end{align*}
where 
\begin{align*}
S(n,m):=\left\{(\{s_j\}_{j=1}^{m+1},
 \{t_k\}_{k=1}^{n-m})\ \Bigg|\
 \begin{array}{l}1=s_1<s_2<\cdots<s_{m+1}\le n,\\
                 1=t_1<t_2<\cdots<t_{n-m}\le n,\\
                \{s_j\}_{j=2}^{m+1}\cup
		 \{t_k\}_{k=2}^{n-m}=\{2,3,\cdots,n\},\\  
               \{s_j\}_{j=2}^{m+1}\cap \{t_k\}_{k=2}^{n-m}=\emptyset.
\end{array}
\right\}.
\end{align*}
Then, set 
\begin{align*}
&V^{0-1-j,0}(u)(\psi):=\sum_{n=1}^{\infty}V^{0-1-j,0,(n)}(u)(\psi),\
 (j=1,2),\\
&V^{0-1,0}(u)(\psi):=V^{0-1-1,0}(u)(\psi)+V^{0-1-2,0}(u)(\psi),\\
&V^{0-2,0}(u)(\psi):=\sum_{n=1}^{\infty}V^{0-2,0,(n)}(u)(\psi),
\end{align*}
on the assumption that these series converge in $\bigwedge \cV$. The reason
why we use the label $0-1$, $0-2$ as the 1st superscript is that
these Grassmann data are independent of the
artificial parameters $\la_1,\la_2$ and thus are classified as the data
of degree 0 with $\la_1,\la_2$. In the next subsection we will introduce
the data $V^{1-j}$ ($j=1,2,3$) and $V^2$ which are of degree 1 and of
degree at least 2 with the parameters $\la_1,\la_2$ respectively.
The 2nd superscripts $1$, $0$ indicate the scale of integration. The
data being integrated with the covariance $\cC_1$ have the 2nd
superscript 1, while the data to be integrated with the covariance
$\cC_0$ have the 2nd superscript 0. Thus, it can be read that
$V^{0,1}$ is independent of $\la_1$, $\la_2$ and to be integrated with
$\cC_1$, $V^{0-1,0}$ is independent of $\la_1$, $\la_2$ and to be integrated with
$\cC_0$ and so on. 

We should explain the structure of the above definitions. 
The idea of the following transformation
is essentially same as the equalities \cite[\mbox{(3.38)}]{M},
\cite[\mbox{(IV.15)}]{L}. It follows from the general formulas 
\eqref{eq_tree_formula_basic}, 
\eqref{eq_tree_formula_simple} that
\begin{align}
&\frac{1}{n!}\left(\frac{d}{dz}\right)^n\log\left(\int
 e^{zV^{0,1}(u)(\psi^1+\psi)} d\mu_{\cC_{1}}(\psi^1)
\right)\Bigg|_{z=0}\label{eq_L_M_transformation}\\
&=V^{0-1-1,0,(n)}(u)(\psi)\notag\\
&\quad+
\frah^{4}\sum_{\bX,\bY\in
 I^2}V_{2,2}^{0-2,1}(u)(\bX,\bY)\frac{1}{n!}Tree(\{1,2,\cdots,n\},\cC_{1})\notag\\
&\qquad\qquad\cdot (\psi^{1}+\psi)_{\bX}(\psi^{1}+\psi)_{\bY}
\prod_{j=2}^{n}V^{0-2,1}(u)(\psi^{j}+\psi)
\Bigg|_{\psi^{j}=0\atop(\forall
 j\in\{1,2,\cdots,n\})}\notag\\
&=V^{0-1-1,0,(n)}(u)(\psi)\notag\\
&\quad +\frah^{4}\sum_{\bX,\bY\in
 I^2}V_{2,2}^{0-2,1}(u)(\bX,\bY)\frac{1}{n!}\prod_{j=1}^n\left(\frac{\partial}{\partial z_j}\right)\notag\\
&\qquad\cdot\log\left(\int
 e^{z_1(\psi^1+\psi)_{\bX}(\psi^1+\psi)_{\bY}+\sum_{j=2}^nz_j
V^{0-2,1}(u)(\psi^1+\psi)}
d\mu_{\cC_{1}}(\psi^1)
\right)
\Bigg|_{z_{j}=0\atop(\forall
 j\in\{1,2,\cdots,n\})}\notag\\
&=V^{0-1-1,0,(n)}(u)(\psi)\notag\\
&\quad+\frah^{4}\sum_{\bX,\bY\in
 I^2}V_{2,2}^{0-2,1}(u)(\bX,\bY)\notag\\
&\qquad\cdot \frac{1}{n!}\prod_{j=2}^n\left(\frac{\partial}{\partial
 z_j}\right)
\int (\psi^1+\psi)_{\bX}(\psi^1+\psi)_{\bY}
e^{\sum_{j=2}^nz_j
V^{0-2,1}(u)(\psi^1+\psi)}d\mu_{\cC_{1}}(\psi^1)\notag\\
&\qquad \cdot 
\left(\int
 e^{\sum_{j=2}^nz_j
V^{0-2,1}(u)(\psi^1+\psi)}
d\mu_{\cC_{1}}(\psi^1)
\right)^{-1}
\Bigg|_{z_{j}=0\atop(\forall
 j\in\{2,3,\cdots,n\})}\notag\\
&=V^{0-1-1,0,(n)}(u)(\psi)\notag\\
&\quad+ \frah^{4}\sum_{\bX,\bY\in
 I^2}V_{2,2}^{0-2,1}(u)(\bX,\bY)\frac{1}{n!}\prod_{j=0}^n\left(\frac{\partial}{\partial
 z_j}\right)\notag\\
&\qquad\cdot\Bigg(\log\left(\int
 e^{z_0(\psi^1+\psi)_{\bX}+z_1(\psi^1+\psi)_{\bY}+
\sum_{j=2}^nz_j
V^{0-2,1}(u)(\psi^1+\psi)}
d\mu_{\cC_{1}}(\psi^1)
\right)\notag\\
&\quad\qquad +\log\left(\int
 e^{z_0(\psi^1+\psi)_{\bX}+
\sum_{j=2}^nz_j
V^{0-2,1}(u)(\psi^1+\psi)}
d\mu_{\cC_{1}}(\psi^1)
\right)\notag\\
&\qquad\qquad\cdot\log\left(\int
 e^{z_1(\psi^1+\psi)_{\bY}+
\sum_{j=2}^nz_j
V^{0-2,1}(u)(\psi^1+\psi)}
d\mu_{\cC_{1}}(\psi^1)
\right)\Bigg)\Bigg|_{z_{j}=0\atop(\forall
 j\in\{0,1,\cdots,n\})}\notag\\
&= V^{0-1-1,0,(n)}(u)(\psi)+V^{0-1-2,0,(n)}(u)(\psi)+V^{0-2,0,(n)}(u)(\psi).\notag
\end{align}
Remark that for any $f^j(\psi)\in \bigwedge_{even}\cV$
$(j=1,2,\cdots,n)$ the maps
\begin{align*}
&(z_1,z_2,\cdots,z_n)\mapsto \log\left(\int
 e^{\sum_{j=1}^nz_jf^j(\psi^1+\psi)}
d\mu_{\cC_{1}}(\psi^1)\right),\\
&(z_1,z_2,\cdots,z_n)\mapsto \left(\int
 e^{\sum_{j=1}^nz_jf^j(\psi^1+\psi)}
d\mu_{\cC_{1}}(\psi^1)\right)^{-1}
\end{align*}
are analytic in a neighborhood of the origin and thus the above
transformation holds true. See e.g. \cite{FKT} for properties of inverse
and logarithm of even Grassmann polynomials. 

In the rest of this subsection we prove the following lemma.
\begin{lemma}\label{lem_UV_without_artificial}
For any $\alpha \in [2^3,\infty)$, 
$$
V^{0-1,0}\in\cQ(2^{-9}c_0^{-2}\alpha^{-4}),\quad V^{0-2,0}\in\cR(2^{-9}c_0^{-2}\alpha^{-4}).
$$
\end{lemma}

\begin{remark} 
The reason why we introduce the norm $\|\cdot\|_{1,\infty,r}$ is that we
 want to make
 use of the following fact. If $f^n\in
 \Map(\overline{D(r)},\bigwedge_{even}\cV)$ $(n\in \N)$ satisfy that
 $u\mapsto f^n(u)(\psi)$ is continuous in $\overline{D(r)}$, analytic in
 $D(r)$ $(\forall n\in \N)$ and $\sum_{n=1}^{\infty}\|f_m^n\|_{1,\infty,r}<\infty$ $(\forall
 m\in \{0,2,\cdots,N\})$, then $\sum_{n=1}^{\infty}f^n(u)(\psi)$
 converges for any $u\in \overline{D(r)}$. Moreover, $u\mapsto
 \sum_{n=1}^{\infty}f^n(u)(\psi)$ is continuous in $\overline{D(r)}$ and
 analytic in $D(r)$.
\end{remark}

\begin{proof}[Proof of Lemma \ref{lem_UV_without_artificial}]
We can derive from 
 \eqref{eq_decay_bound_coupled},
\eqref{eq_0_1_1_initial_kernel},
 \eqref{eq_definition_initial_data_kernel} and the uniqueness of
 anti-symmetric kernel that
\begin{align}
&\|V_2^{0-1,1}\|_{1,\infty,r}\le r L^{-d},\label{eq_0_1_1_initial}\\
&\|V_4^{0-2,1}\|_{1,\infty,r}\le \|V_{2,2}^{0-2,1}\|_{1,\infty,r}\le r,\label{eq_initial_1_infinity}\\
&\sup_{u\in \overline{D(r)}}[V_{2,2}^{0-2,1}(u),\tilde{\cC}_1]_{1,\infty}\le
\sup_{s\in [0,\beta)_h}\sup_{Y_0\in
 I}\frac{1}{h}\sum_{(\eta,\by,t,\zeta)\in I}r
 L^{-d}(h1_{s=t}+\beta^{-1})|\tilde{\cC}_1(Y_0,\eta\by t\zeta)|\label{eq_initial_1_infinity_multiplied}\\
&\qquad\qquad\qquad\qquad\qquad\le
 rL^{-d}\|\tilde{\cC}_1\|\le c_0r
 L^{-d}.\notag
\end{align}
In the following we assume that $\alpha\ge 2^3$ and 
\begin{align}
2^9c_0^2\alpha^4r\le 1.\label{eq_assumption_size_radius}
\end{align}

We can use Lemma \ref{lem_tree_bound} to estimate $V^{0-1-1,0,(n)}$. The
 lemma ensures that the anti-symmetric kernel
 $V(u)_m^{0-1-1,0,(n)}(\cdot)$ satisfies
 \eqref{eq_time_translation_generic}.
Moreover, by using \eqref{eq_tree_1_1_infinity},
 \eqref{eq_determinant_bound}, \eqref{eq_0_1_1_initial} we have that
$$
\|V_m^{0-1-1,0,(1)}\|_{1,\infty,r}\le
 \left(\frac{N}{h}\right)^{1_{m=0}}c_0^{1-\frac{m}{2}}rL^{-d}1_{2\ge
 m}.
$$
Thus,
\begin{align}
&\|V_0^{0-1-1,0,(1)}\|_{1,\infty,r}\le
 \frac{N}{h}c_0rL^{-d},\label{eq_0_1_1_0_1}\\
&\sum_{m=2}^Nc_0^{\frac{m}{2}}\alpha^m\|V_m^{0-1-1,0,(1)}\|_{1,\infty,r}\le
 c_0\alpha^2rL^{-d}.\label{eq_0_1_1_2_1}
\end{align}
Also, by \eqref{eq_tree_1_infinity}, \eqref{eq_determinant_bound},
 \eqref{eq_decay_bound}, \eqref{eq_0_1_1_initial},
 \eqref{eq_initial_1_infinity}, for any $n\in \N_{\ge 2}$, $m\in
 \{0,2,\cdots,N\}$,
\begin{align*}
&\|V_m^{0-1-1,0,(n)}\|_{1,\infty,r}\\
&\le 
\left(\frac{N}{h}\right)^{1_{m=0}}2^{-2m}c_0^{-\frac{m}{2}}
\prod_{j=1}^n\left(
\sum_{b_j\in \{1,2\}}\sum_{p_j=2}^42^{3p_j}c_0^{\frac{p_j}{2}}\|V_{p_j}^{0-b_j,1}\|_{1,\infty,r}\right)\\
&\quad\cdot 1_{\sum_{j=1}^np_j-2(n-1)\ge m}1_{\exists j(b_j=1)}\\
&\le \left(\frac{N}{h}\right)^{1_{m=0}}2^{-2m}c_0^{-\frac{m}{2}}
\sum_{l=1}^n\left(\begin{array}{c}n \\ l\end{array}\right)
(2^6c_0\|V_2^{0-1,1}\|_{1,\infty,r})^l
(2^{12}c_0^2\|V_4^{0-2,1}\|_{1,\infty,r})^{n-l}\\
&\quad\cdot 1_{2l+4(n-l)-2(n-1)\ge m}\\
&\le \left(\frac{N}{h}\right)^{1_{m=0}}2^{-2m}c_0^{-\frac{m}{2}}
\sum_{l=1}^n\left(\begin{array}{c}n \\ l\end{array}\right)
(2^6c_0rL^{-d})^l
(2^{12}c_0^2r)^{n-l}
1_{2n-2l+2\ge m}.
\end{align*}
Therefore, by $c_0\ge 1$,
\begin{align}
&\|V_0^{0-1-1,0,(n)}\|_{1,\infty,r}\le\frac{N}{h}(2^{13}c_0^2r)^nL^{-d},\label{eq_0_1_1_0_n}\\
&\sum_{m=2}^Nc_0^{\frac{m}{2}}\alpha^m
\|V_m^{0-1-1,0,(n)}\|_{1,\infty,r}\label{eq_0_1_1_sum_n}\\
&\le \sum_{l=1}^n\left(\begin{array}{c}n \\ l\end{array}\right)
(2^6c_0rL^{-d})^l
(2^{12}c_0^2r)^{n-l}2(2^{-2}\alpha)^{2n-2l+2}\notag\\
&\le 2(2^{-2}\alpha)^2\sum_{l=1}^n\left(\begin{array}{c}n \\ l\end{array}\right)
(2^6c_0rL^{-d})^l
(2^{8}c_0^2\alpha^2r)^{n-l}\notag\\
&\le \alpha^2(2^9c_0^2\alpha^2r)^nL^{-d},\notag
\end{align}
where we used $\alpha \ge 2^3$ so that $2^{-2}\alpha/(2^{-2}\alpha-1)\le
 2$.

Lemma \ref{lem_tree_double_bound} is the tool to estimate
 $V^{0-1-2,0,(n)}$. According to the lemma, the anti-symmetric kernel
 $V(u)_m^{0-1-2,0,(n)}(\cdot)$ satisfies
 \eqref{eq_time_translation_generic}. By substituting
 \eqref{eq_determinant_bound}, \eqref{eq_initial_1_infinity_multiplied} into
 \eqref{eq_double_1_1_infinity} we obtain that
$$
\|V_m^{0-1-2,0,(1)}\|_{1,\infty,r}\le
 2^{8}\left(\frac{N}{h}\right)^{1_{m=0}}c_0^{2-\frac{m}{2}}rL^{-d}1_{2\ge
 m},
$$
or
\begin{align}
&\|V_0^{0-1-2,0,(1)}\|_{1,\infty,r}\le
 2^8\frac{N}{h}c_0^2rL^{-d},\label{eq_0_1_0_1}\\
&\sum_{m=2}^{N}c_0^{\frac{m}{2}}\alpha^m\|V_m^{0-1-2,0,(1)}\|_{1,\infty,r}\le
 2^8c_0^2\alpha^2r L^{-d}.\label{eq_0_1_sum_1}
\end{align}
Also, by \eqref{eq_double_1_infinity}, \eqref{eq_determinant_bound}, 
\eqref{eq_decay_bound},
\eqref{eq_initial_1_infinity} and
 \eqref{eq_initial_1_infinity_multiplied}, for $n\in \N_{\ge 2}$, $m\in
 \{0,2,\cdots,N\}$,
\begin{align*}
\|V_m^{0-1-2,0,(n)}\|_{1,\infty,r}\le 
\left(\frac{N}{h}\right)^{1_{m=0}}2^{-2m}c_0^{-\frac{m}{2}}(2^{12}c_0^2r)^nL^{-d}1_{2n\ge
 m}.
\end{align*}
Thus, 
\begin{align}
&\|V_0^{0-1-2,0,(n)}\|_{1,\infty,r}\le
 \frac{N}{h}(2^{12}c_0^2r)^nL^{-d},\label{eq_0_1_0_n}\\
&\sum_{m=2}^{N}c_0^{\frac{m}{2}}\alpha^m\|V_m^{0-1-2,0,(n)}\|_{1,\infty,r}\le 2
 (2^8c_0^2\alpha^2r)^n L^{-d},\label{eq_0_1_sum_n}
\end{align}
where we used $\alpha \ge 2^3$ so that $2^{-2}\alpha/(2^{-2}\alpha-1)\le
 2$. Then, we see from \eqref{eq_assumption_size_radius}, 
\eqref{eq_0_1_1_0_1},
\eqref{eq_0_1_1_2_1},
\eqref{eq_0_1_1_0_n},
\eqref{eq_0_1_1_sum_n},
\eqref{eq_0_1_0_1}, \eqref{eq_0_1_sum_1}, \eqref{eq_0_1_0_n},
 \eqref{eq_0_1_sum_n} and $\alpha\ge 2^3$ that
\begin{align*}
&\frac{h}{N}\alpha^2\sum_{n=1}^{\infty}\sum_{j=1}^2\|V_0^{0-1-j,0,(n)}\|_{1,\infty,r}\\
&\le \left(2^{-9}\alpha^{-2}+\alpha^2\sum_{n=2}^{\infty}(2^4\alpha^{-4})^n+2^{-1}\alpha^{-2}+\alpha^2\sum_{n=2}^{\infty}(2^3\alpha^{-4})^n
\right)L^{-d}\le L^{-d},\\
&\sum_{m=2}^Nc_0^{\frac{m}{2}}\alpha^m\sum_{n=1}^{\infty}\sum_{j=1}^2\|V_m^{0-1-j,0,(n)}\|_{1,\infty,r}\\
&\le 
\left(2^{-9}\alpha^{-2}+\alpha^2\sum_{n=2}^{\infty}\alpha^{-2n}+2^{-1}\alpha^{-2}+2\sum_{n=2}^{\infty}(2^{-1}\alpha^{-2})^n
\right)L^{-d}\le L^{-d}.
\end{align*}
This implies that 
 $V^{0-1,0}\in \cQ(2^{-9}c_0^{-2}\alpha^{-4})$.

Let us consider $V^{0-2,0,(n)}$. By Lemma \ref{lem_tree_divided_bound},
 for $n\in\N$,
$m\in \{0,1,\cdots,n-1\}$,
 $(S,T)\in S(n,m)$, $a,b\in
 \{2,4,\cdots,N\}$, $u\in\overline{D(r)}$ there exists a function
 $E_{a,b}^{(n,m,S,T)}(u):I^a\times I^b\to \C$ such that
 $E_{a,b}^{(n,m,S,T)}(u)$ is bi-anti-symmetric, satisfies
 \eqref{eq_time_translation_generic}, \eqref{eq_vanishing_property} and 
\begin{align*}
&V^{0-2,0,(n)}(u)(\psi)\\
&=\frac{1}{n!}\sum_{m=0}^{n-1}\sum_{(S,T)\in
 S(n,m)}\sum_{a,b=2}^{N}1_{a,b\in 2\N}\frah^{a+b}\sum_{\bX\in I^a\atop
 \bY\in I^b}E_{a,b}^{(n,m,S,T)}(u)(\bX,\bY)\psi_{\bX}\psi_{\bY}.
\end{align*}
Define the function $V^{0-2,0,(n)}_{a,b}(u):I^a\times I^b\to\C$ by
\begin{align*}
V^{0-2,0,(n)}_{a,b}(u)(\cdot,\cdot):=\frac{1}{n!}\sum_{m=0}^{n-1}\sum_{(S,T)\in
 S(n,m)}E_{a,b}^{(n,m,S,T)}(u)(\cdot,\cdot).
\end{align*}
Then, $V^{0-2,0,(n)}_{a,b}(u)(\cdot,\cdot)$ is bi-anti-symmetric, satisfies
 \eqref{eq_time_translation_generic}, \eqref{eq_vanishing_property} and 
$$
V^{0-2,0,(n)}(u)(\psi)=\sum_{a,b=2}^{N}1_{a,b\in2\N}\frah^{a+b}
\sum_{\bX\in I^a\atop
 \bY\in I^b}V^{0-2,0,(n)}_{a,b}(u)(\bX,\bY)\psi_{\bX}\psi_{\bY}.
$$
It is also clear from the construction that $u\mapsto
 V^{0-2,0,(n)}_{a,b}(u)(\bX,\bY)$ is continuous in $\overline{D(r)}$ and
 analytic in $D(r)$, $(\forall \bX\in I^a,\bY\in I^b)$.
 Let us prove bound properties of $V^{0-2,0,(n)}_{a,b}(u)$. The
 inequalities proved in Lemma \ref{lem_tree_divided_bound} support our
 analysis. Note that
\begin{align}
\sharp S(n,m)=\left(\begin{array}{c} n-1 \\
		    m\end{array}\right).\label{eq_cardinality_index}
\end{align}
By \eqref{eq_divided_1_1_infinity}, \eqref{eq_initial_1_infinity}, 
\begin{align}
\|V_{2,2}^{0-2,0,(1)}\|_{1,\infty,r}\le r.\label{eq_0_2_1}
\end{align}
Combination of \eqref{eq_divided_1_infinity},
 \eqref{eq_determinant_bound}, \eqref{eq_decay_bound},
 \eqref{eq_initial_1_infinity}, \eqref{eq_cardinality_index} yields that
 for $n\in \N_{\ge 2}$, $a,b\in \{2,4,\cdots,N\}$,
\begin{align*}
&\|V_{a,b}^{0-2,0,(n)}\|_{1,\infty,r}\\
&\le \frac{1}{n!}\sum_{m=0}^{n-1}\left(\begin{array}{c} n-1 \\ m
				       \end{array}\right)
(1_{m\neq 0}(m-1)!+1_{m=0})
(1_{m\neq n-1}(n-m-2)!+1_{m=n-1})\\
&\quad\cdot 2^{-2a-2b}c_0^{-\frac{1}{2}(a+b)}(2^{12}c_0^2r)^n1_{2+2m\ge a}1_{2n-2m\ge
 b}.
\end{align*}
Note that
\begin{align}
\sum_{m=0}^{n-1}\left(\begin{array}{c} n-1 \\ m
				       \end{array}\right)
(1_{m\neq 0}(m-1)!+1_{m=0})
(1_{m\neq n-1}(n-m-2)!+1_{m=n-1})\le n!.\label{eq_combinatorial_simplification}
\end{align}
By using \eqref{eq_combinatorial_simplification} and
 $2^{-2}\alpha/(2^{-2}\alpha-1)\le 2$ we can derive that
\begin{align}
\sum_{a,b=2}^{N}1_{a,b\in 2\N}c_0^{\frac{1}{2}(a+b)}\alpha^{a+b}\|V^{0-2,0,(n)}_{a,b}\|_{1,\infty,r}
\le \alpha^2(2^8c_0^2\alpha^2r)^n.
\label{eq_0_2_sum_n}
\end{align}
It follows from \eqref{eq_assumption_size_radius}, \eqref{eq_0_2_1},
 \eqref{eq_0_2_sum_n} and $\alpha\ge 2^3$ that
\begin{align*}
\sum_{a,b=2}^{N}1_{a,b\in 2\N}c_0^{\frac{1}{2}(a+b)}\alpha^{a+b}\sum_{n=1}^{\infty}\|V^{0-2,0,(n)}_{a,b}\|_{1,\infty,r}
\le 2^{-9}+2^{-1}\alpha^{-2}\le 1.
\end{align*}
Thus we conclude that $V^{0-2,0}\in
 \cR(2^{-9}c_0^{-2}\alpha^{-4})$.
\end{proof}

\subsection{The first integration with the artificial
  term}\label{subsec_UV_with_artificial}

In this subsection we perform a single-scale integration where Grassmann
polynomials are dependent on the artificial parameter
$\bla=(\la_1,\la_2)$. To be specific, we are going to analyze an analytic
continuation of the Grassmann polynomial
\begin{align*}
\log\left(\int e^{-V(u)(\psi+\psi^1)+W(u)(\psi+\psi^1)-A(\psi+\psi^1)}d\mu_{\cC_1}(\psi^1)\right).
\end{align*}
For this purpose we need to introduce sets of Grassmann polynomials
parameterized by $(u,\bla)$. Bound properties of these Grassmann
polynomials are measured in a variant of the $L^1$-norm $\|\cdot\|_{1}$, 
while polynomials belonging to $\cQ(r)$, $\cR(r)$ were measured in the norm $\|\cdot\|_{1,\infty,r}$. To prove uniform
bounds with $(u,\bla)$, we modify the norm $\|\cdot\|_1$ defined in
Subsection \ref{subsec_preliminaries} as follows. For $f\in
\Map(\overline{D(r)}\times\overline{D(r')}^2,\Map(I^m,\C))$ let 
$$
\|f\|_{1,r,r'}:=\sup_{u\in\overline{D(r)}\atop \bla\in \overline{D(r')}^2}\|f(u,\bla)\|_1.
$$
Also for $f\in \Map(\overline{D(r)}\times\overline{D(r')}^2,\C)$ we set
$$
\|f\|_{1,r,r'}:=\sup_{u\in\overline{D(r)}\atop \bla\in \overline{D(r')}^2}|f(u,\bla)|
$$ 
for notational consistency. 

For $r,r'\in \R_{>0}$ we define the subset
$\cQ'(r,r')$ of $\Map(\overline{D(r)}\times\C^2,\bigwedge_{even}\cV)$
as follows. $f$ belongs to $\cQ'(r,r')$ if and only if 
\begin{itemize}
\item 
$$
f\in \Map\left(\overline{D(r)}\times\C^2,\bigwedge_{even}\cV\right).
$$
\item
For any $u\in \overline{D(r)}$, $\bla\mapsto f(u,\bla)(\psi):\C^2\to \bigwedge \cV$ is linear.
\item 
For any $\bla\in\C^2$, $u\mapsto f(u,\bla)(\psi):\overline{D(r)}\to
      \bigwedge \cV$ is continuous in $\overline{D(r)}$ and analytic in
      $D(r)$.
\item
For any $(u,\bla)\in \overline{D(r)}\times \C^2$ the anti-symmetric
      kernels $f(u,\bla)_m:I^m\to \C$ $(m=2,4,\cdots,N)$ satisfy
      \eqref{eq_time_translation_generic} and 
\begin{align}
&\alpha^2\|f_0\|_{1,r,r'}\le L^{-d},\notag\\
&\sum_{m=2}^Nc_0^{\frac{m}{2}}\alpha^m\|f_m\|_{1,r,r'}\le L^{-d}.\label{eq_scale_volume_norm_linear}
\end{align}
\end{itemize}
In other words the set $\cQ'(r,r')$ contains Grassmann polynomials
which are linearly dependent on $\bla$ and become negligibly small as
$L\to \infty$.

We also need a set containing Grassmann polynomials with
bi-anti-symmetric kernels linearly depending on $\bla$. 
For $r,r'\in\R_{>0}$ the subset 
$\cR'(r,r')$ of $\Map(\overline{D(r)}\times\C^2,\bigwedge_{even}\cV)$
is defined as follows. 
 $f$ belongs to $\cR'(r,r')$ if and only if 
\begin{itemize}
\item 
$$
f\in \Map\left(\overline{D(r)}\times\C^2,\bigwedge_{even}\cV\right).
$$
\item
For any $u\in \overline{D(r)}$, $\bla\mapsto f(u,\bla)(\psi):\C^2\to \bigwedge \cV$ is linear.
\item 
For any $\bla\in\C^2$, $u\mapsto f(u,\bla)(\psi):\overline{D(r)}\to
      \bigwedge \cV$ is continuous in $\overline{D(r)}$ and analytic in
      $D(r)$.
\item There exist $f_{p,q}\in
      \Map(\overline{D(r)}\times\C^2,\Map(I^p\times I^q,\C))$
      $(p,q=2,4,\cdots,N)$ such that for any $(u,\bla)\in
      \overline{D(r)}\times\C^2$, $p,q\in \{2,4,\cdots,N\}$,
      $f_{p,q}(u,\bla):I^p\times I^q\to\C$ is bi-anti-symmetric,
      satisfies \eqref{eq_time_translation_generic},
      \eqref{eq_vanishing_property} and
\begin{align}
&f(u,\bla)(\psi)=\sum_{p,q=2}^{N}1_{p,q\in 2\N}\frah^{p+q}\sum_{\bX\in
 I^p\atop \bY\in
 I^q}f_{p,q}(u,\bla)(\bX,\bY)\psi_{\bX}\psi_{\bY},\notag\\
&\sum_{p,q=2}^{N}1_{p,q\in 2\N}
c_0^{\frac{1}{2}(p+q)}\alpha^{p+q}\|f_{p,q}\|_{1,r,r'}\le 1.\label{eq_scale_norm_linear}
\end{align}
\end{itemize}

We introduce another set of Grassmann polynomials with linear dependence
      on $\bla$, which is used to contain the offspring of the
      artificial term $A(\psi)$. 
For $r,r'\in\R_{>0}$, 
 $f$ belongs to $\cS(r,r')$ if and only if 
\begin{itemize}
\item 
$$
f\in \Map\left(\overline{D(r)}\times\C^2,\bigwedge_{even}\cV\right).
$$
\item
For any $u\in \overline{D(r)}$, $\bla\mapsto f(u,\bla)(\psi):\C^2\to \bigwedge \cV$ is linear.
\item 
For any $\bla\in\C^2$, $u\mapsto f(u,\bla)(\psi):\overline{D(r)}\to
      \bigwedge \cV$ is continuous in $\overline{D(r)}$ and analytic in
      $D(r)$.
\item For any $(u,\bla)\in \overline{D(r)}\times\C^2$ the anti-symmetric
      kernels $f(u,\bla)_m:I^m\to\C$ $(m=2,4,\cdots,N)$ satisfy
      \eqref{eq_time_translation_generic} and 
\begin{align}
&\alpha^2\|f_0\|_{1,r,r'}\le 1,\notag\\
&\sum_{m=2}^N
c_0^{\frac{m}{2}}\alpha^{m}\|f_{m}\|_{1,r,r'}\le 1.\label{eq_offspring_norm_bound}
\end{align}
\end{itemize}

Finally we introduce a set of Grassmann polynomials whose degree with
      $\bla$ is more than 1.
For $r,r'\in\R_{>0}$, 
 $f$ belongs to $\cW(r,r')$ if and only if 
\begin{itemize}
\item 
$$
f\in \Map\left(\overline{D(r)}\times\overline{D(r')}^2,\bigwedge_{even}\cV\right).
$$
\item $(u,\bla)\mapsto f(u,\bla)(\psi)$ is continuous in
      $\overline{D(r)}\times\overline{D(r')}^2$ and analytic in
      $D(r)\times D(r')^2$.
\item For any $u\in D(r)$, $j\in \{1,2\}$,
$$
f(u,\b0)(\psi)=\frac{\partial}{\partial\la_j}f(u,\b0)(\psi)=0.
$$
\item For any $(u,\bla)\in \overline{D(r)}\times\overline{D(r')}^2$ the anti-symmetric
      kernels $f(u,\bla)_m:I^m\to\C$ $(m=2,4,\cdots,N)$ satisfy
      \eqref{eq_time_translation_generic} and 
\begin{align}
&\alpha^2\|f_0\|_{1,r,r'}\le 1,\notag\\
&\sum_{m=2}^N
c_0^{\frac{m}{2}}\alpha^{m}\|f_{m}\|_{1,r,r'}\le
 1.\label{eq_scale_norm_quadratic}
\end{align}
\end{itemize}

Here let us systematically define the input and the output of the
single-scale integration. We admit the results of
Lemma \ref{lem_UV_without_artificial} claiming that
$$
V^{0-1,0}\in\cQ(2^{-9}c_0^{-2}\alpha^{-4}),\quad V^{0-2,0}\in\cR(2^{-9}c_0^{-2}\alpha^{-4})
$$
and define $V^{0,0}\in \Map\left(\overline{D(2^{-9}c_0^{-2}\alpha^{-4})},\bigwedge_{even}\cV\right)$
by 
$$V^{0,0}:=V^{0-1,0}+V^{0-2,0}.$$
We define $V^{1,1}\in
\Map(\C^2,\bigwedge_{even}\cV)$ by
$$
V^{1,1}(\bla)(\psi):=-A(\psi),
$$
where $A(\psi)$ is the Grassmann polynomial defined in
\eqref{eq_artificial_Grassmann_term}. 
Then, by recalling the formula \eqref{eq_tree_formula_simple} let us
observe the following expansion.
\begin{align}
&\frac{1}{n!}\left(\frac{d}{dz}\right)^n\log\left(\int
 e^{z(V^{0,1}(u)(\psi^1+\psi)+V^{1,1}(\bla)(\psi^1+\psi))}
d\mu_{\cC_{1}}(\psi^1)
\right)\Bigg|_{z=0}\label{eq_decomposition_start_artificial}\\
&=\frac{1}{n!}Tree(\{1,2,\cdots,n\},\cC_{1})\prod_{j=1}^n
\left(\sum_{b=0}^1V^{b,1}(\psi^j+\psi)\right)\Bigg|_{\psi^{j}=0\atop(\forall
 j\in\{1,2,\cdots,n\})}\notag\\
&=\frac{1}{n!}Tree(\{1,2,\cdots,n\},\cC_{1})\prod_{j=1}^nV^{0,1}(\psi^j+\psi)\Bigg|_{\psi^{j}=0\atop(\forall
 j\in\{1,2,\cdots,n\})}\notag\\
&\quad+1_{n=1}Tree(\{1\},\cC_{1})V^{1,1}(\psi^1+\psi)\Big|_{\psi^{1}=0}\notag\\
&\quad +1_{n\ge 2}
 \frac{1}{(n-1)!}Tree(\{1,2,\cdots,n\},\cC_{1})\notag\\
&\qquad\qquad\qquad\cdot V^{1,1}(\psi^1+\psi)
\prod_{j=2}^nV^{0,1}(\psi^j+\psi)\Bigg|_{\psi^{j}=0\atop(\forall
 j\in\{1,2,\cdots,n\})}\notag\\
&\quad +
 \frac{1}{n!}Tree(\{1,2,\cdots,n\},\cC_{1})
\prod_{j=1}^n\left(\sum_{b_j=0}^1V^{b_j,1}(\psi^j+\psi)
\right)\Bigg|_{\psi^{j}=0\atop(\forall
 j\in\{1,2,\cdots,n\})}1_{\sum_{j=1}^nb_j\ge 2}.\notag
\end{align}
We further decompose or rename each term of this expansion from top to
bottom. It follows from \eqref{eq_L_M_transformation}
that if we set for $n\in
\N$ 
\begin{align*}
V^{0,0,(n)}(\psi):=\frac{1}{n!}Tree(\{1,2,\cdots,n\},\cC_{1})\prod_{j=1}^nV^{0,1}(\psi^j+\psi)\Bigg|_{\psi^{j}=0\atop(\forall
 j\in\{1,2,\cdots,n\})},
\end{align*}
then $V^{0,0}(\psi)=\sum_{n=1}^{\infty}V^{0,0,(n)}(\psi)$. Let us set
\begin{align*}
V^{1-3,0}(\psi):=Tree(\{1\},\cC_{1})V^{1,1}(\psi^1+\psi)\Big|_{\psi^{1}=0}.
\end{align*}
For $n\in \N_{\ge 2}$ we set 
\begin{align*}
&V^{1-1-1,0,(n)}(\psi)\\
&:=\frac{1}{(n-1)!}Tree(\{1,2,\cdots,n\},\cC_1)\\
&\quad\cdot \prod_{j=1}^{n-1}\Bigg(\sum_{b_j=1}^2V^{0-b_j,1}(\psi^j+\psi)\Bigg)V^{1,1}(\psi^n+\psi)
\Bigg|_{\psi^{j}=0\atop(\forall
 j\in\{1,2,\cdots,n\})}1_{\exists j(b_j=1)},\\
&V^{1-1-2,0,(n)}(\psi)\\
&:=
 \frac{1}{(n-1)!}Tree(\{1,2,\cdots,n+1\},\cC_{1})\frah^{4}\sum_{\bX,\bY\in I^2} V_{2,2}^{0-2,1}(\bX,\bY)\\
&\qquad\cdot (\psi^1+\psi)_{\bX}(\psi^2+\psi)_{\bY}
\prod_{j=3}^nV^{0-2,1}(\psi^j+\psi)\cdot V^{1,1}(\psi^{n+1}+\psi)\Bigg|_{\psi^{j}=0\atop(\forall
 j\in\{1,2,\cdots,n+1\})},\\
&V^{1-2,0,(n)}(\psi)\\
&:=\frac{1}{(n-1)!}\sum_{m=0}^{n-1}\sum_{(\{s_j\}_{j=1}^{m+1},\{t_k\}_{k=1}^{n-m})\in
 S(n,m)}\frah^{4}\sum_{\bX,\bY\in I^2} V_{2,2}^{0-2,1}(\bX,\bY)\\
&\qquad\cdot Tree(\{s_j\}_{j=1}^{m+1},\cC_{1})(\psi^{s_1}+\psi)_{\bX}\\
&\qquad\cdot\prod_{j=2}^{m+1}(1_{s_j\neq n}V^{0-2,1}(\psi^{s_j}+\psi)+
1_{s_j= n}V^{1,1}(\psi^{s_j}+\psi))\Bigg|_{\psi^{s_j}=0\atop(\forall
 j\in\{1,2,\cdots,m+1\})}\\
&\qquad\cdot Tree(\{t_k\}_{k=1}^{n-m},\cC_{1})(\psi^{t_1}+\psi)_{\bY}\\
&\qquad\cdot\prod_{k=2}^{n-m}(1_{t_k\neq n}V^{0-2,1}(\psi^{t_k}+\psi)+
1_{t_k= n}V^{1,1}(\psi^{t_k}+\psi))\Bigg|_{\psi^{t_k}=0\atop(\forall
 k\in\{1,2,\cdots,n-m\})}.
\end{align*}
By the same argument as in \eqref{eq_L_M_transformation} we can derive
that 
\begin{align*}
&\frac{1}{(n-1)!}Tree(\{1,2,\cdots,n\},\cC_{1}) V^{1,1}(\psi^1+\psi)
\prod_{j=2}^n V^{0,1}(\psi^j+\psi)\Bigg|_{\psi^{j}=0\atop(\forall
 j\in\{1,2,\cdots,n\})}\\
&=\frac{1}{(n-1)!}Tree(\{1,2,\cdots,n\},\cC_{1})\\
&\quad\cdot \prod_{j=1}^{n-1}
 V^{0,1}(\psi^j+\psi)\cdot 
 V^{1,1}(\psi^n+\psi)\Bigg|_{\psi^{j}=0\atop(\forall
 j\in\{1,2,\cdots,n\})}\\
&=V^{1-1-1,0,(n)}(\psi)+V^{1-1-2,0,(n)}(\psi)+V^{1-2,0,(n)}(\psi).
\end{align*}
Finally we set for $n\in \N_{\ge 2}$,
\begin{align*}
&V^{2,0,(n)}(\psi)\\
&:=\frac{1}{n!}Tree(\{1,2,\cdots,n\},\cC_{1})
\prod_{j=1}^n\left(\sum_{b_j=0}^1V^{b_j,1}(\psi^j+\psi)
\right)\Bigg|_{\psi^{j}=0\atop(\forall
 j\in\{1,2,\cdots,n\})}1_{\sum_{j=1}^nb_j\ge 2}.
\end{align*}
Then, the expansion
\eqref{eq_decomposition_start_artificial} can be equivalently written as
follows.
\begin{align*}
&\frac{1}{n!}\left(\frac{d}{dz}\right)^n\log\left(\int
 e^{z(V^{0,1}(u)(\psi^1+\psi)+V^{1,1}(\bla)(\psi^1+\psi))} d\mu_{\cC_{1}}(\psi^1)
\right)\Bigg|_{z=0}\\
&=V^{0,0,(n)}(\psi)+1_{n=1}V^{1-3,0}(\psi)\\
&\quad + 1_{n\ge
 2}(V^{1-1-1,0,(n)}(\psi)+V^{1-1-2,0,(n)}(\psi)+V^{1-2,0,(n)}(\psi)+V^{2,0,(n)}(\psi)).\end{align*}
By assuming their convergence let us set
\begin{align*}
&V^{1-1-j,0}(\psi):=\sum_{n=2}^{\infty}V^{1-1-j,0,(n)}(\psi),\
 (j=1,2),\quad V^{1-1,0}(\psi):=\sum_{j=1}^2V^{1-1-j,0}(\psi),\\
& V^{1-2,0}(\psi):=\sum_{n=2}^{\infty}V^{1-2,0,(n)}(\psi),\quad V^{2,0}(\psi):=\sum_{n=2}^{\infty}V^{2,0,(n)}(\psi).
\end{align*}
Then, it follows that 
\begin{align*}
&\sum_{n=1}^{\infty}\frac{1}{n!}\left(\frac{d}{dz}\right)^n\log\left(\int
 e^{z(V^{0,1}(u)(\psi^1+\psi)+V^{1,1}(\bla)(\psi^1+\psi))} d\mu_{\cC_{1}}(\psi^1)
\right)\Bigg|_{z=0}\\
&=V^{0,0}(\psi)+\sum_{j=1}^3V^{1-j,0}(\psi)+V^{2,0}(\psi).
\end{align*}
Our purpose is to prove that these Grassmann polynomials
are indeed convergent and they have desired invariant and bound properties. Not to
confuse, we should keep in mind that the data $V^{0,j}$ $(j\in \{0,1\})$
are independent
of the artificial parameter $\bla$, the data $V^{1,1}$, $V^{1-j,0}$ $(j\in\{1,2,3\})$ are
linearly dependent on $\bla$ and the data $V^{2,0}$ depends on $\bla$
at least quadratically. The input have the 2nd superscript 1 and the
output have the 2nd superscript 0 in this single-scale integration. 
More detailed properties of these Grassmann data
are summarized in the following lemma. 

\begin{lemma}\label{lem_UV_with_artificial}
For any $\alpha \in [2^3,\infty)$,
\begin{align*}
&V^{1-1,0}\in \cQ'(2^{-9}c_0^{-2}\alpha^{-4},2^{-9}\beta^{-1}c_0^{-2}\alpha^{-4}),\\
&V^{1-2,0}\in \cR'(2^{-9}c_0^{-2}\alpha^{-4},2^{-9}\beta^{-1}c_0^{-2}\alpha^{-4}),\\
&V^{1-3,0}\in\cS(2^{-9}c_0^{-2}\alpha^{-4},2^{-9}\beta^{-1}c_0^{-2}\alpha^{-4}),\\
&V^{2,0}\in \cW(2^{-9}c_0^{-2}\alpha^{-4},2^{-9}\min\{1,\beta\}
\beta^{-1}c_0^{-2}\alpha^{-4}).
\end{align*}
\end{lemma}

\begin{remark}
It is clear from the definition that $V^{1-3,0}$ is independent of the
 parameter $u$. The condition on the first
 variable assumed in the set $\cS(r,r')$ is in fact
 unnecessary. However, we define the set in this way in accordance
 with the other sets.
\end{remark}

\begin{proof}[Proof of Lemma \ref{lem_UV_with_artificial}]
During the proof we often hide the sign of dependency on the parameter
 $(u,\bla)$ for conciseness. In the following we always assume that
 $\alpha\ge 2^3$, \eqref{eq_assumption_size_radius} and 
\begin{align}
2^9\beta c_0^2 \alpha^4 r'\le 1.
\label{eq_assumption_size_radius_artificial}
\end{align}
Let us start by estimating $V^{1,1}$ and $V^{1-3,0}$.
 Since
 $V_4^{1,1}(\psi)=-\la_2 A^2(\psi)$,
\begin{align*}
&V_4^{1,1}(\rho_1\bx_1 s_1\xi_1,\rho_2\bx_2 s_2\xi_2,\rho_3\bx_3 s_3\xi_3,\rho_4\bx_4 s_4\xi_4)\\
&=-\frac{\la_2h^3}{4!}1_{s_1=s_2=s_3=s_4}\\
&\qquad\cdot \sum_{\s\in
 \S_4}\sgn(\s)1_{((\rho_{\s(1)},\bx_{\s(1)},\xi_{\s(1)}),
(\rho_{\s(2)},\bx_{\s(2)},\xi_{\s(2)}),
(\rho_{\s(3)},\bx_{\s(3)},\xi_{\s(3)}),
(\rho_{\s(4)},\bx_{\s(4)},\xi_{\s(4)}))\atop =
 ((1,r_L(\hat{\bx}),1),(2,r_L(\hat{\bx}),-1),
(2,r_L(\hat{\by}),1),(1,r_L(\hat{\by}),-1))},\\
&(\forall (\rho_j,\bx_j,s_j,\xi_j)\in I\ (j=1,2,3,4)).
\end{align*}
Thus, 
\begin{align}
\|V_4^{1,1}\|_{1,r,r'}=\|V_4^{1-3,0}\|_{1,r,r'}
\le \beta r'.
\label{eq_offspring_bound_4th}
\end{align}
Also,
\begin{align*}
&V_2^{1,1}(\rho_1\bx_1 s_1\xi_1,\rho_2\bx_2 s_2\xi_2)\\
&=-\frac{\la_1h}{2}1_{s_1=s_2}\sum_{\s\in
 \S_2}\sgn(\s)1_{((\rho_{\s(1)},\bx_{\s(1)},\xi_{\s(1)}),
(\rho_{\s(2)},\bx_{\s(2)},\xi_{\s(2)}))=((1,r_L(\hat{\bx}),1),(2,r_L(\hat{\bx}),-1))},\\
&(\forall (\rho_j,\bx_j,s_j,\xi_j)\in I\ (j=1,2)).
\end{align*}
Thus, 
\begin{align}
\|V_2^{1,1}\|_{1,r,r'}\le \beta r'.
\label{eq_offspring_bound_2nd_initial}
\end{align}
We can derive from the definition that
\begin{align*}
&V_2^{1-3,0}(\psi)\\
&=V_2^{1,1}(\psi)\notag\\
&\quad +\frah^2\sum_{\bX\in I^2}
\Bigg(\left(\begin{array}{c}4 \\ 2 \end{array}\right)
\frah^2\sum_{\bY\in I^2}
V_4^{1,1}(\bY,\bX)Tree(\{1\},\cC_{1})\psi_{\bY}^1\Big|_{\psi^1=0}
\Bigg)\psi_{\bX}.
\end{align*}
By using \eqref{eq_determinant_bound}, \eqref{eq_offspring_bound_4th},
 \eqref{eq_offspring_bound_2nd_initial} and $c_0\ge 1$ we have
\begin{align}
\|V_2^{1-3,0}\|_{1,r,r'}&\le \|V_2^{1,1}\|_{1,r,r'}
+ \left(\begin{array}{c}4 \\ 2 \end{array}\right) c_0\|V_4^{1,1}\|_{1,r,r'}
\le 7\beta c_0r'.
\label{eq_offspring_bound_2nd}
\end{align}
It also follows from \eqref{eq_determinant_bound}, 
 \eqref{eq_offspring_bound_4th}, \eqref{eq_offspring_bound_2nd_initial},
 $c_0\ge 1$ and the definition that
\begin{align}
\|V_0^{1-3,0}\|_{1,r,r'}
\le c_0\|V_2^{1,1}\|_{1,r,r'}+c_0^2\|V_4^{1,1}\|_{1,r,r'}\le 2\beta c_0^2r'.\label{eq_offspring_bound_0th}
\end{align}
The inequalities \eqref{eq_offspring_bound_4th},
 \eqref{eq_offspring_bound_2nd}, \eqref{eq_offspring_bound_0th} result
 in 
\begin{align*}
\alpha^2\|V_0^{1-3,0}\|_{1,r,r'}\le 2 \beta c_0^2\alpha^2r',\quad \sum_{m=2}^Nc_0^{\frac{m}{2}}\alpha^m\|V_m^{1-3,0}\|_{1,r,r'}\le
 2^3\beta c_0^{2}\alpha^4 r'.
\end{align*}
Though we can see from the explicit characterization of the kernels, the
 statement of Lemma \ref{lem_tree_bound} ensures that $V_m^{1-3,0}:I^m\to \C$ $(m=2,4)$ satisfy
 \eqref{eq_time_translation_generic}. It is also clear from the
 definition that $\bla\mapsto V^{1-3,0}(\bla)(\psi)$ is linear. Combined with these basic properties, the
 above inequalities and \eqref{eq_assumption_size_radius_artificial} imply that 
\begin{align}
V^{1-3,0}\in \cS(2^{-9}c_0^{-2}\alpha^{-4},2^{-9}\beta^{-1}c_0^{-2}\alpha^{-4}).
\label{eq_offspring_conclusion}
\end{align}

Let us consider $V^{1-1-1,0,(n)}(\psi)$. Here we use Lemma
 \ref{lem_tree_bound}. The lemma states that the anti-symmetric kernels
 of $V^{1-1-1,0,(n)}(\psi)$ satisfy
 \eqref{eq_time_translation_generic}. By definition, $\bla\mapsto
 V^{1-1-1,0,(n)}(\bla)(\psi)$ is linear. Thus, $\sum_{n=2}^{\infty}V^{1-1-1,0,(n)}$
must satisfy these properties if it is convergent. Let us establish
 bound properties of the kernels. By applying \eqref{eq_tree_1} together
 with \eqref{eq_determinant_bound}, \eqref{eq_decay_bound},
 \eqref{eq_0_1_1_initial}, \eqref{eq_initial_1_infinity},
 \eqref{eq_offspring_bound_4th}, \eqref{eq_offspring_bound_2nd_initial}
 we observe that for any $n\in \N_{\ge 2}$, $m\in \{0,2,\cdots,N\}$, 
\begin{align*}
&\|V_m^{1-1-1,0,(n)}\|_{1,r,r'}\\
&\le 2^{-2m}c_0^{-\frac{m}{2}}\prod_{j=1}^{n-1}\Bigg(
\sum_{b_j=1}^2\sum_{p_j\in\{2,4\}}2^{3p_j}c_0^{\frac{p_j}{2}}\|V_{p_j}^{0-b_j,1}\|_{1,\infty,r}\Bigg)\\
&\quad\cdot 
\sum_{p_n\in\{2,4\}}2^{3p_n}c_0^{\frac{p_n}{2}}\|V_{p_n}^{1,1}\|_{1,r,r'}1_{\sum_{j=1}^np_j-2(n-1)\ge m}1_{\exists j(b_j=1)}\\
&\le
 2^{-2m}c_0^{-\frac{m}{2}}\sum_{l=1}^{n-1}\left(\begin{array}{c}n-1\\
						l\end{array}\right)
(2^6c_0\|V_2^{0-1,1}\|_{1,\infty,r})^l
(2^{12}c_0^2\|V_4^{0-2,1}\|_{1,\infty,r})^{n-1-l}\\
&\quad\cdot \sum_{p_n\in\{2,4\}}2^{3p_n}c_0^{\frac{p_n}{2}}\|V_{p_n}^{1,1}\|_{1,r,r'}1_{2l+4(n-1-l)+p_n-2(n-1)\ge m}\\
&\le 2^{-2m+13}c_0^{-\frac{m}{2}}
\sum_{l=1}^{n-1}\left(\begin{array}{c}n-1\\
						l\end{array}\right)
(2^6c_0 rL^{-d})^l
(2^{12}c_0^2 r)^{n-1-l}
c_0^2\beta r'\\
&\quad\cdot 1_{2(n-1-l)+4\ge m}.
\end{align*}
Then, by \eqref{eq_assumption_size_radius},
 \eqref{eq_assumption_size_radius_artificial} and $\alpha \ge 2^3$,
\begin{align}
&\|V_0^{1-1-1,0,(n)}\|_{1,r,r'}\label{eq_1_1_1_0_n}\\
&\le 2^4\sum_{l=1}^{n-1}\left(\begin{array}{c}n-1\\
						l\end{array}\right)
(2^{-3}\alpha^{-4})^l(2^3\alpha^{-4})^{n-1-l}L^{-d}\alpha^{-4}
\le (2^4\alpha^{-4})^nL^{-d},\notag\\
&\sum_{m=2}^Nc_0^{\frac{m}{2}}\alpha^m\|V_m^{1-1-1,0,(n)}\|_{1,r,r'}\label{eq_1_1_1_sum_n}\\
&\le 2^6\sum_{l=1}^{n-1}\left(\begin{array}{c}n-1\\
						l\end{array}\right)(2^6c_0rL^{-d})^l(2^8c_0^2\alpha^2r)^{n-1-l}c_0^2\alpha^4\beta
 r'\notag\\
&\le \sum_{l=1}^{n-1}\left(\begin{array}{c}n-1\\
						l\end{array}\right)(2^{-3}\alpha^{-4})^l(2^{-1}\alpha^{-2})^{n-1-l}L^{-d}\le
 \alpha^{-2(n-1)}L^{-d}.\notag
\end{align}

Let us study properties of $V^{1-1-2,0}$. By Lemma
 \ref{lem_tree_double_bound} the anti-symmetric kernels of
 $V^{1-1-2,0,(n)}(\psi)$ satisfy
 \eqref{eq_time_translation_generic}. Thus, if $\sum_{n=2}^{\infty}V^{1-1-2,0,(n)}(\psi)$ converges, the
 anti-symmetric kernels of $V^{1-1-2,0}(\psi)$ must satisfy
 \eqref{eq_time_translation_generic} as well. We can see from the definition
 that $\bla\mapsto V^{1-1-2,0,(n)}(\bla)(\psi)$ is linear and thus so must
 be $V^{1-1-2,0}(\psi)$ if it converges. Let us find upper bounds on the norms of the kernels of
 $V^{1-1-2,0,(n)}(\psi)$. By substituting \eqref{eq_determinant_bound},
 \eqref{eq_decay_bound}, \eqref{eq_initial_1_infinity},
 \eqref{eq_initial_1_infinity_multiplied},
 \eqref{eq_offspring_bound_4th}, \eqref{eq_offspring_bound_2nd_initial}
 into \eqref{eq_double_1} we have that for any $m\in \{0,2,\cdots,N\}$,
 $n\in \N_{\ge 2}$,
\begin{align*}
\|V_m^{1-1-2,0,(n)}\|_{1,r,r'}&\le
 2^{-2m}c_0^{-\frac{m}{2}}L^{-d}(2^{12}c_0^2r)^{n-1}\sum_{p\in
 \{2,4\}}2^{3p}c_0^{\frac{p}{2}}\beta r'1_{2n-4+p\ge m}\\
&\le
 2^{-2m+1}c_0^{-\frac{m}{2}}L^{-d}(2^{12}c_0^2r)^{n-1}(2^{12}c_0^2\beta
 r')1_{2n\ge m}.
\end{align*}
Thus, by \eqref{eq_assumption_size_radius},
 \eqref{eq_assumption_size_radius_artificial} and the assumption $\alpha \ge 2^3$,
 \begin{align}
&\|V_0^{1-1-2,0,(n)}\|_{1,r,r'}\le
 2L^{-d}(2^{3}\alpha^{-4})^{n},\label{eq_1_1_2_0_n}\\
&\sum_{m= 2}^Nc_0^{\frac{m}{2}}\alpha^m\|V_m^{1-1-2,0,(n)}\|_{1,r,r'}
\le 2^2(2^8c_0^2\alpha^2r)^{n-1}(2^8c_0^2\alpha^2\beta
  r')L^{-d} \le 2^2(2^{-1}\alpha^{-2})^nL^{-d}.\label{eq_1_1_2_sum_n}
\end{align}
It follows from \eqref{eq_1_1_1_0_n}, \eqref{eq_1_1_1_sum_n}, 
\eqref{eq_1_1_2_0_n}, \eqref{eq_1_1_2_sum_n} and $\alpha\ge 2^3$ that 
\begin{align*}
&\alpha^2\sum_{n=2}^{\infty}\sum_{j=1}^2\|V_0^{1-1-j,0,(n)}\|_{1,r,r'}\le L^{-d},\quad\sum_{m=
 2}^Nc_0^{\frac{m}{2}}\alpha^m\sum_{n=2}^{\infty}\sum_{j=1}^2\|V_m^{1-1-j,0,(n)}\|_{1,r,r'}\le L^{-d}.
\end{align*}
These uniform convergence properties imply the well-definedness of
 $V^{1-1,0}$ and its regularity with $(u,\bla)$.
Therefore, $V^{1-1,0}\in \cQ'(r,r')$.

Next let us consider $V^{1-2,0}$. An application of 
Lemma \ref{lem_tree_divided_bound} ensures that 
 there exist bi-anti-symmetric functions $V_{a,b}^{1-2,0,(n)}:I^a\times
 I^b\to \C$ $(a,b\in \{2,4,\cdots,N\})$ satisfying
 \eqref{eq_time_translation_generic}, \eqref{eq_vanishing_property} such
 that
\begin{align*}
V^{1-2,0,(n)}(\psi)=\sum_{a,b=2}^{N}1_{a,b\in
 2\N}\frah^{a+b}\sum_{\bX\in I^a\atop \bY\in
 I^b}V^{1-2,0,(n)}_{a,b}(\bX,\bY)\psi_{\bX}\psi_{\bY}.
\end{align*}
By definition, $\bla\mapsto V^{1-2,0,(n)}(u,\bla)(\psi)$ is linear for
 any $u\in \overline{D(r)}$. Moreover, by construction, $(u,\bla)\mapsto
 V^{1-2,0,(n)}_{a,b}(u,\bla)(\bX,\bY)$ is continuous in
 $\overline{D(r)}\times \overline{D(r')}^2$ and analytic in 
$D(r)\times D(r')^2$, $(\forall \bX\in I^a,\bY\in I^b)$.
Let us establish bound properties of the
 bi-anti-symmetric kernels. By combining \eqref{eq_determinant_bound},
 \eqref{eq_decay_bound}, \eqref{eq_initial_1_infinity},
 \eqref{eq_cardinality_index}, 
 \eqref{eq_offspring_bound_4th}, \eqref{eq_offspring_bound_2nd_initial}
  with \eqref{eq_divided_1} and using $c_0\ge 1$ we observe that for any $a,b\in
 \{2,4,\cdots,N\}$, $n\in \N_{\ge 2}$, 
\begin{align*}
&\|V_{a,b}^{1-2,0,(n)}\|_{1,r,r'}\\
&\le \frac{1}{(n-1)!}\sum_{m=0}^{n-1}\left(\begin{array}{c} n-1 \\ m
				       \end{array}\right)
(1_{m\neq 0}(m-1)!+1_{m=0})
(1_{m\neq n-1}(n-m-2)!+1_{m=n-1})\\
&\quad\cdot 2^{-2a-2b}c_0^{-\frac{1}{2}(a+b)}(2^{12}c_0^2r)^{n-1}
(2^6c_0\beta r'+ 2^{12}c_0^2\beta r')
1_{2+2m\ge a}1_{2n-2m\ge b}\\
&\le \frac{1}{(n-1)!}\sum_{m=0}^{n-1}\left(\begin{array}{c} n-1 \\ m
				       \end{array}\right)
(1_{m\neq 0}(m-1)!+1_{m=0})
(1_{m\neq n-1}(n-m-2)!+1_{m=n-1})\\
&\quad\cdot 2^{-2a-2b+13}c_0^{-\frac{1}{2}(a+b)}(2^{12}c_0^2r)^{n-1}
c_0^2\beta r'1_{2+2m\ge a}1_{2n-2m\ge b}.
\end{align*}
Thus, by \eqref{eq_assumption_size_radius}, \eqref{eq_combinatorial_simplification},
 \eqref{eq_assumption_size_radius_artificial} and $\alpha \ge 2^{3}$,
\begin{align*}
\sum_{a,b=2}^{N}1_{a,b\in
 2\N}c_0^{\frac{1}{2}(a+b)}\alpha^{a+b}\|V_{a,b}^{1-2,0,(n)}\|_{1,r,r'}\le 
2^3n(2^{-2}\alpha)^2
(2^{-1}\alpha^{-2})^n\le \alpha^{2-2n},
\end{align*}
or
\begin{align*}
\sum_{a,b=2}^{N}1_{a,b\in
 2\N}c_0^{\frac{1}{2}(a+b)}\alpha^{a+b}\sum_{n=2}^{\infty}\|V_{a,b}^{1-2,0,(n)}\|_{1,r,r'}\le
 2\alpha^{-2}\le 1.
\end{align*}
This means that $V^{1-2,0}\in \cR'(r,r')$.

It remains to analyze $V^{2,0}$. By Lemma \ref{lem_tree_bound} the
 anti-symmetric kernels of $V^{2,0,(n)}(\psi)$ $(n\in \N_{\ge 2})$
 satisfy \eqref{eq_time_translation_generic}. The constraint
 $1_{\sum_{j=1}^{n}b_j\ge 2}$ implies that $V^{2,0,(n)}(\psi)$ is of
 degree at least 2 with $\la_1,\la_2$. Thus, 
$$
V^{2,0,(n)}(u,\b0)(\psi)=\frac{\partial}{\partial
 \la_j}V^{2,0,(n)}(u,\b0)(\psi)=0,\quad(\forall u\in
 {D(r)},\ j\in \{1,2\}).
$$
Let us prove uniform bound properties of the anti-symmetric
 kernels. Here we
 need to measure $V^{1,1}$ with the $\|\cdot\|_{1,\infty}$-norm as
 well. We can see from the definition that for $m\in \{2,4\}$
\begin{align}
\sup_{\la\in \overline{D(\beta r')}}\|V^{1,1}_m(\la)\|_{1,\infty}\le \beta r'.
\label{eq_offspring_bound_L_1}
\end{align}
By definition,
\begin{align*}
V^{2,0,(n)}(\psi)=\frac{1}{n!}\sum_{l=2}^n
\left(\begin{array}{c}n \\ l \end{array}
\right)&Tree(\{1,2,\cdots,n\},\cC_1)\prod_{j=1}^lV^{1,1}(\psi^j+\psi)\\  
&\cdot \prod_{k=l+1}^nV^{0,1}(\psi^k+\psi)
\Bigg|_{\psi^{j}=0\atop(\forall
 j\in\{1,2,\cdots,n\})}.
\end{align*}
Then, it follows from \eqref{eq_tree_1}, \eqref{eq_determinant_bound},
 \eqref{eq_decay_bound}, \eqref{eq_0_1_1_initial},
\eqref{eq_initial_1_infinity},
 \eqref{eq_offspring_bound_4th}, \eqref{eq_offspring_bound_2nd_initial},
 \eqref{eq_offspring_bound_L_1} and $c_0\ge 1$ that for any $m\in \{0,2,\cdots,N\}$,
 $n\in \N_{\ge 2}$,
\begin{align*}
&\|V^{2,0,(n)}_m\|_{1,r,\min\{1,\beta\}r'}\\
&\le \frac{(n-2)!}{n!}\sum_{l=2}^n
\left(\begin{array}{c}n \\ l \end{array}
\right)c_0^{-\frac{m}{2}}2^{-2m}\sum_{p_1\in \{2,4\}}2^{3p_1}c_0^{\frac{p_1}{2}}\|V^{1,1}_{p_1}\|_{1,r,r'}\\
&\quad\cdot \prod_{j=2}^{l}\Bigg(\sum_{p_j\in
 \{2,4\}}2^{3p_j}c_0^{\frac{p_j}{2}}\sup_{\la\in \overline{D(\beta
 r')}}\|V^{1,1}_{p_j}(\la)\|_{1,\infty}\Bigg)(2^6c_0\|V_2^{0,1}\|_{1,\infty,r}+
2^{12}c_0^2\|V_4^{0,1}\|_{1,\infty,r})^{n-l}\\
&\quad\cdot 1_{\sum_{j=1}^lp_j+4(n-l)-2(n-1)\ge m}\\
&\le \frac{(n-2)!}{n!}\sum_{l=2}^n
\left(\begin{array}{c}n \\ l \end{array}
\right)c_0^{-\frac{m}{2}}2^{-2m}(2^{13}c_0^2\beta r')^l(2^{13}c_0^2r)^{n-l}1_{2n+2\ge m}.
\end{align*}
Moreover by
 \eqref{eq_assumption_size_radius},
 \eqref{eq_assumption_size_radius_artificial},
\begin{align*}
\|V^{2,0,(n)}_m\|_{1,r,\min\{1,\beta\}r'}&\le
 \frac{(n-2)!}{n!}\sum_{l=2}^n
\left(\begin{array}{c}n \\ l \end{array}
\right)c_0^{-\frac{m}{2}}2^{-2m}(2^{4}\alpha^{-4})^l(2^{4}\alpha^{-4})^{n-l}1_{2n+2\ge m}\\
&\le c_0^{-\frac{m}{2}}2^{-2m}(2^{5}\alpha^{-4})^n1_{2n+2\ge m}.
\end{align*}
Thus by $\alpha \ge 2^3$,
\begin{align*}
&\alpha^2\sum_{n=2}^{\infty}\|V_0^{2,0,(n)}\|_{1,r,\min\{1,\beta\}
 r'}\le 2 \alpha^2
 (2^5\alpha^{-4})^{2}\le 1,\\
&\sum_{m=2}^Nc_0^{\frac{m}{2}}\alpha^m\|V_m^{2,0,(n)}\|_{1,r,\min\{1,\beta\}r'}\le
 2(2^{-2}\alpha)^2(2\alpha^{-2})^{n},\\
&\sum_{m=2}^Nc_0^{\frac{m}{2}}\alpha^m\sum_{n=2}^{\infty}\|V_m^{2,0,(n)}\|_{1,r,\min\{1,\beta\}r'}\le
 2^2(2^{-2}\alpha)^2(2\alpha^{-2})^{2}\le 1.
\end{align*}
This implies that $V^{2,0}\in \cW(r,\min\{1,\beta\}r')$. The proof is complete.
\end{proof}

\subsection{The second integration}\label{subsec_second_integration}

Here we establish bound properties of the output of the single-scale
integration with the covariance $\cC_0$. The input to the integration is
the Grassmann polynomials $V^{0-j,0}(\psi)$ $(j=1,2)$, $V^{1-k,0}(\psi)$
$(k=1,2,3)$, $V^{2,0}(\psi)$ whose properties were studied in Lemma
\ref{lem_UV_without_artificial} and Lemma \ref{lem_UV_with_artificial}. 
In fact the object we are going to analyze is an analytic continuation
of 
\begin{align*}
\log\left(\int e^{\sum_{j=1}^2V^{0-j,0}(\psi)+
 \sum_{k=1}^3V^{1-k,0}(\psi) + V^{2,0}(\psi)}d\mu_{\cC_0}(\psi)\right),
\end{align*}
which is also an analytic continuation of 
\begin{align*}
\log\left(
\int e^{-V(u)(\psi)+W(u)(\psi)-A(\psi)} d\mu_{\cC_0+\cC_1}(\psi)\right).
 \end{align*}

Set 
$$
r:=2^{-9}c_0^{-2}\alpha^{-4},\quad r':=\beta^{-1}r,\quad r'':=\min\{1,\beta\}r'.
$$
We define $V^{end}$, $V^{1-3,end}\in \Map(\overline{D(r)}\times
\overline{D(r'')}^2,\C)$ by
\begin{align*}
&V^{end,(n)}:=\frac{1}{n!}Tree(\{1,2,\cdots,n\},\cC_0)\\
&\qquad\qquad\quad\cdot \prod_{j=1}^n\left(
\sum_{m=1}^2V^{0-m,0}(\psi^j)+\sum_{k=1}^3V^{1-k,0}(\psi^j)+V^{2,0}(\psi^j)
\right)\Bigg|_{\psi^{j}=0\atop(\forall
 j\in\{1,2,\cdots,n\})}.\\
&V^{end}:=\sum_{n=1}^{\infty}V^{end,(n)},\\
&V^{1-3,end}:=Tree(\{1\},\cC_0)V^{1-3,0}(\psi^1)\Big|_{\psi^1=0}
\end{align*}
by assuming its convergence. Our purpose here is to prove the following lemma.

\begin{lemma}\label{lem_UV_final}
Assume that $h\ge 1$. Then, the following statements hold for any
 $\alpha \in [2^3,\infty)$, $L\in \N$ with $L^d\ge 2^2 D_c$.
\begin{itemize}
\item $V^{end}$ is continuous in
      $\overline{D(2^{-9}c_0^{-2}\alpha^{-4})}\times
      \overline{D(2^{-11}L^{-d} h^{-1}\beta^{-1}\min\{1,\beta\}
c_0^{-2}\alpha^{-4})}^2$, analytic in ${D(2^{-9}c_0^{-2}\alpha^{-4})}\times
      {D(2^{-11}L^{-d} h^{-1}\beta^{-1}\min\{1,\beta\}
c_0^{-2}\alpha^{-4})}^2$.
\item 
\begin{align}
\frac{h}{N}|V^{end}(u,\b0)|\le 2^8\alpha^{-2}L^{-d},\quad (\forall u\in
 \overline{D(2^{-9}c_0^{-2}\alpha^{-4})}).\label{eq_UV_final_pressure}
\end{align}
\item 
\begin{align}
&\left|
\frac{\partial}{\partial \la_j}V^{end}(u,\b0)-
\frac{\partial}{\partial \la_j}V^{1-3,end}(u,\b0)
\right|\le 2^{10}\beta c_0^2\alpha^4(1+2D_c)L^{-d},\label{eq_UV_final_error_estimate}\\
&(\forall u\in
 {D(2^{-9}c_0^{-2}\alpha^{-4})},\ j\in \{1,2\}).\notag
\end{align}
\end{itemize}
\end{lemma}
\begin{proof} First let us observe that $V^{0-2,0}(\psi)$,
 $V^{1-2,0}(\psi)$ do not contribute to the value of the
 integration. With the aim of proving this, let us take $f(\psi)\in
 \bigwedge \cV$, $p,q\in \{2,4,\cdots,N\}$, $\bX\in (I^0)^p$, $\bY\in
 I^q$. If we define the function $g:[0,\beta)_h^p\to\C$ by 
$$
g(s_1,s_2,\cdots,s_p):=\int \prod_{j=1}^p(\psi_{X_j+s_j})\psi_{\bY}f(\psi)d\mu_{\cC_0}(\psi),
$$
the property \eqref{eq_time_independence} ensures that the function $g$
 satisfies \eqref{eq_test_function_time_translation}. If we expand $\int
 V^{j-2,0}(\psi)f(\psi)d\mu_{\cC_0}(\psi)$ $(j=0,1)$, we see that each
 kernel of $V^{j-2,0}$ is multiplied by a function of the same form as $g$
 and is integrated with
 respect to the time-variables. Thus, the property
 \eqref{eq_vanishing_property} of the bi-anti-symmetric kernels of
 $V^{j-2,0}(\psi)$ $(j=0,1)$ implies that
$$
\int V^{j-2,0}(\psi)f(\psi)d\mu_{\cC_0}(\psi)=0,\quad (j=0,1).
$$
Arbitrariness of $f(\psi)$ implies that for any $z\in \C$
\begin{align*}
&\int
 e^{z(\sum_{k=1}^2V^{0-k,0}(\psi)+\sum_{k=1}^3V^{1-k,0}(\psi)+V^{2,0}(\psi))}
 d\mu_{\cC_0}(\psi)\\
&=
\int
 e^{z(V^{0-1,0}(\psi)+V^{1-1,0}(\psi)+V^{1-3,0}(\psi)+V^{2,0}(\psi))}
 d\mu_{\cC_0}(\psi).
\end{align*}
Therefore, 
\begin{align*}
V^{end,(n)}=\frac{1}{n!}&Tree(\{1,2,\cdots,n\},\cC_0)\\
&\cdot \prod_{j=1}^n(V^{0-1,0}(\psi^j)+
V^{1-1,0}(\psi^j)+V^{1-3,0}(\psi^j)+V^{2,0}(\psi^j))
\Bigg|_{\psi^{j}=0\atop(\forall
 j\in\{1,2,\cdots,n\})}.
\end{align*}
Note that
\begin{align}
&\|V_m^{a,0}(u,\eps \bla)\|_1\le \eps
 \|V_m^{a,0}\|_{1,r,r''},\label{eq_artificial_parameter_scaling_again}\\
&\|V_m^{a,0}(u,\eps \bla)\|_{1,\infty}\le h\eps
 \|V_m^{a,0}\|_{1,r,r''},\label{eq_artificial_parameter_scaling_again_h}\\
&(\forall u\in \overline{D(r)},\ \bla\in \overline{D(r'')}^2,\ \eps\in
 [0,1/2],\ a\in \{1-1,1-3,2\}).\notag
\end{align}
For $a=1-1,1-3$, \eqref{eq_artificial_parameter_scaling_again} and
 \eqref{eq_artificial_parameter_scaling_again_h} are clear. 
For $a=2$ we can use the following equality based on Cauchy's integral
 formula to derive \eqref{eq_artificial_parameter_scaling_again},
 \eqref{eq_artificial_parameter_scaling_again_h}.
\begin{align*}
&V_m^{2,0}(u,\eps\bla)=\sum_{n=2}^{\infty}\frac{1}{2\pi
 i}\oint_{|z|=\delta}dz \frac{V_m^{2,0}(u,z\bla)}{z^{n+1}}\eps^n
=\frac{1}{2\pi i}\oint_{|z|=\delta}dz
 V_m^{2,0}(u,z\bla)\frac{\eps^2}{z^2(z-\eps)},\\
&(\forall u\in \overline{D(r)},\ \bla\in \overline{D(r'')}^2,\ \eps\in
 [0,1/2],\ \delta\in (1/2,1)).
\end{align*}
 In the following we let $\eps=\frac{1}{3}L^{-d}h^{-1}$, $\alpha \ge
 2^3$. The assumption $h\ge 1$ implies that $\eps \in (0,1/2]$.
Take any $u\in \overline{D(r)}$, $\bla \in \overline{D(r'')}^2$. 
By \eqref{eq_tree_1_1}, \eqref{eq_determinant_bound},
 \eqref{eq_scale_volume_norm}, \eqref{eq_scale_volume_norm_linear},
 \eqref{eq_offspring_norm_bound}, \eqref{eq_scale_norm_quadratic},
 \eqref{eq_artificial_parameter_scaling_again} we have that
\begin{align*}
&|V^{end,(1)}(u,\eps\bla)|\\
&\le\frac{N}{h}L^{-d}\alpha^{-2}+3\eps
 \alpha^{-2}+\sum_{m=2}^Nc_0^{\frac{m}{2}}\left(\frac{N}{h}\|V_m^{0-1,0}\|_{1,\infty,r}
+\eps\sum_{a\in \{1-1,1-3,2\}}\|V_m^{a,0}\|_{1,r,r''}\right)\\
&\le
 2\left(\frac{N}{h}L^{-d}+3\eps\right)\alpha^{-2}=2(N+1)h^{-1}L^{-d}\alpha^{-2}.
\end{align*}
Also by \eqref{eq_tree_1}, \eqref{eq_determinant_bound},
 \eqref{eq_decay_bound_final}, \eqref{eq_scale_volume_norm},
 \eqref{eq_scale_volume_norm_linear}, \eqref{eq_offspring_norm_bound},
 \eqref{eq_scale_norm_quadratic},
 \eqref{eq_artificial_parameter_scaling_again}, 
\eqref{eq_artificial_parameter_scaling_again_h} and $\alpha\ge 2^3$,
 for $n\in \N_{\ge 2}$
\begin{align*}
&|V^{end,(n)}(u,\eps\bla)|\\
&\le D_c^{n-1}
\left(\sum_{p=2}^N2^{3p}c_0^{\frac{p}{2}}\left(\frac{N}{h}\|V_p^{0-1,0}\|_{1,\infty,r}
+\eps \sum_{a\in \{1-1,1-3,2\}}\|V_p^{a,0}\|_{1,r,r''}\right)\right)\\
&\quad\cdot \left(\sum_{q=2}^N2^{3q}c_0^{\frac{q}{2}}\left(\|V_q^{0-1,0}\|_{1,\infty,r}+h\eps
\sum_{a\in \{1-1,1-3,2\}}\|V_q^{a,0}\|_{1,r,r''}
\right)
\right)^{n-1}\\
&\le D_c^{n-1}\left(\frac{N}{h}L^{-d}+3\eps\right)(L^{-d}+3h\eps)^{n-1}(2^6\alpha^{-2})^n\\
&= (N+1)h^{-1}L^{-d}(2D_cL^{-d})^{n-1}(2^6\alpha^{-2})^n.
\end{align*}
Thus, if $2D_cL^{-d}\le 1/2$,
\begin{align*}
\sum_{n=1}^{\infty}\sup_{u\in\overline{D(r)}\atop \bla\in
 \overline{D(\eps r'')}^2}|V^{end,(n)}(u,\bla)|
&\le 2(N+1)h^{-1}L^{-d}\alpha^{-2}+2^{12}
 (N+1)h^{-1}L^{-d}\alpha^{-4}\\
&\le 2^7(N+1)h^{-1}L^{-d}\alpha^{-2}\le
 2^{8}Nh^{-1}L^{-d}\alpha^{-2}.
\end{align*}
This estimation implies that $V^{end}$ is continuous in
 $\overline{D(r)}\times \overline{D(\eps r'')}^2$, analytic in
 $D(r)\times D(\eps r'')^2$ and 
\begin{align*}
\frac{h}{N}|V^{end}(u,\b0)|\le 2^8L^{-d}\alpha^{-2},\quad (\forall u\in
 \overline{D(r)}).
\end{align*}
Moreover, observe that for any $u\in D(r)$, $j\in \{1,2\}$,
\begin{align*}
\frac{\partial}{\partial \la_j}V^{end}(u,\b0)
&=\frac{1}{r'}\sum_{n=1}^{\infty}\frac{1}{(n-1)!}Tree(\{1,2,\cdots,n\},\cC_0)\\
&\quad\cdot \sum_{a\in
 \{1-1,1-3\}}V^{a,0}(u,r'\be_j)(\psi^1)
\prod_{k=2}^nV^{0-1,0}(u)(\psi^k)
\Bigg|_{\psi^{k}=0\atop(\forall
 k\in\{1,2,\cdots,n\})}.
\end{align*}
Thus, by \eqref{eq_tree_1_1}, \eqref{eq_tree_1}, \eqref{eq_determinant_bound},
 \eqref{eq_decay_bound_final}, \eqref{eq_scale_volume_norm},
 \eqref{eq_scale_volume_norm_linear}, \eqref{eq_offspring_norm_bound}
 and the assumptions $\alpha\ge 2^3$, $2^2D_cL^{-d}\le 1$,
\begin{align*}
&\left|
\frac{\partial}{\partial \la_j}V^{end}(u,\b0)-
\frac{\partial}{\partial \la_j}V^{1-3,end}(u,\b0)
\right|\\
&\le
 \frac{1}{r'}\left|Tree(\{1\},\cC_0)V^{1-1,0}(u,r'\be_j)(\psi^1)\Big|_{\psi^1=0}\right|\\
&\quad +\frac{1}{r'}\Bigg|
\sum_{n=2}^{\infty}\frac{1}{(n-1)!}Tree(\{1,2,\cdots,n\},\cC_0)\\
&\qquad\qquad \qquad\cdot \sum_{a\in
 \{1-1,1-3\}}V^{a,0}(u,r'\be_j)(\psi^1)
\prod_{k=2}^nV^{0-1,0}(u)(\psi^k)
\Bigg|_{\psi^{k}=0\atop(\forall
 k\in\{1,2,\cdots,n\})}\Bigg|\\
&\le \frac{1}{r'}\sum_{m=0}^Nc_0^{\frac{m}{2}}\|V_m^{1-1,0}\|_{1,r,r'}\\
&\quad +\frac{1}{r'}\sum_{n=2}^{\infty}D_c^{n-1}\sum_{m=2}^N2^{3m}c_0^{\frac{m}{2}}\sum_{a\in
 \{1-1,1-3\}}\|V_m^{a,0}\|_{1,r,r'}
\left(\sum_{p=2}^N2^{3p}c_0^{\frac{p}{2}}\|V_p^{0-1,0}\|_{1,\infty,r}\right)^{n-1}\\
&\le \frac{2}{r'}L^{-d}+\frac{2}{r'}\sum_{n=2}^{\infty}(D_cL^{-d})^{n-1}\le
\frac{2}{r'}(1+2D_c)L^{-d}.
\end{align*}
We can see from above that the claims of the lemma have been proved.
\end{proof}

\begin{remark}
There is no essential necessity to complete the generalized double-scale
 integration by explicitly estimating the combinatorial factors as in 
 Lemma \ref{lem_UV_without_artificial}, Lemma
 \ref{lem_UV_with_artificial}, Lemma \ref{lem_UV_final}. We did so only
 to feature the explicitness of our construction. In fact the following
 statements, which are less explicit but are sufficient to achieve the
 main goal of this paper, can be proved by shorter arguments. There
 exists a positive constant $c$ independent of any parameter such that
 if $h\ge 1$, $\alpha\ge c$, $L^d\ge cD_c$, 
\begin{itemize}
\item 
\begin{align*}
&V^{0-1,0}\in \cQ(c^{-1}c_0^{-2}\alpha^{-4}),\quad V^{0-2,0}\in
 \cR(c^{-1}c_0^{-2}\alpha^{-4}),\\
&V^{1-1,0}\in \cQ'(c^{-1}c_0^{-2}\alpha^{-4},
 c^{-1}\beta^{-1}c_0^{-2}\alpha^{-4}),\\
&V^{1-2,0}\in \cR'(c^{-1}c_0^{-2}\alpha^{-4},
 c^{-1}\beta^{-1}c_0^{-2}\alpha^{-4}),\\
&V^{1-3,0}\in \cS(c^{-1}c_0^{-2}\alpha^{-4},
 c^{-1}\beta^{-1}c_0^{-2}\alpha^{-4}),\\
&V^{2,0}\in \cW(c^{-1}c_0^{-2}\alpha^{-4},
 c^{-1}\beta^{-1}\min\{1,\beta\}c_0^{-2}\alpha^{-4}).
\end{align*}
\item $V^{end}$ is continuous in
      $\overline{D(c^{-1}c_0^{-2}\alpha^{-4})}\times
      \overline{D(c^{-1}L^{-d}h^{-1}\beta^{-1}\min\{1,\beta\}
c_0^{-2}\alpha^{-4})}^2$ and analytic in 
${D(c^{-1}c_0^{-2}\alpha^{-4})}\times
      {D(c^{-1}L^{-d}h^{-1}\beta^{-1}\min\{1,\beta\}
c_0^{-2}\alpha^{-4})}^2$.
\item 
$$
\frac{h}{N}|V^{end}(u,\b0)|\le c\alpha^{-2}L^{-d},\quad (\forall u\in \overline{D(c^{-1}c_0^{-2}\alpha^{-4})}).
$$
\item 
\begin{align*}
&\left|
\frac{\partial}{\partial \la_j}V^{end}(u,\b0)-
\frac{\partial}{\partial \la_j}V^{1-3,end}(u,\b0)
\right|\le c\beta c_0^2\alpha^4(1+D_c)L^{-d},\\
&(\forall u\in {D(c^{-1}c_0^{-2}\alpha^{-4})},\ j\in \{1,2\}).
\end{align*}
\end{itemize}
\end{remark}

\begin{remark} In practice $D_c$ will be the biggest parameter as
 $\theta$ approaches to $2\pi/\beta$. The essential benefit of Lemma
 \ref{lem_UV_final} is that the parameter $D_c$ does not affect the
 domain of analyticity with the extended coupling constant $u$. This is because
 the heavy contribution from $D_c$ was absorbed by the inverse of
 the volume factor.
\end{remark}

\section{Proof of the theorem}\label{sec_proof_theorem}

In this section we will prove Theorem \ref{thm_main_theorem}. In view of
the formulation \eqref{eq_grassmann_formulation_2_band},
\eqref{eq_grassmann_formulation_2_band_correlation} we must know
to what the Grassmann Gaussian integrals converge inside the normal Gaussian
integral as $h\to \infty$, $L\to \infty$. One part of this question will
be answered by realizing the general results of the double-scale
integration prepared in the previous section. To do so, we need to
confirm that the actual covariances satisfy the properties required
in the previous section. It follows from the double-scale integration,
especially from the bound \eqref{eq_UV_final_pressure} that the spatial
mean of logarithm of the Grassmann Gaussian integral converges to zero
in the infinite-volume limit. However, it will turn out necessary to make sure that the Grassmann Gaussian integral
itself, not the spatial mean, converges in the time-continuum,
infinite-volume limit. To prove this, which cannot be deduced from
 the results of the previous section, 
we will study detailed
convergent properties of each term of the perturbative expansion of
logarithm of the Grassmann Gaussian integral. After these preparations
we will move on to the proof of Theorem \ref{thm_main_theorem}.
To shorten formulas, we set 
$$
\Theta:=\left|\frac{\theta}{2}-\frac{\pi}{\beta}\right|
$$
throughout this section.

\subsection{Application of Pedra-Salmhofer's determinant
  bound}\label{subsec_P_S_bound}

Here we derive a uniform bound on the determinant of $C(\phi)$ by
applying Pedra-Salmhofer's determinant bound (\cite{PS}). We especially
use the general theorem \cite[\mbox{Theorem 1.3}]{PS} which is a
generalization of Gram's inequality to covariances with
time-discontinuity typically caused by time-ordering.
We restrict our
attention to what is sufficient to solve the current problem. 
The following proposition, which is a specific version of 
\cite[\mbox{Theorem 1.3}]{PS}, is in fact sufficient.

\begin{proposition}\label{prop_P_S_bound}
Let $C:(\{1,2\}\times \G\times [0,\beta))^2\to \C$. Assume that there
 is a complex Hilbert space $\cH$ and $f_j^{\ge}$, $g_j^{\ge}$, $f_j^{<}$,
 $g_j^{<}\in \Map(\{1,2\}\times\G\times \R,\cH)$ $(j=1,2)$ such that
\begin{align}
&C(\rho\bx s,\eta\by t)=1_{s\ge t}\sum_{j\in\{1,2\}}
\<f_j^{\ge}(\rho\bx s),g_j^{\ge}(\eta\by t)\>_{\cH}+
1_{s< t}\sum_{j\in\{1,2\}}
\<f_j^{<}(\rho\bx s),g_j^{<}(\eta\by t)\>_{\cH},\label{eq_P_S_representation}\\
&(\forall (\rho,\bx,s),(\eta,\by,t)\in \{1,2\}\times \G\times
 [0,\beta)),\notag
\end{align}
where $\<\cdot,\cdot\>_{\cH}$ is the inner product of $\cH$. Moreover,
 assume that there exists $D\in \R_{>0}$ such that 
\begin{align*}
&\|f_j^{\ge}(X)\|_{\cH},\ \|g_j^{\ge}(X)\|_{\cH},\
 \|f_j^{<}(X)\|_{\cH},\ 
 \|g_j^{<}(X)\|_{\cH}\le D,\\
&(\forall X\in \{1,2\}\times \G\times \R,\ j\in \{1,2\})
\end{align*}
and the maps $s\mapsto f_j^{\ge}(\rho\bx s)$,
                  $s\mapsto g_j^{\ge}(\rho\bx s)$,
                   $s\mapsto f_j^{<}(\rho\bx s)$,
                  $s\mapsto g_j^{<}(\rho\bx s)$ $(j=1,2)$ are continuous
 in $\R$ for any $\rho\in \{1,2\}$, $\bx\in \G$. Then,
\begin{align*}
&|\det(\<\bu_i,\bv_j\>_{\C^m}C(X_i,Y_j))_{1\le i,j\le n}|\le
 (4D)^{2n},\\
&(\forall m,n\in \N,\ \bu_i,\bv_i\in\C^m\text{ with
 }\|\bu_i\|_{\C^m},\|\bv_i\|_{\C^m}\le 1,\\
&\quad X_i,Y_i\in
 \{1,2\}\times\G\times [0,\beta)\ (i=1,2,\cdots,n)).
\end{align*}
\end{proposition}
Proposition \ref{prop_P_S_bound} is a direct implication of 
\cite[\mbox{Theorem 1.3}]{PS}. For readers' convenience we provide a
 proof for this proposition in Appendix \ref{app_P_S_bound}. In
 fact we added the continuity condition of $f_j^{\ge}$, $g_j^{\ge}$,
 $f_j^{<}$, $g_j^{<}$ with the time variable, which is not assumed in the
 original \cite[\mbox{Theorem 1.3}]{PS}, to shorten the proof.
By applying this proposition we obtain the following.

\begin{proposition}\label{prop_P_S_bound_application}
\begin{align}
&|\det(\<\bu_i,\bv_j\>_{\C^m}C(\phi)(X_i,Y_j))_{1\le i,j\le n}|\label{eq_P_S_bound_application}\\
&\le \left(
\frac{2^4}{L^d}\sum_{\bk\in\G^*}\left(
1+2\cos\left(\frac{\beta\theta}{2}\right)e^{-\beta \sqrt{e(\bk)^2+|\phi|^2}}+e^{-2\beta \sqrt{e(\bk)^2+|\phi|^2}}
\right)^{-\frac{1}{2}}\right)^n,\notag\\
&(\forall m,n\in \N,\ \bu_i,\bv_i\in\C^m\text{ with
 }\|\bu_i\|_{\C^m},\|\bv_i\|_{\C^m}\le 1,\notag\\
&\quad X_i,Y_i\in
 \{1,2\}\times\G\times [0,\beta)\ (i=1,2,\cdots,n),\ \phi\in \C).\notag
\end{align}
\end{proposition}
\begin{remark} In the next subsection we will derive a $\phi$-independent
 upper bound on 
$$
\frac{1}{L^d}\sum_{\bk\in\G^*}\left(
1+2\cos\left(\frac{\beta\theta}{2}\right)e^{-\beta \sqrt{e(\bk)^2+|\phi|^2}}+e^{-2\beta \sqrt{e(\bk)^2+|\phi|^2}}\right)^{-\frac{1}{2}}.
$$
See \eqref{eq_bare_covariance_integrand_bound}.
\end{remark}

\begin{remark} We need to find a representation of the form
 \eqref{eq_P_S_representation}. Such a representation was constructed
 for one-band models with a real-valued dispersion relation in
 \cite[\mbox{Subsection 4.1}]{PS}. It is straightforward to modify the
 construction of \cite[\mbox{Subsection 4.1}]{PS} to fit in our 2-band
 model with the complex-valued dispersion relation. We should also
 mention that an extension of the construction of \cite[\mbox{Subsection
 4.1}]{PS} to one-band models with a complex-valued dispersion relation
 was reported in \cite[\mbox{Subsection V.A}]{K10}. Though it is close
 to both \cite[\mbox{Subsection
 4.1}]{PS} and \cite[\mbox{Subsection V.A}]{K10}, we will provide a
 concrete representation of the form \eqref{eq_P_S_representation} for our 2-band
 model for completeness of the paper.
\end{remark}

\begin{proof}[Proof of Proposition \ref{prop_P_S_bound_application}]
Define the functions $e_j:\G^*\to \C$ $(j=1,2)$ by 
$e_j(\bk):=i\frac{\theta}{2}+(-1)^{1_{j=2}}e(\phi)(\bk)$, where $e(\phi)(\cdot)$
 is the function defined in \eqref{eq_full_eigen_value}. Since $\G^*$ is
 the finite set, for any sufficiently small $\eps \in\R_{>0}$, 
$e_j(\bk)+\eps\neq 0$ $(\forall \bk\in\G^*)$. Set
 $e_{j,\eps}(\bk):=e_j(\bk)+\eps$ and 
\begin{align}
C_{\eps}(\rho\bx s,\eta \by t):=\frac{1}{L^d}\sum_{\bk\in
 \G^*}\sum_{j\in \{1,2\}}&e^{i\<\bk,\bx-\by\>}U(\phi)(\bk)(\rho,j)U(\phi)(\bk)^*(j,\eta)\label{eq_covariance_diagonalization_inside}\\
&\cdot e^{(s-t)e_{j,\eps}(\bk)}\left(\frac{1_{s\ge t}}{1+e^{\beta
 e_{j,\eps}(\bk)}}-\frac{1_{s<t}}{1+e^{-\beta
 e_{j,\eps}(\bk)}}\right),\notag
\end{align}
where $U(\phi)(\bk)$ is the $2\times 2$ matrix defined in
 \eqref{eq_unitary_for_diagonalization}.
Let us find a determinant bound of $C_{\eps}$ and send $\eps\searrow 0$
 afterward. We can see from \eqref{eq_covariance_2_band_pre} that 
$\lim_{\eps \searrow 0}C_{\eps}(\bX)=C(\phi)(\bX)$ $(\forall \bX\in
 I_0^2)$. 

Remark that $L^2(\G^*\times \R)$ is the Hilbert space whose inner
 product $\<\cdot,\cdot\>_{L^2(\G^*\times \R)}$ is defined by 
$$
\<f,g\>_{L^2(\G^*\times \R)}:=\frac{1}{L^d}\sum_{\bk\in\G^*}\int_{\R}dv \overline{f(\bk,v)}g(\bk,v).
$$
For $(\rho,\bx,s)\in \{1,2\}\times \G\times \R$, $j\in \{1,2\}$, $a\in
 \{1,-1\}$ we define $f_{\rho\bx s}^{j,a}$, $g_{\rho\bx s}^{j,a}\in
 L^{2}(\G^*\times \R)$ by 
\begin{align*}
f_{\rho\bx s}^{j,a}(\bk,v)
:=&1_{a\Re e_{j,\eps}(\bk)>0}
 \overline{U(\phi)(\bk)(\rho,j)}e^{-i\<\bk,\bx\>-is(a\Im
 e_{j,\eps}(\bk)-v)}\\
&\cdot\frac{1+e^{-\beta a e_{j,\eps}(\bk)}}{|1+e^{-\beta a
 e_{j,\eps}(\bk)}|^{\frac{3}{2}}}\sqrt{\frac{|\Re
 e_{j,\eps}(\bk)|}{\pi}}
 \frac{1}{iv+\Re e_{j,\eps}(\bk)},\\
g_{\rho\bx s}^{j,a}(\bk,v):=&1_{a\Re e_{j,\eps}(\bk)>0}
 \overline{U(\phi)(\bk)(\rho,j)}e^{-i\<\bk,\bx\>-is(a\Im
 e_{j,\eps}(\bk)-v)}\\
&\cdot \frac{1}{|1+e^{-\beta a
 e_{j,\eps}(\bk)}|^{\frac{1}{2}}}\sqrt{\frac{|\Re
 e_{j,\eps}(\bk)|}{\pi}}
\frac{1}{iv+\Re e_{j,\eps}(\bk)}.
\end{align*}
Then, let us define the maps $f_j^{\ge}$, $g_j^{\ge}$, $f_j^<$,
 $g_j^<\in \Map(\{1,2\}\times \G\times \R, L^2(\G^*\times \R))$
 $(j=1,2)$ by
\begin{align*}
&f_j^{\ge}(\rho,\bx,s)=f_j^{<}(\rho,\bx,s):=f_{\rho\bx s}^{j,1}+f_{\rho
 \bx (-s)}^{j,-1},\\
&g_j^{\ge}(\rho,\bx,s):=g_{\rho\bx (\beta+s)}^{j,1}+g_{\rho
 \bx (-s)}^{j,-1},\quad g_j^{<}(\rho,\bx,s):=-g_{\rho\bx s}^{j,1}-g_{\rho
 \bx (\beta-s)}^{j,-1},\\
&(\forall j\in \{1,2\},\ (\rho,\bx,s)\in \{1,2\}\times \G\times
 \R).
\end{align*}
By using the formula
$$
e^{-tA}=\frac{A}{\pi}\int_{\R}dv\frac{e^{itv}}{v^2+A^2},\quad (\forall
 t\in \R_{\ge 0},\ A\in \R_{>0})
$$
and the uniform bound $|U(\phi)(\bk)(\rho,\eta)|\le 1$ $(\forall \bk\in
 \G^*,\ \rho,\eta\in \{1,2\})$ one can check that 
\begin{align}
&C_{\eps}(\rho\bx s,\eta\by t)=1_{s\ge t}\sum_{j\in\{1,2\}}
\<f_j^{\ge}(\rho\bx s),g_j^{\ge}(\eta\by t)\>_{L^2(\G^*\times \R)}\label{eq_perturbed_covariance_Gram_representation}\\
&\qquad\qquad\qquad\quad
+
1_{s< t}\sum_{j\in\{1,2\}}
\<f_j^{<}(\rho\bx s),g_j^{<}(\eta\by t)\>_{L^2(\G^*\times \R)},\notag\\
&(\forall (\rho,\bx,s),(\eta,\by,t)\in \{1,2\}\times \G\times
 [0,\beta)),\notag\\
&\|f_j^{\ge}(X)\|_{L^2(\G^*\times \R)},\
 \|g_j^{\ge}(X)\|_{L^2(\G^*\times \R)},\ \|f_j^{<}(X)\|_{L^2(\G^*\times
 \R)},\ \|g_j^{<}(X)\|_{L^2(\G^*\times
 \R)}\label{eq_each_vector_norm_bound}\\
&\le \left(
\frac{1}{L^d}\sum_{\bk\in\G^*}\left(
1+2\cos\left(\frac{\beta\theta}{2}\right)e^{-\beta |(-1)^{1_{j=2}}e(\phi)(\bk)+\eps|}+e^{-2\beta |(-1)^{1_{j=2}}e(\phi)(\bk)+\eps|}
\right)^{-\frac{1}{2}}\right)^{\frac{1}{2}},\notag\\
&(\forall X\in \{1,2\}\times \G\times
 \R).\notag
\end{align}
It is clear that $f_j^{\ge},g_j^{\ge},f_j^{<},g_j^{<}$ $(j=1,2)$ are continuous
 with respect to the time variable as the maps from $\R$ to $L^2(\G^*\times \R)$.
Here we can apply Proposition \ref{prop_P_S_bound} to the perturbed
 matrix $C_{\eps}$. Then, by sending $\eps\searrow 0$ we obtain the
 claimed bound.
\end{proof}
  
\subsection{Completion of the double-scale
  integration}\label{subsec_double_scale}

The analysis of the previous section was constructed on the basic
assumptions on the
two generalized covariances. We have to demonstrate that the actual full covariance
can be decomposed into a sum of 2 covariances and each of them satisfies
the required bound properties. Our plan is to reformulate the full
covariance into a sum over the Matsubara frequency and let $\cC_0$ be one
portion with only one Matsubara frequency closest to $\theta/2$ and let
$\cC_1$ be the one with the rest of the Matsubara frequencies. Concerning
the determinant bound, Gram's inequality applies to $\cC_0$, while it
does not to $\cC_1$. However, since the Pedra-Salmhofer's type
determinant bound obtained in the previous subsection applies to
$\cC_0+\cC_1$, we can derive the determinant bound on $\cC_1$ by
decomposing $\cC_1$ as $(\cC_0+\cC_1)-\cC_0$.
 In order to derive the
$L^1$-type norm bounds, we introduce a family of scale-dependent UV
cut-off and estimate the norm of scale-dependent covariances with the
Matsubara UV cut-off. This is a normal technique used in multi-scale analysis
over the Matsubara frequency. Since the $L^1$-type norm bound of the
covariance with UV cut-off is summable with the scale
index, we can obtain an upper bound on the norm of $\cC_1$. 

The momentum variable dual to the time variable is the Matsubara
frequency $\frac{\pi}{\beta}(2\Z+1)$. Since we discretized $[0,\beta)$ by the
step size $\frac{1}{h}$, we automatically have a cut-off in the infinite set
$\frac{\pi}{\beta}(2\Z+1)$. Set
$$
\cM_h:=\left\{\o\in \frac{\pi}{\beta}(2\Z+1)\ \big|\ |\o|<\pi h
\right\}.
$$
To begin with, let us reformulate the restriction of $C(\phi)(\cdot)$ into a sum over $\cM_h$.

\begin{lemma}\label{lem_covariance_matsubara_sum}
\begin{align}
&C(\phi)(\rho\bx s,\eta \by t)\label{eq_covariance_matsubara_sum}\\
&=\frac{1}{\beta L^d}\sum_{\bk\in\G^*}\sum_{\o\in
 \cM_h}e^{i\<\bk,\bx-\by\>+i\o(s-t)}
h^{-1}(I_2-e^{-\frac{i}{h}(\o-\frac{\theta}{2})I_2+\frac{1}{h}E(\phi)(\bk)})^{-1}(\rho,\eta),\notag\\
&(\forall (\rho,\bx,s),(\eta,\by,t)\in \{1,2\}\times \G\times
 [0,\beta)_h).\notag
\end{align}
\end{lemma}

\begin{proof}
One can derive \eqref{eq_covariance_matsubara_sum} by using
 \eqref{eq_dispersion_matrix_diagonalization}, \eqref{eq_covariance_2_band_pre} and the equality
\begin{align*}
&e^{sA}\left(\frac{1_{s\ge 0}}{1+e^{\beta A}}- \frac{1_{s< 0}}{1+e^{-\beta A}}
\right)=\frac{1}{\beta}\sum_{\o\in \cM_h}\frac{e^{i\o
 s}}{h(1-e^{-i\frac{\o}{h}+\frac{A}{h}})},\\
&\left(\forall
 s\in\left\{-\beta,-\beta+\frac{1}{h},\cdots,\beta-\frac{1}{h}\right\},\ 
A\in
 \C\big\backslash i\frac{\pi}{\beta}(2\Z+1)\right).
\end{align*}
See \cite[\mbox{Appendix C}]{K9} for the proof of the above formula.
\end{proof}

Let us take a function $\chi\in C^{\infty}(\R)$ satisfying that 
\begin{align*}
&\chi(x)=1,\quad (\forall x\in (-\infty,1]),\\
&\chi(x)=0,\quad (\forall x\in [2,\infty)),\\
&\chi(x)\in (0,1),\quad (\forall x\in (1,2)), \\
&\frac{d}{dx}\chi(x)\le 0,\quad (\forall x\in\R).
\end{align*}
We do not need more detailed information on the function $\chi$. 
See e.g. \\
\cite[\mbox{Problem II.6. Solution}]{FKT} for an explicit
construction of cut-off functions of this type. Let us take the
parameter $M$ from $[2\pi,\infty)$. With the aim of dealing with small
as well as large $\beta$ at the same time, we set the smallest scale of
cut-off to be $\beta$-dependent, which is the idea implemented in
\cite[\mbox{Section 3}]{K14}. We define the function $\chi^M:\R\to \R$
by 
$$
\chi^M(x):=\chi\left(\frac{x-M}{M^2-M}+1\right).
$$
Note that 
\begin{align*}
&\chi^M(x)=1,\quad (\forall x\in (-\infty,M]),\\
&\chi^M(x)=0,\quad (\forall x\in [M^2,\infty)),\\
&\chi^M(x)\in (0,1),\quad (\forall x\in (M,M^2)), \\
&\frac{d}{dx}\chi^M(x)\le 0,\quad (\forall x\in\R).
\end{align*} 
For $h\in \frac{2}{\beta}\N$, set 
$$
N_h:=\left\lfloor \frac{\log(2h)}{\log M}\right\rfloor,\quad
N_{\beta}:=\max\left\{\left\lfloor \frac{\log(1/\beta)}{\log M}\right\rfloor+1,1
\right\},
$$ 
where $\lfloor x\rfloor$ denotes the largest integer which does not
exceed $x$ for $x\in\R$.
We want $N_h$ to be larger than $N_{\beta}$. One can find a sufficient
condition as follows. 
\begin{lemma}\label{lem_basic_h_largeness}
If $h\ge \frac{1}{2}\max\{1,\beta^{-1}\}M^2$, $N_h\ge N_{\beta}+1$.
\end{lemma}
\noindent Since we will need the condition $h\ge 4d$ later, let us assume from now
that 
\begin{align}
h\ge \max\left\{\frac{1}{2}\max\{1,\beta^{-1}\}M^2,4d\right\}.\label{eq_basic_h_largeness}
\end{align}
It follows that 
\begin{align}
&1_{\beta\ge 1}M+1_{\beta<1}\beta^{-1}\le M^{N_{\beta}}\le
 \max\{1,\beta^{-1}\}M.\label{eq_M_beta_relation}\\
&M^l\le 2h,\quad (\forall l\in \{N_{\beta},
 N_{\beta}+1,\cdots,N_h\}).\label{eq_M_h_relation}
\end{align}
Then, let us define the cut-off functions $\chi_l\in C^{\infty}(\R)$
$(l=N_{\beta},N_{\beta}+1,\cdots,N_h)$ by 
\begin{align*}
&\chi_{N_{\beta}}(\o):=\chi^M(M^{-N_{\beta}}h|1-e^{i\frac{\o}{h}}|),\\
&\chi_l(\o):=\chi^M(M^{-l}h|1-e^{i\frac{\o}{h}}|)-\chi^M(M^{-(l-1)}h|1-e^{i\frac{\o}{h}}|),\\
&(\forall l\in \{N_{\beta}+1,N_{\beta}+2,\cdots,N_h\}).
\end{align*}
It follows from the inequalities $h|1-e^{i\frac{\o}{h}}|\le 2h\le
M^{N_h+1}$ that $\chi^M(M^{-N_h}h|1-e^{i\frac{\o}{h}}|)=1$ $(\forall
\o\in\R)$. Thus,
\begin{align}
\sum_{l=N_{\beta}}^{N_h}\chi_l(\o)=1,\quad (\forall
 \o\in\R).\label{eq_partition_unity}
\end{align}
The values of the cut-off functions are summarized as follows.
\begin{align}\label{eq_support_property_cut_off}
&\chi_{N_{\beta}}(\o)=\left\{\begin{array}{ll}1 & \text{if
		       }h|1-e^{i\frac{\o}{h}}|\le M^{N_{\beta}+1},\\
                   \in (0,1) & \text{if
		    }M^{N_{\beta}+1}<h|1-e^{i\frac{\o}{h}}|<M^{N_{\beta}+2},\\
0 & \text{if  }h|1-e^{i\frac{\o}{h}}|\ge M^{N_{\beta}+2},\end{array}
\right.\\
&\chi_{l}(\o)=\left\{\begin{array}{ll}0 & \text{if
		       }h|1-e^{i\frac{\o}{h}}|\le M^{l},\\
                   \in (0,1] & \text{if
		    }M^{l}<h|1-e^{i\frac{\o}{h}}|<M^{l+2},\\
0 & \text{if  }h|1-e^{i\frac{\o}{h}}|\ge M^{l+2},\end{array}
\right.\notag\\
&(\forall \o\in\R,\ l\in \{N_{\beta}+1,N_{\beta}+2,\cdots,N_h\}).\notag
\end{align}
We show a couple of necessary properties in the following lemma. To be
correct, we should remark that the lemma holds for any $\beta\in
\R_{>0}$.
\begin{lemma}\label{lem_cut_off_support}
\begin{enumerate}[(i)]
\item\label{item_cut_off_support}
There exists a positive constant $c$ independent of any parameter such
     that
\begin{align*}
&\frac{1}{\beta}\sum_{\o\in \cM_h}1_{\chi_l(\o+x)\neq 0}\le c
 M^{l+2},\quad (\forall l\in \{N_{\beta},N_{\beta}+1,\cdots,N_h\},\ x\in \R).
\end{align*}
\item\label{item_cut_off_trivial_remark}
If $\o\in [-\pi h,\pi h]\cap \supp \chi_l(\cdot)$ for some $l\in
     \{N_{\beta}+1,N_{\beta}+2,\cdots,N_h\}$, then
     $|\o-\frac{\theta}{2}|\ge \frac{1}{2}|\o|$, $(\forall \theta\in
     [0,\frac{2\pi}{\beta}))$.
\end{enumerate}
\end{lemma}

\begin{proof}
\eqref{item_cut_off_support}: By the periodicity that $\chi_l(x+2\pi
 h)=\chi_l(x)$ $(\forall x\in\R)$, \eqref{eq_M_beta_relation} and \eqref{eq_support_property_cut_off},
\begin{align*}
\frac{1}{\beta}\sum_{\o\in\cM_h}1_{\chi_l(\o+x)\neq 0}&\le
\sup_{r\in
 [0,\frac{2\pi}{\beta})}\Bigg(\frac{1}{\beta}\sum_{m=-\frac{\beta
 h}{2}}^{\frac{\beta h}{2}-1}1_{\chi_l(\frac{2\pi}{\beta}m+r)\neq 0}\Bigg) \\
&\le \sup_{r\in
 [0,\frac{2\pi}{\beta})}\Bigg(\frac{1}{\beta}\sum_{m=-\frac{\beta
 h}{2}}^{\frac{\beta h}{2}-1}1_{|\frac{2\pi}{\beta}m+r|\le c
 M^{l+2}}\Bigg)\le c M^{l+2}.
\end{align*}

\eqref{item_cut_off_trivial_remark}: It follows from the assumption $M\ge 2\pi$
 and \eqref{eq_M_beta_relation}, \eqref{eq_support_property_cut_off}
 that $|\o|\ge h|1-e^{i\frac{\o}{h}}|\ge M^{N_{\beta}+1}\ge \frac{2\pi}{\beta}$,
 which implies the result.
\end{proof}

Here we introduce the covariances with scale-dependent UV cut-off. In
the following we fix $\phi\in\C$ unless otherwise stated. For
$(\rho,\bx,s),(\eta,\by,t)\in I_0$, set
\begin{align*}
&C_l(\rho\bx s,\eta \by t)\\
&:=\frac{1}{\beta L^d}\sum_{\bk\in\G^*}\sum_{\o\in
 \cM_h}e^{i\<\bk,\bx-\by\>+i(s-t)(\o-\frac{\pi}{\beta})}\\
&\qquad\qquad\qquad\cdot \chi_{l+N_{\beta}-1}(\o)h^{-1}(I_2-e^{-\frac{i}{h}(\o-\frac{\theta}{2})I_2+\frac{1}{h}E(\phi)(\bk)})^{-1}(\rho,\eta),\notag\\
&(l\in \{2,3,\cdots,N_h-N_{\beta}+1\}).\\
&C_1(\rho\bx s,\eta \by t)\\
&:=\frac{1}{\beta L^d}\sum_{\bk\in\G^*}\sum_{\o\in
 \cM_h\backslash
 \{\frac{\pi}{\beta}\}}e^{i\<\bk,\bx-\by\>+i(s-t)(\o-\frac{\pi}{\beta})}\\
&\qquad\qquad\qquad\cdot \chi_{N_{\beta}}(\o)h^{-1}(I_2-e^{-\frac{i}{h}(\o-\frac{\theta}{2})I_2+\frac{1}{h}E(\phi)(\bk)})^{-1}(\rho,\eta),\notag\\
&C_0(\rho\bx s,\eta \by t)\\
&:=\frac{1}{\beta
 L^d}\sum_{\bk\in\G^*}e^{i\<\bk,\bx-\by\>}h^{-1}(I_2-e^{-\frac{i}{h}(\frac{\pi}{\beta}-\frac{\theta}{2})I_2+\frac{1}{h}E(\phi)(\bk)})^{-1}(\rho,\eta).
\end{align*}
We can deduce from \eqref{eq_covariance_matsubara_sum}, \eqref{eq_partition_unity},
\eqref{eq_support_property_cut_off} and the
inequality $h|1-e^{i\frac{\pi}{\beta h}}|\le
\frac{\pi}{\beta}\le M^{N_{\beta}+1}$ that
\begin{align}
\sum_{l=0}^{N_h-N_{\beta}+1}C_l(\rho\bx s,\eta \by
 t)=e^{-i\frac{\pi}{\beta}(s-t)}C(\phi)(\rho\bx s, \eta \by t),\quad
 (\forall (\rho,\bx,s),(\eta,\by,t)\in
 I_0).\label{eq_covariance_decomposition}
\end{align}
We want to consider $\sum_{l=1}^{N_h-N_{\beta}+1}C_l$, $C_0$ as $\cC_1$,
$\cC_0$ introduced in Subsection \ref{subsec_generalized_covariances}
respectively. For this purpose we are going to study properties of $C_l$
$(l=0,1,\cdots,N_h-N_{\beta}+1)$. 

Let us make an inequality which will
be used in the estimation of $C_l$. Recall the function
$g_d:(0,\infty)\to \R$ defined in \eqref{eq_magnitude_function}.

\begin{lemma}\label{lem_covariance_crucial_inequality}
Let $K\in \R_{>0}$. There exists a positive constant $c(d)$ depending
 only on $d$ such that for any $L\in \N$ satisfying $L\ge
 K^{-3}g_d(K)^{-1}$, 
$$
\frac{1}{L^d}\sum_{\bk\in \G^*}\frac{1}{\sqrt{K^2+e(\bk)^2}}\le c(d)g_d(K).
$$
\end{lemma}
\begin{remark}
The claimed inequality crucially affects the possible magnitude of the
 coupling constant in our double-scale integration process. In terms of the
 order with $K$ as $K\searrow 0$, the claimed upper bound is better than
 the crude upper bound $K^{-1}$, which is out of use for our purpose of
 proving SSB and ODLRO. However, it is unlikely to be optimal especially
 in the case $d\ge 2$. More delicate analysis specifying $d$ and $\mu$
 can improve the result. In this paper we prefer to obtain an order
 with which the coupling constant can satisfy both the condition for the
 convergence of the Grassmann integration and the condition for the
 solvability of the gap equation under the minimum assumption on $d$ and
 $\mu$, rather than to obtain the optimal order with some complication.
\end{remark}
\begin{proof}[Proof of Lemma \ref{lem_covariance_crucial_inequality}]
Note that for any $\bk'\in \{0,\frac{2\pi}{L},\frac{2\pi}{L}\cdot
 2,\cdots,2\pi -\frac{2\pi}{L}\}^{d-1}$,
\begin{align*}
&\left|
\frac{1}{2\pi}\int_0^{2\pi}dk
 \frac{1}{\sqrt{K^2+e(k,\bk')^2}}-\frac{1}{L}\sum_{l=0}^{L-1}\frac{1}{\sqrt{K^2+e(\frac{2\pi}{L}l,\bk')^2}}\right|\le c(d)K^{-3}L^{-1},
\end{align*}
where we used the assumption $|\mu|\le 2d$ to suppress the dependency of
 the error on $\mu$. By repeating this
 estimation for each coordinate we obtain that
\begin{align}
\left|
\frac{1}{(2\pi)^d}\int_{[0,2\pi]^d}d\bk\frac{1}{\sqrt{K^2+e(\bk)^2}}-\frac{1}{L^d}\sum_{\bk\in\G^*}\frac{1}{\sqrt{K^2+e(\bk)^2}}\right|\le
 c(d)K^{-3}L^{-1}.\label{eq_difference_discrete_continuous}
\end{align}
Take any $\eps\in (0,\frac{\pi}{2})$ and set $I_{\eps}:=[0,\eps]\cup
 [\pi-\eps,\pi+\eps]\cup [2\pi-\eps,2\pi]$. 
Note that $\inf_{\bk\in [0,2\pi]^d\backslash
 I_{\eps}^d}\|\nabla e(\bk)\|_{\R^d}\ge c\eps$. 
It follows from the coarea
 formula and this inequality that 
\begin{align}
&\int_{[0,2\pi]^d}d\bk\frac{1}{\sqrt{K^2+e(\bk)^2}}\label{eq_crucial_continuous_part}\\
&\le c \eps^{-1}\int_{[0,2\pi]^d\backslash I_{\eps}^d}d\bk\frac{\|\nabla
 e(\bk)\|_{\R^d}}{\sqrt{K^2+e(\bk)^2}}+c(d)\eps^dK^{-1}\notag\\
&\le c \eps^{-1}\int_{-2d-\mu}^{2d-\mu}d\eta \frac{\cH^{d-1}(\{\bk\in
 [0,2\pi]^d\ |\ e(\bk)=\eta\})}{\sqrt{K^2+\eta^2}}+c(d)\eps^d K^{-1}\notag\\
&\le c \eps^{-1}\sup_{r\in \R}\cH^{d-1}\left(\left\{\bk\in [0,2\pi]^d\ \Big|\
 \sum_{j=1}^d\cos k_j=r\right\}\right)\int_{-4d}^{4d}d\eta
 \frac{1}{\sqrt{K^2+\eta^2}} +c(d)\eps^d K^{-1}\notag\\
&\le c(d)(\eps^{-1}\log(K^{-1}+1)+\eps^dK^{-1}),\notag
\end{align}
where $\cH^{d-1}$ denotes the $d-1$ dimensional Hausdorff measure.
One can check that the function $x\mapsto
 x^{-1}\log(K^{-1}+1)+x^dK^{-1}$ $:(0,\infty)\to\R$ attains its minimum
 at 
$$x_0=(d^{-1}\log(K^{-1}+1)\cdot K)^{\frac{1}{d+1}}$$ 
and the minimum
 value is
 $$(d^{\frac{1}{d+1}}+d^{-\frac{d}{d+1}})(\log(K^{-1}+1))^{\frac{d}{d+1}}K^{-\frac{1}{d+1}}.$$ 
Since $\log(K^{-1}+1)\le K^{-1}$, $x_0\in (0,\frac{\pi}{2})$. By taking $\eps$
 to be $x_0$ we have from \eqref{eq_crucial_continuous_part} that
\begin{align}
&\int_{[0,2\pi]^d}d\bk\frac{1}{\sqrt{K^2+e(\bk)^2}}\le c(d)(\log(K^{-1}+1))^{\frac{d}{d+1}}K^{-\frac{1}{d+1}}.\label{eq_crucial_continuous_general_dimension}
\end{align}

Let us improve the upper bound in the case $d=1$.
$$
\int_{[0,2\pi]}dk\frac{1}{\sqrt{K^2+e(k)^2}}\le
 c\int_0^{\frac{\pi}{2}}dk\frac{1}{\sqrt{K^2+(2\cos k - |\mu|)^2}}.
$$
Let $\arccos:(-1,1)\to (0,\pi)$ be the inverse function of
 $\cos|_{(0,\pi)}$. Note that for $k\in
 [0,\frac{\pi}{2}]$, 
\begin{align*}
\left|\cos
 k-\frac{|\mu|}{2}\right|&=\left|\int_{\arccos(|\mu|/2)}^kdp\sin p\right|
\ge \frac{2}{\pi}\left|\int_{\arccos(|\mu|/2)}^kdp
 p\right|\\
&\ge \frac{1}{\pi} \arccos\left(\frac{|\mu|}{2}\right)\left|k-\arccos\left(\frac{|\mu|}{2}\right)\right|.
\end{align*}
By substituting this inequality we have 
\begin{align}\label{eq_crucial_continuous_1_dimension}
\int_{[0,2\pi]}dk\frac{1}{\sqrt{K^2+e(k)^2}}\le\frac{c}{\arccos(\frac{|\mu|}{2})}\int_0^{\frac{\pi}{2}}dk\frac{1}{\sqrt{K^2+k^2}}\le
 \frac{c}{\arccos(\frac{|\mu|}{2})}\log(K^{-1}+1).
\end{align}
The claim follows from \eqref{eq_difference_discrete_continuous},
 \eqref{eq_crucial_continuous_general_dimension},
 \eqref{eq_crucial_continuous_1_dimension}.
\end{proof}

In the following, unless otherwise stated, we assume that 
\begin{align}
L\ge
 \Theta^{-3}g_d(\Theta)^{-1}\label{eq_large_L_condition}
\end{align}
so that
Lemma \ref{lem_covariance_crucial_inequality} holds for
$K=\Theta$. In the next lemma we collect
bound properties of $C_l$. We should emphasize that the constant
$c(d,M,\chi)$ appearing in the lemma is independent of $\phi$.

\begin{lemma}\label{lem_real_covariances}
There exists a positive constant $c(d,M,\chi)$ depending only on
 $d,M,\chi$ such that the following statements hold.
\begin{enumerate}[(i)]
\item\label{item_real_covariances_determinant_bound}
\begin{align*}
&|\det(\<\bu_i,\bv_j\>_{\C^m}C_0(X_i,Y_j))_{1\le i,j\le n}|,\quad
\left|\det\left(\<\bu_i,\bv_j\>_{\C^m}\sum_{l=1}^{N_h-N_{\beta}+1}C_l(X_i,Y_j)
\right)_{1\le i,j\le n}\right|\\
&\le
 (c(d,M,\chi)(1+\beta^{-1}g_d(\Theta)))^n,\notag\\
&(\forall m,n\in \N,\ \bu_i,\bv_i\in\C^m\text{ with
 }\|\bu_i\|_{\C^m},\|\bv_i\|_{\C^m}\le 1,\ X_i,Y_i\in I_0\
 (i=1,2,\cdots,n)).
\end{align*}
\item\label{item_real_covariances_decay_bound}
\begin{align*}
&\|\tilde{C}_l\|_{1,\infty}\le c(d,M,\chi)\min\{1,\beta\}M^{-l},\quad
 (\forall l\in \{2,3,\cdots,N_h-N_{\beta}+1\}),\\
&\|\tilde{C}_1\|_{1,\infty}\le c(d,M,\chi)\beta(1+\beta)^{d+1},\\
&\|\tilde{C}_0\|_{1,\infty}\le c(d,M,\chi)\Theta^{-1}(1+\Theta^{-1})^d.
\end{align*}
\item\label{item_covariances_semi_decay_bound}
\begin{align*}
\left\|\sum_{l=1}^{N_h-N_{\beta}+1}\tilde{C}_l\right\|_{1,\infty}'\le
 c(d,M,\chi)(\beta^{-1}g_d(\Theta)+(1+\beta)^{d+1}).
\end{align*}
\end{enumerate}
Here $\tilde{C}_l:I^2\to \C$ is the anti-symmetric extension of $C_l$
 defined as in \eqref{eq_anti_symmetric_extension}.
\end{lemma}
\begin{proof} In the following `$c$' denotes a generic positive constant
 and `$c(a_1,a_2,\cdots,a_n)$' denotes a positive constant depending
 only on parameters $a_1,a_2,\cdots,a_n$. First of all let us make some
 inequalities concerning the integrand inside the covariances. 
 For $x\in [-\pi h,\pi h]$, $\delta\in \{1,-1\}$,
\begin{align*}
&|h(1-e^{-i\frac{x}{h}+\frac{\delta}{h}\sqrt{e(\bk)^2+|\phi|^2}})|^2\\
&\ge
 h^2(1-e^{-\frac{1}{h}\sqrt{e(\bk)^2+|\phi|^2}})^2+4h^{2}e^{-\frac{1}{h}\sqrt{e(\bk)^2+|\phi|^2}}\sin^2\left(\frac{x}{2h}\right)\notag\\
&\ge 1_{\sqrt{e(\bk)^2+|\phi|^2}>h}h^2(1-e^{-1})^2+
1_{\sqrt{e(\bk)^2+|\phi|^2}\le h}\left(
e^{-2}e(\bk)^2+
4h^{2}e^{-1}\sin^2\left(\frac{x}{2h}\right)
\right)\notag\\
&\ge c(e(\bk)^2+x^2),\notag
\end{align*}
where we used that $h\ge 4d\ge \sup_{\bk\in \R^d}|e(\bk)|$. Since the
 eigen values of $E(\phi)(\bk)$ are $\pm \sqrt{e(\bk)^2+|\phi|^2}$, this
 implies that
\begin{align}
&\|h^{-1}(I_2-e^{-i\frac{x}{h}I_2+\frac{1}{h}E(\phi)(\bk)})^{-1}\|_{2\times
 2}\le \frac{c}{\sqrt{x^2+e(\bk)^2}},\ (\forall x\in [-\pi h,\pi
 h]\backslash\{0\},\ \bk\in \R^d),\label{eq_integrand_core_inequality}
\end{align}
where $\|\cdot\|_{2\times 2}$ is the operator norm for $2\times 2$
 matrices. Similarly we can prove that 
\begin{align}
&\|h^{-1}(I_2-e^{-i\frac{x}{h}I_2+\frac{1}{h}E(\phi)(\bk)})^{-1}
e^{\frac{1}{h}E(\phi)(\bk)}\|_{2\times
 2}\le \frac{c}{|x|},\quad (\forall x\in [-\pi h,\pi
 h]\backslash\{0\},\ \bk\in \R^d).\label{eq_integrand_core_inequality_multiplied}
\end{align}

\eqref{item_real_covariances_determinant_bound}: We derive the claimed
 determinant bound on $C_0$ by means of Gram's inequality. Set
 $\cH:=L^2(\{1,2\}\times \G^*\times \cM_h)$, which is a Hilbert space
 endowed with the inner product $\<\cdot,\cdot\>_{\cH}$ defined by
$$
\<f,g\>_{\cH}:=\frac{1}{\beta L^d}\sum_{K\in \{1,2\}\times \G^*\times \cM_h}\overline{f(K)}g(K).$$
Let us define vectors $f_X,g_X\in \cH$ $(X\in I_0)$ by
\begin{align*}
&f_{\rho\bx
 s}(\eta,\bk,\o):=e^{-i\<\bk,\bx\>-is(\o-\frac{\pi}{\beta})}1_{\o=\frac{\pi}{\beta}}\delta_{\rho,\eta}\left(\left(\o-\frac{\theta}{2}\right)^2+e(\bk)^2\right)^{-\frac{1}{4}},\\
&g_{\rho\bx
 s}(\eta,\bk,\o):=e^{-i\<\bk,\bx\>-is(\o-\frac{\pi}{\beta})}1_{\o=\frac{\pi}{\beta}}\left(\left(\o-\frac{\theta}{2}\right)^2+e(\bk)^2\right)^{\frac{1}{4}}\\
&\qquad\qquad\qquad\quad\cdot h^{-1}(I_2-e^{-\frac{i}{h}(\o-\frac{\theta}{2})I_2+\frac{1}{h}E(\phi)(\bk)})^{-1}(\eta,\rho).
\end{align*}
We can deduce from Lemma \ref{lem_covariance_crucial_inequality} and
 \eqref{eq_integrand_core_inequality} that 
$$
\|f_X\|_{\cH},\ \|g_X\|_{\cH}\le
 c(d)(\beta^{-1}g_d(\Theta))^{\frac{1}{2}},\quad
 (\forall X\in I_0).
$$
Since $C_0(X,Y)=\<f_X,g_Y\>_{\cH}$ $(\forall X,Y\in I_0)$, Gram's
 inequality in the Hilbert space $\C^m\times \cH$ ensures that 
\begin{align}
&|\det(\<\bu_i,\bv_j\>_{\C^m}C_0(X_i,Y_j))_{1\le i,j\le n}|\label{eq_0th_covariance_determinant_bound_pre}\le \prod_{j=1}^n\|f_{X_j}\|_{\cH}\|g_{Y_j}\|_{\cH}\le 
 (c(d)\beta^{-1}g_d(\Theta))^n,\\
&(\forall m,n\in \N,\ \bu_i,\bv_i\in\C^m\text{ with
 }\|\bu_i\|_{\C^m},\|\bv_i\|_{\C^m}\le 1,\ X_i,Y_i\in I_0\
 (i=1,2,\cdots,n)),\notag
\end{align}
which implies that $C_0$ satisfies the claimed determinant bound.

To derive the determinant bound on $\sum_{l=1}^{N_h-N_{\beta}+1}C_l$ we
 use Proposition \ref{prop_P_S_bound_application}. Note
 that by Lemma \ref{lem_covariance_crucial_inequality}
\begin{align}
&\frac{1}{L^d}\sum_{\bk\in \G^*}\left(
1+2\cos\left(\frac{\beta\theta}{2}\right)e^{-\beta \sqrt{e(\bk)^2+|\phi|^2}}+ e^{-2\beta \sqrt{e(\bk)^2+|\phi|^2}}
\right)^{-\frac{1}{2}}\label{eq_bare_covariance_integrand_bound}\\
&\le \frac{c}{L^d}\sum_{\bk\in \G^*}\left(1_{\beta \sqrt{e(\bk)^2+|\phi|^2}>1}+
1_{\beta \sqrt{e(\bk)^2+|\phi|^2}\le
 1}\beta^{-1}(\Theta^2+e(\bk)^2)^{-\frac{1}{2}}\right)\notag\\
&\le c(d)(1+\beta^{-1}g_d(\Theta)).\notag
\end{align}
Thus, by \eqref{eq_P_S_bound_application} and
 \eqref{eq_covariance_decomposition}
\begin{align}
&\left|\det\left(\<\bu_i,\bv_j\>_{\C^m}
\sum_{l=0}^{N_h-N_{\beta}+1}C_l(X_i,Y_j)\right)_{1\le i,j\le n}\right|\le
 (c(d)(1+\beta^{-1}g_d(\Theta)))^n,\label{eq_full_covariance_determinant_bound_pre}\\
&(\forall m,n\in \N,\ \bu_i,\bv_i\in\C^m\text{ with
 }\|\bu_i\|_{\C^m},\|\bv_i\|_{\C^m}\le 1,\ X_i,Y_i\in I_0\
 (i=1,2,\cdots,n)).\notag
\end{align}
Here we can apply Lemma \ref{lem_application_cauchy_binet} with
 \eqref{eq_0th_covariance_determinant_bound_pre},
\eqref{eq_full_covariance_determinant_bound_pre} to derive the claimed
 determinant bound on $\sum_{l=1}^{N_h-N_{\beta}+1}C_l$.

\eqref{item_real_covariances_decay_bound}:
We can deduce from  \eqref{eq_support_property_cut_off}, Lemma
 \ref{lem_cut_off_support} \eqref{item_cut_off_support},
 \eqref{item_cut_off_trivial_remark} and
 \eqref{eq_integrand_core_inequality} that 
\begin{align}
&|C_l(\bX)|\le cM^2,\quad |C_1(\bX)|\le cM^{N_{\beta}+2}\beta,\quad
 |C_0(\bX)|\le c\beta^{-1}\Theta^{-1},\label{eq_covariances_uniform_bound}\\
&(\forall l\in \{2,3,\cdots,N_h-N_{\beta}+1\},\ \bX\in I_0^2).\notag
\end{align}
Let $n\in \N$, $j\in \{1,2,\cdots,d\}$. Note that by periodicity,
\begin{align*}
&\left(\frac{L}{2\pi}\left(e^{-i\frac{2\pi}{L}\<\bx-\by,\be_j\>}-1
\right)\right)^nC_l(\cdot\bx s,\cdot \by t)\\
&=\frac{1}{\beta L^d}\sum_{\bk\in \G^*}\sum_{\o\in
 \cM_h}e^{i\<\bx-\by,\bk\>+i(s-t)(\o-\frac{\pi}{\beta})}
\left(\chi_{l+N_{\beta}-1}(\o)(1_{l\ge 2}+1_{l=1}1_{\o\neq
 \frac{\pi}{\beta}})+1_{l=0}1_{\o=\frac{\pi}{\beta}}\right)\\
&\quad\cdot
 \prod_{m=1}^n\left(\frac{L}{2\pi}\int_0^{\frac{2\pi}{L}}dp_m\right)\left(\frac{\partial}{\partial q_j}\right)^nh^{-1}(I_2-e^{-\frac{i}{h}(\o-\frac{\theta}{2})I_2+\frac{1}{h}E(\phi)(\bq)})^{-1}\Big|_{\bq=\bk+\sum_{m=1}^np_m\be_j}.
\end{align*}

We need to find a $\phi$-independent upper bound on 
$$
\left\|\left(\frac{\partial}{\partial
 k_j}\right)^nh^{-1}(I_2-e^{-\frac{i}{h}(\o-\frac{\theta}{2})I_2+\frac{1}{h}E(\phi)(\bk)})^{-1}\right\|_{2\times 2}
$$
for $\o\in\cM_h$, $\bk\in\R^d$.
First let us consider the case that $\sqrt{e(\bk)^2+|\phi|^2}\le
 h$. Observe that
\begin{align*}
&\left\|\left(\frac{\partial}{\partial
 k_j}\right)^nh^{-1}(I_2-e^{-\frac{i}{h}(\o-\frac{\theta}{2})I_2+\frac{1}{h}E(\phi)(\bk)})^{-1}\right\|_{2\times
 2}\\
&\le
 c(n)h^{-1}\sum_{m=1}^n\prod_{u=1}^m\left(\sum_{l_u=1}^n\right)1_{\sum_{u=1}^ml_u=n}\\
&\quad\cdot \prod_{u=1}^m\left\|
(I_2-e^{-\frac{i}{h}(\o-\frac{\theta}{2})I_2+\frac{1}{h}E(\phi)(\bk)})^{-1}
\left(\frac{\partial}{\partial k_j}\right)^{l_u}e^{\frac{1}{h}E(\phi)(\bk)}
\right\|_{2\times 2}\\
&\quad\cdot \|(I_2-e^{-\frac{i}{h}(\o-\frac{\theta}{2})I_2+\frac{1}{h}E(\phi)(\bk)})^{-1}\|_{2\times 2}.
\end{align*}
See e.g. the formula \cite[\mbox{(C.1)}]{K16} for derivatives of inverse
 of a matrix-valued function. Since
$$
\left\|\left(\frac{\partial}{\partial
 k_j}\right)^ne^{\frac{1}{h}E(\phi)(\bk)}\right\|_{2\times 2}\le c(n)h^{-1}
$$
in this case, we deduce from \eqref{eq_integrand_core_inequality} and
 the above inequality that
\begin{align}
&\left\|\left(\frac{\partial}{\partial
 k_j}\right)^nh^{-1}(I_2-e^{-\frac{i}{h}(\o-\frac{\theta}{2})I_2+\frac{1}{h}E(\phi)(\bk)})^{-1}\right\|_{2\times
 2}
\le c(n) \left(
\left|\o-\frac{\theta}{2}\right|^{-2}+\left|\o-\frac{\theta}{2}\right|^{-n-1}
\right).\label{eq_integrand_derivative_normal_case}
\end{align}
Next let us consider the case that $\sqrt{e(\bk)^2+|\phi|^2}> h$.
It is convenient to use the equality
\begin{align}
&(I_2-e^{-\frac{i}{h}(\o-\frac{\theta}{2})I_2+\frac{1}{h}E(\phi)(\bk)})^{-1}\label{eq_inverse_matrix_formula}\\
&=\prod_{\delta\in
 \{1,-1\}}(1-e^{-\frac{i}{h}(\o-\frac{\theta}{2})+\frac{\delta}{h}\sqrt{e(\bk)^2+|\phi|^2}})^{-1}\notag\\
&\quad\cdot \left(\begin{array}{cc} 1-e^{-\frac{i}{h}(\o-\frac{\theta}{2})}e^{\frac{1}{h}E(\phi)(\bk)}(2,2) & e^{-\frac{i}{h}(\o-\frac{\theta}{2})}e^{\frac{1}{h}E(\phi)(\bk)}(1,2) \\
e^{-\frac{i}{h}(\o-\frac{\theta}{2})}e^{\frac{1}{h}E(\phi)(\bk)}(2,1) &
 1-e^{-\frac{i}{h}(\o-\frac{\theta}{2})}e^{\frac{1}{h}E(\phi)(\bk)}(1,1)
												       \end{array}\right).\notag
\end{align}
Let us estimate derivatives of each component. Remark that
\begin{align}
|(1-e^{-\frac{i}{h}(\o-\frac{\theta}{2})-\frac{1}{h}\sqrt{e(\bk)^2+|\phi|^2}})^{-1}|
=
|(1-e^{-\frac{i}{h}(\o-\frac{\theta}{2})+\frac{1}{h}\sqrt{e(\bk)^2+|\phi|^2}})^{-1}e^{\frac{1}{h}\sqrt{e(\bk)^2+|\phi|^2}}|\le
 c.\label{eq_eigen_value_inverse}
\end{align}
Moreover, for any $m\in \N\cup \{0\}$, $\delta\in \{1,-1\}$,
\begin{align}
&\left|\left(\frac{\partial}{\partial
 k_j}\right)^me^{\frac{\delta}{h}\sqrt{e(\bk)^2+|\phi|^2}}\right|\le
 c(m)(1_{m=0}+1_{m\ge
 1}h^{-2})e^{\frac{\delta}{h}\sqrt{e(\bk)^2+|\phi|^2}},\label{eq_eigen_exponential}\\
&\left\|\left(\frac{\partial}{\partial
 k_j}\right)^me^{\frac{1}{h}E(\phi)(\bk)}\right\|_{2\times 2}\le
 c(m)(1_{m=0}+1_{m\ge
 1}h^{-1})e^{\frac{1}{h}\sqrt{e(\bk)^2+|\phi|^2}}.\label{eq_hopping_exponential_derivative}
\end{align}
To derive \eqref{eq_hopping_exponential_derivative}, the following
 formula can be repeatedly used.
$$
\frac{\partial}{\partial
 k_j}e^{\frac{1}{h}E(\phi)(\bk)}=\frac{1}{h}\int_0^1ds
 e^{\frac{1-s}{h}E(\phi)(\bk)}
\frac{\partial}{\partial
 k_j}E(\phi)(\bk) e^{\frac{s}{h}E(\phi)(\bk)}.
$$
By \eqref{eq_eigen_value_inverse} and \eqref{eq_eigen_exponential}
\begin{align}
&\left|\left(\frac{\partial}{\partial
 k_j}\right)^m (1-e^{-\frac{i}{h}(\o-\frac{\theta}{2})-\frac{1}{h}\sqrt{e(\bk)^2+|\phi|^2}})^{-1}
\right|\le c(m)(1_{m=0}+1_{m\ge
 1}h^{-2}),\label{eq_negative_eigen_inverse}\\
&\left|\left(\frac{\partial}{\partial
 k_j}\right)^m (1-e^{-\frac{i}{h}(\o-\frac{\theta}{2})+\frac{1}{h}\sqrt{e(\bk)^2+|\phi|^2}})^{-1}
\right|\label{eq_positive_eigen_inverse}\\
&\le c(m)(1_{m=0}+1_{m\ge
 1}h^{-2})|1-e^{-\frac{i}{h}(\o-\frac{\theta}{2})+\frac{1}{h}\sqrt{e(\bk)^2+|\phi|^2}}|^{-1}.\notag
\end{align}
Thus, we have that
\begin{align}
&\left|\left(\frac{\partial}{\partial
 k_j}\right)^n\prod_{\delta\in \{1,-1\}}
 (1-e^{-\frac{i}{h}(\o-\frac{\theta}{2})+\frac{\delta}{h}\sqrt{e(\bk)^2+|\phi|^2}})^{-1}\right|\le
 c(n)h^{-2}.\label{eq_eigen_value_inverse_product}
\end{align}
Also, by \eqref{eq_eigen_value_inverse},
 \eqref{eq_hopping_exponential_derivative},
 \eqref{eq_positive_eigen_inverse}
\begin{align*}
&\left\|\left(\frac{\partial}{\partial
 k_j}\right)^m
 (1-e^{-\frac{i}{h}(\o-\frac{\theta}{2})+\frac{1}{h}\sqrt{e(\bk)^2+|\phi|^2}})^{-1}
e^{\frac{1}{h}E(\phi)(\bk)}\right\|_{2\times 2}\le
 c(m)(1_{m=0}+1_{m\ge 1}h^{-1}),
\end{align*}
which combined with \eqref{eq_negative_eigen_inverse} implies that
\begin{align}
&\left\|\left(\frac{\partial}{\partial
 k_j}\right)^n\prod_{\delta\in \{1,-1\}}
 (1-e^{-\frac{i}{h}(\o-\frac{\theta}{2})+\frac{\delta}{h}\sqrt{e(\bk)^2+|\phi|^2}})^{-1}e^{\frac{1}{h}E(\phi)(\bk)}\right\|_{2\times 2}\le
 c(n)h^{-1}.\label{eq_eigen_inverse_product_hopping_exponential}
\end{align}
We can see from \eqref{eq_inverse_matrix_formula},
 \eqref{eq_eigen_value_inverse_product},
 \eqref{eq_eigen_inverse_product_hopping_exponential} and $\pi h\ge |\o-\theta/2|$
 that
\begin{align}
&\left\|\left(\frac{\partial}{\partial
 k_j}\right)^nh^{-1}(I_2-e^{-\frac{i}{h}(\o-\frac{\theta}{2})I_2+\frac{1}{h}E(\phi)(\bk)})^{-1}\right\|_{2\times
 2}\label{eq_integrand_derivative_abnormal_case}\\
&\le c(n)h^{-2}\le c(n) \left(
\left|\o-\frac{\theta}{2}\right|^{-2}+\left|\o-\frac{\theta}{2}\right|^{-n-1}
\right).\notag
\end{align}

By using \eqref{eq_support_property_cut_off}, Lemma \ref{lem_cut_off_support}
 \eqref{item_cut_off_support}, \eqref{item_cut_off_trivial_remark},
\eqref{eq_integrand_derivative_normal_case} and
 \eqref{eq_integrand_derivative_abnormal_case} we can derive that
\begin{align}
&\left\|\left(\frac{L}{2\pi}(e^{-i\frac{2\pi}{L}\<\bx-\by,\be_j\>}-1)\right)^nC_l(\cdot\bx
 s,\cdot \by t)\right\|_{2\times
 2}\label{eq_covariances_spatial_power}\\
&\le \frac{c(n)}{\beta}\sum_{\o\in \cM_h}\left(
\chi_{l+N_{\beta}-1}(\o)(1_{l\ge 2}+1_{l=1}1_{\o\neq
 \frac{\pi}{\beta}})+1_{l=0}1_{\o=\frac{\pi}{\beta}}
\right)\notag\\
&\qquad\qquad\qquad\cdot\left(
\left|\o-\frac{\theta}{2}\right|^{-2}+\left|\o-\frac{\theta}{2}\right|^{-n-1}
\right)\notag\\
&\le c(n)\left(1_{l\ge
 2}M^{-l-N_{\beta}+3}+1_{l=1}M^{N_{\beta}+2}(\beta^{2}+\beta^{n+1})
+ 1_{l=0}\beta^{-1}(\Theta^{-2}+\Theta^{-n-1})\right).\notag
\end{align}
 The inequality
 \eqref{eq_covariances_spatial_power} coupled with
 \eqref{eq_covariances_uniform_bound} yields that
\begin{align}
&|C_1(\rho\bx s,\eta\by t)|\le
 \frac{c(d,M)M^{N_{\beta}}\beta}{1+(1+\beta)^{-d-1}\sum_{j=1}^d(\frac{L}{2\pi}|e^{i\frac{2\pi}{L}\<\bx-\by,\be_j\>}-1|)^{d+1}},\label{eq_terminal_covariances_spatial_decay}\\
&|C_0(\rho\bx s,\eta\by t)|\le
 \frac{c(d)\beta^{-1}\Theta^{-1}}{1+(1+\Theta^{-1})^{-d-1}\sum_{j=1}^d(\frac{L}{2\pi}|e^{i\frac{2\pi}{L}\<\bx-\by,\be_j\>}-1|)^{d+1}},\notag\\
&(\forall (\rho,\bx,s),(\eta,\by,t)\in I_0).\notag
\end{align}
Thus, by using \eqref{eq_M_beta_relation},
\begin{align*}
&\|\tilde{C}_1\|_{1,\infty}\le c(d,M)M^{N_{\beta}}\beta^2(1+\beta)^d\le
 c(d,M)\beta (1+\beta)^{d+1},\\
&\|\tilde{C}_0\|_{1,\infty}\le
 c(d)\Theta^{-1}(
1+\Theta^{-1})^d.
\end{align*}

Let $l\in \{2,3,\cdots,N_h-N_{\beta}+1\}$. Let us estimate decay of the covariances with the time
 variable. By periodicity,
\begin{align*}
&\left(\frac{\beta}{2\pi}\left(e^{-i\frac{2\pi}{\beta}(s-t)}-1
\right)\right)^nC_l(\cdot\bx s,\cdot \by t)\\
&=\frac{1}{\beta L^d}\sum_{\bk\in \G^*}\sum_{\o\in
 \cM_h}e^{i\<\bx-\by,\bk\>+i(s-t)(\o-\frac{\pi}{\beta})}
\prod_{m=1}^n\left(\frac{\beta}{2\pi}\int_0^{\frac{2\pi}{\beta}}dr_m\frac{\partial}{\partial
 r_m}\right)\\
&\qquad\qquad\qquad\cdot\chi_{l+N_{\beta}-1}\left(\o+\sum_{m=1}^nr_m\right)
h^{-1}(I_2-e^{-\frac{i}{h}(\o-\frac{\theta}{2}+\sum_{m=1}^nr_m)I_2+\frac{1}{h}E(\phi)(\bk)})^{-1}.
\end{align*}
Thus, 
\begin{align*}
&\left\|\left(\frac{\beta}{2\pi}(e^{-i\frac{2\pi}{\beta}(s-t)}-1)
\right)^nC_l(\cdot\bx s,\cdot \by t)\right\|_{2\times 2}\\
&\le \sup_{x\in\R}\left(\frac{1}{\beta}\sum_{\o\in
 \cM_h}1_{\chi_{l+N_{\beta}-1}(\o+x)\neq 0}\right)\\
&\quad\cdot 
\sup_{\o\in [-\pi h,\pi h]\atop \bk\in\R^d}\left\|
\left(\frac{\partial}{\partial \o}\right)^n
\chi_{l+N_{\beta}-1}(\o)
h^{-1}(I_2-e^{-\frac{i}{h}(\o-\frac{\theta}{2})I_2+\frac{1}{h}E(\phi)(\bk)})^{-1}
\right\|_{2\times 2}.
\end{align*}
Note that 
\begin{align}
&\left\|\left(\frac{\partial}{\partial
 \o}\right)^nh^{-1}(I_2-e^{-\frac{i}{h}(\o-\frac{\theta}{2})I_2+\frac{1}{h}E(\phi)(\bk)})^{-1}\right\|_{2\times
 2}\label{eq_integrand_matsubara_derivative}\\
&\le
 c(n)h^{-1}\sum_{m=1}^n\prod_{u=1}^m\left(\sum_{l_u=1}^n\right)1_{\sum_{u=1}^ml_u=n}\notag\\
&\quad\cdot \prod_{u=1}^m\left\|
(I_2-e^{-\frac{i}{h}(\o-\frac{\theta}{2})I_2+\frac{1}{h}E(\phi)(\bk)})^{-1}
\left(\frac{\partial}{\partial \o}\right)^{l_u}e^{-\frac{i}{h}(\o-\frac{\theta}{2})I_2+\frac{1}{h}E(\phi)(\bk)}
\right\|_{2\times 2}\notag\\
&\quad\cdot
 \|(I_2-e^{-\frac{i}{h}(\o-\frac{\theta}{2})I_2+\frac{1}{h}E(\phi)(\bk)})^{-1}\|_{2\times
 2}\notag\\
&\le c(n)h^{-1-n}\sum_{m=1}^n\|
(I_2-e^{-\frac{i}{h}(\o-\frac{\theta}{2})I_2+\frac{1}{h}E(\phi)(\bk)})^{-1}
e^{\frac{1}{h}E(\phi)(\bk)}\|_{2\times 2}^m\notag\\
&\quad\cdot\|
(I_2-e^{-\frac{i}{h}(\o-\frac{\theta}{2})I_2+\frac{1}{h}E(\phi)(\bk)})^{-1}
\|_{2\times 2}.\notag
\end{align}
Then by using \eqref{eq_M_h_relation},
 \eqref{eq_support_property_cut_off}, Lemma \ref{lem_cut_off_support}
 \eqref{item_cut_off_support}, \eqref{item_cut_off_trivial_remark},
 \eqref{eq_integrand_core_inequality},
 \eqref{eq_integrand_core_inequality_multiplied}, \eqref{eq_integrand_matsubara_derivative} and the fact that 
\begin{align*}
\left|\left(\frac{d}{d \o}\right)^n\chi_{l+N_{\beta}-1}(\o)\right|\le
 c(n,M,\chi)M^{-(l+N_{\beta})n},\quad (\forall \o\in\R,\ n\in \N),
\end{align*}
we have that
\begin{align}
&\left\|
\left(\frac{\beta}{2\pi}(e^{-i\frac{2\pi}{\beta}(s-t)}-1)\right)^nC_l(\cdot
 \bx s,\cdot \by t)\right\|_{2\times 2}\label{eq_covariance_time_decay}\\
&\le c(n,M)M^{l+N_{\beta}}\notag\\
&\quad\cdot \Bigg(\sup_{\o\in [-\pi h,\pi h]\atop \bk\in\R^d}\left|
\left(\frac{d}{d \o}\right)^n
\chi_{l+N_{\beta}-1}(\o)\right|\|h^{-1}(I_2-e^{-\frac{i}{h}(\o-\frac{\theta}{2})I_2+\frac{1}{h}E(\phi)(\bk)})^{-1}\|_{2\times
 2}\notag\\
&\qquad +\sum_{m=0}^{n-1}\sup_{\o\in [-\pi h,\pi h]\atop
 \bk\in\R^d}\left|
\left(\frac{d}{d \o}\right)^m
\chi_{l+N_{\beta}-1}(\o)\right|\notag\\
&\qquad\qquad\cdot 
\left\|\left(\frac{\partial}{\partial \o}\right)^{n-m}h^{-1}(I_2-e^{-\frac{i}{h}(\o-\frac{\theta}{2})I_2+\frac{1}{h}E(\phi)(\bk)})^{-1}\right\|_{2\times
 2}\Bigg)\notag\\
&\le c(n,M,\chi)M^{l+N_{\beta}}\notag\\
&\quad\cdot \left(
M^{-(l+N_{\beta})n-(l+N_{\beta})}
+ \sum_{m=0}^{n-1}M^{-(l+N_{\beta})m}h^{-1-(n-m)}\sum_{u=1}^{n-m}h^{u+1}M^{-(l+N_{\beta})(u+1)}
\right)\notag\\
& \le c(n,M,\chi)M^{-(l+N_{\beta})n}.\notag
\end{align}
By combining \eqref{eq_covariances_uniform_bound},
 \eqref{eq_covariances_spatial_power} with
 \eqref{eq_covariance_time_decay} we reach the inequality 
\begin{align}
&|C_l(\rho\bx s,\eta \by t)|\label{eq_covariance_full_decay}\\
&\le
 \frac{c(n,M,\chi)}{1+M^{n(l+N_{\beta})}|\frac{\beta}{2\pi}(e^{i\frac{2\pi}{\beta}(s-t)}-1)|^n+M^{l+N_{\beta}}\sum_{j=1}^d|\frac{L}{2\pi}(e^{i\frac{2\pi}{L}\<\bx-\by,\be_j\>}-1)|^n},\notag\\
&(\forall (\rho,\bx,s),(\eta,\by,t)\in I_0,\ n\in \N).\notag
\end{align}
From \eqref{eq_M_beta_relation}, \eqref{eq_M_h_relation} and 
 \eqref{eq_covariance_full_decay} we can deduce that
\begin{align*}
&\|\tilde{C}_l\|_{1,\infty}\le c(d,M,\chi)M^{-l-N_{\beta}}\le
 c(d,M,\chi)\min\{1,\beta\}M^{-l}.
\end{align*}

\eqref{item_covariances_semi_decay_bound}:
By using the result of \eqref{item_real_covariances_determinant_bound}
 for $n=1$, \eqref{eq_terminal_covariances_spatial_decay},
 \eqref{eq_covariance_full_decay} and \eqref{eq_M_beta_relation} we have
that for any $\rho,\eta\in \{1,2\}$, $s,t\in [0,\beta)_h$, $\by\in \G$, 
\begin{align*}
&\sum_{\bx\in \G}\left|\sum_{l=1}^{N_h-N_{\beta}+1}C_l(\rho\bx s,\eta\by
 t)\right|=\sum_{\bx\in \G}\left|\sum_{l=1}^{N_h-N_{\beta}+1}C_l(\rho\by s,\eta\bx
 t)\right|\\
&=\left|\sum_{l=1}^{N_h-N_{\beta}+1}C_l(\rho\b0 s,\eta\b0
 t)\right|+\sum_{\bx\in \G\atop \bx\neq \b0}\left|\sum_{l=1}^{N_h-N_{\beta}+1}C_l(\rho\bx s,\eta\b0
 t)\right|\\
&\le
 c(d,M,\chi)(1+\beta^{-1}g_d(\Theta)) +
\sum_{\bx\in \G\atop \bx\neq \b0}
 \frac{c(d,M)M^{N_{\beta}}\beta}{1+(1+\beta)^{-d-1}\sum_{j=1}^d|\frac{L}{2\pi}(e^{i\frac{2\pi}{L}\<\bx,\be_j\>}-1)|^{d+1}}\\
&\quad +\sum_{l=2}^{N_h-N_{\beta}+1}\sum_{\bx\in \G\atop \bx\neq \b0}
 \frac{c(d,M,\chi)M^{-l-N_{\beta}}}{\sum_{j=1}^d|\frac{L}{2\pi}(e^{i\frac{2\pi}{L}\<\bx,\be_j\>}-1)|^{d+1}}\\
&\le
 c(d,M,\chi)(1+\beta^{-1}g_d(\Theta)+M^{N_{\beta}}\beta(1+\beta)^d)\\
&\le
 c(d,M,\chi)(\beta^{-1}g_d(\Theta)+(1+\beta)^{d+1}),
\end{align*}
which implies the result.
\end{proof}

By summarizing the results of Lemma \ref{lem_real_covariances} we can
reach the following conclusion.

\begin{corollary}\label{cor_real_covariances}
Set $\cC_1:=\sum_{l=1}^{N_h-N_{\beta}+1}C_l$, $\cC_0:=C_0$. There exists
 a constant $c(d,M,\chi)\in\R_{\ge 1}$ depending only on $d$, $M$,
 $\chi$ such that the following statements hold with $c_0$, $D_c$ defined by 
$$c_0:=c(d,M,\chi)(1+\beta^{d+2}+
\beta^{-1}g_d(\Theta)),\quad D_c:=\Theta^{-1}(1+\Theta^{-1})^d.$$
\begin{itemize}
\item $\cC_1$ satisfies \eqref{eq_time_translation_generic_covariance},
      \eqref{eq_determinant_bound}, 
      \eqref{eq_decay_bound}, \eqref{eq_decay_bound_coupled} with $c_0$.
\item $\cC_0$ satisfies \eqref{eq_time_independence},
      \eqref{eq_determinant_bound}, \eqref{eq_decay_bound_final} with
      $c_0$ and $D_c$. 
\end{itemize}
\end{corollary}
\begin{proof} The claims directly follow from Lemma
 \ref{lem_real_covariances}. Since $\cM_h-\frac{\pi}{\beta}\subset
 \frac{2\pi}{\beta}\Z$, $\cC_1$ satisfies
 \eqref{eq_time_translation_generic_covariance}. Let us only show how to
 derive the claimed upper bounds on $\|\tilde{\cC}_1\|_{1,\infty}$,
 $\|\tilde{\cC}_1\|$ from the results of Lemma
 \ref{lem_real_covariances}. 
\begin{align*}
\|\tilde{\cC}_1\|_{1,\infty}&\le
 \|\tilde{C}_1\|_{1,\infty}+\sum_{l=2}^{N_h-N_{\beta}+1}\|\tilde{C}_l\|_{1,\infty}\\
&\le
 c(d,M,\chi)\beta(1+\beta)^{d+1}+c(d,M,\chi)\min\{1,\beta\}\sum_{l=2}^{N_h-N_{\beta}+1}M^{-l}\le c_0,\\
\|\tilde{\cC}_1\|&\le\left\|\sum_{l=1}^{N_h-N_{\beta}+1}\tilde{C}_l\right\|_{1,\infty}'+\beta^{-1}
\sum_{l=1}^{N_h-N_{\beta}+1}\|\tilde{C}_l\|_{1,\infty}\\
&\le
 c(d,M,\chi)(\beta^{-1}g_d(\Theta)+(1+\beta)^{d+1})+c(d,M,\chi)(1+\beta)^{d+1}\\
&\quad +c(d,M,\chi)\sum_{l=2}^{N_h-N_{\beta}+1}M^{-l}\\
&\le c_0.
\end{align*}
\end{proof}

\begin{remark}
Since we assumed the simple conditions on $\cC_0$, $\cC_1$ in Subsection
 \ref{subsec_generalized_covariances}, it became necessary to
 overestimate the covariances' bounds in terms of the order with
 $\beta$, $\Theta^{-1}$
in the derivation of Corollary \ref{cor_real_covariances}
 from Lemma \ref{lem_real_covariances}. In this paper we choose to
 simplify the Grassmann integration process at the expense of the order
 with these parameters. The most important information in the statements
 of Corollary \ref{cor_real_covariances} is the dependency of $c_0$ on
 $\Theta$, which ultimately decides whether SSB and ODLRO can be proven.
\end{remark}

Since we have verified the necessary properties of the real covariances,
we can apply the general results in the previous section to our model problem.

\begin{lemma}\label{lem_real_UV_integration}
There exists a positive constant $c(d)$ depending only on $d$ such that
 the following statements hold with any $a\in \R_{\ge 1}$, $h\in
 \frac{2}{\beta}\N$ satisfying $h\ge c(d)\max\{1,\beta^{-1}\}$, $L\in
 \N$ satisfying 
\begin{align}
L^d\ge c(d)\max\{
\Theta^{-3d}g_d(\Theta)^{-d}, 
\Theta^{-1}(1+
\Theta^{-1})^d\}\label{eq_L_largeness_final_condition}
\end{align}
and $r,r'\in \R_{>0}$ defined by
\begin{align*}
&r:=c(d)^{-1}(1+\beta^{d+2}+
\beta^{-1}g_d(\Theta))^{-2}a^{-4},\\
&r':=c(d)^{-1} L^{-d}h^{-1}\beta^{-1}\min\{1,\beta\}
(1+\beta^{d+2}+
\beta^{-1}g_d(\Theta))^{-2}a^{-4}.
\end{align*}
There exists a function $V^{end}:\C\times
     \overline{D(2r)}\times\overline{D(2r')}^2\to\C$ such that for any
     $\phi\in \C$, $(u,\bla)\mapsto V^{end}(\phi,u,\bla)$ is continuous
     in $\overline{D(2r)}\times\overline{D(2r')}^2$ and analytic in
     $D(2r)\times D(2r')^2$. Moreover, for any
 $(\phi,u,\bla)\in\C\times\overline{D(r)}\times\overline{D(r')}^2$, $j=1,2$,
\begin{align}
&\frac{h}{N}|V^{end}(\phi,u,\b0)|\le a^{-2}L^{-d},\label{eq_real_pressure_decay}\\
&\left|
\frac{\partial}{\partial \la_j}V^{end}(\phi,u,\b0)+\int
 A^j(\psi)d\mu_{C(\phi)}(\psi)\right|\label{eq_real_correction_estimate}\\
&\le c(d)\beta 
(1+\beta^{d+2}+
\beta^{-1}g_d(\Theta))^{2}(1+\Theta^{-1}
(1+\Theta^{-1})^d)a^{4}L^{-d},\notag\\
&e^{V^{end}(\phi,u,\bla)}=\int
 e^{-V(u)(\psi)+W(u)(\psi)-A(\psi)}d\mu_{C(\phi)}(\psi).\label{eq_model_analytic_continuation}
\end{align}
\end{lemma}

\begin{proof} 
By the division formula of Grassmann Gaussian integral (see
 e.g. \\
\cite[\mbox{Proposition I.21}]{FKT}),
\begin{align*}
-\int A(\psi)d\mu_{\cC_1+\cC_0}(\psi)&=\int\int
 V^{1,1}(\psi+\psi^1)d\mu_{\cC_1}(\psi^1)d\mu_{\cC_0}(\psi)\\
&=\int
 V^{1-3,0}(\psi)d\mu_{\cC_0}(\psi)=V^{1-3,end}.
\end{align*}
As in Corollary \ref{cor_real_covariances} we set
 $\cC_1=\sum_{l=1}^{N_h-N_{\beta}+1}C_l$, $\cC_0=C_0$. Then, by
 \eqref{eq_covariance_decomposition} and the fact that $A(\psi)$ is
 invariant under the transform $\psi_{\rho\bx s\xi}\mapsto
 e^{-i\xi\frac{\pi}{\beta}s}\psi_{\rho\bx s\xi}$ $((\rho,\bx,s,\xi)\in
 I)$,
$$\int A(\psi)d\mu_{\cC_1+\cC_0}(\psi)= \int
 A(\psi)d\mu_{C(\phi)}(\psi).$$ 
Thus, $V^{1-3,end}=-\int
 A(\psi)d\mu_{C(\phi)}(\psi)$.
By substituting this equality, $c_0$, $D_c$ defined in Corollary
 \ref{cor_real_covariances} and $\alpha =2^4 a$ we see that Lemma
 \ref{lem_UV_final} implies the existence of the function $V^{end}$
 satisfying the claimed regularity and \eqref{eq_real_pressure_decay}, \eqref{eq_real_correction_estimate}.

The equality
 \eqref{eq_model_analytic_continuation} can be derived by a basic
 argument close to \\
\cite[\mbox{Proposition 6.4 (3)}]{K16}. However, we
 provide the proof for completeness. We fix $\phi\in\C$ and hide the
 dependency on $\phi$ in the following. 
Since the expansion of
 $e^{-V(u)(\psi)+W(u)(\psi)-A(\psi)}$ terminates at a finite order, there
 exists a $(\beta,L,h)$-dependent positive constant $r(\beta,L,h)$ such
 that 
\begin{align}
\Re \int e^{-V(u)(\psi)+W(u)(\psi)-A(\psi)}d\mu_{\cC_0+\cC_1}(\psi)>0,\quad
 (\forall (u,\bla)\in
 \overline{D(r(\beta,L,h))}^3).\label{eq_real_part_positive_full}
\end{align}
Let us set 
\begin{align*}
&V^{1}(\psi):=\sum_{j=0}^1V^{j,1}(\psi),\quad V^{0}(\psi):=\sum_{j=1}^2V^{0-j,0}(\psi)+\sum_{j=1}^3V^{1-j,0}(\psi)+V^{2,0}(\psi).
\end{align*}
By the results of Lemma \ref{lem_UV_without_artificial}, Lemma
 \ref{lem_UV_with_artificial} and \eqref{eq_0_1_1_initial},
\eqref{eq_initial_1_infinity},
 \eqref{eq_offspring_bound_4th}, \eqref{eq_offspring_bound_2nd_initial}
there exists a $(\beta,L,h)$-dependent
 positive constant $c(\beta,L,h)$ such that 
\begin{align*}
&\sum_{m=0}^N\|V_m^l(u,\bla)\|_1\le c(\beta,L,h)\alpha^{-2},\\
&(\forall (u,\bla)\in \overline{D(c(\beta,L,h)^{-1}\alpha^{-4})}^3,\
 \alpha \in [2^3,\infty),\ l\in \{0,1\}).
\end{align*}
This implies that we can choose a $(\beta,L,h)$-dependent positive
 constant $c(\beta,L,h)'$ such that if $\alpha\ge c(\beta,L,h)'$,
\begin{align}
&\Re \int e^{z V^l(u,\bla)(\psi)}d\mu_{\cC_l}(\psi)>0,\label{eq_real_part_positive_one_step}\\
&(\forall (u,\bla)\in
 \overline{D(c(\beta,L,h)^{-1}\alpha^{-4})}^3,\ z\in \overline{D(2)},\
 l\in \{0,1\}).\notag
\end{align}
Let us take $\alpha_0\in [c(\beta,L,h)',\infty)$ satisfying that
 $c(\beta,L,h)^{-1}\alpha_0^{-4}\le r(\beta,L,h)$ and fix $(u,\bla)\in
 \overline{D(c(\beta,L,h)^{-1}\alpha_0^{-4})}^3$ so that 
both \eqref{eq_real_part_positive_full} and \eqref{eq_real_part_positive_one_step}
hold for this $(u,\bla)$. 
By \eqref{eq_real_part_positive_one_step}, for
 $l\in \{0,1\}$, 
$$
z\mapsto \log\left(\int e^{z V^l(u,\bla)(\psi+\psi^1)}d\mu_{\cC_l}(\psi^1)\right)
$$
is analytic in $D(2)$ and thus
\begin{align}
&\log\left(\int e^{V^1(u,\bla)(\psi+\psi^1)}d\mu_{\cC_1}(\psi^1)\right)\label{eq_one_step_analytic_continuation}\\
&=\sum_{n=1}^{\infty}\frac{1}{n!}\left(\frac{d}{dz}\right)^n\log\left(\int
 e^{z
 V^1(u,\bla)(\psi+\psi^1)}d\mu_{\cC_1}(\psi^1)\right)\Bigg|_{z=0}=V^0(u,\bla)(\psi).\notag
\end{align}
Similarly,
\begin{align}
\log\left(\int e^{V^0(u,\bla)(\psi)}d\mu_{\cC_0}(\psi)\right)
=V^{end}(u,\bla).\label{eq_next_analytic_continuation}
\end{align}
By \eqref{eq_real_part_positive_one_step} for $l=1$, $z=1$ and
 \eqref{eq_one_step_analytic_continuation} we can apply the
 basic lemma \cite[\mbox{Lemma C.2}]{K14} to ensure that
$$
e^{V^0(u,\bla)(\psi)}=\int e^{V^1(u,\bla)(\psi+\psi^1)}d\mu_{\cC_1}(\psi^1).
$$
By substituting this equality into \eqref{eq_next_analytic_continuation} and
 using the division formula again,
\begin{align}\label{eq_end_log_relation}
V^{end}(u,\bla)=\log\left(\int e^{V^1(u,\bla)(\psi)}d\mu_{\cC_0+\cC_1}(\psi)\right).
\end{align}
Moreover, by \eqref{eq_real_part_positive_full},
$$
e^{V^{end}(u,\bla)}=\int e^{V^1(u,\bla)(\psi)}d\mu_{\cC_0+\cC_1}(\psi).
$$
By the identity theorem, this equality must hold for any
 $(u,\bla)\in \overline{D(r)}\times \overline{D(r')}^2$. By the gauge
 transform $\psi_{\rho\bx s\xi}\mapsto
 e^{-i\xi\frac{\pi}{\beta}s}\psi_{\rho\bx s\xi}$ $((\rho,\bx,s,\xi)\in
 I)$ the right-hand side of the above equality becomes that of
 \eqref{eq_model_analytic_continuation}.
\end{proof}

Lemma \ref{lem_real_UV_integration} can be reduced to the following
explicit statements.
\begin{proposition}\label{prop_model_UV_application}
There exists a positive constant $c(d)\in\R_{\ge 1}$ depending only on
 $d$ such that the following statements hold with any $h\in
 \frac{2}{\beta}\N$ satisfying $h\ge c(d)\max\{1,\beta^{-1}\}$, $L\in
 \N$ satisfying \eqref{eq_L_largeness_final_condition} and $r\in\R_{>0}$
 defined by
$$
r:=c(d)^{-1}(1+\beta^{d+3}+(1+\beta^{-1})g_d(\Theta))^{-2}.
$$
\begin{align}
&\left|
\int e^{-V(u)(\psi)+W(u)(\psi)}d\mu_{C(\phi)}(\psi)-1\right|\le
 \frac{1}{2},\quad (\forall (\phi,u)\in\C\times
 \overline{D(r)}).\label{eq_model_good_smallness}\\
&\left|\frac{\int
 e^{-V(u)(\psi)+W(u)(\psi)}A^j(\psi)d\mu_{C(\phi)}(\psi)}{\int
 e^{-V(u)(\psi)+W(u)(\psi)}d\mu_{C(\phi)}(\psi)}-\int A^j(\psi)d\mu_{C(\phi)}(\psi)
\right|\label{eq_model_error_estimate}\\
&\le c(d)\beta 
(1+\beta^{d+3}+
(1+\beta^{-1})g_d(\Theta))^{2}(1+\Theta^{-1}
(1+\Theta^{-1})^d)L^{-d},\notag\\
&(\forall (\phi,u)\in\C\times
 \overline{D(r)},\ j\in \{1,2\}).\notag
\end{align}
\end{proposition}
\begin{proof}
Take $a$ to be 
$$
2\left(\log\left(\frac{3}{2}\right)\right)^{-\frac{1}{2}}(1+\beta)^{\frac{1}{2}},
$$
which is larger than 1.
Then, the $r$ in Lemma \ref{lem_real_UV_integration} and the right-hand
 side of \eqref{eq_real_correction_estimate} are rescaled to be the $r$ in
 this proposition and the right-hand side of
 \eqref{eq_model_error_estimate} respectively. Moreover,
 \eqref{eq_real_pressure_decay} implies that
\begin{align*}
|V^{end}(\phi,u,\b0)|\le \log\left(\frac{3}{2}\right),\quad (\forall
 (\phi,u)\in\C\times \overline{D(r)}).
\end{align*}
By combining this inequality with \eqref{eq_model_analytic_continuation}
 we obtain
 \eqref{eq_model_good_smallness}. Also, by 
differentiating \eqref{eq_model_analytic_continuation} the left-hand
 side of \eqref{eq_real_correction_estimate} is proved to be equal to that of
 \eqref{eq_model_error_estimate}.
\end{proof}

\subsection{Existence of the infinite-volume limit of the correction
  term}\label{subsec_L_limit}

According to \eqref{eq_real_pressure_decay},
\eqref{eq_end_log_relation}, \eqref{eq_model_good_smallness},
we know that
\begin{align*}
\lim_{L\to \infty\atop L\in \N}\limsup_{h\to\infty\atop h\in
 \frac{2}{\beta}\N}\sup_{(\phi,u)\in \C\times
 \overline{D(r)}}\frac{h}{N}\left|
\log\left(\int e^{-V(u)(\psi)+W(u)(\psi)}d\mu_{C(\phi)}(\psi)\right)
\right|=0.
\end{align*}
However, this property does not imply a uniform convergence of
$$
\int e^{-V(u)(\psi)+W(u)(\psi)}d\mu_{C(\phi)}(\psi)
$$
with $(\phi,u)$ as $h\to \infty$, $L\to \infty$. As we need such a
stronger convergence property to complete the proof of the main theorem,
let us prove it beforehand. Let us start by confirming spatial decay properties
of $C(\phi)$.

\begin{lemma}\label{lem_full_covariance_decay}
There exists a positive constant $c(d,\beta,\theta)$ depending only on
 $d,\beta,\theta$ such that 
\begin{align}
&|C(\phi)(\rho\bx s,\eta \by t)|\le
 \frac{c(d,\beta,\theta)}{1+\sum_{j=1}^d|\frac{L}{2\pi}(e^{i\frac{2\pi}{L}\<\bx-\by,\be_j\>}-1)|^{d+1}},\label{eq_bare_covariance_decay}\\
&(\forall (\rho,\bx,s),(\eta,\by,t)\in \{1,2\}\times \Z^d\times
 [0,\beta),\ \phi\in \C).\notag\\
&|C(\phi)(\rho\bx s,\eta \by t)|\le
 \frac{c(d,\beta,\theta)}{1+(\frac{2}{\pi})^{d+1}
\sum_{j=1}^d|\<\bx-\by,\be_j\>|^{d+1}},\label{eq_bare_covariance_decay_L_independent}\\
&(\forall (\rho,\bx,s),(\eta,\by,t)\in \{1,2\}\times \Z^d\times
 [0,\beta)\text{ with }\bx-\by\in [-L/2,L/2]^d,\ \phi\in \C).\notag
\end{align}
\end{lemma}

\begin{proof}
It follows from Lemma \ref{lem_real_covariances}
 \eqref{item_real_covariances_determinant_bound} for $n=1$, \eqref{eq_covariance_decomposition} and
 \eqref{eq_covariances_spatial_power} that the inequality
 \eqref{eq_bare_covariance_decay} holds for any
 $(\rho,\bx,s),(\eta,\by,t)\in \{1,2\}\times \Z^d\times [0,\beta)_h$,
 $\phi\in\C$. Since $(s,t)\mapsto C(\phi)(\rho\bx s,\eta \by t)$ is
 continuous in $\{(s,t)\in [0,\beta)^2\ |\ s\neq t\}$,
 \eqref{eq_bare_covariance_decay} can be claimed for any $s,t\in
 [0,\beta)$ with $s\neq t$ by approximating $s,t$ by converging
 sequences. The inequality
 \eqref{eq_bare_covariance_decay} also holds in the case $s=t$, since
 $C(\phi)(\rho\bx s,\eta \by s)=C(\phi)(\rho\bx 0,\eta \by 0)$. The
 inequality \eqref{eq_bare_covariance_decay_L_independent} follows from 
\eqref{eq_bare_covariance_decay}.
\end{proof}

By definition, $\phi\mapsto C(\phi)(\bX)$ is continuous in $\C$ for any
$\bX\in (\{1,2\}\times \Z^d\times [0,\beta))^2$ and thus
$$
(\phi,u)\mapsto \int e^{-V(u)(\psi)+W(u)(\psi)}d\mu_{C(\phi)}(\psi)
$$
is continuous in $\C^2$. In the rest of this subsection we prove the
following proposition, which requires deeper analysis of the tree expansion
than that performed in Subsection \ref{subsec_general_estimation}.

\begin{proposition}\label{prop_L_limit}
Let $r$ be the radius set in Proposition
 \ref{prop_model_UV_application}. For any non-empty compact set $Q$ of
 $\C$, 
\begin{align*}
&\lim_{h\to \infty\atop h\in \frac{2}{\beta}\N}\int
 e^{-V(u)(\psi)+W(u)(\psi)}d\mu_{C(\phi)}(\psi),\quad
\lim_{L\to \infty\atop L\in \N}\lim_{h\to \infty\atop h\in \frac{2}{\beta}\N}\int
 e^{-V(u)(\psi)+W(u)(\psi)}d\mu_{C(\phi)}(\psi)
\end{align*}
converge in $C(Q\times \overline{D(r/2)},\C)$ as sequences of function
 with the variable $(\phi,u)$.
\end{proposition}

\begin{proof} Recalling the definition \eqref{eq_0_1_1_initial_kernel}, 
 \eqref{eq_definition_initial_data_kernel}, let us define 
the anti-symmetric function $V_2:I^2\to \C$ and 
the bi-anti-symmetric function $V_{2,2}:I^2\times I^2\to\C$ by
$V_2(\bX):=V_2^{0-1,1}(1)(\bX)$, 
$V_{2,2}(\bX,\bY):=V_{2,2}^{0-2,1}(1)(\bX,\bY)$ and set 
\begin{align*}
&\hat{V}_2(\psi):=\frah^2\sum_{\bX\in
 I^2}V_{2}(\bX)\psi_{\bX},\quad  
\hat{V}_4(\psi):=\frah^4\sum_{\bX,\bY\in
 I^2}V_{2,2}(\bX,\bY)\psi_{\bX}\psi_{\bY},\\
&\hat{V}(\psi):=\hat{V}_2(\psi)+\hat{V}_4(\psi).
\end{align*}
It follows that $u\hat{V}(\psi)=-V(u)(\psi)+W(u)(\psi)$. For $\phi\in\C$ we
 define the function $\cG(\phi):(\{1,2\}\times \Z^d\times
 [0,\beta))^2\to \C$ by 
$\cG(\phi)(\rho\bx s,\eta \by t):=
 e^{-i\frac{\pi}{\beta}(s-t)}C(\phi)(\rho\bx s,\eta \by t)$.
By the gauge transform $\psi_{\rho\bx s\xi}\to e^{i\xi
 \frac{\pi}{\beta}s}\psi_{\rho\bx s\xi}$ $((\rho,\bx,s,\xi)\in I)$, 
$$
\int e^{u\hat{V}(\psi)}d\mu_{\cG(\phi)}(\psi)= \int e^{u\hat{V}(\psi)}d\mu_{C(\phi)}(\psi).
$$
Though we often drop the sign of $\phi$-dependency from $\cG(\phi)$ for
 simplicity in the following, the dependency on $\phi$ should be reminded especially
 when we establish uniform bounds with $\phi$. By
 \eqref{eq_model_good_smallness},
$$
(\phi,u)\mapsto \log\left(\int e^{u\hat{V}(\psi)}d\mu_{\cG(\phi)}(\psi)
\right)
$$
is continuous in $\C\times \overline{D(r)}$ and 
$$
u\mapsto \log\left(\int e^{u\hat{V}(\psi)}d\mu_{\cG(\phi)}(\psi)
\right)
$$
is analytic in $D(r)$ for any $\phi\in\C$.
Take a non-empty compact set $Q$ of $\C$. For $n\in \N$, $\phi\in Q$,
 set 
\begin{align*}
\alpha_{n,L,h}(\phi):=\frac{1}{n!}\left(\frac{d}{d u}\right)^n\log\left(
\int e^{u\hat{V}(\psi)}d\mu_{\cG}(\psi)\right)\Bigg|_{u=0}.
\end{align*}
Let us prove that $\alpha_{n,L,h}$ converges in $C(Q,\C)$ as a function of
 $\phi$ in the limit $h\to
 \infty$, $L\to \infty$. By the transformation close to
 \eqref{eq_L_M_transformation} we have
\begin{align*}
\alpha_{n,L,h}
&=\frac{1}{n!}Tree(\{1,2,\cdots,n\},\cG)\prod_{j=1}^n\hat{V}_2(\psi^j)
\Bigg|_{\psi^j=0\atop(\forall j\in\{1,2,\cdots,n\})}\\
&\quad + \sum_{l=1}^n\left(\begin{array}{c} n \\ l\end{array}\right)
\frac{1}{n!}Tree(\{1,2,\cdots,n\},\cG)\prod_{j=1}^l\hat{V}_4(\psi^j)
\prod_{k=l+1}^n\hat{V}_2(\psi^k)\Bigg|_{\psi^j=0\atop(\forall j\in\{1,2,\cdots,n\})}\\
&=\frac{1}{n!}Tree(\{1,2,\cdots,n\},\cG)\prod_{j=1}^n\hat{V}_2(\psi^j)
\Bigg|_{\psi^j=0\atop(\forall j\in\{1,2,\cdots,n\})}\\
&\quad + \sum_{l=1}^n\left(\begin{array}{c} n \\ l\end{array}\right)
\frah^4\sum_{\bX,\bY\in
 I^2}V_{2,2}(\bX,\bY)\\
&\qquad\cdot
 \frac{1}{n!}Tree(\{1,2,\cdots,n+1\},\cG)\psi_{\bX}^1\psi_{\bY}^2\prod_{j=3}^{l+1}\hat{V}_4(\psi^j)
\prod_{k=l+2}^{n+1}\hat{V}_2(\psi^k)
\Bigg|_{\psi^j=0\atop(\forall j\in\{1,2,\cdots,n+1\})}\\
&\qquad
 +\sum_{l=1}^n\left(\begin{array}{c} n \\ l\end{array}\right) \frac{1}{n!}\sum_{m=0}^{n-1}\sum_{(\{s_j\}_{j=1}^{m+1},\{t_k\}_{k=1}^{n-m})\in S(n,m)}\frah^4
\sum_{\bX,\bY\in
 I^2}V_{2,2}(\bX,\bY)\\
&\qquad\quad\cdot
 Tree(\{s_j\}_{j=1}^{m+1},\cG)\psi_{\bX}^{s_1}\prod_{j=2}^{m+1}
(1_{s_j\le l}\hat{V}_4(\psi^{s_j})+ 1_{s_j> l}\hat{V}_2(\psi^{s_j}))
\Bigg|_{\psi^{s_j}=0\atop(\forall j\in\{1,2,\cdots,m+1\})}\\
&\qquad\quad\cdot
 Tree(\{t_k\}_{k=1}^{n-m},\cG)\psi_{\bY}^{t_1}\prod_{k=2}^{n-m}
(1_{t_k\le l}\hat{V}_4(\psi^{t_k})+ 1_{t_k> l}\hat{V}_2(\psi^{t_k}))
\Bigg|_{\psi^{t_k}=0\atop(\forall k\in\{1,2,\cdots,n-m\})}.
\end{align*}
Note that the translation invariances
 \eqref{eq_time_translation_generic_covariance},
 \eqref{eq_time_translation_generic} are satisfied by $\cG$ and the
 kernels of $\hat{V}_2,\hat{V}_4$. This implies that for any $b_j\in
 \{2,4\}$ $(j=2,3,\cdots,m+1)$,
\begin{align*}
&\sum_{\bX\in
 I^2}V_{2,2}(\bX,\bY)
 Tree(\{s_j\}_{j=1}^{m+1},\cG)\psi_{\bX}^{s_1}\prod_{j=2}^{m+1}\hat{V}_{b_j}(\psi^{s_j})
\Bigg|_{\psi^{s_j}=0\atop(\forall j\in\{1,2,\cdots,m+1\})}\\
&=\sum_{\bX\in (I^0)^2\atop s\in
 [0,\beta)_h}V_{2,2}(\bX+s,\bY)Tree(\{s_j\}_{j=1}^{m+1},\cG)\psi_{\bX}^{s_1}
\prod_{j=2}^{m+1}\hat{V}_{b_j}(\psi^{s_j})
\Bigg|_{\psi^{s_j}=0\atop(\forall j\in\{1,2,\cdots,m+1\})}=0.
\end{align*}
Thus, 
$$ \alpha_{n,L,h}=\alpha'_{n,L,h}+a_{n,L,h},$$ 
where
\begin{align*}
&\alpha'_{n,L,h}\\
&:= \frac{1}{n!}Tree(\{1,2,\cdots,n\},\cG)\prod_{j=1}^n\hat{V}_2(\psi^j)
\Bigg|_{\psi^j=0\atop(\forall j\in\{1,2,\cdots,n\})}\\
&\quad + 1_{n\ge 2}\sum_{l=1}^{n-1}\left(\begin{array}{c} n \\ l\end{array}\right)
\frah^4\sum_{\bX,\bY\in
 I^2}V_{2,2}(\bX,\bY)\\
&\qquad\cdot
 \frac{1}{n!}Tree(\{1,2,\cdots,n+1\},\cG)\psi_{\bX}^1\psi_{\bY}^2\prod_{j=3}^{l+1}\hat{V}_4(\psi^j)
\prod_{k=l+2}^{n+1}\hat{V}_2(\psi^k)
\Bigg|_{\psi^j=0\atop(\forall j\in\{1,2,\cdots,n+1\})},\\
&a_{n,L,h}\\
&:=\frah^4\sum_{\bX,\bY\in
 I^2}V_{2,2}(\bX,\bY)\\
&\qquad\cdot
 \frac{1}{n!}Tree(\{1,2,\cdots,n+1\},\cG)\psi_{\bX}^1\psi_{\bY}^2\prod_{j=3}^{n+1}\hat{V}_4(\psi^j)\Bigg|_{\psi^j=0\atop(\forall j\in\{1,2,\cdots,n+1\})}.
\end{align*}
Let us show that $a_{n,L,h}$ converges in $C(Q,\C)$ as $h\to \infty$. Let us set
\begin{align*}
&\nu(s,t):=\frac{1}{\beta}-h1_{s=t},\quad (s,t\in [0,\beta)_h),\\
&V_{\bx s 1}(\psi):=\psi_{1\bx s 1}\psi_{2 \bx s -1},\quad 
V_{\bx s -1}(\psi):=\psi_{2\bx s 1}\psi_{1 \bx s -1},\quad (\bx\in\G,\
 s\in [0,\beta)_h).
\end{align*}
Then, by the periodicity and the translation
 invariance with the spatial variable we observe that
\begin{align}
&a_{n,L,h}\label{eq_expression_for_h_limit}\\
&=\frah^2\sum_{\bx\in\G\atop s,t\in [0,\beta)_h}\nu(s,t)\notag\\
&\qquad\cdot
 \frac{1}{n!}Tree(\{1,2,\cdots,n+1\},\cG)V_{\b0 s 1}(\psi^1)V_{\bx t -1}(\psi^2)
\prod_{j=3}^{n+1}\hat{V}_4(\psi^j)\Bigg|_{\psi^j=0\atop(\forall
 j\in\{1,2,\cdots,n+1\})}\notag\\
&=L^{-d(n-1)}\sum_{a_1\in \{1,2\}}\sum_{\bx_2\in
 \G}\prod_{j=3}^{n+1}\left(\sum_{\bx_j,\by_j\in
 \G}\sum_{a_j\in \{1,2\}}\right)\notag\\
&\quad\cdot \prod_{j=1\atop j\neq 2}^{n+1}\left(
1_{a_j=2}\frac{1}{\beta h^2}\sum_{\bs_j\in
 [0,\beta)_h^{a_j}}
- 1_{a_j=1}\frac{1}{h}\sum_{\bs_j\in
 [0,\beta)_h^{a_j}}\right)\notag\\
&\quad\cdot
 f(\phi)((a_1,(a_j)_{j=3}^{n+1}),(\bs_1,(\bs_j)_{j=3}^{n+1}),(\bx_2,(\bx_j,\by_j)_{j=3}^{n+1})),\notag
\end{align}
where we set
\begin{align*}
&f(\phi)((a_1,(a_j)_{j=3}^{n+1}),(\bs_1,(\bs_j)_{j=3}^{n+1}),(\bx_2,(\bx_j,\by_j)_{j=3}^{n+1}))\\&:=\frac{1}{n!}Tree(\{1,2,\cdots,n+1\},\cG(\phi))\\
&\quad\cdot (1_{a_1=1}V_{\b0
 s_1 1}(\psi^1)V_{\bx_2 s_1 -1}(\psi^2)+ 1_{a_1=2}V_{\b0
 s_{1,1} 1}(\psi^1)V_{\bx_2 s_{1,2} -1}(\psi^2))\\
&\quad\cdot\prod_{j=3}^{n+1}(1_{a_j=1}V_{\bx_j
 s_j 1}(\psi^j)V_{\by_j s_j -1}(\psi^j)+ 1_{a_j=2}V_{\bx_j
 s_{j,1} 1}(\psi^j)V_{\by_j s_{j,2} -1}(\psi^j))
\Bigg|_{\psi^j=0\atop(\forall
 j\in\{1,2,\cdots,n+1\})}.
\end{align*}
By recalling the definition of $Tree(\{1,2,\cdots,n+1\},\cG(\phi))$ we
 see that $f$ consists of a finite sum of products of $\cG(\phi)$ and thus
 the domain of the function $\bs\mapsto f(\phi)(\ba,\bs,\bX)$ is
 naturally extended to be $[0,\beta)^{a_1}\times
 \prod_{j=3}^{n+1}[0,\beta)^{a_j}$. For simplicity, set
 $[0,\beta)^{\ba}:=[0,\beta)^{a_1}\times
 \prod_{j=3}^{n+1}[0,\beta)^{a_j}$, $|\ba|:=a_1+\sum_{j=3}^{n+1}a_j$
for
 $\ba=(a_1,(a_j)_{j=3}^{n+1})$. We can see from the definition that 
for any $s_0,t_0\in [0,\beta)$ with $s_0\neq t_0$, $\rho,\eta\in \{1,2\}$,
 $\bx,\by\in\Z^d$,
\begin{align}
\lim_{(s,t)\to (s_0,t_0)}\sup_{\phi\in Q}|\cG(\phi)(\rho\bx s,\eta \by
 t)- \cG(\phi)(\rho\bx s_0,\eta \by
 t_0)|=0.\label{eq_convergence_first_principle}
\end{align}
Define the subset $S$ of $[0,\beta)^{\ba}$ by
\begin{align*}
S:=\{\bs\in [0,\beta)^{\ba}\ |\
 \bs=(s_1,s_2,\cdots,s_{|\ba|}),\ (i,j\in \{1,\cdots,|\ba|\})\land
i\neq j\to s_i\neq s_j\}.
\end{align*}
Note that the Lebesgue measure of $[0,\beta)^{\ba}\backslash S$ is zero.
It follows from the properties \eqref{eq_phi_good_property},
 \eqref{eq_exponent_matrix_bound}, \eqref{eq_bare_covariance_decay}, 
\eqref{eq_convergence_first_principle} that
\begin{align}
&\lim_{\bs\to \bs_0\atop \bs\in [0,\beta)^{\ba}}\sup_{\phi\in
 Q}|f(\phi)(\ba,\bs,\bX)-f(\phi)(\ba,\bs_0,\bX)|=0,\quad (\forall
 \bs_0\in S),\label{eq_convergence_from_first_principle}\\
&\sup_{\bs\in [0,\beta)^{\ba}}\sup_{\phi\in
 Q}|f(\phi)(\ba,\bs,\bX)|<\infty.\label{eq_uniform_bound_from_first_principle}
\end{align}
We can consider $f(\cdot)(\ba,\cdot,\bX)$ as an element of
 $L^1([0,\beta)^{\ba},C(Q,\C))$. For any $s\in [0,\beta)$ there uniquely exists $s'\in [0,\beta)_h$ such
 that $s\in [s',s'+\frac{1}{h})$. Let us define the map $p_h:[0,\beta)\to
 [0,\beta)_h$ by $p_h(s):=s'$. Then, define the map $P_h:[0,\beta)^n\to
 [0,\beta)_h^n$ by
 $P_h(s_1,\cdots,s_n):=(p_h(s_1),\cdots,p_h(s_n))$. It
 follows from \eqref{eq_convergence_from_first_principle} that 
\begin{align}
\lim_{h\to\infty\atop h\in \frac{2}{\beta}\N}\sup_{\phi\in
 Q}|f(\phi)(\ba,P_h(\bs),\bX)- f(\phi)(\ba,\bs,\bX)|=0,\quad (\forall
 \bs\in S).\label{eq_convergence_from_first_principle_application} 
\end{align}
By \eqref{eq_uniform_bound_from_first_principle},
 \eqref{eq_convergence_from_first_principle_application} we can apply
 the dominated convergence theorem in $L^1([0,\beta)^{\ba},$ $C(Q,\C))$ to
 ensure that
\begin{align*}
\lim_{h\to\infty\atop h\in\frac{2}{\beta}\N}\int_{[0,\beta)^{\ba}}d\bs
 f(\cdot)(\ba,P_h(\bs),\bX)=
 \int_{[0,\beta)^{\ba}}d\bs
 f(\cdot)(\ba,\bs,\bX)
\text{ in }C(Q,\C).
\end{align*}
By using this convergence property in \eqref{eq_expression_for_h_limit}
 we can reach the conclusion that $a_{n,L,h}(\cdot)$ converges in $C(Q,\C)$
 as $h\to \infty$. In the same way as above we can prove that
 $\alpha'_{n,L,h}(\cdot)$ converges in $C(Q,\C)$ as $h\to \infty$.   

Next we will prove the convergence property as $L\to \infty$. 
Let us prove the convergence of $a_{n,L,h}$ first. The proof for the
 convergence of $\alpha'_{n,L,h}$ is much simpler because of the presence of
 the term $\hat{V}_2$. We will see it after completing the proof for $a_{n,L,h}$.
For this
 purpose we need to disclose the operator
 $Tree(\{1,2,\cdots,n+1\},\cG)$. For $p,q\in \{1,2,\cdots,n+1\}$ set
$$
B_{\{p,q\}}:=\sum_{\bX\in I^2}\tilde{\cG}(\bX)\frac{\partial}{\partial \psi_{X_1}^p}\frac{\partial}{\partial \psi_{X_2}^q},
$$
where $\tilde{\cG}$ is the anti-symmetric extension of $\cG$ defined as
 in \eqref{eq_anti_symmetric_extension}. Note that by anti-symmetry
 $B_{\{p,q\}}=B_{\{q,p\}}$. 
We can rewrite $a_{n,L,h}$ as
 follows.
\begin{align*}
a_{n,L,h}=\frac{2^n}{(n!)^2}\sum_{T\in \T(\{1,2,\cdots,n+1\})}a_{n,L,h}(T),
\end{align*}
where
\begin{align}
&a_{n,L,h}(T)\label{eq_expansion_fixed_tree}\\
&:=L^{-d(n-1)}\int_{[0,1]^n}d\bs\sum_{\s\in \S_{n+1}(T)}\varphi(T,\s,\bs)
\left(\sum_{p,q=1}^{n+1}M(T,\s,\bs)_{p,q}B_{\{p,q\}}\right)^n\prod_{\{p,q\}\in
 T}B_{\{p,q\}}\notag\\
&\quad\cdot 
\frah^2\sum_{\bx\in\G\atop s,t\in [0,\beta)_h}\nu(s,t)V_{\b0 s
 1}(\psi^1)V_{\bx t -1}(\psi^2)\notag\\
&\quad\cdot\prod_{j=3}^{n+1}\Bigg(
\frah^2\sum_{\bx,\by\in\G\atop s,t\in [0,\beta)_h}\nu(s,t)V_{\bx s
 1}(\psi^j)V_{\by t -1}(\psi^j)
\Bigg)\Bigg|_{\psi^j=0\atop (\forall j\in \{1,2,\cdots,n+1\})}.\notag
\end{align}
We have to study case by case depending on the tree's configuration.
Let us begin with the simplest case $n=1$. Set 
$$
\G_L:=\left\{-\left\lfloor\frac{L}{2}\right\rfloor,
 -\left\lfloor\frac{L}{2}\right\rfloor+1,\cdots, -\left\lfloor\frac{L}{2}\right\rfloor+L-1\right\}^d.
$$
Since $\T(\{1,2\})=\{1,2\}$ and $M(T,\s,\bs)$ is symmetric,
\begin{align*}
a_{1,L,h}(T)
&=\int_{[0,1]}ds\sum_{\s\in \S_2(T)}\varphi(T,\s,s)
2M(T,\s,s)_{1,2}B_{\{1,2\}}^2\\
&\quad\cdot \frah^2\sum_{\bx\in\G\atop u,t\in [0,\beta)_h}\nu(u,t)V_{\b0 u
 1}(\psi^1)V_{\bx t -1}(\psi^2)\Bigg|_{\psi^1=\psi^2=0}\\
&=\int_{[0,1]}ds\sum_{\s\in \S_2(T)}\varphi(T,\s,s)
4M(T,\s,s)_{1,2}\sum_{\bx\in\Z^d}1_{\bx\in \G_L}\\
&\quad\cdot \Bigg(\frac{1}{\beta}\int_0^{\beta}du\int_0^{\beta}dt\tilde{\cG}(2\b0
 p_h(u)(-1),2\bx p_h(t) 1)
\tilde{\cG}(1\b0
 p_h(u) 1,1\bx p_h(t) (-1))\\
&\qquad -\beta \tilde{\cG}(2\b0 0(-1),2\bx 01)\tilde{\cG}(1\b0 01,1\bx 0(-1))
\Bigg).
\end{align*}
Thus,
\begin{align*}
&\lim_{h\to \infty\atop h\in \frac{2}{\beta}\N}a_{1,L,h}(T)\\
&=\int_{[0,1]}ds\sum_{\s\in \S_2(T)}\varphi(T,\s,s)
4M(T,\s,s)_{1,2}\sum_{\bx\in\Z^d}1_{\bx\in
 \G_L}\frac{1}{\beta}\int_0^{\beta}du\int_0^{\beta}dt\\
&\quad\cdot \left(\tilde{\cG}(2\b0
 u (-1),2\bx t 1)
\tilde{\cG}(1\b0
 u 1,1\bx t (-1))- \tilde{\cG}(2\b0 0(-1),2\bx 01)\tilde{\cG}(1\b0 01,1\bx 0(-1))\right).
\end{align*}
We can deduce from the definition that
for any $\bX\in (\{1,2\}\times\Z^d\times[0,\beta)\times \{1,-1\})^2$,
 $\lim_{L\to \infty}\tilde{\cG}(\bX)$ converges in
 $C(Q,\C)$. Moreover by \eqref{eq_bare_covariance_decay_L_independent},
\begin{align*}
&\sup_{\phi\in Q}1_{\bx\in \G_L}|\tilde{\cG}(2\b0
 u (-1),2\bx t 1)
\tilde{\cG}(1\b0
 u 1,1\bx t (-1))- \tilde{\cG}(2\b0 0(-1),2\bx 01)\tilde{\cG}(1\b0 01,1\bx 0(-1))|\\
&\le \frac{c\cdot
 c(d,\beta,\theta)^2}{(1+(\frac{2}{\pi})^{d+1}\sum_{j=1}^d|\<\bx,\be_j\>|^{d+1})^2}.
\end{align*}
Therefore, the dominated convergence theorem in $L^1(\Z^d\times
 [0,\beta)^2,C(Q,\C))$ guarantees that 
$\lim_{L\to\infty, L\in\N}\lim_{h\to\infty, h\in
 \frac{2}{\beta}\N}a_{1,L,h}$ converges in $C(Q,\C)$. 

Let us consider the case $n\in \N_{\ge 2}$.
To make clear the structure, let us add the superscript $1,-1$ to the
 Grassmann variables and rewrite the formula
 \eqref{eq_expansion_fixed_tree} as follows.
\begin{align*}
a_{n,L,h}(T)
&=\frah^2\sum_{s_1,t_2\in[0,\beta)_h}\nu(s_1,t_2)\prod_{j=3}^{n+1}
\left(\frah^2\sum_{s_j,t_j\in[0,\beta)_h}\nu(s_j,t_j)\right)
L^{-d(n-1)}E R^n\\
&\quad\cdot \prod_{\{p,q\}\in
 T}\left(\sum_{f,g\in\{1,-1\}}B_{(p,f),(q,g)}\right)
\sum_{\bx\in\G}V_{\b0 s_1
 1}(\psi^{1,1})V_{\bx t_2 -1}(\psi^{2,-1})\\
&\quad\cdot \prod_{j=3}^{n+1}\Bigg(
\sum_{\bx,\by\in\G}V_{\bx s_j 1}(\psi^{j,1})V_{\by t_j -1}(\psi^{j,-1})
\Bigg)\Bigg|_{\psi^{j,\delta}=0\atop (\forall j\in
 \{1,2,\cdots,n+1\},\delta\in \{1,-1\})},
\end{align*}
where 
\begin{align*}
&E:=\int_{[0,1]^n}d\bs\sum_{\s\in \S_{n+1}(T)}\varphi(T,\s,\bs),\\
&B_{(p,f),(q,g)}:=\sum_{\bX\in
 I^2}\tilde{\cG}(\bX)\frac{\partial}{\partial
 \psi_{X_1}^{p,f}}\frac{\partial}{\partial \psi_{X_2}^{q,g}},\\
&R:=\sum_{a,b\in
 \{1,-1\}}\sum_{p,q=1}^{n+1}M(T,\s,\bs)_{p,q}B_{(p,a),(q,b)}.
\end{align*}
Since the integration and the summation with $\bs$, $\s$ are irrelevant
 in the following argument, we do not indicate the dependency of $R$ on
 these variables. For $T\in \T(\{1,2,\cdots,n+1\})$ we consider the
 vertex 1 as the root of $T$. For $j\in \{1,2,\cdots,n+1\}$ let
 $\dis_T(1,j)$ denote the length of the shortest path between 1 and $j$ in
 $T$. Let us consider the case that 
\begin{align}
\exists v\in \{3,4,\cdots,n+1\}(\{j,v\}\in T\to \dis_T(j,1)+1=\dis_T(v,1)).\label{eq_vertex_logical}
\end{align}
In this case $v$ is the terminal of a branch of $T$ and thus
there uniquely exists $v'\in
 \{1,2,\cdots,n+1\}\backslash\{v\}$ such that $\{v',v\}\in T$. 
The operator $B_{(v,a),(v',b)}$ erases one Grassmann variable from 
$V_{\bx s_v 1}(\psi^{v,1})V_{\by t_v -1}(\psi^{v,-1})$. The
 remaining 2 variables with the superscript $`v,-a'$ must be erased by $R^n$.
Thus, there is at least one operator, at most two operators among $n$
 of $R$ such that they are to act on the Grassmann
 variables with the superscript $`v,-a'$. By
 decomposing these operators we have that
\begin{align*}
&a_{n,L,h}(T)\\
&=\frah^2\sum_{s_1,t_2\in[0,\beta)_h}\nu(s_1,t_2)\prod_{j=3}^{n+1}
\left(\frah^2\sum_{s_j,t_j\in[0,\beta)_h}\nu(s_j,t_j)\right)
L^{-d(n-1)}E\\
&\quad\cdot \prod_{\{p,q\}\in
 T\backslash\{\{v,v'\}\}}\left(\sum_{f,g\in\{1,-1\}}B_{(p,f),(q,g)}\right)
\sum_{a,b\in\{1,-1\}}B_{(v,a),(v',b)}\\
&\quad\cdot \Bigg(n R^{n-1}B_{(v,-a),(v,-a)}\\
&\qquad\quad+\left(\begin{array}{c} n \\ 2 \end{array}\right) R^{n-2}\Bigg(
2\sum_{c\in\{1,-1\}}\sum_{p=1}^{n+1}1_{(p,c)\neq (v,-a)}
M(T,\s,\bs)_{p,v}B_{(p,c),(v,-a)}\Bigg)^2\Bigg)\\
&\quad\cdot \sum_{\bx\in\G}V_{\b0 s_1
 1}(\psi^{1,1})V_{\bx t_2 -1}(\psi^{2,-1})\\
&\quad\cdot \prod_{j=3}^{n+1}\Bigg(
\sum_{\bx,\by\in\G}V_{\bx s_j 1}(\psi^{j,1})V_{\by t_j -1}(\psi^{j,-1})
\Bigg)\Bigg|_{\psi^{j,\delta}=0\atop (\forall j\in
 \{1,2,\cdots,n+1\},\delta\in \{1,-1\})}\\
&=n\frah^2\sum_{s_1,t_2\in[0,\beta)_h}\nu(s_1,t_2)\prod_{j=3}^{n+1}
\left(\frah^2\sum_{s_j,t_j\in[0,\beta)_h}\nu(s_j,t_j)\right)
L^{-d(n-1)}ER^{n-1}\\
&\quad\cdot \prod_{\{p,q\}\in
 T\backslash\{\{v,v'\}\}}\left(\sum_{f,g\in\{1,-1\}}B_{(p,f),(q,g)}\right)
\sum_{a,b\in\{1,-1\}}B_{(v,a),(v',b)}\\
&\quad\cdot\sum_{c\in\{1,-1\}}\sum_{p=1}^{n+1}
M(T,\s,\bs)_{p,v}B_{(p,c),(v,-a)}\\
&\quad\cdot \sum_{\bx\in\G}V_{\b0 s_1
 1}(\psi^{1,1})V_{\bx t_2 -1}(\psi^{2,-1})\\
&\quad\cdot \prod_{j=3}^{n+1}\Bigg(
\sum_{\bx,\by\in\G}V_{\bx s_j 1}(\psi^{j,1})V_{\by t_j -1}(\psi^{j,-1})
\Bigg)\Bigg|_{\psi^{j,\delta}=0\atop (\forall j\in
 \{1,2,\cdots,n+1\},\delta\in \{1,-1\})},
\end{align*}
where we used that
\begin{align*}
M(T,\s,\bs)_{p,q}=M(T,\s,\bs)_{q,p},\quad M(T,\s,\bs)_{p,p}=1,\quad
 B_{(p,f),(q,g)}=B_{(q,g),(p,f)}.
\end{align*}
Remark that
\begin{align}
 &\sum_{s\in [0,\beta)_h}\nu(s,t)B_{(v,1),(v,1)}V_{\bx s 1}(\psi^{v,1})=
\sum_{t\in [0,\beta)_h}\nu(s,t)B_{(v,-1),(v,-1)}V_{\by t
 -1}(\psi^{v,-1})=0.\label{eq_vanishing_case_inside}
\end{align}
Thus, the operator $B_{(v,-a),(v,-a)}$ in the above expansion
  can be eliminated. As the result, 
\begin{align*}
&a_{n,L,h}(T)\\
&=n
\frah^2\sum_{s_1,t_2\in[0,\beta)_h}\nu(s_1,t_2)\prod_{j=3}^{n+1}
\left(\frah^2\sum_{s_j,t_j\in[0,\beta)_h}\nu(s_j,t_j)\right)
L^{-d(n-1)}E R^{n-1}\\
&\quad\cdot \sum_{a\in\{1,-1\}}\Bigg(\Bigg(
\sum_{c\in \{1,-1\}}\sum_{p=1\atop p\neq v}^{n+1}
 M(T,\s,\bs)_{p,v}B_{(p,c),(v,-a)}+B_{(v,1),(v,-1)}\Bigg)\\
&\qquad\cdot \prod_{\{p,q\}\in
 T\backslash\{\{v,v'\}\}}\Bigg(\sum_{f,g\in\{1,-1\}}B_{(p,f),(q,g)}\Bigg)
\sum_{b\in\{1,-1\}}B_{(v,a),(v',b)}\Bigg)\\
&\quad\cdot \sum_{\bx\in\G}V_{\b0 s_1
 1}(\psi^{1,1})V_{\bx t_2 -1}(\psi^{2,-1})\\
&\quad\cdot \prod_{j=3}^{n+1}\Bigg(
\sum_{\bx,\by\in\G}V_{\bx s_j 1}(\psi^{j,1})V_{\by t_j -1}(\psi^{j,-1})
\Bigg)\Bigg|_{\psi^{j,\delta}=0\atop (\forall j\in
 \{1,2,\cdots,n+1\},\delta\in \{1,-1\})}.
\end{align*}
Set for $a\in \{1,-1\}$ 
\begin{align*}
&R(a):=\sum_{c\in\{1,-1\}}\sum_{p=1\atop p\neq
 v}^{n+1}M(T,\s,\bs)_{p,v}B_{(p,c),(v,-a)}+
B_{(v,1),(v,-1)},\\
&B(a):=\prod_{\{p,q\}\in
 T\backslash\{\{v,v'\}\}}\Bigg(\sum_{f,g\in\{1,-1\}}B_{(p,f),(q,g)}\Bigg)
\sum_{b\in\{1,-1\}}B_{(v,a),(v',b)}
\end{align*}
to simplify the following explanation. We carry out a recursive
 estimation along the tree lines from younger branches to the root
 1. Here we need to estimate along the straight line whose
 terminal is the vertex $v$ first of all. We apply $B(a)$ and then $R(a)$ to
 the given Grassmann polynomial. The rest of the Grassmann variables are
 erased by $R^{n-1}$. The application of $R(a)$ yields another
 $\tilde{\cG}(\cdot)$ which together with $\tilde{\cG}(\cdot)$ created by
 $B(a)$ are integrated with respect to the variables at the
 vertex $v$. The application of $B(a)$ combinatorially yields at most 
$$
\prod_{j=1}^2\left(\begin{array}{c}2 \\ d_j(T)\end{array}\right)
 d_j(T)!\cdot \prod_{k=3}^{n+1}\left(\begin{array}{c}4 \\ d_k(T)\end{array}\right)d_k(T)!
$$
factors, which is bounded by $c^n$ with a generic positive
 constant $c$. Recall that $d_j(T)$ is the degree of the vertex $j$ in
$T$. After applying $B(a)$ and $R(a)$ we have Grassmann monomials of degree
 $2(n-1)$.
 Applying $R^{n-1}$ to each of the remaining monomials combinatorially gives at most
 $(2(n-1))!$ factors.  By performing the recursive estimation as
 described above and
 using \eqref{eq_phi_good_property},
 \eqref{eq_exponent_matrix_bound}, \eqref{eq_bare_covariance_decay} we
 observe that
\begin{align}
&|a_{n,L,h}(T)|\label{eq_4_field_terminal_bound}\\
&\le
 L^{-d(n-1)}(2(n-1))!c^nc(d,\beta,\theta)^{n-1}\notag\\
&\quad\cdot \frac{1}{h}\sum_{s\in[0,\beta)_h}\sup_{X\in I,\eta\in \{1,2\}\atop \zeta\in \{1,-1\}}\Bigg(
\frac{1}{h}\sum_{\by\in\G\atop t\in [0,\beta)_h}|\tilde{\cG}(X,\eta\by t\zeta)||\nu(s,t)|\Bigg)\notag\\
&\quad\cdot \Bigg(\sup_{X\in I,\eta\in \{1,2\}\atop \zeta\in \{1,-1\}}
\Bigg(\frac{1}{h^2}\sum_{\bx,\by\in\G\atop s,t\in
 [0,\beta)_h}|\tilde{\cG}(X,\eta\by
 t\zeta)||\nu(s,t)|\Bigg)\Bigg)^{n-2}\notag\\
&\quad\cdot 
\Bigg(\sup_{X,Z\in I,\rho,\eta\in \{1,2\}\atop \xi,\zeta\in \{1,-1\}}
\Bigg(\frac{1}{h^2}\sum_{\bx,\by\in\G\atop s,t\in
 [0,\beta)_h}|\tilde{\cG}(X,\rho\bx s\xi)||\tilde{\cG}(Z,\eta\by
 t\zeta)||\nu(s,t)|\Bigg)\notag\\
&\qquad\quad+
\sup_{X\in I,\rho,\rho',\eta\in \{1,2\}\atop \xi,\xi',\zeta\in \{1,-1\}}
\Bigg(\frac{1}{h^2}\sum_{\bx,\by\in\G\atop s,t\in
 [0,\beta)_h}|\tilde{\cG}(X,\rho\bx s\xi)||\tilde{\cG}(\rho'\bx s\xi',\eta\by
 t\zeta)||\nu(s,t)|\Bigg)\Bigg)\notag\\
&\le 
 L^{-d(n-1)}(2(n-1))!c^nc(d,\beta,\theta)^{2n}\beta \sum_{\bx\in
 \G}\frac{1}{1+\sum_{j=1}^d|\frac{L}{2\pi}(e^{i\frac{2\pi}{L}\<\bx,\be_j\>}-1)|^{d+1}}\notag\\
&\quad\cdot \left(
\beta
 L^d\sum_{\bx\in\G}\frac{1}{1+\sum_{j=1}^d|\frac{L}{2\pi}(e^{i\frac{2\pi}{L}\<\bx,\be_j\>}-1)|^{d+1}}\right)^{n-2}\notag\\
&\quad \cdot 
\beta\left(\sum_{\bx\in
 \G}\frac{1}{1+\sum_{j=1}^d|\frac{L}{2\pi}(e^{i\frac{2\pi}{L}\<\bx,\be_j\>}-1)|^{d+1}}
\right)^2\notag\\
&\le c^n(2n)!c(d,\beta,\theta)^{2n}\beta^nL^{-d}.\notag
\end{align}

Next let us consider the case that \eqref{eq_vertex_logical} does not
 hold. In this case the tree $T$ is one straight line whose terminal is
 the vertex 2. By changing the numbers if necessary we may assume that
$$
T=\{\{1,n+1\},\{n+1,n\},\cdots,\{4,3\},\{3,2\}\}.
$$
The term $a_{n,L,h}(T)$ can be further decomposed as follows. 
$$
a_{n,L,h}(T)=\sum_{a,b,c,d\in \{1,-1\}}a_{n,L,h}^{(a,b,c,d)}(T),
$$
where
\begin{align*}
a_{n,L,h}^{(a,b,c,d)}(T)
&:=\frah^2\sum_{s_1,t_2\in[0,\beta)_h}\nu(s_1,t_2)\prod_{j=3}^{n+1}\left(\frah^{2}\sum_{s_j,t_j\in
 [0,\beta)_h}\nu(s_j,t_j)\right)
L^{-d(n-1)}ER^n\\
&\quad\cdot 
\prod_{\{p,q\}\in T\backslash \{\{2,3\},\{3,4\}\}}\left(
\sum_{f,g\in \{1,-1\}}B_{(p,f),(q,g)}\right)
B_{(2,a),(3,b)}B_{(3,c),(4,d)}\\
&\quad\cdot \sum_{\bx\in\G}V_{\b0 s_1 1}(\psi^{1,1})V_{\bx t_2 -1}(\psi^{2,-1})\\
&\quad\cdot \prod_{j=3}^{n+1}\left(
\sum_{\bx,\by\in\G}V_{\bx s_j 1}(\psi^{j,1})V_{\by t_j -1}(\psi^{j,-1})
\right)\Bigg|_{\psi^{j,\delta}=0\atop (\forall j\in \{1,2,\cdots,n+1\},
 \delta\in \{1,-1\})}.
\end{align*}

Let us consider $a_{n,L,h}^{(a,1,1,d)}(T)$. 
In this case one Grassmann variable with the superscript $`2,-1'$ and two
 Grassmann variables with the superscript $`3,-1'$ are untouched by the
 derivatives along the tree lines and thus must be erased by the
 operator $R^n$. Set
\begin{align*}
&R':=2^2M(T,\s,\bs)_{3,2}B_{(3,-1),(2,-1)}+
2\sum_{\delta\in
 \{1,-1\}}\sum_{p=1}^{n+1}1_{(p,\delta)\neq
 (3,-1)}M(T,\s,\bs)_{p,2}B_{(p,\delta),(2,-1)},\\
&R'':=\sum_{\delta\in
 \{1,-1\}}\sum_{p=1}^{n+1} 1_{(p,\delta)\neq (2,-1),(3,-1)}
M(T,\s,\bs)_{p,3} B_{(p,\delta),(3,-1)},\\
&B^{(a,1,1,d)}:=
\prod_{\{p,q\}\in T\backslash \{\{2,3\},\{3,4\}\}}\left(
\sum_{f,g\in \{1,-1\}}B_{(p,f),(q,g)}\right)
B_{(2,a),(3,1)}B_{(3,1),(4,d)}\\
&\qquad\qquad\quad\cdot \sum_{\bx\in\G}V_{\b0 s_1 1}(\psi^{1,1})V_{\bx t_2 -1}(\psi^{2,-1})\prod_{j=3}^{n+1}\left(
\sum_{\bx,\by\in\G}V_{\bx s_j 1}(\psi^{j,1})V_{\by t_j -1}(\psi^{j,-1})
\right).
\end{align*}
Let us observe that
\begin{align*}
&a_{n,L,h}^{(a,1,1,d)}(T)\\
&=\frah^2\sum_{s_1,t_2\in[0,\beta)_h}\nu(s_1,t_2)\prod_{j=3}^{n+1}\left(\frah^{2}\sum_{s_j,t_j\in
 [0,\beta)_h}\nu(s_j,t_j)\right)L^{-d(n-1)}E\\
&\quad\cdot n R^{n-1}\left(
2\sum_{\delta\in
 \{1,-1\}}\sum_{p=1}^{n+1}M(T,\s,\bs)_{p,2}B_{(p,\delta),(2,-1)}\right)
B^{(a,1,1,d)}\Bigg|_{\psi^{j,\delta}=0\atop (\forall j\in \{1,2,\cdots,n+1\},
 \delta\in \{1,-1\})}\\
&=n\frah^2\sum_{s_1,t_2\in[0,\beta)_h}\nu(s_1,t_2)\prod_{j=3}^{n+1}\left(\frah^{2}\sum_{s_j,t_j\in
 [0,\beta)_h}\nu(s_j,t_j)\right) L^{-d(n-1)}E\\
&\quad\cdot\Bigg((n-1) R^{n-2}\Bigg(
2\sum_{\eta\in
 \{1,-1\}}\sum_{q=1}^{n+1}1_{(q,\eta)\neq (3,-1)}
M(T,\s,\bs)_{q,3}B_{(q,\eta),(3,-1)}\Bigg)\\
&\qquad\qquad\cdot 2M(T,\s,\bs)_{3,2}B_{(3,-1),(2,-1)}\\
&\qquad+\Bigg( (n-1)R^{n-2}B_{(3,-1),(3,-1)}\\
&\qquad\quad+1_{n\ge 3}\left(\begin{array}{c}n-1 \\ 2 \end{array}\right) R^{n-3}
\Bigg(2\sum_{\eta\in \{1,-1\}}\sum_{q=1}^{n+1}1_{(q,\eta)\neq (3,-1)}
M(T,\s,\bs)_{q,3}B_{(q,\eta),(3,-1)}\Bigg)^2\Bigg)\\
&\qquad\qquad\cdot \Bigg(
2\sum_{\delta\in \{1,-1\}}\sum_{p=1}^{n+1}1_{(p,\delta)\neq (3,-1)}
M(T,\s,\bs)_{p,2}B_{(p,\delta),(2,-1)}\Bigg)\Bigg)\\
&\quad\cdot B^{(a,1,1,d)}\Bigg|_{\psi^{j,\delta}=0\atop (\forall j\in \{1,2,\cdots,n+1\},
 \delta\in \{1,-1\})}\\
&=n\frah^2\sum_{s_1,t_2\in[0,\beta)_h}\nu(s_1,t_2)\prod_{j=3}^{n+1}\left(\frah^{2}\sum_{s_j,t_j\in
 [0,\beta)_h}\nu(s_j,t_j)\right)L^{-d(n-1)}E\\
&\quad\cdot \Bigg((n-1) R^{n-2}\Bigg(
2\sum_{\eta\in
 \{1,-1\}}\sum_{q=1}^{n+1}1_{(q,\eta)\neq (3,-1)}
M(T,\s,\bs)_{q,3}B_{(q,\eta),(3,-1)}\Bigg)\\
&\qquad\qquad\cdot 2M(T,\s,\bs)_{3,2}B_{(3,-1),(2,-1)}\\
&\qquad+(n-1)R^{n-2}\Bigg(
\sum_{\eta\in
 \{1,-1\}}\sum_{q=1}^{n+1} M(T,\s,\bs)_{q,3}B_{(q,\eta),(3,-1)}\Bigg)\\
&\qquad\qquad\cdot 
\Bigg(2
\sum_{\delta\in
 \{1,-1\}}\sum_{p=1}^{n+1}
1_{(p,\delta)\neq (3,-1)} M(T,\s,\bs)_{p,2}B_{(q,\delta),(2,-1)}\Bigg)\Bigg)\\
&\quad\cdot B^{(a,1,1,d)}\Bigg|_{\psi^{j,\delta}=0\atop (\forall j\in \{1,2,\cdots,n+1\},
 \delta\in \{1,-1\})}\\
&=n(n-1)\frah^2\sum_{s_1,t_2\in[0,\beta)_h}\nu(s_1,t_2)\prod_{j=3}^{n+1}\left(\frah^{2}\sum_{s_j,t_j\in
 [0,\beta)_h}\nu(s_j,t_j)\right)L^{-d(n-1)}E\\
&\quad\cdot  R^{n-2}R' R''
 B^{(a,1,1,d)}\Bigg|_{\psi^{j,\delta}=0\atop (\forall j\in \{1,2,\cdots,n+1\},
 \delta\in \{1,-1\})}.
\end{align*}
In the derivation of the last equality we used the fact that since
 \eqref{eq_vanishing_case_inside} with $v=3$ holds, the term with
 $B_{(3,-1),(3,-1)}$ does not contribute to the result.
It is important that there is no link between the vertex 2 and the
 vertex 3 in the operator $R''$.
By using \eqref{eq_phi_good_property},
 \eqref{eq_exponent_matrix_bound}, \eqref{eq_bare_covariance_decay} we
 perform the recursive estimation from the terminal 2 to the root 1.
 The important point is that the extra $\tilde{\cG}(\cdot)$ produced by
 $R''$ is added to the integration on the vertex 3. We uniformly bound
 all the $\tilde{\cG}(\cdot)$s produced by
 $R^{n-2}R'$, not integrating them.
As the result,
\begin{align}
|a_{n,L,h}^{(a,1,1,d)}(T)|\label{eq_truncation_bound_1_1}
&\le L^{-d(n-1)}(2(n-1))!c^n c(d,\beta,\theta)^{n-1}\\
&\quad\cdot \frac{1}{h}\sum_{s\in[0,\beta)_h}\sup_{X\in I,\eta\in \{1,2\}\atop \zeta\in \{1,-1\}}\Bigg(
\frac{1}{h}\sum_{\by\in\G\atop t\in [0,\beta)_h}|\tilde{\cG}(X,\eta\by t\zeta)||v(s,t)|\Bigg)\notag\\
&\quad\cdot \Bigg(\sup_{X\in I,\eta\in \{1,2\}\atop \zeta\in \{1,-1\}}
\Bigg(\frac{1}{h^2}\sum_{\bx,\by\in\G\atop s,t\in
 [0,\beta)_h}|\tilde{\cG}(X,\eta\by
 t\zeta)||v(s,t)|\Bigg)\Bigg)^{n-2}\notag\\
&\quad\cdot 
\Bigg(\sup_{X,Z\in I,\rho,\eta\in \{1,2\}\atop \xi,\zeta\in \{1,-1\}}
\Bigg(\frac{1}{h^2}\sum_{\bx,\by\in\G\atop s,t\in
 [0,\beta)_h}|\tilde{\cG}(X,\rho\bx s\xi)||\tilde{\cG}(Z,\eta\by
 t\zeta)||v(s,t)|\Bigg)\Bigg)\notag\\
&\le (2n)!c^n  c(d,\beta,\theta)^{2n} \beta^n
 L^{-d}\left(\sum_{\bx\in
 \G}\frac{1}{1+\sum_{j=1}^d|\frac{L}{2\pi}(e^{i\frac{2\pi}{L}\<\bx,\be_j\>}-1)|^{d+1}}\right)^{n+1}\notag\\
&\le (2n)!c^n  c(d,\beta,\theta)^{2n} \beta^n
 L^{-d}.\notag
\end{align}
The same argument as above shows that
$$
|a_{n,L,h}^{(a,-1,-1,d)}(T)|\le (2n)!c^n  c(d,\beta,\theta)^{2n} \beta^n
 L^{-d}.
$$

Next let us consider $a_{n,L,h}^{(a,1,-1,d)}(T)$. In this case one
 Grassmann variable, which originally belongs to $V_{\bx s
 1}(\psi^{3,1})$, remains after applying the operators along the tree
 lines. This Grassmann variable must be erased by $R^n$. Thus, inside
 $a_{n,L,h}^{(a,1,-1,d)}(T)$, the operator $R^n$ can be decomposed as follows.
\begin{align*}
&2n R^{n-1}M(T,\s,\bs)_{2,3}B_{(2,a),(3,1)} + 2n R^{n-1}B_{(3,-1),(3,1)}\\
&+ 2n R^{n-1}\sum_{p=4}^{n+1}\sum_{\delta \in
\{1,-1\}}M(T,\s,\bs)_{p,3}B_{(p,\delta),(3,1)}.
\end{align*}
Let $\tilde{a}_{n,L,h}^{(a,1,-1,d)}(T)$ denote the following.
\begin{align*}
&\frah^2\sum_{s_1,t_2\in[0,\beta)_h}\nu(s_1,t_2)\prod_{j=3}^{n+1}\left(\frah^{2}\sum_{s_j,t_j\in [0,\beta)_h}\nu(s_j,t_j)\right)
L^{-d(n-1)}E\\
&\quad\cdot 2n R^{n-1}M(T,\s,\bs)_{2,3}B_{(2,a),(3,1)}\\
&\quad\cdot 
\prod_{\{p,q\}\in T\backslash \{\{2,3\},\{3,4\}\}}\left(
\sum_{f,g\in \{1,-1\}}B_{(p,f),(q,g)}\right)
B_{(2,a),(3,1)}B_{(3,-1),(4,d)}\\
&\quad\cdot \sum_{\bx\in\G}V_{\b0 s_1 1}(\psi^{1,1})V_{\bx t_2 -1}(\psi^{2,-1})\\
&\quad\cdot \prod_{j=3}^{n+1}\left(
\sum_{\bx,\by\in\G}V_{\bx s_j 1}(\psi^{j,1})V_{\by t_j -1}(\psi^{j,-1})
\right)\Bigg|_{\psi^{j,\delta}=0\atop (\forall j\in \{1,2,\cdots,n+1\},
 \delta\in \{1,-1\})}.
\end{align*}
In fact $\tilde{a}_{n,L,h}^{(a,1,-1,d)}(T)$ is derived by
 replacing $R^n$ by $2n
 R^{n-1}M(T,\s,\bs)_{2,3}B_{(2,a),(3,1)}$ inside
 $a_{n,L,h}^{(a,1,-1,d)}(T)$. Since the application of
 $B_{(3,-1),(3,1)}$, $B_{(p,\delta),(3,1)}$ $(p\in \{4,5,\cdots,n+1\},\
 \delta\in \{1,-1\})$ gives an additional free propagator at the vertex
 3, the same calculation as that leading to \eqref{eq_truncation_bound_1_1}
 yields that 
\begin{align*}
|a_{n,L,h}^{(a,1,-1,d)}(T)-\tilde{a}_{n,L,h}^{(a,1,-1,d)}(T)|\le (2n)!c^n
 c(d,\beta,\theta)^{2n}\beta^n L^{-d}.
\end{align*}
The term $a_{n,L,h}^{(a,-1,1,d)}(T)$ can be analyzed in the same way as
 above. The result is that
\begin{align*}
|a_{n,L,h}^{(a,-1,1,d)}(T)-\tilde{a}_{n,L,h}^{(a,-1,1,d)}(T)|\le (2n)!c^n
 c(d,\beta,\theta)^{2n}\beta^n L^{-d},
\end{align*}
where
\begin{align*}
\tilde{a}_{n,L,h}^{(a,-1,1,d)}(T)
&:=\frah^2\sum_{s_1,t_2\in[0,\beta)_h}\nu(s_1,t_2)\prod_{j=3}^{n+1}\left(\frah^{2}\sum_{s_j,t_j\in [0,\beta)_h}\nu(s_j,t_j)\right)
L^{-d(n-1)}E\\
&\quad\cdot 2n R^{n-1}M(T,\s,\bs)_{2,3}B_{(2,a),(3,-1)}\\
&\quad\cdot 
\prod_{\{p,q\}\in T\backslash \{\{2,3\},\{3,4\}\}}\left(
\sum_{f,g\in \{1,-1\}}B_{(p,f),(q,g)}\right)
B_{(2,a),(3,-1)}B_{(3,1),(4,d)}\\
&\quad\cdot \sum_{\bx\in\G}V_{\b0 s_1 1}(\psi^{1,1})V_{\bx t_2 -1}(\psi^{2,-1})\\
&\quad\cdot \prod_{j=3}^{n+1}\left(
\sum_{\bx,\by\in\G}V_{\bx s_j 1}(\psi^{j,1})V_{\by t_j -1}(\psi^{j,-1})
\right)\Bigg|_{\psi^{j,\delta}=0\atop (\forall j\in \{1,2,\cdots,n+1\},
 \delta\in \{1,-1\})}.
\end{align*}

By combining these results we conclude that
\begin{align*}
|a_{n,L,h}(T)-\tilde{a}_{n,L,h}(T)|\le (2n)!c^n
 c(d,\beta,\theta)^{2n}\beta^n L^{-d},
\end{align*}
where
$$
\tilde{a}_{n,L,h}(T):=\sum_{a,d\in\{1,-1\}}(\tilde{a}_{n,L,h}^{(a,1,-1,d)}(T)+\tilde{a}_{n,L,h}^{(a,-1,1,d)}(T)).
$$
We can reform $\tilde{a}_{n,L,h}(T)$ as follows. 
\begin{align*}
\tilde{a}_{n,L,h}(T)
&=\frah^2\sum_{s_1,t_2\in[0,\beta)_h}\nu(s_1,t_2)\prod_{j=3}^{n+1}\left(\frah^{2}\sum_{s_j,t_j\in
 [0,\beta)_h}\nu(s_j,t_j)\right)\\
&\quad\cdot \sum_{b\in\{1,-1\}}\sum_{\bx\in
 \G}B_{(2,-1),(3,b)}B_{(2,-1),(3,b)}\\
&\qquad\cdot V_{\bx t_2
 -1}(\psi^{2,-1})(1_{b=1}V_{\b0 s_3 1}(\psi^{3,1})+ 1_{b=-1}V_{\b0 t_3
 -1}(\psi^{3,-1}))\\
&\quad\cdot 2n L^{-d(n-2)}E M(T,\s,\bs)_{2,3}R^{n-1}
\prod_{\{p,q\}\in T\backslash \{\{2,3\}\}}\left(
\sum_{f,g\in \{1,-1\}}B_{(p,f),(q,g)}\right)\\
&\quad\cdot \sum_{\bz\in\G}V_{\b0 s_1 1}(\psi^{1,1})
(1_{b=1}V_{\bz t_3 -1}(\psi^{3,-1})+ 1_{b=-1}V_{\bz s_3 1}(\psi^{3,1}))\\
&\quad\cdot \prod_{j=4}^{n+1}\left(
\sum_{\bx,\by\in\G}V_{\bx s_j 1}(\psi^{j,1})V_{\by t_j -1}(\psi^{j,-1})
\right)\Bigg|_{\psi^{j,\delta}=0\atop (\forall j\in \{1,2,\cdots,n+1\},
 \delta\in \{1,-1\})}.
\end{align*}
Let us observe that there is a recursive structure here. 
We can repeat the same procedure as above on the tree
 $T\backslash\{\{2,3\}\}$. The result is that
\begin{align*}
\left|a_{n,L,h}(T)-\tilde{\tilde{a}}_{n,L,h}(T)\right|\le (2n)!c^n
 c(d,\beta,\theta)^{2n}\beta^n L^{-d},
\end{align*}
where
\begin{align*}
&\tilde{\tilde{a}}_{n,L,h}(T)\\
&:=\frah^2\sum_{s_1,t_2\in[0,\beta)_h}\nu(s_1,t_2)\prod_{j=3}^{n+1}\left(\frah^{2}\sum_{s_j,t_j\in
 [0,\beta)_h}\nu(s_j,t_j)\right)\\
&\quad\cdot \sum_{b\in\{1,-1\}}\sum_{\bx\in
 \G}B_{(2,-1),(3,b)}B_{(2,-1),(3,b)}\\
&\qquad\cdot V_{\bx t_2
 -1}(\psi^{2,-1})(1_{b=1}V_{\b0 s_3 1}(\psi^{3,1})+ 1_{b=-1}V_{\b0 t_3
 -1}(\psi^{3,-1}))\\
&\quad\cdot \sum_{c\in\{1,-1\}}\sum_{\bz\in
 \G}B_{(3,-b),(4,c)}B_{(3,-b),(4,c)}\\
&\qquad\cdot (1_{b=1}V_{\bz t_3 -1}(\psi^{3,-1})+ 1_{b=-1}V_{\bz s_3
 1}(\psi^{3,1}))
(1_{c=1}V_{\b0 s_4 1}(\psi^{4,1})+ 1_{c=-1}V_{\b0 t_4
 -1}(\psi^{4,-1}))\\
&\quad\cdot
2^2n(n-1) L^{-d(n-3)}E M(T,\s,\bs)_{2,3}M(T,\s,\bs)_{3,4}
R^{n-2}\\
&\quad\cdot \prod_{\{p,q\}\in T\backslash \{\{2,3\},\{3,4\}\}}\left(
\sum_{f,g\in \{1,-1\}}B_{(p,f),(q,g)}\right)\\
&\quad\cdot \sum_{\bw\in\G}V_{\b0 s_1 1}(\psi^{1,1})
(1_{c=1}V_{\bw t_4 -1}(\psi^{4,-1})+ 1_{c=-1}V_{\bw s_4 1}(\psi^{4,1}))\\
&\quad\cdot \prod_{j=5}^{n+1}\left(
\sum_{\bx,\by\in\G}V_{\bx s_j 1}(\psi^{j,1})V_{\by t_j -1}(\psi^{j,-1})
\right)\Bigg|_{\psi^{j,\delta}=0\atop (\forall j\in \{1,2,\cdots,n+1\},
 \delta\in \{1,-1\})}.
\end{align*}
By repeating this procedure we eventually have
\begin{align}
\left|a_{n,L,h}(T)-b_{n,L,h}(T)\right|\le (2n)!c^n
 c(d,\beta,\theta)^{2n}\beta^n L^{-d},\label{eq_2_field_subtraction_bound}
\end{align}
where
\begin{align*}
&b_{n,L,h}(T):=\prod_{j=2}^{n+1}\Bigg(\sum_{\bx_j\in
 \G}\Bigg)g_{L,h}(\phi)(\bx_2,\bx_3,\cdots,\bx_{n+1}),\\
&g_{L,h}(\phi)(\bx_2,\bx_3,\cdots,\bx_{n+1})\\
&:=\frah^2\sum_{s_1,s_2^{-1}\in[0,\beta)_h}\nu(s_1,s_2^{-1})\prod_{j=3}^{n+1}\Bigg(\frah^{2}\sum_{s_j^1,s_j^{-1}\in
 [0,\beta)_h}\nu(s_j^1,s_j^{-1})\Bigg)\\
&\quad\cdot \sum_{b_3\in\{1,-1\}}B_{(2,-1),(3,b_3)}^2 V_{\bx_2 s_2^{-1}
 -1}(\psi^{2,-1}) V_{\b0 s_3^{b_3}
 b_3}(\psi^{3,b_3})\\
&\quad\cdot \sum_{b_4\in\{1,-1\}}B_{(3,-b_3),(4,b_4)}^2 V_{\bx_3 s_3^{-b_3}
 -b_3}(\psi^{3,-b_3}) V_{\b0 s_4^{b_4}
 b_4}(\psi^{4,b_4})\\
&\quad\vdots\\
&\quad \cdot B_{(n+1,-b_{n+1}),(1,1)}^2 V_{\bx_{n+1} s_{n+1}^{-b_{n+1}}
 -b_{n+1}}(\psi^{n+1,-b_{n+1}}) V_{\b0 s_1 1}(\psi^{1,1})\\
&\quad\cdot
2^n n! E \prod_{\{p,q\}\in T}M(T,\s,\bs)_{p,q}.
\end{align*}
Since $g_{L,h}(\phi)(\bX)$ is a finite sum of products of $\tilde{\cG}$, 
we can naturally define $g_{L,h}$ as a map from $(\Z^d)^n$ to $C(Q,\C)$.
By the same argument as the proof of the convergence $\lim_{h\to
 \infty,h\in \frac{2}{\beta}\N}a_{n,L,h}$ in $C(Q,\C)$ we can prove that 
 $\lim_{h\to
 \infty,h\in \frac{2}{\beta}\N}g_{L,h}(\cdot)(\bX)$ converges in
 $C(Q,\C)$ for any $\bX\in (\Z^d)^n$ and so does $\lim_{h\to
 \infty,h\in \frac{2}{\beta}\N}b_{n,L,h}(T)$. 
We can also deduce from the definition of $\tilde{\cG}$ and $g_{L,h}$
 that $\lim_{L\to
 \infty,L\in \N}\lim_{h\to
 \infty,h\in \frac{2}{\beta}\N}g_{L,h}(\cdot)(\bX)$ converges in
 $C(Q,\C)$ for any $\bX\in (\Z^d)^n$. It follows from
 \eqref{eq_phi_good_property}, \eqref{eq_exponent_matrix_bound},
 \eqref{eq_bare_covariance_decay_L_independent} that
\begin{align*}
&\sup_{\phi\in Q}|g_{L,h}(\phi)(\bx_2,\bx_3,\cdots,\bx_{n+1})|1_{\bx_j\in
 \G_L\ (j=2,3,\cdots,n+1)}\\
&\le n! c^n
 c(d,\beta,\theta)^{2n}\beta^n\prod_{l=2}^{n+1}\frac{1}{1+(\frac{2}{\pi})^{d+1}\sum_{j=1}^d|\<\bx_l,\be_j\>|^{d+1}},\\
&(\forall \bx_j\in \Z^d\ (j=2,3,\cdots,n+1)).
\end{align*}
The right-hand side of the above inequality is summable over
 $(\Z^d)^n$. Since
\begin{align*}
b_{n,L,h}(T)=\prod_{j=2}^{n+1}\Bigg(\sum_{\bx_j\in\Z^d}\Bigg)
1_{\bx_j\in\G_L
 (j=2,3,\cdots,n+1)}g_{L,h}(\phi)(\bx_2,\bx_3,\cdots,\bx_{n+1}),
\end{align*}
we can apply the dominated convergence theorem in
 $L^1((\Z^d)^n,C(Q,\C))$ to conclude that
$\lim_{L\to\infty,L\in \N}\lim_{h\to \infty, h\in
 \frac{2}{\beta}\N}b_{n,L,h}(T)$ converges in $C(Q,\C)$. Observe that
\begin{align*}
a_{n,L,h}
&=\frac{2^n}{(n!)^2}\sum_{T\in
 \T(\{1,2,\cdots,n+1\})}1_{\eqref{eq_vertex_logical}}
 a_{n,L,h}(T)\\
&\quad + \frac{2^n}{(n!)^2}\sum_{T\in
 \T(\{1,2,\cdots,n+1\})}1_{\lnot\eqref{eq_vertex_logical}}(a_{n,L,h}(T)-b_{n,L,h}(T))\\
&\quad + \frac{2^n}{(n!)^2}\sum_{T\in
 \T(\{1,2,\cdots,n+1\})}1_{\lnot\eqref{eq_vertex_logical}}b_{n,L,h}(T).
\end{align*}
By \eqref{eq_4_field_terminal_bound},
 \eqref{eq_2_field_subtraction_bound} and the fact that $\sharp
 \T(\{1,2,\cdots,n+1\})\le c^n n!$,
\begin{align*}
&\sup_{\phi\in Q}\left|
a_{n,L,h}(\phi)- \frac{2^n}{(n!)^2}\sum_{T\in
 \T(\{1,2,\cdots,n+1\})}1_{\lnot\eqref{eq_vertex_logical}}b_{n,L,h}(T)(\phi)\right|\\
&\le \frac{(2n)!}{n!}c^n
 c(d,\beta,\theta)^{2n}\beta^n L^{-d}.
\end{align*}
Since we have checked that
$\lim_{h\to \infty, h\in \frac{2}{\beta}\N}a_{n,L,h}$, 
$\lim_{h\to \infty, h\in \frac{2}{\beta}\N}b_{n,L,h}(T)$,\\
$\lim_{L\to \infty, L\in \N}\lim_{h\to \infty, h\in
 \frac{2}{\beta}\N}b_{n,L,h}(T)$
 converge in $C(Q,\C)$, we can deduce from this inequality that  
$\lim_{L\to \infty, L\in \N}\lim_{h\to \infty, h\in
 \frac{2}{\beta}\N}a_{n,L,h}$
 converges in $C(Q,\C)$.

Let us confirm the convergence of $\alpha'_{n,L,h}$. By definition,
$$\alpha'_{1,L,h}=-\beta \cG(\phi)(1\b0 0,1\b0 0),$$ 
which converges in
 $C(Q,\C)$ as $h\to \infty$, $L\to \infty$. Assume that $n\ge 2$. Let us
 estimate $|\alpha'_{n,L,h}|$ by using the general lemmas Lemma
 \ref{lem_tree_bound}, Lemma \ref{lem_tree_double_bound}, which is a
 simpler way than decomposing the operator $Tree(\{1,\cdots,n\},\cG)$ as
 above. By \eqref{eq_covariance_decomposition},
 $\cG(\phi)(\bX)=\sum_{l=0}^{N_h-N_{\beta}+1}C_l(\bX)$, ($\forall \bX\in
 I_0^2$). Thus, by
 \eqref{eq_full_covariance_determinant_bound_pre}
\begin{align*}
&|\det(\<\bu_i,\bv_j\>_{\C^m}\cG(\phi)(X_i,Y_j))_{1\le i,j\le n}|\le
 (c(d)(1+\beta^{-1}g_d(\Theta)))^n\le c'(d,\beta,\theta)^n,\\
&(\forall m,n\in \N,\bu_i,\bv_i\in\C^m\text{ with
 }\|\bu_i\|_{\C^m},\|\bv_i\|_{\C^m}\le 1,X_i,Y_i\in I_0\
 (i=1,2,\cdots,n)),
\end{align*}
where $c'(d,\beta,\theta)(\in\R_{\ge 1})$ is a positive constant
 depending only on $d,\beta,\theta$. Moreover, by Lemma
 \ref{lem_full_covariance_decay}
$$
\|\tilde{\cG}(\phi)\|_{1,\infty},\|\tilde{\cG}(\phi)\|\le \sum_{\bx\in\G}\frac{c'(d,\beta,\theta)}{1+\sum_{j=1}^d|\frac{L}{2\pi}(e^{i\frac{2\pi}{L}\<\bx,\be_j\>}-1)|^{d+1}}.
$$
Also, by definition
\begin{align*}
&\|V_2\|_{1,\infty}\le L^{-d},\quad \|V_{2,2}\|_{1,\infty}\le 1,\\ 
&[V_{2,2},\tilde{\cG}(\phi)]_{1,\infty}\le
 L^{-d}\|\tilde{\cG}(\phi)\|\le L^{-d}
 \sum_{\bx\in\G}\frac{c'(d,\beta,\theta)}{1+\sum_{j=1}^d|\frac{L}{2\pi}(e^{i\frac{2\pi}{L}\<\bx,\be_j\>}-1)|^{d+1}}.
\end{align*}
With these inequalities and the fact that the
 $\|\cdot\|_{1,\infty}$-norm of the anti-symmetric kernel of $\hat{V}_4$
 is bounded by $\|V_{2,2}\|_{1,\infty}$ we can apply
 \eqref{eq_tree_1_infinity}, \eqref{eq_double_1_infinity} to derive that
\begin{align*}
&|\alpha'_{n,L,h}|\\
&\le
 \frac{N}{h}\left(\sum_{\bx\in\G}\frac{1}{1+\sum_{j=1}^d|\frac{L}{2\pi}(e^{i\frac{2\pi}{L}\<\bx,\be_j\>}-1)|^{d+1}}\right)^{n-1}(2^6c'(d,\beta,\theta)L^{-d})^n\\
&\quad +\sum_{l=1}^{n-1}\left(\begin{array}{c} n \\ l\end{array}\right)
\frac{N}{h}2^{12} c'(d,\beta,\theta)^2
 \left(\sum_{\bx\in\G}\frac{1}{1+\sum_{j=1}^d|\frac{L}{2\pi}(e^{i\frac{2\pi}{L}\<\bx,\be_j\>}-1)|^{d+1}}\right)^{n}\\
&\qquad\cdot L^{-d}
(2^{12}c'(d,\beta,\theta)^2)^{l-1}(2^{6}c'(d,\beta,\theta)L^{-d})^{n-l}\\
&\le \frac{N}{h}L^{-2d}c^nc'(d,\beta,\theta)^{2n}\sum_{a=0}^1
\left(\sum_{\bx\in\G}\frac{1}{1+\sum_{j=1}^d|\frac{L}{2\pi}(e^{i\frac{2\pi}{L}\<\bx,\be_j\>}-1)|^{d+1}}\right)^{n-a},
\end{align*}
which implies that $\lim_{L\to \infty, L\in \N}\lim_{h\to \infty, h\in
 \frac{2}{\beta}\N}\alpha'_{n,L,h}=0$ in $C(Q,\C)$. Thus, we have seen
 that $\alpha_{n,L,h}$ converges in $C(Q,\C)$ as $h\to \infty(h\in
 \frac{2}{\beta}\N)$, $L\to \infty(L\in \N)$ for any $n\in \N$. 

Let us complete the proof of the proposition. The inequality
 \eqref{eq_model_good_smallness} implies that 
$$
u\mapsto \log\left(\int e^{-V(u)(\psi)+W(u)(\psi)}d\mu_{C(\phi)}(\psi)\right)
$$
is analytic in $D(r)$ for any $\phi\in\C$ and 
$$
\sup_{(\phi,u)\in\C\times\overline{D(r)}}\left|\log\left(\int
 e^{-V(u)(\psi)+W(u)(\psi)}d\mu_{C(\phi)}(\psi)\right)\right|\le 1.
$$ 
Thus,
\begin{align*}
&\log\left(\int
 e^{-V(u)(\psi)+W(u)(\psi)}d\mu_{C(\phi)}(\psi)\right)=\sum_{n=1}^{\infty}\alpha_{n,L,h}(\phi)u^n,\quad
 (\forall (\phi,u)\in Q\times \overline{D(r/2)}),\\
&\sup_{(\phi,u)\in Q\times
 \overline{D(\frac{r}{2})}}|\alpha_{n,L,h}(\phi)u^n|\\
&\le \sup_{\phi\in Q}\left|
\frac{1}{2\pi
 i}\oint_{|z|=(1+\eps)\frac{r}{2}}dz\frac{1}{z^{n+1}}\log\left(\int
 e^{-V(z)(\psi)+W(z)(\psi)}d\mu_{C(\phi)}(\psi)\right)\right|\left(\frac{r}{2}\right)^n\\
&\le
 \frac{1}{(1+\eps)^n},\quad (\forall \eps \in (0,1)).
\end{align*}
Therefore we can use the dominated convergence theorem in
 $l^1(\N,C(Q\times \overline{D(r/2)}))$ to ensure that 
\begin{align*}
&\lim_{h\to \infty\atop h\in
 \frac{2}{\beta}\N}
\log\left(\int
 e^{-V(u)(\psi)+W(u)(\psi)}d\mu_{C(\phi)}(\psi)\right),\\
&\lim_{L\to \infty\atop L\in \N}
\lim_{h\to \infty\atop h\in
 \frac{2}{\beta}\N}
\log\left(\int
 e^{-V(u)(\psi)+W(u)(\psi)}d\mu_{C(\phi)}(\psi)\right)
\end{align*}
converge in $C(Q\times  \overline{D(r/2)})$.
Since 
\begin{align*}
&\int e^{-V(u)(\psi)+W(u)(\psi)}d\mu_{C(\phi)}(\psi)=e^{\log\left(\int
 e^{-V(u)(\psi)+W(u)(\psi)}d\mu_{C(\phi)}(\psi)\right)},\\
&(\forall
 (\phi,u)\in Q\times \overline{D(r/2)}),
\end{align*}
the claims of the proposition follow.
\end{proof}

\subsection{Completion of the proof of the main
  theorem}\label{subsec_proof_theorem}

In this subsection we will complete the proof of Theorem
\ref{thm_main_theorem}. The main necessary tools have already been prepared. 
It remains to study the solvability of the gap equation which is
different from the conventional BCS gap equation due to the presence of
the imaginary magnetic field. Let us start by showing an inequality
which will be used to give a sufficient condition for the solvability of our gap equation.

\begin{lemma}\label{lem_estimation_free_fermi_surface}
Set $K:=\frac{1}{2}(2d-|\mu|)$. Then, 
\begin{align*}
\inf_{\eta\in [-K,K]}\cH^{d-1}(\{\bk\in [0,2\pi]^d\ |\
 e(\bk) =\eta\})\ge 1_{d=1}+1_{d\ge
 2}\left(\frac{2d-|\mu|}{10(d-1)d}\right)^{d-1}.
\end{align*}
\end{lemma}
\begin{proof}
Since $|\mu+\eta|<2d$, 
$$
\{\bk\in [0,2\pi]^d\ |\
 e(\bk) =\eta \}\neq \emptyset,\quad (\forall \eta\in [-K,K]).
$$
This implies the lower bound for $d=1$. Let us assume that $d\ge 2$.
Note that
\begin{align*}
&\cH^{d-1}(\{\bk\in [0,2\pi]^d\ |\ e(\bk) =\eta \})\\
&\ge \cH^{d-1}\left(\left\{\bk\in \left[0,\frac{\pi}{2}\right]^{d-1}\times[0,\pi]\ \Big|\
 \sum_{j=1}^d\cos k_j=\frac{1}{2}|\eta +\mu|\right\}\right).
\end{align*}
In the following we assume that $\bk\in
 [0,\frac{\pi}{2}]^{d-1}\times[0,\pi]$. Set
$$\eps :=\frac{\min\{1,2d-|\mu|\}}{5(d-1)}.$$

Assume that $\frac{1}{2}|\eta+\mu|\in [l,l+\frac{1}{2})$ for some $l\in
 \{0,1,\cdots,d-1\}$. If 
\begin{align}
k_j\in [0,\eps]\quad (\forall j\in
 \{1,2,\cdots,l\}),\quad k_j\in
 \left[\frac{\pi}{2}-\eps,\frac{\pi}{2}\right]\quad (\forall j\in
 \{l+1,\cdots,d-1\}),\label{eq_momentum_restriction_1}
\end{align}
then
$$
\sum_{j=1}^{d-1}\cos k_j\in [l(1-\eps),l+(d-1-l)\eps]\subset
 \left[l-\frac{1}{5}, l+\frac{1}{5}\right].
$$ 
Thus,
$$
\frac{1}{2}|\eta+\mu|-\sum_{j=1}^{d-1}\cos k_j\in \left[-\frac{1}{5},\frac{7}{10}\right].
$$
Recall that we defined $\arccos$ as a map from $(-1,1)$ to $(0,\pi)$ in
 the proof of Lemma \ref{lem_covariance_crucial_inequality}. 
Then, we see that if \eqref{eq_momentum_restriction_1}
 holds and $k_d\in [0,\pi]$, the equality $\sum_{j=1}^{d}\cos
 k_j=\frac{1}{2}|\eta+\mu|$ is equivalent to 
$$
k_d=\arccos\left(\frac{1}{2}|\eta+\mu|-\sum_{j=1}^{d-1}\cos k_j\right).
$$
Thus, 
\begin{align*}
&\cH^{d-1}\left(\left\{\bk\in \left[0,\frac{\pi}{2}\right]^{d-1}\times[0,\pi]\ \Big|\
 \sum_{j=1}^d\cos k_j=\frac{1}{2}|\eta +\mu|\right\}\right)\\
&\ge \prod_{i=1}^l\left(\int_0^{\eps}dk_i\right)
\prod_{j=l+1}^{d-1}\left(\int_{\frac{\pi}{2}-\eps}^{\frac{\pi}{2}}dk_j\right)
\Bigg(1+\sum_{j=1}^{d-1}\frac{\sin^2k_j}{1-(\frac{1}{2}|\eta+\mu|-\sum_{m=1}^{d-1}\cos
 k_m)^2}\Bigg)^{\frac{1}{2}}\\
&\ge \eps^{d-1}.
\end{align*}

Assume that $\frac{1}{2}|\eta+\mu|\in [l+\frac{1}{2},l+1)$ for some
 $l\in \{0,1,\cdots,d-2\}$. If 
\begin{align}
k_j\in [0,\eps]\quad (\forall j\in
 \{1,2,\cdots,l+1\}),\quad k_j\in
 \left[\frac{\pi}{2}-\eps,\frac{\pi}{2}\right]\quad (\forall j\in
 \{l+2,\cdots,d-1\}),\label{eq_momentum_restriction_2}
\end{align}
then
$$
\sum_{j=1}^{d-1}\cos k_j\in [(l+1)(1-\eps),l+1+(d-l-2)\eps]\subset
 \left[l+\frac{4}{5}, l+\frac{6}{5}\right].
$$ 
Thus,
$$
\frac{1}{2}|\eta+\mu|-\sum_{j=1}^{d-1}\cos k_j\in \left[-\frac{7}{10},\frac{1}{5}\right].
$$
Therefore, if \eqref{eq_momentum_restriction_2} and $k_d\in
 [0,\pi]$ hold, the equality $\sum_{j=1}^{d}\cos
 k_j=\frac{1}{2}|\eta+\mu|$ is equivalently written as 
$$
k_d=\arccos\left(\frac{1}{2}|\eta+\mu|-\sum_{j=1}^{d-1}\cos k_j\right).
$$
Thus, by the same calculation as above we have that
\begin{align}
\cH^{d-1}\left(\left\{\bk\in \left[0,\frac{\pi}{2}\right]^{d-1}\times[0,\pi]\ \Big|\
 \sum_{j=1}^d\cos k_j=\frac{1}{2}|\eta +\mu|\right\}\right)\ge\eps^{d-1}.\label{eq_target_fermi_estimate}
\end{align}

Assume that  $\frac{1}{2}|\eta+\mu|\in [d-\frac{1}{2},d)$. 
If $k_j\in [0,\eps]$ $(\forall j\in
 \{1,2,\cdots,d-1\})$, 
then
$$
\sum_{j=1}^{d-1}\cos k_j\in [(d-1)(1-\eps),d-1]\subset
 \left[d-1-\frac{1}{5}(2d-|\mu|), d-1\right].
$$
By assumption, $\frac{1}{2}|\eta+\mu|\le
 \frac{1}{2}|\mu|+\frac{1}{4}(2d-|\mu|)$. 
Thus,
$$
\frac{1}{2}|\eta+\mu|-\sum_{j=1}^{d-1}\cos k_j\in
 \left[\frac{1}{2},1-\frac{1}{20}(2d-|\mu|)\right]\subset
 \left[\frac{1}{2},1\right).$$
Therefore, the same argument as above yields the estimate
\eqref{eq_target_fermi_estimate}. Since $\min\{1,2d-|\mu|\}\ge
 (2d-|\mu|)/(2d)$, the claimed inequality has been derived.
\end{proof}

Using Lemma \ref{lem_estimation_free_fermi_surface},
let us give a sufficient condition for the solvability and the
non-solvability of the gap
equation \eqref{eq_gap_equation}.

\begin{lemma}\label{lem_sufficient_condition_gap_equation}
The following statements hold true.
\begin{enumerate}[(i)]
\item\label{item_sufficient_solvability}
There exists a positive constant $c(d)$ depending only on $d$ such that
if 
\begin{align*}
|U|>c(d)(2d-|\mu|)^{1-d}\beta\Theta\left(1_{\Theta\le \frac{1}{2}(2d-|\mu|)}
+
1_{\Theta> \frac{1}{2}(2d-|\mu|)}(2d-|\mu|)^{-1}\Theta\right),
\end{align*}
\begin{align*}
-\frac{2}{|U|}+\frac{1}{(2\pi)^d}\int_{[0,2\pi]^d}d\bk\frac{\sinh(\beta
 |e(\bk)|)}{(\cos(\beta\theta/2)+\cosh(\beta e(\bk)))|e(\bk)|}>0.
\end{align*}
\item\label{item_sufficient_nonsolvability}
Assume that $\theta\in [0,\frac{\pi}{\beta}]$. If $|U|<2\beta^{-1}$, 
\begin{align*}
-\frac{2}{|U|}+\frac{1}{(2\pi)^d}\int_{[0,2\pi]^d}d\bk\frac{\sinh(\beta
 |e(\bk)|)}{(\cos(\beta\theta/2)+\cosh(\beta e(\bk)))|e(\bk)|}<0.
\end{align*}
\end{enumerate}
\end{lemma}
\begin{proof} Set $K:=\frac{1}{2}(2d-|\mu|)$. 

\eqref{item_sufficient_solvability}: Let us define the function $f:\R\to
 \R$ by 
$$
f(x):=\frac{\sinh(\beta|x|)}{(\cos(\beta\theta/2)+\cosh(\beta x))|x|}.
$$ 
We can see from the definition that $f\in C^{\infty}(\R)$. Moreover, by
 the coarea formula and Lemma \ref{lem_estimation_free_fermi_surface}, 
\begin{align*}
\int_{[0,2\pi]^d}d\bk f(e(\bk))&\ge
 \frac{1}{2\sqrt{d}}\int_{[0,2\pi]^d}d\bk f(e(\bk))\|\nabla e(\bk)\|_{\R^d}\\
&=\frac{1}{2\sqrt{d}}\int_{-\infty}^{\infty}d\eta
 f(\eta)\cH^{d-1}(\{\bk\in [0,2\pi]^d\ |\ e(\bk)=\eta\})\\
&\ge \frac{1}{2\sqrt{d}}\left(1_{d=1}+1_{d\ge 2}\left(
\frac{2d-|\mu|}{10(d-1)d}\right)^{d-1}\right)\int_{-K}^Kd\eta f(\eta).
\end{align*}
Note that 
\begin{align*}
\int_{-K}^Kd\eta f(\eta)&\ge \int_{-K}^Kd\eta \frac{\beta}{\cosh(\beta
 \eta)-1+2\sin^2(\beta \Theta/2)}\\
&\ge c\beta^{-1}\int_0^{\min\{\Theta,K\}}d\eta \frac{1}{\eta^2+\Theta^2}
=c\beta^{-1}\Theta^{-1}\arctan(\min\{1,K\Theta^{-1}\})\\
&\ge c\beta^{-1}\Theta^{-1}(1_{\Theta \le K}+1_{\Theta >
 K}K\Theta^{-1}).
\end{align*}
By combining this inequality with the above inequality we obtain 
$$
\int_{[0,2\pi]^d}d\bk f(e(\bk))\ge c(d)(2d-|\mu|)^{d-1}\beta^{-1}
\Theta^{-1}(1_{\Theta \le K}+1_{\Theta >
 K}K\Theta^{-1}),
$$
which implies the claim \eqref{item_sufficient_solvability}.

\eqref{item_sufficient_nonsolvability}: By the assumption on $\theta$, 
$f(x)\le \tanh(\beta |x|)/|x|\le \beta$ for any $x\in \R$. The claim
 follows from this inequality.
\end{proof}

Before giving the proof of the main theorem, let us confirm a few more
simple facts.

\begin{lemma}\label{lem_monotoniticity}
Let $\eps\in (-1,1]$. The function 
$$x\mapsto \frac{\sinh(x)}{x(\eps +\cosh(x))}:[0,\infty)\to \R$$
is strictly monotone decreasing and converges to $0$ as $x\to \infty$.
\end{lemma}
\begin{proof}
Observe that
$$
\frac{\sinh(x)}{x(\eps+\cosh(x))}=\frac{\tanh(x/2)}{x}\cdot
 \frac{1+\cosh(x)}{\eps +\cosh(x)}.
$$
One can check that the derivative of $\tanh(x/2)/x$, $(1+\cosh(x))/(\eps
 +\cosh(x))$ $(\eps\in (-1,1))$ are negative in $(0,\infty)$, which
 implies the strict  monotone decreasing property of the function.
The convergence property is clear.
\end{proof}

\begin{lemma}\label{lem_explicit_calculation_covariance}
For any $\bx,\by\in \G$, $\phi\in\C$, $\rho,\eta \in \{1,2\}$, $\rho\neq
 \eta$, 
\begin{align*}
&C(\phi)(\rho\bx 0,\rho \by 0)\\
&=\frac{1}{L^d}\sum_{\bk\in\G^*}e^{i\<\bx-\by,\bk\>}\Bigg(
\frac{e^{-i\frac{\beta\theta}{2}}+
 \cosh(\beta \sqrt{e(\bk)^2+|\phi|^2})}
{2(\cos(\beta\theta/2)
+\cosh(\beta\sqrt{e(\bk)^2+|\phi|^2}))}\\
&\qquad\qquad\qquad +\frac{(-1)^{\rho}
\sinh(\beta
 \sqrt{e(\bk)^2+|\phi|^2})e(\bk)}{2\sqrt{e(\bk)^2+|\phi|^2}(\cos(\beta\theta/2)
+\cosh(\beta\sqrt{e(\bk)^2+|\phi|^2}))}\Bigg),\\
&C(\phi)(\rho\bx 0,\eta \by 0)\\
&=\frac{1}{L^d}\sum_{\bk\in\G^*}e^{i\<\bx-\by,\bk\>}\frac{
-(1_{(\rho,\eta)=(1,2)}\overline{\phi}+
 1_{(\rho,\eta)=(2,1)}\phi)
\sinh(\beta  \sqrt{e(\bk)^2+|\phi|^2})}{
2\sqrt{e(\bk)^2+|\phi|^2}(\cos(\beta\theta/2)+
\cosh(\beta
 \sqrt{e(\bk)^2+|\phi|^2}))}.
\end{align*}
\end{lemma}
\begin{proof}
By using the unitary matrix $U(\phi)(\bk)$ defined in
 \eqref{eq_unitary_for_diagonalization}, which diagonalizes
 $E(\phi)(\bk)$ as shown in
 \eqref{eq_dispersion_matrix_diagonalization}, we can derive the claimed
 equalities.
\end{proof}

We are ready to prove Theorem \ref{thm_main_theorem}.
\begin{proof}[Proof of Theorem \ref{thm_main_theorem}]
The claims \eqref{item_confirmation_condition_superconductivity},
 \eqref{item_imply_no_superconductivity} have been proved right after
 the statement of Theorem \ref{thm_main_theorem}. Let us prove the
 claims \eqref{item_partition_positivity}, \eqref{item_superconducting_phase},
 \eqref{item_no_superconductivity}. With the constant $c(d)(\in \R_{\ge
 1})$ introduced in Proposition \ref{prop_model_UV_application}, set
 $c_2(d):=(2c(d))^{-1}$ and assume that
$$
|U|<c_2(d)(1+\beta^{d+3}+(1+\beta^{-1})g_d(\Theta))^{-2}
$$
throughout the proof. By assuming so the coupling constant $U$ is inside a disk on which
 all the results of Proposition \ref{prop_model_UV_application} and Proposition
 \ref{prop_L_limit} hold.
The inequality \eqref{eq_model_good_smallness} implies that 
\begin{align*}
&\Re \lim_{h\to \infty\atop h\in \frac{2}{\beta}\N}\int
 e^{-V(\psi)+W(\psi)}d\mu_{C(\phi)}(\psi)\ge \frac{1}{2},\\
&(\forall \phi\in\C,\ L\in \N\text{ satisfying
 }\eqref{eq_L_largeness_final_condition}).
\end{align*}
Then, it follows from Lemma \ref{lem_free_partition_function} and
 \eqref{eq_grassmann_formulation_2_band} that
\begin{align*}
&\Re \Tr e^{-\beta(\sH+i\theta\sS_z+\sF)}>0,\quad (\forall  L\in \N\text{ satisfying
 }\eqref{eq_L_largeness_final_condition},\ \g\in [0,1]).
\end{align*}
Then, by taking into account Lemma \ref{lem_real_value_confirmation} we
 observe that the claim \eqref{item_partition_positivity} holds.

By considering Lemma \ref{lem_monotoniticity} we see that the following
 statements hold.
If
\begin{align}
-\frac{2}{|U|}+\frac{1}{(2\pi)^d}\int_{[0,2\pi]^d}d\bk\frac{\sinh(\beta
 |e(\bk)|)}{(\cos(\beta\theta/2)+\cosh(\beta e(\bk)))|e(\bk)|}>0,\label{eq_gap_equation_positive}
\end{align}
there uniquely exists $\D\in (0,\infty)$ such that 
\begin{align}
-\frac{2}{|U|}+\frac{1}{(2\pi)^d}\int_{[0,2\pi]^d}d\bk\frac{\sinh(\beta
\sqrt{e(\bk)^2+\D^2})}{(\cos(\beta\theta/2)+\cosh(\beta \sqrt{e(\bk)^2+\D^2}))\sqrt{e(\bk)^2+\D^2}}=0.\label{eq_gap_equation_delta}
\end{align}
If
\begin{align}
-\frac{2}{|U|}+\frac{1}{(2\pi)^d}\int_{[0,2\pi]^d}d\bk\frac{\sinh(\beta
 |e(\bk)|)}{(\cos(\beta\theta/2)+\cosh(\beta e(\bk)))|e(\bk)|}<0,\label{eq_gap_equation_negative}
\end{align}
there is no positive solution to the equation \eqref{eq_gap_equation_delta}. In this case we set
 $\D:=0$. During the proof we assume
 that either \eqref{eq_gap_equation_positive} or
 \eqref{eq_gap_equation_negative} occurs and $\D(\in \R_{\ge 0})$ is
 defined as above.

Let us prove the claims concerning SSB. We assume that $\g\in (0,1]$
 unless otherwise stated. Then, there uniquely exists $a(\g)\in
 (\D,\infty)$ such that
\begin{align}
&a(\g)\Bigg(-\frac{2}{|U|}\label{eq_exact_gap_equation_continuous}\\
&+\frac{1}{(2\pi)^d}\int_{[0,2\pi]^d}d\bk\frac{\sinh(\beta
 \sqrt{e(\bk)^2+a(\g)^2})}{(\cos(\beta\theta/2)+\cosh(\beta\sqrt{e(\bk)^2+a(\g)^2}))\sqrt{e(\bk)^2+a(\g)^2}}\Bigg)=-\frac{2\g}{|U|}\notag
\end{align}
and $\lim_{\g\searrow 0}a(\g)=\D$. Let us set $\ba:=(a(\g),0)$. Here we
 define the function $F:\R^2\to \R$ by 
\begin{align*}
F(\bx)
&:=-\frac{1}{|U|}((x_1-\g)^2+x_2^2)\\
&\qquad +\frac{1}{\beta
 (2\pi)^{d}}\int_{[0,2\pi]^d}d\bk\log\left(\cos\left(\frac{\beta\theta}{2}\right) +\cosh\left(\beta\sqrt{e(\bk)^2+\|\bx\|_{\R^2}^2}\right)\right)\\
&\qquad -\frac{1}{\beta
 (2\pi)^{d}}\int_{[0,2\pi]^d}d\bk\log\left(\cos\left(\frac{\beta\theta}{2}\right)+\cosh(\beta
 e(\bk))\right).
\end{align*}
For $r\in \R_{>0}$, $\bb\in\R^2$ we set $B_r(\bb):=\{\bx\in\R^2\ |\
 \|\bx-\bb\|_{\R^2}<r\}$. Some remarks concerning the function $F$ are
 in order.
\begin{itemize}
\item $F\in C^{\infty}(\R^2)$.
\item $F$ takes its global maximum at and only at $\bx=\ba$. 
\item 
\begin{align}
&\frac{\partial^2 F}{\partial x_1^2}(\ba)\le -\frac{2\g}{|U|a(\g)},\quad 
\frac{\partial^2 F}{\partial x_1\partial x_2}(\ba)=0,\quad 
\frac{\partial^2 F}{\partial x_2^2}(\ba)= -\frac{2\g}{|U|a(\g)}.\label{eq_L_independent_potential_derivatives}
\end{align}
\item For any $r\in \R_{>0}$,
$$
-\infty <\sup_{\bx\in\R^2\backslash \overline{B_r(\ba)}}(F(\bx)-F(\ba))<0.
$$
\end{itemize}
Since these are the properties of the explicitly defined function, we omit
 the proof. It is also necessary to deal with the discrete analogue $F_L$ of
 $F$. Set for $\bx\in \R^2$
\begin{align*}
F_L(\bx)
&:=-\frac{1}{|U|}((x_1-\g)^2+x_2^2)\\
&\qquad +\frac{1}{\beta L^d}\sum_{\bk\in\G^*}\log\left(\cos\left(\frac{\beta\theta}{2}\right) +\cosh\left(\beta\sqrt{e(\bk)^2+\|\bx\|_{\R^2}^2}\right)\right)\\
&\qquad -\frac{1}{\beta L^d}\sum_{\bk\in\G^*}\log\left(\cos\left(\frac{\beta\theta}{2}\right)+\cosh(\beta
 e(\bk))\right).
\end{align*}
For sufficiently large $L$ we can assume that 
\begin{align}
-\frac{2}{|U|}+\frac{1}{L^d}\sum_{\bk\in \G^*}
\frac{\sinh(\beta
 |e(\bk)|)}{(\cos(\beta\theta/2)+\cosh(\beta e(\bk)))|e(\bk)|}\neq 0.\label{eq_L_dependent_gap_equation}
\end{align}
Since the situation is parallel to that of $F(x)$, it follows that
\begin{itemize}
\item $F_L\in C^{\infty}(\R^2)$.
\item $F_L$ takes its global maximum at and only at
      $\bx=\ba_L=(a_L(\g),0)$,
where $a_L(\g)\in (0,\infty)$ and 
\begin{align*}
&a_L(\g)\Bigg(-\frac{2}{|U|}\\ 
&+\frac{1}{L^d}\sum_{\bk\in\G^*}\frac{\sinh(\beta
 \sqrt{e(\bk)^2+a_L(\g)^2})}{(\cos(\beta\theta/2)+\cosh(\beta\sqrt{e(\bk)^2+a_L(\g)^2}))\sqrt{e(\bk)^2+a_L(\g)^2}}\Bigg)=-\frac{2\g}{|U|}.\notag
\end{align*} 
\end{itemize}
Moreover, we observe that
\begin{itemize}
\item There exists a positive constant $c(\beta,d,\theta,|U|)$ depending
      only on $\beta,d,\theta,|U|$ such that
\begin{align}
&F_L(\bx)\le
 -\frac{\|\bx\|_{\R^2}^2}{|U|}+\left(\frac{2}{|U|}+1\right)\|\bx\|_{\R^2}+
c(\beta,d,\theta,|U|),\quad (\forall \bx\in \R^2,\
 L\in\N).\label{eq_uniform_upper_bound_effective_potential}
\end{align}
\item For any compact set $Q$ of $\R^2$ and $i,j\in \N\cup\{0\}$ with
      $i+j\le 2$,
\begin{align}
\lim_{L\to \infty\atop L\in\N}\sup_{\bx\in Q}\left|
\frac{\partial^{i+j}}{\partial x_1^i\partial x_2^j}F_L(\bx)-
\frac{\partial^{i+j}}{\partial x_1^i\partial x_2^j}F(\bx)
\right|=0.\label{eq_uniform_convergence_effective_potential}
\end{align}
\end{itemize}
By making use of the properties
 \eqref{eq_uniform_upper_bound_effective_potential},
 \eqref{eq_uniform_convergence_effective_potential} we can prove that
\begin{align}
\lim_{L\to\infty\atop L\in
 \N}\ba_L=\ba.\label{eq_maximum_point_convergence}
\end{align}
Let $H(F)(\bx)$, $H(F_L)(\bx)$ denote the Hessian of $F$, $F_L$
 respectively. The property
 \eqref{eq_L_independent_potential_derivatives} implies that 
$$
H(F)(\ba)\le -\frac{2\g}{|U|a(\g)}.
$$
By applying \eqref{eq_uniform_upper_bound_effective_potential}, 
\eqref{eq_uniform_convergence_effective_potential}, 
\eqref{eq_maximum_point_convergence} we can establish necessary
 basic properties as follows. There exist $\delta \in\R_{>0}$ and
 $L_0\in\N$ such that the following statements hold true for any $L\in
 \N$ with $L\ge L_0$.
\begin{itemize}
\item For any $\bx\in \overline{B_{\delta}(\ba_L)}$, 
\begin{align}
&F_L(\bx)=F_L(\ba_L)+\int_0^1dt (1-t)\<\bx-\ba_L,H(F_L)(t(\bx-\ba_L)+\ba_L)(\bx-\ba_L)\>,\label{eq_L_dependent_taylor_expansion}\\
&H(F_L)(t(\bx-\ba_L)+\ba_L)\le
 \frac{1}{2}H(F)(\ba)<0,\quad (\forall t\in [0,1]).\label{eq_hessian_uniform_bound}
\end{align}
\item For any $\bx\in \R^2\backslash \overline{B_{\delta}(\ba_L)}$,
\begin{align}
F_L(\bx)-F_L(\ba_L)\le
 \frac{1}{2}\sup_{\bx\in\R^2\backslash\overline{B_{\delta/2}(\ba)}}(F(\bx)-F(\ba))<0.\label{eq_difference_from_maximum}
\end{align}
\end{itemize}

At this point we go back to the Grassmann integral formulations. By
 Proposition \ref{prop_model_UV_application} and \eqref{eq_full_covariance_determinant_bound_pre},
\begin{align} 
&\sup_{L\in \N\atop \text{satisfying
 }\eqref{eq_L_largeness_final_condition}}
\sup_{h\in\frac{2}{\beta}\N\atop h\ge
 c(d)\max\{1,\beta^{-1}\}}\sup_{\phi\in\C}\Bigg(
\left|\int e^{-V(\psi)+W(\psi)}d\mu_{C(\phi)}(\psi)
\right|\label{eq_grassmann_formulations_all_bound}\\
&\qquad+\sum_{j\in\{1,2\}}\left|\int e^{-V(\psi)+W(\psi)}A^j(\psi) d\mu_{C(\phi)}(\psi)
\right|
+\sum_{j\in\{1,2\}}\left|\int A^j(\psi) d\mu_{C(\phi)}(\psi)
\right|\Bigg)<\infty.\notag
\end{align}
Since the functions inside the modulus above are continuous with $\phi$
 over $\C$, the following transformation is justified.
\begin{align}
&\int_{\R^2}d\phi_1d\phi_2 e^{-\frac{\beta L^d}{|U|}|\phi-\g|^2}
\frac{\prod_{\bk\in\G^*}(\cos(\beta
 \theta/2)+\cosh(\beta\sqrt{e(\bk)^2+|\phi|^2}))}{\prod_{\bk\in\G^*}(\cos(\beta \theta/2)+\cosh(\beta e(\bk)))}\label{eq_2_point_formulation_decomposition}\\
&\quad\cdot \int e^{-V(\psi)+W(\psi)}A^1(\psi)d\mu_{C(\phi)}(\psi)\notag\\
&=\int_{\R^2}d\phi_1d\phi_2 e^{\beta L^dF_L(\bphi)}
\int e^{-V(\psi)+W(\psi)}d\mu_{C(\phi)}(\psi)
\int A^1(\psi)d\mu_{C(\phi)}(\psi)\notag\\
&\quad + 
\int_{\R^2}d\phi_1d\phi_2 e^{\beta L^dF_L(\bphi)}\Bigg(
\int e^{-V(\psi)+W(\psi)}A^1(\psi)d\mu_{C(\phi)}(\psi)\notag\\
&\qquad\qquad -\int e^{-V(\psi)+W(\psi)}d\mu_{C(\phi)}(\psi)
\int A^1(\psi)d\mu_{C(\phi)}(\psi)\Bigg)\notag\\
&=e^{\beta L^dF_L(\ba_L)}L^{-d}
\int_{\R^2}d\phi_1d\phi_2 1_{\|\bphi\|_{\R^2}\le L^{\frac{d}{2}}\delta}
e^{\beta\int_0^1dt(1-t)\<\bphi,H(F_L)(tL^{-\frac{d}{2}}\bphi+\ba_L)\bphi\>}\notag\\
&\qquad\cdot
\int e^{-V(\psi)+W(\psi)}d\mu_{C(L^{-\frac{d}{2}}\phi+a_L(\g))}(\psi)
\int A^1(\psi)d\mu_{C(L^{-\frac{d}{2}}\phi+a_L(\g))}(\psi)\notag\\
&\quad + 
e^{\beta L^dF_L(\ba_L)}
\int_{\R^2}d\phi_1d\phi_2 1_{\bphi\notin \overline{B_{\delta}(\ba_L)}}
e^{\beta L^d(F_L(\bphi)-F_L(\ba_L))}\notag\\
&\qquad\cdot 
\int e^{-V(\psi)+W(\psi)}d\mu_{C(\phi)}(\psi)
\int A^1(\psi)d\mu_{C(\phi)}(\psi)\notag\\
&\quad + 
e^{\beta L^dF_L(\ba_L)}L^{-d}
\int_{\R^2}d\phi_1d\phi_2 1_{\|\bphi\|_{\R^2}\le L^{\frac{d}{2}}\delta}
e^{\beta\int_0^1dt(1-t)\<\bphi,H(F_L)(tL^{-\frac{d}{2}}\bphi+\ba_L)\bphi\>}\notag\\
&\qquad\cdot \Bigg(
\int e^{-V(\psi)+W(\psi)}A^1(\psi)
d\mu_{C(L^{-\frac{d}{2}}\phi+a_L(\g))}(\psi)\notag\\
&\qquad\qquad -
\int e^{-V(\psi)+W(\psi)}d\mu_{C(L^{-\frac{d}{2}}\phi+a_L(\g))}(\psi)
\int A^1(\psi)
d\mu_{C(L^{-\frac{d}{2}}\phi+a_L(\g))}(\psi)\Bigg)\notag\\
&\quad + e^{\beta L^dF_L(\ba_L)}\int_{\R^2}d\phi_1d\phi_2
1_{\bphi\notin  \overline{B_{\delta}(\ba_L)}}
e^{\beta L^d(F_L(\bphi)-F_L(\ba_L))}\notag\\
&\qquad\cdot \Bigg(
\int e^{-V(\psi)+W(\psi)}A^1(\psi)d\mu_{C(\phi)}(\psi)\notag\\
&\qquad\qquad -\int e^{-V(\psi)+W(\psi)}d\mu_{C(\phi)}(\psi)
\int A^1(\psi)d\mu_{C(\phi)}(\psi)\Bigg),\notag
\end{align}
where $\bphi=(\phi_1,\phi_2)$, $\phi=\phi_1+i\phi_2$ and 
$\delta(\in\R_{>0})$ is the parameter appearing in
\eqref{eq_L_dependent_taylor_expansion},
 \eqref{eq_hessian_uniform_bound}, \eqref{eq_difference_from_maximum}.
It follows from Lemma \ref{lem_grassmann_formulation_2_band}
 \eqref{item_integrand_time_continuum_limit}, 
\eqref{eq_model_error_estimate},
  \eqref{eq_uniform_upper_bound_effective_potential}, 
\eqref{eq_hessian_uniform_bound},
 \eqref{eq_difference_from_maximum},
\eqref{eq_grassmann_formulations_all_bound} that
\begin{align}
&\lim_{L\to \infty\atop L\in \N}\lim_{h\to \infty\atop h\in
 \frac{2}{\beta}\N}L^{d}
\int_{\R^2}d\phi_1d\phi_2 1_{\bphi\notin \overline{B_{\delta}(\ba_L)}}
e^{\beta L^d(F_L(\bphi)-F_L(\ba_L))}\label{eq_2_point_formulation_vanishing_part}\\
&\qquad\cdot 
\int e^{-V(\psi)+W(\psi)}d\mu_{C(\phi)}(\psi)
\int A^1(\psi)d\mu_{C(\phi)}(\psi)=0,\notag\\
&\lim_{L\to \infty\atop L\in\N}\lim_{h\to \infty\atop
 h\in\frac{2}{\beta}\N}
\int_{\R^2}d\phi_1d\phi_2 1_{\|\bphi\|_{\R^2}\le L^{\frac{d}{2}}\delta}
e^{\beta\int_0^1dt(1-t)\<\bphi,H(F_L)(tL^{-\frac{d}{2}}\bphi+\ba_L)\bphi\>}\notag\\
&\qquad\cdot \Bigg(
\int e^{-V(\psi)+W(\psi)}A^1(\psi)
d\mu_{C(L^{-\frac{d}{2}}\phi+a_L(\g))}(\psi)\notag\\
&\qquad\qquad -
\int e^{-V(\psi)+W(\psi)}d\mu_{C(L^{-\frac{d}{2}}\phi+a_L(\g))}(\psi)
\int A^1(\psi)
d\mu_{C(L^{-\frac{d}{2}}\phi+a_L(\g))}(\psi)\Bigg)=0,\notag\\
&\lim_{L\to \infty\atop L\in \N}\lim_{h\to \infty\atop h\in
 \frac{2}{\beta}\N}L^d\int_{\R^2}d\phi_1d\phi_2 1_{\bphi\notin \overline{B_{\delta}(\ba_L)}}
e^{\beta L^d(F_L(\bphi)-F_L(\ba_L))}\notag\\
&\cdot\Bigg(
\int e^{-V(\psi)+W(\psi)}A^1(\psi)d\mu_{C(\phi)}(\psi) -\int e^{-V(\psi)+W(\psi)}d\mu_{C(\phi)}(\psi)
\int A^1(\psi)d\mu_{C(\phi)}(\psi)\Bigg)\notag\\
&=0.\notag
\end{align}
Moreover, we can apply 
 Proposition \ref{prop_L_limit},
 \eqref{eq_uniform_convergence_effective_potential},
 \eqref{eq_maximum_point_convergence}, \eqref{eq_hessian_uniform_bound},
\eqref{eq_grassmann_formulations_all_bound} to conclude that
\begin{align}
&\lim_{L\to \infty\atop L\in \N}\lim_{h\to \infty\atop h\in
 \frac{2}{\beta}\N}
\int_{\R^2}d\phi_1d\phi_2 1_{\|\bphi\|_{\R^2}\le L^{\frac{d}{2}}\delta}
e^{\beta\int_0^1dt(1-t)\<\bphi,H(F_L)(tL^{-\frac{d}{2}}\bphi+\ba_L)\bphi\>}\label{eq_2_point_formulation_surviving_part}\\
&\qquad\cdot
\int e^{-V(\psi)+W(\psi)}d\mu_{C(L^{-\frac{d}{2}}\phi+a_L(\g))}(\psi)
\int A^1(\psi)d\mu_{C(L^{-\frac{d}{2}}\phi+a_L(\g))}(\psi)\notag\\
&=\int_{\R^2}d\phi_1d\phi_2e^{\frac{\beta}{2}\<\bphi,H(F)(\ba)\bphi\>}
\lim_{L\to \infty\atop L\in \N}\lim_{h\to \infty\atop h\in
 \frac{2}{\beta}\N}
\int e^{-V(\psi)+W(\psi)}d\mu_{C(a(\g))}(\psi)\notag\\
&\quad\cdot \beta
\lim_{L\to \infty\atop L\in \N}C(a(\g))(1\b0 0,2\b0 0).\notag
\end{align}
Similarly we have that 
\begin{align}
&\int_{\R^2}d\phi_1d\phi_2 e^{\beta L^dF_L(\phi)}
\int e^{-V(\psi)+W(\psi)}d\mu_{C(\phi)}(\psi)
\label{eq_0_point_formulation_decomposition}\\
&=e^{\beta L^dF_L(\ba_L)}L^{-d}
\int_{\R^2}d\phi_1d\phi_2 1_{\|\bphi\|_{\R^2}\le L^{\frac{d}{2}}\delta}
e^{\beta\int_0^1dt(1-t)\<\bphi,H(F_L)(tL^{-\frac{d}{2}}\bphi+\ba_L)\bphi\>}\notag\\
&\qquad\cdot
\int e^{-V(\psi)+W(\psi)}d\mu_{C(L^{-\frac{d}{2}}\phi+a_L(\g))}(\psi)\notag\\
&\quad + 
e^{\beta L^dF_L(\ba_L)}
\int_{\R^2}d\phi_1d\phi_2 1_{\bphi\notin \overline{B_{\delta}(\ba_L)}}
e^{\beta L^d(F_L(\bphi)-F_L(\ba_L))}
\int e^{-V(\psi)+W(\psi)}d\mu_{C(\phi)}(\psi).\notag
\end{align}
\begin{align}
&\lim_{L\to \infty\atop L\in \N}\lim_{h\to \infty\atop h\in
 \frac{2}{\beta}\N}
\int_{\R^2}d\phi_1d\phi_2 1_{\|\bphi\|_{\R^2}\le L^{\frac{d}{2}}\delta}
e^{\beta\int_0^1dt(1-t)\<\bphi,H(F_L)(tL^{-\frac{d}{2}}\bphi+\ba_L)\bphi\>}\label{eq_0_point_formulation_surviving_part}\\
&\qquad\cdot
\int e^{-V(\psi)+W(\psi)}d\mu_{C(L^{-\frac{d}{2}}\phi+a_L(\g))}(\psi)\notag\\
&=\int_{\R^2}d\phi_1d\phi_2e^{\frac{\beta}{2}\<\bphi,H(F)(\ba)\bphi\>}
\lim_{L\to \infty\atop L\in \N}\lim_{h\to \infty\atop h\in
 \frac{2}{\beta}\N}
\int e^{-V(\psi)+W(\psi)}d\mu_{C(a(\g))}(\psi),\notag\\
&\lim_{L\to \infty\atop L\in \N}\lim_{h\to \infty\atop h\in
 \frac{2}{\beta}\N}L^{d}
\int_{\R^2}d\phi_1d\phi_2 1_{\bphi\notin \overline{B_{\delta}(\ba_L)}}
e^{\beta L^d(F_L(\bphi)-F_L(\ba_L))}
\int e^{-V(\psi)+W(\psi)}d\mu_{C(\phi)}(\psi)=0.\notag
\end{align}
The inequality \eqref{eq_model_good_smallness} implies that
\begin{align}
\lim_{L\to \infty\atop L\in \N}\lim_{h\to \infty\atop h\in
 \frac{2}{\beta}\N}
\int e^{-V(\psi)+W(\psi)}d\mu_{C(a(\g))}(\psi)\neq
 0.\label{eq_correction_term_limit_non_zero}
\end{align}
Note that by Lemma \ref{lem_grassmann_formulation_2_band}
 \eqref{item_integrand_time_continuum_limit} and
 \eqref{eq_grassmann_formulations_all_bound} we can change the order of
 the integral over $\R^2$ and the limit operation with $h$ in
 \eqref{eq_grassmann_formulation_2_band} with $\bla=(0,0)$ and
 \eqref{eq_grassmann_formulation_2_band_correlation}. Then, by using
  \eqref{eq_2_point_formulation_decomposition},
\eqref{eq_2_point_formulation_vanishing_part},
\eqref{eq_2_point_formulation_surviving_part},
 \eqref{eq_0_point_formulation_decomposition},
\eqref{eq_0_point_formulation_surviving_part},
\eqref{eq_correction_term_limit_non_zero} 
and \eqref{eq_exact_gap_equation_continuous}, 
Lemma \ref{lem_explicit_calculation_covariance}
we can derive from \eqref{eq_grassmann_formulation_2_band},
 \eqref{eq_grassmann_formulation_2_band_correlation} that
\begin{align*}
&\lim_{L\to \infty\atop L\in\N}\frac{\Tr (e^{-\beta
 (\sH+i\theta\sS_z+\sF)}\sA_1)}{\Tr e^{-\beta (\sH+i\theta\sS_z+\sF)}}\\
&=\lim_{L\to \infty\atop L\in \N}
\frac{\lim_{h\to \infty\atop h\in
 \frac{2}{\beta}\N}\int_{\R^2}d\phi_1d\phi_2  e^{\beta L^d F_L(\bphi)}
\int e^{-V(\psi)+W(\psi)}A^1(\psi)d\mu_{C(\phi)}(\psi)}
{\lim_{h\to \infty\atop h\in
 \frac{2}{\beta}\N}\beta\int_{\R^2}d\phi_1d\phi_2 e^{\beta L^d F_L(\bphi)}
\int e^{-V(\psi)+W(\psi)}d\mu_{C(\phi)}(\psi)}\\
&=\int_{\R^2}d\phi_1d\phi_2  
e^{\frac{\beta}{2}\<\bphi,H(F)(\ba)\bphi\>}
\lim_{L\to \infty\atop L\in \N}\lim_{h\to \infty\atop h\in
 \frac{2}{\beta}\N}
\int e^{-V(\psi)+W(\psi)}d\mu_{C(a(\g))}(\psi)\\
&\quad\cdot 
\lim_{L\to \infty\atop L\in \N}C(a(\g))(1\b0 0,2\b0 0)\\
&\quad\cdot\Bigg/\int_{\R^2}d\phi_1d\phi_2  
e^{\frac{\beta}{2}\<\bphi,H(F)(\ba)\bphi\>}
\lim_{L\to \infty\atop L\in \N}\lim_{h\to \infty\atop h\in
 \frac{2}{\beta}\N}
\int e^{-V(\psi)+W(\psi)}d\mu_{C(a(\g))}(\psi)\\
&=-\frac{a(\g)}{|U|}+\frac{\g}{|U|}.
\end{align*}
Thus, 
\begin{align*}
&\lim_{\g\searrow 0\atop \g\in(0,1]}\lim_{L\to \infty\atop L\in\N}\frac{\Tr (e^{-\beta
 (\sH+i\theta\sS_z+\sF)}\sA_1)}{\Tr e^{-\beta (\sH+i\theta\sS_z+\sF)}}=-\frac{\D}{|U|}.
\end{align*}
We let $c_1(d)$ be the constant $c(d)$ appearing in Lemma
 \ref{lem_sufficient_condition_gap_equation}
 \eqref{item_sufficient_solvability}. Then, if
\begin{align*}
|U|>c_1(d)(2d-|\mu|)^{1-d}\beta\Theta\left(
1_{\Theta\le \frac{1}{2}(2d-|\mu|)}+
1_{\Theta> \frac{1}{2}(2d-|\mu|)}
(2d-|\mu|)^{-1}\Theta\right),
\end{align*}
\eqref{eq_gap_equation_positive} holds and thus $\D>0$. 
This proves the claims \eqref{eq_gap_equation}, \eqref{eq_SSB}.
Note that 
$$
c_2(d)(1+\beta^{d+3}+(1+\beta^{-1})g_d(\Theta))^{-2}\le 
(1+\beta^{d+3})^{-2}\le 2\beta^{-1},\quad (\forall \beta \in \R_{>0}).
$$
Thus, if $\theta\in [0,\pi/\beta]$, Lemma
 \ref{lem_sufficient_condition_gap_equation}
 \eqref{item_sufficient_nonsolvability} implies that
 \eqref{eq_gap_equation_negative} holds and thus $\D=0$. Therefore, the first statement of
 \eqref{item_no_superconductivity} and the claim concerning SSB in
 \eqref{item_no_superconductivity} hold true.

Next let us prove the claim \eqref{eq_ODLRO} and the claim concerning
 ODLRO in \eqref{item_no_superconductivity}. The proof is in
 fact close to the proof of SSB above. However we present it for
 completeness. Let us define the function $f:\R\to \R$ by
\begin{align*}
f(x):=&-\frac{x^2}{|U|}+\frac{1}{\beta (2\pi)^d}\int_{[0,2\pi]^d}d\bk
\log\left(\cos\left(\frac{\beta\theta}{2}\right)+\cosh(\beta
 \sqrt{e(\bk)^2+x^2})\right)\\
&-\frac{1}{\beta (2\pi)^d}\int_{[0,2\pi]^d}d\bk
\log\left(\cos\left(\frac{\beta\theta}{2}\right)+\cosh(\beta
 e(\bk))\right).
\end{align*}
Then, we see that
\begin{itemize}
\item $f\in C^{\infty}(\R)$.
\item $f|_{[0,\infty)}:[0,\infty)\to \R$ takes its global maximum at and
      only at $x=\D$, where $f|_{[0,\infty)}$ denotes the restriction of
      $f$ on $[0,\infty)$.
\item 
$$
\frac{d^2 f}{d x^2}(\D)<0.
$$
\end{itemize}
The third statement above can be confirmed as follows.
\begin{align*}
&\frac{d^2 f}{d x^2}(\D)\\
&\le
 1_{\D=0}\left(-\frac{2}{|U|}+\frac{1}{(2\pi)^d}\int_{[0,2\pi]^d}d\bk
\frac{\sinh(\beta|e(\bk)|)}{(\cos(\beta\theta/2)+\cosh(\beta e(\bk)))|e(\bk)|}
\right)\\
&\quad +1_{\D>0} \D \sup_{\eta\in [-2d-|\mu|,2d+|\mu|]}\left(\frac{d}{dx}\left(
\frac{\sinh(\beta\sqrt{\eta^2+x^2})}{(\cos(\beta\theta/2)+\cosh(\beta\sqrt{\eta^2+x^2}))\sqrt{\eta^2+x^2}}\right)\Bigg|_{x=\D}\right)\\
&<0.
\end{align*}
Again we need to introduce the $L$-dependent version of $f$ as follows.
\begin{align*}
f_L(x):=&-\frac{x^2}{|U|}+\frac{1}{\beta L^d}\sum_{\bk\in\G^*}
\log\left(\cos\left(\frac{\beta\theta}{2}\right)+\cosh(\beta
 \sqrt{e(\bk)^2+x^2})\right)\\
&-\frac{1}{\beta L^d}\sum_{\bk\in\G^*}
\log\left(\cos\left(\frac{\beta\theta}{2}\right)+\cosh(\beta
 e(\bk))\right).
\end{align*}
We may assume that \eqref{eq_L_dependent_gap_equation} holds. When the
 left-hand side of \eqref{eq_L_dependent_gap_equation} is positive,
 there uniquely exists $\D_L\in (0,\infty)$ such that
\begin{align*}
-\frac{2}{|U|}+\frac{1}{L^d}\sum_{\bk\in\G^*}
\frac{\sinh(\beta\sqrt{e(\bk)^2+\D_L^2})}{(\cos(\beta\theta/2)+\cosh(\beta\sqrt{e(\bk)^2+\D_L^2}))\sqrt{e(\bk)^2+\D_L^2}}=0.
\end{align*}
If the left-hand side of \eqref{eq_L_dependent_gap_equation} is
 negative, we set $\D_L:=0$. It follows that
\begin{itemize}
\item $f_L|_{[0,\infty)}:[0,\infty)\to\R$ takes its global maximum at
      and only at $x=\D_L$, where $f_L|_{[0,\infty)}$ is the restriction
      of $f_L$ on $[0,\infty)$.
\end{itemize}
Based on this fact and that $f_L$ and the derivatives of $f_L$ locally
 uniformly converge to $f$ and those of $f$ respectively as $L\to \infty$, we can prove that
\begin{align}
\lim_{L\to \infty\atop
 L\in\N}\D_L=\D.\label{eq_maximum_point_convergence_1D}
\end{align}
Moreover, there exist $\delta\in\R_{>0}$ and $L_0\in\N$ such that the
 following statements hold true for any $L\in \N$ with $L\ge L_0$.
\begin{itemize}
\item For any $x\in [\D_L-\delta,\D_L+\delta]$, 
\begin{align}
&f_L(x)=f_L(\D_L)+\int_0^1dt(1-t)\frac{d^2f_L}{dx^2}(t(x-\D_L)+\D_L)(x-\D_L)^2,\label{eq_L_dependent_taylor_expansion_1_D}\\
&\frac{d^2f_L}{dx^2}(t(x-\D_L)+\D_L)\le
 \frac{1}{2}\frac{d^2f}{dx^2}(\D)<0,\quad (\forall t\in [0,1]).\label{eq_hessian_uniform_bound_1D}
\end{align}
\item For any $x\in [0,\infty)\backslash [\D_L-\delta,\D_L+\delta]$,
\begin{align}
f_L(x)-f_L(\D_L)\le
 \frac{1}{2}\sup_{x\in[0,\infty)\backslash
 [\D-\frac{\delta}{2},\D+\frac{\delta}{2}]}(f(x)-f(\D))<0.
\label{eq_difference_from_maximum_1D}
\end{align}
\end{itemize}
Observe that
\begin{align}
&\int_{\R^2}d\phi_1d\phi_2 e^{-\frac{\beta L^d}{|U|}|\phi|^2}
\frac{\prod_{\bk\in\G^*}(\cos(\beta
 \theta/2)+\cosh(\beta\sqrt{e(\bk)^2+|\phi|^2}))}{\prod_{\bk\in\G^*}(\cos(\beta \theta/2)+\cosh(\beta e(\bk)))}\label{eq_4_point_formulation_decomposition}\\
&\quad\cdot \int e^{-V(\psi)+W(\psi)}A^2(\psi)d\mu_{C(\phi)}(\psi)\notag\\
&=\int_{0}^{2\pi}d\xi \int_0^{\infty}dr r  e^{\beta L^df_L(r)}
\int e^{-V(\psi)+W(\psi)}d\mu_{C(re^{i\xi})(\psi)}
\int A^2(\psi)d\mu_{C(re^{i\xi})}(\psi)\notag\\
&\quad + \int_{0}^{2\pi}d\xi \int_0^{\infty}dr r  e^{\beta L^df_L(r)}
\Bigg(
\int e^{-V(\psi)+W(\psi)}A^2(\psi)d\mu_{C(re^{i\xi})}(\psi)\notag\\
&\qquad\qquad -\int e^{-V(\psi)+W(\psi)}d\mu_{C(re^{i\xi})}(\psi)
\int A^2(\psi)d\mu_{C(re^{i\xi})}(\psi)\Bigg)\notag\\
&=e^{\beta L^df_L(\D_L)}L^{-\frac{d}{2}}\int_0^{2\pi}d\xi 
\int_{-L^{d/2}\min\{\delta,\D_L\}}^{L^{d/2}\delta}dr
(L^{-\frac{d}{2}}r+\D_L)
e^{\beta \int_0^1dt (1-t)f_L''(tL^{-\frac{d}{2}}r+\D_L)r^2}\notag\\
&\qquad\cdot
\int e^{-V(\psi)+W(\psi)}d\mu_{C((L^{-\frac{d}{2}}r+\D_L)e^{i\xi})}(\psi)
\int A^2(\psi)d\mu_{C((L^{-\frac{d}{2}}r+\D_L)e^{i\xi})}(\psi)\notag\\
&\quad + 
e^{\beta L^df_L(\D_L)}
\int_0^{2\pi}d\xi 
\int_{[0,\infty)\backslash [\D_L-\delta,\D_L+\delta]}dr
r e^{\beta L^d(f_L(r)-f_L(\D_L))}\notag\\
&\qquad\cdot 
\int e^{-V(\psi)+W(\psi)}d\mu_{C(re^{i\xi})}(\psi)
\int A^2(\psi)d\mu_{C(re^{i\xi})}(\psi)\notag\\
&\quad + 
e^{\beta L^df_L(\D_L)}L^{-\frac{d}{2}}
\int_0^{2\pi}d\xi 
\int_{-L^{d/2}\min\{\delta,\D_L\}}^{L^{d/2}\delta}
dr(L^{-\frac{d}{2}}r+\D_L)
e^{\beta \int_0^1dt (1-t)f_L''(tL^{-\frac{d}{2}}r+\D_L)r^2}\notag\\
&\qquad\cdot
\Bigg(
\int
 e^{-V(\psi)+W(\psi)}A^2(\psi)d\mu_{C((L^{-\frac{d}{2}}r+\D_L)e^{i\xi})}(\psi)\notag\\&\qquad\qquad-
\int
 e^{-V(\psi)+W(\psi)}d\mu_{C((L^{-\frac{d}{2}}r+\D_L)e^{i\xi})}(\psi)
\int A^2(\psi)d\mu_{C((L^{-\frac{d}{2}}r+\D_L)e^{i\xi})}(\psi)\Bigg)\notag\\
&\quad + 
e^{\beta L^df_L(\D_L)}
\int_0^{2\pi}d\xi 
\int_{[0,\infty)\backslash[\D_L-\delta,\D_L+\delta]}dr
r e^{\beta L^d(f_L(r)-f_L(\D_L))}\notag\\
&\qquad \cdot \Bigg(
\int e^{-V(\psi)+W(\psi)}A^2(\psi)
d\mu_{C(re^{i\xi})}(\psi)\notag\\
&\qquad\qquad -\int e^{-V(\psi)+W(\psi)}d\mu_{C(re^{i\xi})}(\psi)
\int A^2(\psi)
d\mu_{C(re^{i\xi})}(\psi)\Bigg),\notag
\end{align}
where $\delta(\in \R_{>0})$ is that appearing in
\eqref{eq_L_dependent_taylor_expansion_1_D},  
\eqref{eq_hessian_uniform_bound_1D}, \eqref{eq_difference_from_maximum_1D}.
To prove convergent properties, we need to multiply different volume
 factors depending on whether $\D>0$ or $\D=0$. Lemma \ref{lem_grassmann_formulation_2_band}
 \eqref{item_integrand_time_continuum_limit}, the inequalities  
                 \eqref{eq_grassmann_formulations_all_bound},
                 \eqref{eq_difference_from_maximum_1D} and a variant of
 the inequality \eqref{eq_uniform_upper_bound_effective_potential}
 ensure that 
\begin{align}
&\lim_{L\to \infty\atop L\in\N}\lim_{h\to\infty\atop h\in
 \frac{2}{\beta}\N}
(1_{\D>0}L^{\frac{d}{2}}+1_{\D=0}L^d)
\int_0^{2\pi}d\xi 
\int_{[0,\infty)\backslash [\D_L-\delta,\D_L+\delta]}dr
r e^{\beta L^d(f_L(r)-f_L(\D_L))}\label{eq_4_point_formulation_vanishing_part}\\
&\qquad\cdot 
\int e^{-V(\psi)+W(\psi)}d\mu_{C(re^{i\xi})}(\psi)
\int A^2(\psi)d\mu_{C(re^{i\xi})}(\psi)=0,\notag\\
&\lim_{L\to \infty\atop L\in\N}\lim_{h\to\infty\atop h\in
 \frac{2}{\beta}\N}
(1_{\D>0}L^{\frac{d}{2}}+1_{\D=0}L^d)
\int_0^{2\pi}d\xi 
\int_{[0,\infty)\backslash [\D_L-\delta,\D_L+\delta]}dr
r e^{\beta L^d(f_L(r)-f_L(\D_L))}\notag\\
&\qquad \cdot \Bigg(
\int e^{-V(\psi)+W(\psi)}A^2(\psi)
d\mu_{C(re^{i\xi})}(\psi)\notag\\
&\qquad\qquad -\int e^{-V(\psi)+W(\psi)}d\mu_{C(re^{i\xi})}(\psi)
\int A^2(\psi)
d\mu_{C(re^{i\xi})}(\psi)\Bigg)=0.\notag
\end{align}
Moreover, it follows from Lemma \ref{lem_grassmann_formulation_2_band}
 \eqref{item_integrand_time_continuum_limit}, 
\eqref{eq_model_error_estimate},
Proposition \ref{prop_L_limit}, \eqref{eq_grassmann_formulations_all_bound},
 \eqref{eq_maximum_point_convergence_1D},
 \eqref{eq_hessian_uniform_bound_1D} and a variant of
 \eqref{eq_uniform_convergence_effective_potential}
 that
\begin{align}
&\lim_{L\to \infty\atop L\in\N}\lim_{h\to\infty\atop h\in
 \frac{2}{\beta}\N}
\int_0^{2\pi}d\xi 
\int_{-L^{d/2}\min\{\delta,\D_L\}}^{L^{d/2}\delta}dr
(1_{\D>0}(L^{-\frac{d}{2}}r+\D_L)+1_{\D=0}r)
e^{\beta \int_0^1dt (1-t)f_L''(tL^{-\frac{d}{2}}r+\D_L)r^2}\label{eq_4_point_formulation_surviving_part}\\
&\qquad\cdot
\int e^{-V(\psi)+W(\psi)}d\mu_{C((L^{-\frac{d}{2}}r+\D_L)e^{i\xi})}(\psi)
\int A^2(\psi)d\mu_{C((L^{-\frac{d}{2}}r+\D_L)e^{i\xi})}(\psi)\notag\\
&= \left(1_{\D>0}\int_{-\infty}^{\infty}dr\D+1_{\D=0}\int_0^{\infty}dr
r\right)e^{\frac{\beta}{2}f''(\D)r^2}\int_0^{2\pi}d\xi\notag\\ 
&\qquad\cdot \lim_{L\to \infty\atop L\in\N}
\lim_{h\to \infty\atop h\in\frac{2}{\beta}\N}
\int e^{-V(\psi)+W(\psi)}d\mu_{C(\D e^{i\xi})}(\psi)\notag\\
&\quad\cdot (-\beta)\lim_{L\to \infty\atop
 L\in\N}\det\left(\begin{array}{cc} C(\D)(1\hat{\bx}0,1\hat{\by} 0) &
	 C(\D)(1\b0 0, 2 \b0 0) \\
C(\D)(2 \b0 0,1\b0 0) &
	 C(\D)(2\hat{\by} 0, 2 \hat{\bx} 0)\end{array}\right),\notag\\
&\lim_{L\to \infty\atop L\in\N}\lim_{h\to\infty\atop h\in
 \frac{2}{\beta}\N}
\int_0^{2\pi}d\xi 
\int_{-L^{d/2}\min\{\delta,\D_L\}}^{L^{d/2}\delta}dr
(1_{\D>0}(L^{-\frac{d}{2}}r+\D_L)+1_{\D=0}r)
e^{\beta \int_0^1dt (1-t)f_L''(tL^{-\frac{d}{2}}r+\D_L)r^2}\notag\\
&\qquad\cdot\Bigg(
\int
 e^{-V(\psi)+W(\psi)}A^2(\psi)d\mu_{C((L^{-\frac{d}{2}}r+\D_L)e^{i\xi})}(\psi)\notag\\&\qquad\qquad -
\int e^{-V(\psi)+W(\psi)}d\mu_{C((L^{-\frac{d}{2}}r+\D_L)e^{i\xi})}(\psi)
\int A^2(\psi)d\mu_{C((L^{-\frac{d}{2}}r+\D_L)e^{i\xi})}(\psi)\Bigg)=0.\notag
\end{align}
For the same reason as above we have that 
\begin{align}
&\int_{\R^2}d\phi_1d\phi_2 e^{\beta L^df_L(|\phi|)}
 \int
 e^{-V(\psi)+W(\psi)}d\mu_{C(\phi)}(\psi)\label{eq_no_field_formulation_decomposition}\\
&=e^{\beta L^df_L(\D_L)}L^{-\frac{d}{2}}\int_0^{2\pi}d\xi 
\int_{-L^{d/2}\min\{\delta,\D_L\}}^{L^{d/2}\delta}dr
(L^{-\frac{d}{2}}r+\D_L)
e^{\beta \int_0^1dt (1-t) f_L''(tL^{-\frac{d}{2}}r+\D_L)r^2}\notag\\
&\qquad\cdot
\int e^{-V(\psi)+W(\psi)}d\mu_{C((L^{-\frac{d}{2}}r+\D_L)e^{i\xi})}(\psi)\notag\\
&\quad + 
e^{\beta L^df_L(\D_L)}
\int_0^{2\pi}d\xi 
\int_{[0,\infty)\backslash [\D_L-\delta,\D_L+\delta]}dr
r e^{\beta L^d(f_L(r)-f_L(\D_L))}\notag\\
&\qquad\cdot  
\int e^{-V(\psi)+W(\psi)}d\mu_{C(re^{i\xi})}(\psi)\notag
\end{align}
and
\begin{align}
&\lim_{L\to \infty\atop L\in\N}\lim_{h\to\infty\atop h\in
 \frac{2}{\beta}\N}
\int_0^{2\pi}d\xi 
\int_{-L^{d/2}\min\{\delta,\D_L\}}^{L^{d/2}\delta}dr
(1_{\D>0}(L^{-\frac{d}{2}}r+\D_L)+1_{\D=0}r)
e^{\beta \int_0^1dt (1-t) f_L''(tL^{-\frac{d}{2}}r+\D_L)r^2}\label{eq_no_field_formulation_surviving_part}\\
&\qquad\cdot
\int e^{-V(\psi)+W(\psi)}d\mu_{C((L^{-\frac{d}{2}}r+\D_L)e^{i\xi})}(\psi)\notag\\
&=\left(1_{\D>0}\int_{-\infty}^{\infty}dr\D+1_{\D=0}\int_0^{\infty}dr
r\right) e^{\frac{\beta}{2}f''(\D)r^2}\int_0^{2\pi}d\xi\notag\\ 
&\qquad\cdot\lim_{L\to \infty\atop L\in\N}
\lim_{h\to \infty\atop h\in\frac{2}{\beta}\N}
\int e^{-V(\psi)+W(\psi)}d\mu_{C(\D e^{i\xi})}(\psi),\notag\\
&\lim_{L\to \infty\atop L\in\N}\lim_{h\to\infty\atop h\in
 \frac{2}{\beta}\N}(1_{\D>0}L^{\frac{d}{2}}+1_{\D=0}L^d)
\int_0^{2\pi}d\xi \int_{[0,\infty)\backslash
 [\D_L-\delta,\D_L+\delta]}dr r
e^{\beta L^{d}(f_L(r)-f_L(\D_L))}\notag\\
&\qquad\cdot\int e^{-V(\psi)+W(\psi)}d\mu_{C(r e^{i\xi})}(\psi)=0.\notag
\end{align}
Furthermore, by \eqref{eq_model_good_smallness}
\begin{align}
\int_{0}^{2\pi}d\xi
\lim_{L\to \infty\atop L\in\N}
\lim_{h\to \infty\atop h\in\frac{2}{\beta}\N}
\int e^{-V(\psi)+W(\psi)}d\mu_{C(\D e^{i\xi})}(\psi)\neq
 0.\label{eq_correction_term_limit_non_zero_integral}
\end{align}
To prove the claim \eqref{eq_ODLRO}, we first change the order of the
 integration over $\R^2$ and the limit operation $h\to \infty$ in
 \eqref{eq_grassmann_formulation_2_band},
 \eqref{eq_grassmann_formulation_2_band_correlation}, which is justified
 by the uniform bound \eqref{eq_grassmann_formulations_all_bound} and the dominated
 convergence theorem, and then apply 
\eqref{eq_4_point_formulation_decomposition},
\eqref{eq_4_point_formulation_vanishing_part},
\eqref{eq_4_point_formulation_surviving_part},
\eqref{eq_no_field_formulation_decomposition},
\eqref{eq_no_field_formulation_surviving_part},
\eqref{eq_correction_term_limit_non_zero_integral}.
As the result,
\begin{align*}
&\lim_{L\to \infty\atop L\in\N}
\frac{\Tr (e^{-\beta
 (\sH+i\theta\sS_z)}\sA_2)}{\Tr e^{-\beta (\sH+i\theta\sS_z)}}\\
&=-\left(1_{\D>0}\int_{-\infty}^{\infty}dr\D+1_{\D=0}\int_0^{\infty}dr
r\right) e^{\frac{\beta}{2}f''(\D)r^2}\int_0^{2\pi}d\xi \notag\\
&\qquad\cdot \lim_{L\to \infty\atop L\in\N}
\lim_{h\to \infty\atop h\in\frac{2}{\beta}\N}
\int e^{-V(\psi)+W(\psi)}d\mu_{C(\D e^{i\xi})}(\psi)\notag\\
&\quad\cdot \lim_{L\to \infty\atop
 L\in\N}\det\left(\begin{array}{cc} C(\D)(1\hat{\bx}0,1\hat{\by} 0) &
	 C(\D)(1\b0 0, 2 \b0 0) \\
C(\D)(2 \b0 0,1\b0 0) &
	 C(\D)(2\hat{\by} 0, 2 \hat{\bx} 0)\end{array}\right)\notag\\
&\quad\cdot\Bigg/\Bigg(
\left(1_{\D>0}\int_{-\infty}^{\infty}dr\D+1_{\D=0}\int_0^{\infty}dr
r\right)e^{\frac{\beta}{2}f''(\D)r^2}\int_0^{2\pi}d\xi \notag\\
&\qquad\qquad\cdot\lim_{L\to \infty\atop L\in\N}
\lim_{h\to \infty\atop h\in\frac{2}{\beta}\N}
\int e^{-V(\psi)+W(\psi)}d\mu_{C(\D e^{i\xi})}(\psi)\Bigg)\notag\\
&=-\lim_{L\to \infty\atop
 L\in\N}\det\left(\begin{array}{cc} C(\D)(1\hat{\bx}0,1\hat{\by} 0) &
	 C(\D)(1\b0 0, 2 \b0 0) \\
C(\D)(2 \b0 0,1\b0 0) &
	 C(\D)(2\hat{\by} 0, 2 \hat{\bx} 0)\end{array}\right),\notag
\end{align*}
or by Lemma \ref{lem_explicit_calculation_covariance} and
 \eqref{eq_gap_equation_delta},
\begin{align*}
&\lim_{\|\hat{\bx}-\hat{\by}\|_{\R^d}\to\infty}
\lim_{L\to \infty\atop L\in\N}\frac{\Tr (e^{-\beta
 (\sH+i\theta\sS_z)}\sA_2)}{\Tr e^{-\beta (\sH+i\theta\sS_z)}}=\lim_{L\to \infty\atop L\in\N}C(\D)(1\b0 0, 2 \b0 0)
C(\D)(2 \b0 0,1\b0 0)
=\frac{\D^2}{U^2}.
\end{align*}
After reaching this equality we only need to repeat the same argument as
 in the end of the proof for SSB
 to complete the proof of the claim \eqref{eq_ODLRO} and the claim
 concerning ODLRO in \eqref{item_no_superconductivity}.

It remains to prove the claim \eqref{eq_free_energy_density}. Remark that by
 \eqref{eq_model_good_smallness}
\begin{align*}
&\Re \int_0^{2\pi}d\xi 
\int_{-L^{d/2}\min\{\delta,\D_L\}}^{L^{d/2}\delta}dr
(1_{\D>0}(L^{-\frac{d}{2}}r+\D_L)+1_{\D=0}r)
e^{\beta\int_0^1dt(1-t)f_L''(tL^{-\frac{d}{2}}r+\D_L)r^2}\notag\\
&\qquad\cdot
\int
 e^{-V(\psi)+W(\psi)}d\mu_{C((L^{-\frac{d}{2}}r+\D_L)e^{i\xi})}(\psi)\\
&\ge \pi \int_{-L^{d/2}\min\{\delta,\D_L\}}^{L^{d/2}\delta}dr
(1_{\D>0}(L^{-\frac{d}{2}}r+\D_L)+1_{\D=0}r)
e^{\beta\int_0^1dt(1-t)f_L''(tL^{-\frac{d}{2}}r+\D_L)r^2}>0,\\
&\Re \int_0^{2\pi}d\xi \int_{[0,\infty)\backslash
 [\D_L-\delta,\D_L+\delta]}dr r
e^{\beta L^{d}(f_L(r)-f_L(\D_L))}\int e^{-V(\psi)+W(\psi)}d\mu_{C(r
 e^{i\xi})}(\psi)\\
&\ge \pi\int_{[0,\infty)\backslash
 [\D_L-\delta,\D_L+\delta]}dr r
e^{\beta L^{d}(f_L(r)-f_L(\D_L))}>0
\end{align*}
for sufficiently large $L,h$. The following transformation based
 on \eqref{eq_no_field_formulation_decomposition} is justified.
\begin{align*}
&-\frac{1}{\beta L^d}\log\left(
\frac{\beta L^d}{\pi |U|}\int_{\R^2}d\phi_1d\phi_2
e^{\beta L^d f_L(|\phi|)}
\int e^{-V(\psi)+W(\psi)}d\mu_{C(\phi)}(\psi)
\right)\\
&=-\frac{1}{\beta L^d}\log\left(
\frac{\beta L^d}{\pi |U|}e^{\beta
 L^df_L(\D_L)}
(1_{\D>0}L^{-\frac{d}{2}}+1_{\D=0}L^{-d}) \right)\\
&\quad 
-\frac{1}{\beta L^d}\log\Bigg(
 \int_0^{2\pi}d\xi 
\int_{-L^{d/2}\min\{\delta,\D_L\}}^{L^{d/2}\delta}dr
(1_{\D>0}(L^{-\frac{d}{2}}r+\D_L)+1_{\D=0}r)\\
&\qquad\qquad\qquad\qquad\qquad\cdot e^{\beta\int_0^1dt(1-t)f_L''(tL^{-\frac{d}{2}}r+\D_L)r^2}
\int
 e^{-V(\psi)+W(\psi)}d\mu_{C((L^{-\frac{d}{2}}r+\D_L)e^{i\xi})}(\psi)\\
&\qquad\qquad\qquad+(1_{\D>0}L^{\frac{d}{2}}+1_{\D=0}L^d) \int_0^{2\pi}d\xi 
\int_{[0,\infty)\backslash [\D_L-\delta,\D_L+\delta]}dr r
e^{\beta L^{d}(f_L(r)-f_L(\D_L))}\\
&\qquad\qquad\qquad\qquad\qquad\cdot\int e^{-V(\psi)+W(\psi)}d\mu_{C(r e^{i\xi})}(\psi)
\Bigg).
\end{align*}
Then, by \eqref{eq_grassmann_formulation_2_band}, 
\eqref{eq_no_field_formulation_surviving_part}
and a variant of \eqref{eq_uniform_convergence_effective_potential},
\eqref{eq_maximum_point_convergence_1D},
\begin{align*}
&\lim_{L\to \infty\atop L\in \N}\left(-\frac{1}{\beta L^d}\log\left(
\frac{\Tr e^{-\beta (\sH+i\theta\sS_z)}}{\Tr e^{-\beta
 (\sH_0+i\theta\sS_z)}}\right)\right)\\
&=\lim_{L\to \infty\atop L\in \N}
\lim_{h\to \infty\atop h\in \frac{2}{\beta}\N}
\left(-\frac{1}{\beta L^d}\log\left(
\frac{\beta L^d}{\pi |U|}\int_{\R^2}d\phi_1d\phi_2
e^{\beta L^d f_L(|\phi|)}
\int e^{-V(\psi)+W(\psi)}d\mu_{C(\phi)}(\psi)
\right)\right)\\
&=-\lim_{L\to\infty\atop L\in \N}f_L(\D_L)=-f(\D).
\end{align*}
By combining this with \eqref{eq_free_partition_function} we finally
 obtain the equality \eqref{eq_free_energy_density}. 
\end{proof}

\appendix
\section{Proof of Proposition \ref{prop_P_S_bound}}\label{app_P_S_bound}

Here we provide a short proof of Proposition \ref{prop_P_S_bound} for 
readers' convenience. We should remark that the proof below is
essentially a digest of the general construction of \cite{PS}.
First let us recall a simple fact based on the Cauchy-Binet formula.
\begin{lemma}\label{lem_application_cauchy_binet}
Assume that $n\times n$ matrices $A=(A(i,j))_{1\le i,j\le n}$,
 $B=(B(i,j))_{1\le i,j\le n}$ satisfy that
$$
|\det(A(k_i,l_j))_{1\le i,j\le m}|\le D_A^{2m},\quad 
|\det(B(k_i,l_j))_{1\le i,j\le m}|\le D_B^{2m}
$$
with $D_A$, $D_B\in \R_{\ge 0}$ for any $\{k_i\}_{i=1}^m$,
 $\{l_i\}_{i=1}^m\subset\{1,2,\cdots,n\}$ satisfying $k_1<\cdots<k_m$,
 $l_1<\cdots<l_m$. Then,
$$
|\det(A+B)|\le (D_A+D_B)^{2n}.
$$
\end{lemma}
\begin{proof}
By applying the Cauchy-Binet formula to the decomposition
$$
A+B=(\begin{array}{cc}A &
	  I_n\end{array})\left(\begin{array}{c}I_n \\ B \end{array}\right)
$$
we observe that 
\begin{align*}
&|\det(A+B)|\\
&\le \sum_{\g:\{1,\cdots,n\}\to \{1,\cdots,2n\}\atop
 \text{with }\g(1)<\cdots<\g(n)}|\det((\begin{array}{cc}A &
	  I_n\end{array})(i,\g(j)))_{1\le i,j\le n}|\left|
\det\left(\left(\begin{array}{c}I_n \\ B
		\end{array}\right)(\g(i),j)\right)_{1\le i,j\le
 n}\right|\\
&\le\sum_{m=0}^n\sum_{\g:\{1,\cdots,n\}\to \{1,\cdots,2n\}\atop
 \text{with }\g(1)<\cdots<\g(n)}1_{\g(m)\le
 n<\g(m+1)}D_A^{2m}D_B^{2(n-m)}=\sum_{m=0}^n\left(\begin{array}{c}n \\ m\end{array}\right)^2
D_A^{2m}D_B^{2(n-m)}\\
&\le \sum_{m=0}^n\left(\begin{array}{c}2n \\ 2m\end{array}\right)
D_A^{2m}D_B^{2(n-m)}\le (D_A+D_B)^{2n},
\end{align*}
where we set $\g(0):=0$, $\g(n+1):=n+1$.
\end{proof}

\begin{proof}[Proof of Proposition \ref{prop_P_S_bound}]
Take any $m,n\in \N$, $\bu_i,\bv_i\in\C^m$ with $\|\bu_i\|_{\C^m}$,
$\|\bv_i\|_{\C^m}\le 1$, $(\rho_i,\bx_i,s_i),(\eta_i,\by_i,t_i)\in
 \{1,2\}\times \G\times [0,\beta)$ $(i=1,2,\cdots,n)$, $j\in
 \{1,2\}$. Define the $n\times n$ matrix $M=(M_{k,l})_{1\le k,l\le n}$ by
\begin{align*}
&M_{k,l}:=\<\bu_k,\bv_l\>_{\C^m}1_{s_k\ge
 t_l}\<f_j^{\ge}(\rho_k\bx_ks_k), g_j^{\ge}(\eta_l\by_l
 t_l)\>_{\cH},\quad (l,k=1,2,\cdots,n).
\end{align*}
Let us prove that $|\det M|\le D^{2n}$. By permutating rows and columns
 if necessary we may assume that $s_1\ge \cdots \ge s_n$, $t_1\ge \cdots
 \ge t_n$. By the assumption of continuity the function 
\begin{align*}
&(\eps_1,\cdots,\eps_n,\delta_1,\cdots,\delta_n)\mapsto \\
&|\det(\<\bu_k,\bv_l\>_{\C^m}1_{s_k+\eps_k\ge t_l-\delta_l}
\<f_j^{\ge}(\rho_k\bx_k(s_k+\eps_k)),g_j^{\ge}(\eta_l\by_l(t_l-\delta_l))\>_{\cH})_{1\le
 k,l\le n}|\\
&:\R^{2n}_{\ge 0}\to\R
\end{align*}
is continuous at $\b0$. Thus, we can choose real sequences
 $(s_k^p)_{p=1}^{\infty}$, $(t_k^p)_{p=1}^{\infty}$ $(k=1,2,\cdots,n)$
such that
$s_1^{p}>\cdots> s_n^p$, $t_1^{p}>\cdots> t_n^p$,
$\{s_k^p\}_{k=1}^n\cap \{t_k^p\}_{k=1}^n=\emptyset$ for any $p\in\N$ and 
\begin{align*}
\lim_{p\to \infty}|\det(\<\bu_k,\bv_l\>_{\C^m}1_{s_k^p\ge t_l^p}
\<f_j^{\ge}(\rho_k\bx_k s_k^p),g_j^{\ge}(\eta_l\by_l t_l^p)\>_{\cH})_{1\le
 k,l\le n}|=|\det M|.
\end{align*}
Thus, by keeping in mind that we perform the limit operation in the end
 we may also assume that 
$s_1>\cdots> s_n$, $t_1>\cdots> t_n$,
$\{s_k\}_{k=1}^n\cap \{t_k\}_{k=1}^n=\emptyset$.

Define the vectors $f(s_k)$, $g(t_k)$ $(k=1,2,\cdots,n)$ of
 $\C^m\otimes\cH$ by 
$f(s_k):=\bu_k\otimes f_j^{\ge}(\rho_k\bx_k s_k)$, 
$g(t_k):=\bv_k\otimes g_j^{\ge}(\eta_k\by_k t_k)$. Let $\hat{\cH}$ be the
 finite-dimensional subspace of $\C^m\otimes \cH$ spanned by $f(s_k)$, $g(t_k)$ $(k=1,2,\cdots,n)$. For $f\in \hat{\cH}$ let $a(f)$ $(a(f)^*)$ be
 the annihilation (creation) operator on the Fermionic Fock space
 $F_f(\hat{\cH})$. It is well-known (see e.g. \cite[\mbox{Subsection 5.2.1}]{BR})
 that 
\begin{align}
&\{a(f),a(g)^*\}=\<f,g\>_{\hat{\cH}},\label{eq_CAR}\\
&\|a(f)\|_{\cB(F_f(\hat{\cH}))}=\|a(f)^*\|_{\cB(F_f(\hat{\cH}))}=\|f\|_{\hat{\cH}},\quad
 (\forall f,g\in \hat{\cH}),\label{eq_fermionic_operator_norm}
\end{align}
where $\|\cdot\|_{\cB(F_f(\hat{\cH}))}$ is the operator norm for operators on 
$F_f(\hat{\cH})$. 
For $(b,\xi)\in (\{s_k\}_{k=1}^n\times
 \{-1\})\cup (\{t_k\}_{k=1}^n\times \{1\})$ we set $a_{(b,\xi)}:=a(f(b))$
 if $\xi=-1$, $a(g(b))^*$ if $\xi=1$. Let $l\in \{1,\cdots,2n\}$. For any
 distinct $(b_1,\xi_1),\cdots,(b_l,\xi_l)\in (\{s_k\}_{k=1}^n\times \{-1\})\cup (\{t_k\}_{k=1}^n\times \{1\})$
there uniquely exists $\s\in \S_l$ such that
$b_{\s(1)}> b_{\s(2)}>\cdots >b_{\s(l)}$.  Then, we set
$$
\bT(a_{(b_1,\xi_1)}\cdots a_{(b_l,\xi_l)}):=\sgn(\s)a_{(b_{\s(1)},\xi_{\s(1)})}\cdots a_{(b_{\s(l)},\xi_{\s(l)})}.
$$

Let us prove that 
\begin{align}
\det M=(-1)^{\frac{n(n-1)}{2}}\<\hat{\O},\bT(a(f(s_1))\cdots a(f(s_n))a(g(t_1))^*\cdots a(g(t_n))^*)\hat{\O}\>_{F_f(\hat{\cH})}\label{eq_target_equality_P_S_proof}
\end{align}
by induction with $n$, where $\hat{\O}$ denotes the vacuum and
 $\<\cdot,\cdot\>_{F_f(\hat{\cH})}$ is the inner product of $F_f(\hat{\cH})$. It 
 clearly holds for $n=1$. Let us assume that it holds for $n-1$ with
 $n\ge 2$. If
 $t_1>s_1$,
\begin{align*}
&(\text{R.H.S of }\eqref{eq_target_equality_P_S_proof})\\
&=(-1)^{\frac{n(n-1)}{2}+n}
\<\hat{\O},a(g(t_1))^*\bT(a(f(s_1))\cdots a(f(s_n))a(g(t_2))^*\cdots
 a(g(t_n))^*)\hat{\O}\>_{F_f(\hat{\cH})}\\
&=0=\det M.
\end{align*}
Consider the case that $t_1\le s_1$. Then, there exists $k\in
 \{1,2,\cdots,n\}$ such that $(k\le n-1)\land (s_k>t_1>s_{k+1})$ or
 $(k=n)\land (s_k>t_1)$. Then, by \eqref{eq_CAR} and the induction
 hypothesis,
\begin{align*}
&(\text{R.H.S of }\eqref{eq_target_equality_P_S_proof})\\
&=(-1)^{\frac{n(n-1)}{2}+n-k}
\<\hat{\O},a(f(s_1))\cdots a(f(s_k))
a(g(t_1))^*\\
&\qquad\qquad\qquad\qquad\quad\cdot \bT(a(f(s_{k+1}))\cdots a(f(s_n))a(g(t_2))^*\cdots
 a(g(t_n))^*)\hat{\O}\>_{F_f(\hat{\cH})}\\
&=\sum_{l=1}^k(-1)^{\frac{n(n-1)}{2}+n+l}\<f(s_l),g(t_1)\>_{\hat{\cH}}\\
&\quad\cdot \<\hat{\O},a(f(s_1))\cdots a(f(s_{l-1}))a(f(s_{l+1}))\cdots a(f(s_k))\\
&\qquad\quad\cdot \bT(a(f(s_{k+1}))\cdots a(f(s_n))a(g(t_2))^*\cdots
 a(g(t_n))^*)\hat{\O}\>_{F_f(\hat{\cH})}\\
&=\sum_{l=1}^k(-1)^{\frac{n(n-1)}{2}+n+l}\<f(s_l),g(t_1)\>_{\hat{\cH}}\\
&\quad\cdot\<\hat{\O},\bT(a(f(s_1))\cdots a(f(s_{l-1}))a(f(s_{l+1}))\cdots a(f(s_n))
a(g(t_2))^*\cdots
 a(g(t_n))^*)\hat{\O}\>_{F_f(\hat{\cH})}\\
&=\sum_{l=1}^k(-1)^{l+1}M_{l,1}\det(M_{p,q})_{1\le p\le n,p\neq l\atop 2\le
 q\le n}
=\det M.
\end{align*}
Here we used that
$$
(-1)^{\frac{n(n-1)}{2}+n+l+\frac{(n-1)(n-2)}{2}}=(-1)^{l+1}.
$$
Thus, by induction \eqref{eq_target_equality_P_S_proof} holds for any
 $n\in \N$.

Then, by using \eqref{eq_fermionic_operator_norm} we can derive from
 \eqref{eq_target_equality_P_S_proof} that
\begin{align*}
&|\det M|\le \prod_{k=1}^n\|f_j^{\ge}(\rho_k\bx_k
 s_k)\|_{\cH}\|g_j^{\ge}(\eta_k\by_k t_k)\|_{\cH}\le D^{2n}.
\end{align*}
A parallel argument shows that 
$$
|\det(\<\bu_k,\bv_l\>_{\C^m}1_{s_k< t_l}
\<f_j^{<}(\rho_k\bx_k s_k),g_j^{<}(\eta_l\by_l t_l)\>_{\cH})_{1\le
 k,l\le n}|\le D^{2n}.
$$
In fact in this case we may assume that $s_1<\cdots<s_n$,
 $t_1<\cdots<t_n$, $\{s_k\}_{k=1}^n\cap \{t_k\}_{k=1}^n=\emptyset$ by
 the continuity argument. Then we only need to define $\bT$ to arrange
 in the opposite order. Now coming back to the decomposition
 \eqref{eq_P_S_representation}, we can repeatedly apply Lemma
 \ref{lem_application_cauchy_binet} to derive the claimed inequality.
\end{proof} 

\section*{Acknowledgments}
This work was supported by JSPS KAKENHI Grant Number 
26870110.

\section*{Notation}
\subsection*{Parameters and constants}
\begin{center}
\begin{tabular}{lll}
Notation & Description & Reference \\
\hline
$d$ & spatial dimension & Subsection \ref{subsec_main_results}\\ 
$L$ & size of the spatial lattice & Subsection
 \ref{subsec_main_results}\\
$hop$ & 0 or 1, parameter to determine sign of hopping & Subsection
 \ref{subsec_main_results}\\  
$\mu$  & chemical potential & Subsection \ref{subsec_main_results}\\
$U$  &  negative coupling constant  & Subsection \ref{subsec_main_results}\\ 
$\g$  & magnitude of symmetry breaking external field & Subsection \ref{subsec_main_results}\\
$\beta$ & inverse temperature  & Subsection \ref{subsec_main_results} \\ 
$\theta$  & magnitude of imaginary magnetic field & Subsection \ref{subsec_main_results}\\
$\la_1,\la_2$ & artificial parameters & Subsection \ref{subsec_one_band}\\
$h$ & inverse step size of time-discretization & Subsection \ref{subsec_one_band}\\ 
$N$ & $4\beta h L^d$, cardinality of $I$ & Subsection
 \ref{subsec_one_band}\\
$\Theta$ & $|\theta/2-\pi/\beta|$ & beginning of Section
 \ref{sec_proof_theorem}
\end{tabular}
\end{center}

\subsection*{Sets and spaces}
\begin{center}
\begin{tabular}{lll}
Notation & Description & Reference \\
\hline
$\G$ & $\{0,1,\cdots,L-1\}^d$ &  Subsection \ref{subsec_main_results}\\
$\G^*$ & $\{0,\frac{2\pi}{L},\frac{2\pi}{L}\cdot 2,\cdots, \frac{2\pi}{L}(L-1)\}^d$ &  Subsection \ref{subsec_one_band}\\
$[0,\beta)_h$ & $\{0,\frac{1}{h},\frac{2}{h},\cdots,\beta-\frac{1}{h}\}$
 & Subsection \ref{subsec_one_band}\\
$D(r)$ & $\{z\in\C\ |\ |z|<r\}$ 
 & Subsection \ref{subsec_one_band}\\
$I_0$ & $\{1,2\}\times\G\times [0,\beta)_h$ 
 & Subsection \ref{subsec_two_band}\\
$I$ & $I_0\times\{1,-1\}$ & Subsection \ref{subsec_two_band}\\
$\cV$ & complex vector space spanned by $\{\psi_X\}_{X\in I}$ &
 Subsection \ref{subsec_two_band}\\
$\bigwedge\cV$ & Grassmann algebra generated by $\{\psi_X\}_{X\in I}$ & Subsection \ref{subsec_two_band}\\
$I^0$ & $\{1,2\}\times\G\times \{0\}\times \{1,-1\}$ 
 & Subsection \ref{subsec_preliminaries}\\
$\bigwedge_{even}\cV$ & Subspace of $\bigwedge \cV$ consisting of even
 polynomials & Subsection \ref{subsec_preliminaries}\\
$\Map(A,B)$ & set of maps from $A$ to $B$ & Subsection \ref{subsec_UV_without_artificial}
\end{tabular}
\end{center}

\newpage
\subsection*{Functions and maps}
\begin{center}
\begin{tabular}{lll}
Notation & Description & Reference \\
\hline
$r_L$ & map from $\Z^d$ to $\G$ &  Subsection
 \ref{subsec_main_results}\\
$\sH_0$  & $\sum_{\bx\in \G,\s\in
 \spin}((-1)^{hop}\sum_{j=1}^d(\psi_{\bx\s}^*\psi_{\bx+\be_j\s}+\psi_{\bx\s}^*\psi_{\bx-\be_j\s})$ & Subsection \ref{subsec_main_results}\\
 & $\qquad\qquad\qquad -\mu\psi_{\bx\s}^*\psi_{\bx\s})$ & \\
$\sV$  & $\frac{U}{L^d}\sum_{\bx,\by\in
 \G}\psi_{\bx\ua}^*\psi_{\bx\da}^*\psi_{\by\da}\psi_{\by\ua}$ &
 Subsection \ref{subsec_main_results}\\
$\sH$ & $\sH_0+\sV$ & Subsection \ref{subsec_main_results}\\
$\sF$ & $\g \sum_{\bx\in
 \G}(\psi_{\bx\ua}^*\psi_{\bx\da}^*+\psi_{\bx\da}\psi_{\bx\ua})$ &
 Subsection \ref{subsec_main_results}\\
$\sS_z$ &
 $\frac{1}{2}\sum_{\bx\in\G}(\psi_{\bx\ua}^*\psi_{\bx\ua}-\psi_{\bx\da}^*\psi_{\bx\da})$
& Subsection \ref{subsec_main_results}\\
$\sA_1$ & $\psi_{\hat{\bx}\ua}^*\psi_{\hat{\bx}\da}^*$ &
 Subsection \ref{subsec_main_results}\\
$\sA_2$ & $\psi_{\hat{\bx}\ua}^*\psi_{\hat{\bx}\da}^*\psi_{\hat{\by}\da}\psi_{\hat{\by}\ua}$ &
Subsection \ref{subsec_main_results}\\
$e(\cdot)$ & $(-1)^{hop}2\sum_{j=1}^d\cos k_j-\mu$, free dispersion relation & Subsection
 \ref{subsec_main_results}\\
$g_d$ & function to control possible magnitude of &
 Equation \eqref{eq_magnitude_function}\\
 & coupling constant & \\
$C(\phi)$ & 2-band free covariance parameterized by $\phi\in\C$ &
  Equation \eqref{eq_covariance_2_band}\\
$E(\phi)$ & $(2\times 2)$-matrix-valued function parameterized &  Equation \eqref{eq_dispersion_matrix} \\
 & by $\phi\in\C$ &  \\
$Tree(S,\cC)$ & operator consisting of Grassmann left-derivatives &
 Subsection \ref{subsec_preliminaries}\\
$r_{\beta}$ & map from $\frac{1}{h}\Z$ to $[0,\beta)_h$ &
 Subsection \ref{subsec_general_estimation}\\
$\cR_{\beta}$ & map from $(\{1,2\}\times \G\times \frac{1}{h}\Z\times
 \{1,-1\})^n$ to $I^n$ & Subsection \ref{subsec_general_estimation}\\
   & or from $(\{1,2\}\times \G\times \frac{1}{h}\Z)^n$ to $I_0^n$ & 
\end{tabular}
\end{center}

\subsection*{Norms and semi-norms}
\begin{center}
\begin{tabular}{lll}
Notation & Description & Reference \\
\hline
$\|\cdot\|_{1,\infty}$ & integrating with all but one fixed variable &  Subsection
 \ref{subsec_preliminaries}\\
$\|\cdot\|_{1}$ & integrating with all variables &  Subsection
 \ref{subsec_preliminaries}\\
$\|\cdot\|_{1,\infty}'$ & norm defined on anti-symmetric function on $I^2$ &  Subsection
 \ref{subsec_preliminaries}\\
$\|\cdot\|$ & $\|\cdot\|_{1,\infty}'+\beta^{-1}\|\cdot\|_{1,\infty}$  &
 Subsection \ref{subsec_preliminaries}\\
$[\cdot,\cdot]_{1,\infty}$ & measurement of function on $I^m\times I^n$
 & Subsection
 \ref{subsec_preliminaries}\\
  & coupled with a function on $I^2$ &  \\
$[\cdot,\cdot]_{1}$ & measurement of function on $I^m\times I^n$ & Subsection
 \ref{subsec_preliminaries}\\
  & coupled with a function on $I^2$ & \\
$\|\cdot\|_{1,\infty,r}$ &
 $\sup_{u\in\overline{D(r)}}\|f(u)\|_{1,\infty}$ & Subsection
 \ref{subsec_UV_without_artificial}\\
$\|\cdot\|_{1,r,r'}$ &
 $\sup_{u\in\overline{D(r)},\bla\in\overline{D(r')}^2}\|f(u,\bla)\|_{1}$ & Subsection \ref{subsec_UV_with_artificial}
\end{tabular}
\end{center}

\subsection*{Other notations}
\begin{center}
\begin{tabular}{lll}
Notation & Description & Reference \\
\hline
$\be_j\ (j=1,\cdots,d)$ & standard basis of $\R^d$ &  Subsection
 \ref{subsec_main_results}\\
$V(\psi)$ & sum of quadratic and quartic polynomials &  Equation
 \eqref{eq_2_band_Grassmann_V} \\ 
 &  of $\bigwedge \cV$ &   \\
$W(\psi)$ & quartic polynomial of $\bigwedge \cV$ &  Equation \eqref{eq_2_band_Grassmann_W} \\
$A^1(\psi)$ & quadratic polynomial of $\bigwedge \cV$ &  Equation \eqref{eq_2_band_Grassmann_A_parts} \\
$A^2(\psi)$ & quartic polynomial of $\bigwedge \cV$ &  Equation \eqref{eq_2_band_Grassmann_A_parts} \\
$A(\psi)$ & $\la_1A^1(\psi)+\la_2A^2(\psi)$ &  Equation \eqref{eq_artificial_Grassmann_term} \\
$d_j(T)$ & degree of vertex $j$ in tree $T$ & Subsection
 \ref{subsec_general_estimation}
\end{tabular}
\end{center}

\end{document}